\documentclass[english,leqno]{article}

\usepackage{latexsym}
\usepackage{linearb}
\usepackage[LGR, OT1]{fontenc}
\usepackage{amsmath,amsthm}
\usepackage{amssymb}
\usepackage{MnSymbol}
\usepackage{graphicx}
\usepackage{marvosym}
\usepackage{eufrak}
\usepackage[margin=1.25in,dvips]{geometry}
\usepackage{color}
\usepackage{comment}
\usepackage{mathrsfs}\usepackage[LGR,T1]{fontenc}
\usepackage[latin9]{inputenc}
\usepackage{geometry}
\usepackage{babel}
\usepackage{enumitem}
\usepackage{amsthm}
\usepackage{amstext}
\usepackage{amssymb}
\usepackage{esint}
\usepackage{float}
\usepackage[unicode=true,pdfusetitle,
 bookmarks=true,bookmarksnumbered=false,bookmarksopen=false,
 breaklinks=false,pdfborder={0 0 1},backref=false,colorlinks=false]
 {hyperref}
\makeatletter

\DeclareRobustCommand{\greektext}{%
  \fontencoding{LGR}\selectfont\def\encodingdefault{LGR}}
\DeclareRobustCommand{\textgreek}[1]{\leavevmode{\greektext #1}}
\DeclareFontEncoding{LGR}{}{}
\DeclareTextSymbol{\~}{LGR}{126}

\numberwithin{equation}{section}
\numberwithin{figure}{section}
\usepackage{enumitem}		
\newcommand{\lyxaddress}[1]{
\par {\raggedright #1
\vspace{1.4em}
\noindent\par}
}
  \theoremstyle{remark}
  \newtheorem*{rem*}{\protect\remarkname}
 \newlist{casenv}{enumerate}{4}
 \setlist[casenv]{leftmargin=*,align=left,widest={iiii}}
 \setlist[casenv,1]{label={{\itshape\ \casename} \arabic*.},ref=\arabic*}
 \setlist[casenv,2]{label={{\itshape\ \casename} \roman*.},ref=\roman*}
 \setlist[casenv,3]{label={{\itshape\ \casename\ \alph*.}},ref=\alph*}
 \setlist[casenv,4]{label={{\itshape\ \casename} \arabic*.},ref=\arabic*}
 \theoremstyle{definition}
 \newtheorem*{defn*}{\protect\definitionname}
  \theoremstyle{definition}
  \newtheorem{defn}{\protect\definitionname}[section]
  \theoremstyle{plain}
  \newtheorem{prop}{\protect\propositionname}[section]
  \theoremstyle{plain}
  \newtheorem{lem}{\protect\lemmaname}[section]

\newtheorem{innercustomthm}{Theorem}
\newenvironment{customthm}[1]
  {\renewcommand\theinnercustomthm{#1}\innercustomthm}
  {\endinnercustomthm}

\newtheorem{innercustomprop}{Proposition}
\newenvironment{customprop}[1]
  {\renewcommand\theinnercustomprop{#1}\innercustomprop}
  {\endinnercustomprop}

\theoremstyle{definition}
\newtheorem{innercustomdef}{Definition}
\newenvironment{customdef}[1]
  {\renewcommand\theinnercustomdef{#1}\innercustomdef}
  {\endinnercustomdef}

\makeatother

\usepackage{babel}
  \providecommand{\definitionname}{Definition}
  \providecommand{\lemmaname}{Lemma}
  \providecommand{\propositionname}{Proposition}
  \providecommand{\remarkname}{Remark}
 \providecommand{\casename}{Case}

\begin{document}
\title{A proof of the instability of AdS\\ for the Einstein--null dust system with an inner mirror}

\author{Georgios Moschidis}

\maketitle

\lyxaddress{Princeton University, Department of Mathematics, Fine Hall, Washington
Road, Princeton, NJ 08544, United States, \tt gm6@math.princeton.edu}
\begin{abstract}
In 2006, Dafermos and Holzegel \cite{DafHol,DafermosTalk} formulated
the so-called AdS instability conjecture, stating that there exist
\emph{arbitrarily small} perturbations to AdS initial data which,
under evolution by the Einstein vacuum equations for $\Lambda<0$
with reflecting boundary conditions on conformal infinity $\mathcal{I}$,
lead to the formation of black holes. The numerical study of this
conjecture in the simpler setting of the spherically symmetric Einstein--scalar
field system was initiated by Bizon and Rostworowski \cite{bizon2011weakly},
followed by a vast number of numerical and heuristic works by several
authors.

In this paper, we provide the first rigorous proof of the AdS instability
conjecture in the simplest possible setting, namely for the spherically
symmetric Einstein--massless Vlasov system, in the case when the Vlasov
field is moreover supported only on radial geodesics. This system
is equivalent to the Einstein--null dust system, allowing for both
ingoing and outgoing dust. In order to overcome the break down of
this system occuring once the null dust reaches the centre $r=0$,
we place an inner mirror at $r=r_{0}>0$ and study the evolution of
this system on the exterior domain $\{r\ge r_{0}\}$. The structure
of the maximal development and the Cauchy stability properties of
general initial data in this setting are studied in our companion
paper \cite{MoschidisMaximalDevelopment}. 

The statement of the main theorem is as follows: We construct a family
of mirror radii $r_{0\text{\textgreek{e}}}>0$ and initial data $\mathcal{S}_{\text{\textgreek{e}}}$,
$\text{\textgreek{e}}\in(0,1]$, converging, as $\text{\textgreek{e}}\rightarrow0$,
to the AdS initial data $\mathcal{S}_{0}$ in a suitable norm, such
that, for any $\text{\textgreek{e}}\in(0,1]$, the maximal development
$(\mathcal{M}_{\text{\textgreek{e}}},g_{\text{\textgreek{e}}})$ of
$\mathcal{S}_{\text{\textgreek{e}}}$ contains a black hole region.
Our proof is based on purely physical space arguments and involves
the arrangement of the null dust into a large number of beams which
are successively reflected off $\{r=r_{0\text{\textgreek{e}}}\}$
and $\mathcal{I}$, in a configuration that forces the energy of a
certain beam to increase after each successive pair of reflections.
As $\text{\textgreek{e}}\rightarrow0$, the number of reflections
before a black hole is formed necessarily goes to $+\infty$. We expect
that this instability mechanism can be applied to the case of more
general matter fields.
\end{abstract}
\tableofcontents{}

\section{\label{sec:Introduction}Introduction}

Anti-de Sitter spacetime $(\mathcal{M}_{AdS}^{n+1},g_{AdS})$, $n\ge3$,
is the simplest solution of the \emph{Einstein vacuum equations} 
\begin{equation}
Ric_{\text{\textgreek{m}\textgreek{n}}}-\frac{1}{2}Rg_{\text{\textgreek{m}\textgreek{n}}}+\Lambda g_{\text{\textgreek{m}\textgreek{n}}}=0\label{eq:VacuumEinsteinEquations}
\end{equation}
with a negative cosmological constant $\Lambda$. In the standard
polar coordinate chart on $\mathcal{M}_{AdS}$, the AdS metric takes
the form 
\begin{equation}
g_{AdS}=-\big(1-\frac{2}{n(n-1)}\Lambda r^{2}\big)dt^{2}+\big(1-\frac{2}{n(n-1)}\Lambda r^{2}\big)^{-1}dr^{2}+r^{2}g_{\mathbb{S}^{n-1}},\label{eq:AdSMetricPolarCoordinates}
\end{equation}
where $g_{\mathbb{S}^{n-1}}$ is the round metric on the $n-1$ dimensional
sphere (see \cite{HawkingEllis1973}).

Despite being geodesically complete, $(\mathcal{M}_{AdS},g_{AdS})$
fails to be globally hyperbolic. In particular, it can be conformally
identified with the interior of $(\mathbb{R}\times\mathbb{S}_{+}^{n},g_{E})$,
where $\mathbb{S}_{+}^{n}$ is the closed upper hemisphere of $\mathbb{S}^{n}$
and $g_{E}$ is the metric 
\begin{equation}
g_{E}=-d\bar{t}^{2}+g_{\mathbb{S}^{n}}.\label{eq:EinsteinMetric}
\end{equation}
Through this identification, the timelike boundary 
\begin{equation}
\mathcal{I}^{n}=\mathbb{R}\times\partial\mathbb{S}_{+}^{n}\simeq\mathbb{R}\times\mathbb{S}^{n-1}\label{eq:ConformalBoundaryAtInfinity}
\end{equation}
 of $(\mathbb{R}\times\mathbb{S}_{+}^{n},g_{E})$ is naturally attached
to $(\mathcal{M}_{AdS},g_{AdS})$ as a ``conformal boundary at infinity''
(see \cite{HawkingEllis1973}).

In 1998, Maldacena, Gubser--Klebanov--Polyakov and Witten \cite{Maldacena,gubser1998gauge,witten1998anti}
proposed the \emph{AdS/CFT conjecture}, suggesting a correspondence
between certain conformal field theories defined on $\mathcal{I}^{n}$
(in the strongly coupled regime) and supergravity on spacetimes asymptotically
of the form $(\mathcal{M}_{AdS}^{n+1}\times S^{k},g_{AdS}+g_{S^{k}})$,
where $(S^{k},g_{S^{k}})$ is a suitable compact Riemannian manifold
of dimension $k$. Following the introduction of this conjecture,
asymptotically AdS spacetimes (i.\,e.~spacetimes $(\mathcal{M},g)$
with an asymptotic region with geometry resembling that of $(\mathcal{M}_{AdS},g_{AdS})$
in the vicinity of $\mathcal{I}$) became a subject of intense study
in the high energy physics literature (see e.\,g.~\cite{AGMOO2000,Hartnoll2009,AmmonErdmenger}
and references therein).

The correct setting for the study of the dynamics of asymptotically
AdS solutions $(\mathcal{M},g)$ to (\ref{eq:VacuumEinsteinEquations})
is that of an \emph{initial value problem} with appropriate \emph{boundary
conditions} prescribed asymptotically on $\mathcal{I}$. The issue
of the right boundary conditions on $\mathcal{I}$ leading to well
posedness for the resulting initial-boundary value problem for (\ref{eq:VacuumEinsteinEquations})
was first addressed by Friedrich in \cite{Friedrich1995}. Well posedness
for more general boundary conditions and matter fields in the spherically
symmetric case was obtained in \cite{HolzSmul2012,HolzWarn2015} (see
also \cite{HolzegelLukSmuleviciWarnick,Friedrich2014}). In general,
most physically interesting boundary conditions on $\mathcal{I}$
leading to a well posed initial-boundary value problem can be classified
as either \emph{reflecting} (for which an appropriate ``energy flux''
for $g$ through $\mathcal{I}$ vanishes) or \emph{dissipative }(allowing
for a non-vanishing outgoing ``energy flux'' for $g$ through $\mathcal{I}$),
with substantially different global dynamics associated to each case;
see the discussion in \cite{HolzegelLukSmuleviciWarnick}. 

In 2006, Dafermos and Holzegel \cite{DafHol,DafermosTalk} suggested
the following conjecture:

\medskip{}

\noindent \textbf{AdS instability conjecture.} \emph{There exist
arbitrarily small perturbations to the initial data of $(\mathcal{M}_{AdS},g_{AdS})$
for the vacuum Einstein equations (\ref{eq:VacuumEinsteinEquations})
with a reflecting boundary condition on $\mathcal{I}$ which lead
to the development of trapped surfaces and, thus, black hole regions.
In particular, $(\mathcal{M}_{AdS},g_{AdS})$ is non-linearly unstable. }

\medskip{}

\noindent This conjecture was motivated in \cite{DafHol} by the study
of asymptotically AdS solutions to (\ref{eq:VacuumEinsteinEquations})
with biaxial Bianchi IX symmetry in $4+1$ dimensions, a symmetry
class in which the vacuum Einstein equations \emph{(\ref{eq:VacuumEinsteinEquations})}
reduce to a $1+1$ hyperbolic system with non-trivial dynamics. This
model was introduced in \cite{BizonChmajSchmidt2005}. In this setting,
it was observed in \cite{DafHol} that perturbations of the initial
data of $(\mathcal{M}_{AdS},g_{AdS})$ (which, if not trivial, necessarily
have strictly positive ADM mass $M_{ADM}$, in view of \cite{GibbonsEtAl1983})
can not settle down to a horizonless static spacetime, since $M_{ADM}$
is conserved along $\mathcal{I}$ under reflecting boundary conditions
and no static asymptotically AdS solution of (\ref{eq:VacuumEinsteinEquations})
with $M_{ADM}>0$ exists (according to \cite{BoucherGibbonsHorowitz}).
This picture was supported by results of Anderson \cite{AndersonAdS}.

The following remarks should be made regarding the statement of the
AdS instability conjecture: 

\begin{itemize}

\item{ The perturbations referred to in the conjecture are assumed
to be small with respect to a norm for which (\ref{eq:VacuumEinsteinEquations})
is well-posed and $(\mathcal{M}_{AdS},g_{AdS})$ is \emph{Cauchy stable}
as a solution to (\ref{eq:VacuumEinsteinEquations}) (otherwise, the
conjecture is trivial).%
\footnote{\noindent Here, Cauchy stability of $(\mathcal{M}_{AdS},g_{AdS})$
refers to Cauchy stability of the conformal compactification of $(\mathcal{M}_{AdS},g_{AdS})$
(including, therefore, the timelike boundary $\mathcal{I}$); see
the discussion in the next section.%
} For such perturbations, Cauchy stability implies that the ``time''
elapsed before the formation of a trapped surface tends to $+\infty$
as the size of the initial perturbation shrinks to $0$.}

\item{ The AdS instability conjecture stands in contrast to the non-linear
stability of Minkowski space $(\mathbb{R}^{3+1},\text{\textgreek{h}})$,
in the case $\Lambda=0$ (see Christodoulou--Klainerman \cite{Christodoulou1993}),
or de~Sitter space $(\mathcal{M}_{dS},g_{dS})$, in the case $\Lambda>0$
(see Friedrich \cite{Friedrich1986}). The proof of the non-linear
stability of $(\mathbb{R}^{3+1},\text{\textgreek{h}})$ and $(\mathcal{M}_{dS},g_{dS})$
is based on a stability mechanism related to the fact that linear
fields on those spacetimes satisfy sufficiently strong decay rates.
The decay rates are, however, borderline in the case $\Lambda=0$,
and thus the stability of $(\mathbb{R}^{3+1},\text{\textgreek{h}})$
is a deep fact depending on the precise non-linear structure of the
system (\ref{eq:VacuumEinsteinEquations}), whereas, in the case $\Lambda>0$,
the decay is exponential and stability can be inferred relatively
easily. In contrast, on $(\mathcal{M}_{AdS},g_{AdS})$, it can be
shown that linear fields satisfying a reflecting boundary condition
on $\mathcal{I}$ remain bounded, but do \emph{not} decay in time.
It is precisely the lack of a sufficently fast decay rate at the linear
level which is associated to the possibility of non-linear instability.}

\item{ The prescription of a reflecting boundary condition on $\mathcal{I}$
is essential for the conjecture: For maximally dissipative boundary
conditions, it is expected that $(\mathcal{M}_{AdS},g_{AdS})$ is
non-linearly \emph{stable}, in view of the quantitative decay rates
obtained for the linearised vacuum Einstein equations (and other linear
fields) around $(\mathcal{M}_{AdS},g_{AdS})$ by Holzegel--Luk--Smulevici--Warnick
in \cite{HolzegelLukSmuleviciWarnick}. }

\item{ In the biaxial Bianchi IX symmetry class, all perturbations
of $(\mathcal{M}_{AdS},g_{AdS})$ leading to the formation of a trapped
surface can be shown to possess a complete conformal infinity $\mathcal{I}$
and are expected to settle down to a member of the Schwarzschild--AdS
family (see \cite{DafHol,DafHolStability2006}). However, in the absence
of any symmetry, the picture regarding the end state of the evolution
of general vacuum perturbations of $(\mathcal{M}_{AdS},g_{AdS})$
is complicated; see the discussion in the next section.}

\end{itemize}

Starting from the pioneering work \cite{bizon2011weakly} of Bizon
and Rostworoski in 2011, a plethora of numerical and heuristic results
have been obtained in the direction of establishing the AdS instability
conjecture, mainly in the context of the spherically symmetric Einstein--scalar
field system. See the discussion in Section \ref{sub:Numerics}.

In this paper, we will prove the AdS instability conjecture in the
simplest possible setting, namely for the Einstein--massless Vlasov
system in spherical symmetry, further reduced to the case when the
Vlasov field $f$ is supported only on radial geodesics. We will call
this system the \emph{spherically symmetric} \emph{Einstein--radial
massless Vlasov system}. In fact, this is a singular reduction; the
resulting system is equivalent to the \emph{spherically symmetric
Einstein--null dust system}, allowing for both ingoing and outgoing
dust. This system has been studied in the $\Lambda=0$ case by Poisson
and Israel \cite{PoissonIsrael1990}.

A serious problem with the spherically symmetric Einstein--null dust
system is that it suffers from a severe break down when the null dust
reaches the centre $r=0$. In particular, in any reasonable initial
data topology, the spherically symmetric Einstein--null dust system
is \underline{not} well posed and $(\mathcal{M}_{AdS},g_{AdS})$
is \underline{not} a Cauchy stable solution of it. One way to restore
the well posedness of this system (a necessary step for the study
of the AdS instability conjecture in this setting) is to place an\emph{
inner mirror} at some radius sphere $\{r=r_{0}\}$ with $r_{0}>0$
and study the evolution of the system in the exterior region $\{r\ge r_{0}\}$.
However, fixing the mirror radius $r_{0}$ results in a trivial global
stability statement for $(\mathcal{M}_{AdS},g_{AdS})$, as initial
data perturbations with total ADM mass $\tilde{m}_{ADM}<\frac{1}{2}r_{0}$
cannot form a black hole. Thus, it is necessary to allow the radius
$r_{0}$ to shrink to $0$ as the total ADM mass of the initial data
shrinks to $0$, in order to address the AdS instability conjecture
in this setting. See the discussion in Section \ref{sub:NeedOfAMirror}. 

A non-technical statement of our result is the following: 

\begin{customthm}{1}[rough version]\label{Theorem_Intro_Rough} The
AdS spacetime $(\mathcal{M}_{AdS}^{3+1},g_{AdS})$ is non-linearly
unstable under evolution by the spherically symmetric Einstein--radial
massless Vlasov system with a reflecting boundary condition on $\mathcal{I}$
and an inner mirror, in the following sense: 

There exists a one parameter family of spherically symmetric initial
data $\mathcal{S}_{\text{\textgreek{e}}}$, $\text{\textgreek{e}}\in(0,1]$
and a family of inner mirror radii $r=r_{0\text{\textgreek{e}}}$
(with $r_{0\text{\textgreek{e}}}\xrightarrow{\text{\textgreek{e}}\rightarrow0}0$)
satisfying the following properties:

\begin{enumerate}

\item As $\text{\textgreek{e}}\rightarrow0$, $||\mathcal{S}_{\text{\textgreek{e}}}||_{\mathcal{CS}}\rightarrow0$,
i.\,e.~$\mathcal{S}_{\text{\textgreek{e}}}$ converge to the initial
data $\mathcal{S}_{0}$ of $(\mathcal{M}_{AdS},g_{AdS})$. 

\item For any $\text{\textgreek{e}}>0$, the maximal future development
$(\mathcal{M}_{\text{\textgreek{e}}},g_{\text{\textgreek{e}}})$ of
$\mathcal{S}_{\text{\textgreek{e}}}$ contains a trapped surface and,
thus, a black hole region. Moreover, $(\mathcal{M}_{\text{\textgreek{e}}},g_{\text{\textgreek{e}}})$
possesses a complete conformal infinity $\mathcal{I}$.

\end{enumerate}

The norm $||\cdot||_{\mathcal{CS}}$ in 1 measures the concentation
of the energy of $\mathcal{S}_{\text{\textgreek{e}}}$ in annuli of
width $\sim r_{0\text{\textgreek{e}}}$ and has the property that
the radial Einstein--massless Vlasov system is well-posed and $(\mathcal{M}_{AdS},g_{AdS})$
is Cauchy stable with respect to $||\cdot||_{\mathcal{CS}}$ \underline{independently of the precise value of $r_{0\text{\textgreek{e}}}$}.

\end{customthm}

For progressively more detailed statements of Theorem \ref{Theorem_Intro_Rough},
see Sections \ref{sub:The-main-result:Intro} and \ref{sec:The-main-result:Details}.
For further discussion on the need of an inner mirror at $r\sim r_{0\text{\textgreek{e}}}$
and its relation to natural dispersive mechanisms appearing in other
matter models, see Section \ref{sub:NeedOfAMirror}.

We should also remark the following:
\begin{itemize}
\item Except for the condition $r_{0\text{\textgreek{e}}}<2(M_{ADM})_{\text{\textgreek{e}}}$
referred to earlier, where $(M_{ADM})_{\text{\textgreek{e}}}$ is
the ADM mass of $\mathcal{S}_{\text{\textgreek{e}}}$, there is considerable
flexibility in the choice of the mirror radii $r_{0\text{\textgreek{e}}}$
in the statement of Theorem \ref{Theorem_Intro_Rough} and this can
be exploited to one's advantage. For simplicity, we choose $r_{0\text{\textgreek{e}}}$
to satisfy $r_{0\text{\textgreek{e}}}\sim(M_{ADM})_{\text{\textgreek{e}}}$
(see also the discussion in Section \ref{sub:The-main-result:Intro}).
\item While we do not address the issue of the end state of the evolution
of $\mathcal{S}_{\text{\textgreek{e}}}$, it can be easily inferred
from our proof of Theorem \ref{Theorem_Intro_Rough} that the spacetimes
$(\mathcal{M}_{\text{\textgreek{e}}},g_{\text{\textgreek{e}}})$ settle
down to a member of the Schwarzschild--AdS family (see also \cite{MoschidisMaximalDevelopment}).
\end{itemize}
\begin{figure}[h] 
\centering 
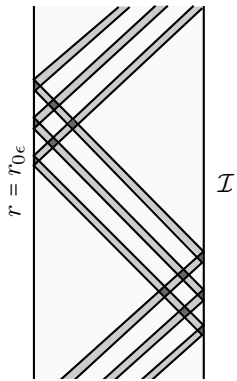 
\caption{The  family of initial data $\mathcal{S}_{\epsilon}$ that we construct for the proof of Theorem \ref{Theorem_Intro_Rough}  give rise to a large number of Vlasov beams, that are successively reflected off $\mathcal{I}$  and the inner mirror at $r=r_{0\epsilon}$. The Cauchy stability statement for $||\cdot ||_{CS}$ implies that the number of reflections necessarily goes to $+\infty $ as $\epsilon \rightarrow 0$ (cf.~the remark after the statement of the AdS instability conjecture). }
\end{figure}

The trivial instability at $r=0$ occuring for the spherically symmetric
Einstein--null dust system is absent in the case of smooth solutions
to the general spherically symmetric Einstein--massless Vlasov system
(not reduced to the radial case). In particular, the smooth initial
value problem for the spherically symmetric Einstein--massless Vlasov
system is well-posed, and placing an inner mirror at $r=r_{0}>0$
is not necessary.%
\footnote{In fact, well-posedness for the smooth initial value problem for the
Einstein--Vlasov system also holds outside spherical symmetry, see
\cite{ChoquetBruhat1971}. In the case $\Lambda=0$, the stability
of Minkowski spacetime for the Einstein--massless Vlasov system without
any symmetry assumptions was recently established by Taylor \cite{Taylor2017}.%
} For a proof of the AdS instability in this setting, see our forthcoming
\cite{MoschidisVlasov}.

\subsection{\label{sub:Numerics}Earlier numerical and heuristic works}

Restricted under spherical symmetry, all solutions to the Einstein
vacuum equations (\ref{eq:VacuumEinsteinEquations}) are locally isometric
to a member of the Schwarzschild--AdS family (see \cite{Eiesland1925}).
\textgreek{T}hus, any attempt to search for unstable vacuum perturbations
of $(\mathcal{M}_{AdS},g_{AdS})$ for (\ref{eq:VacuumEinsteinEquations})
in $3+1$ dimensions can not be reduced to a problem for a $1+1$
hyperbolic system (where the wide variety of available tools would
make the problem more tractable).%
\footnote{This problem is circumvented in 4+1 dimensions by the biaxial Bianchi
IX symmetry class referred to earlier (see \cite{BizonChmajSchmidt2005}).%
} For this reason, instead of (\ref{eq:VacuumEinsteinEquations}),
numerical and heuristic works on the AdS instability have so far mainly
focused on the \emph{Einstein--scalar field system }
\begin{equation}
\begin{cases}
Ric_{\text{\textgreek{m}\textgreek{n}}}-\frac{1}{2}Rg_{\text{\textgreek{m}\textgreek{n}}}+\Lambda g_{\text{\textgreek{m}\textgreek{n}}}=8\pi T_{\text{\textgreek{m}\textgreek{n}}}[\text{\textgreek{f}}],\\
\square_{g}\text{\textgreek{f}}=0,\\
T_{\text{\textgreek{m}\textgreek{n}}}[\text{\textgreek{f}}]\doteq\partial_{\text{\textgreek{m}}}\text{\textgreek{f}}\partial_{\text{\textgreek{n}}}\text{\textgreek{f}}-\frac{1}{2}g_{\text{\textgreek{m}\textgreek{n}}}\partial^{\text{\textgreek{a}}}\text{\textgreek{f}}\partial_{\text{\textgreek{a}}}\text{\textgreek{f}}.
\end{cases}\label{eq:EinsteinScalarField}
\end{equation}
The system (\ref{eq:EinsteinScalarField}), whose mathematical study
in the case $\Lambda=0$ was pioneered by Christodoulou \cite{Christodoulou1999},
admits non-trivial dynamics in spherical symmetry and spherically
symmetric solutions to (\ref{eq:EinsteinScalarField}) share many
qualitative properties with general solutions of (\ref{eq:VacuumEinsteinEquations}).
Reduced under spherical symmetry in a double null gauge $(u,v)$ in
$3+1$ dimensions, i.\,e.~a gauge where 
\begin{equation}
g=-\text{\textgreek{W}}^{2}dudv+r^{2}g_{\mathbb{S}^{2}},
\end{equation}
the system (\ref{eq:EinsteinScalarField}) takes the form 
\begin{equation}
\begin{cases}
\partial_{u}\partial_{v}(r^{2}) & =-\frac{1}{2}(1-\Lambda r^{2})\text{\textgreek{W}}^{2},\\
\partial_{u}\partial_{v}\log(\text{\textgreek{W}}^{2}) & =\frac{\text{\textgreek{W}}^{2}}{2r^{2}}\big(1+4\text{\textgreek{W}}^{-2}\partial_{u}r\partial_{v}r\big)-8\pi\partial_{u}\text{\textgreek{f}}\partial_{v}\text{\textgreek{f}},\\
\partial_{v}(\text{\textgreek{W}}^{-2}\partial_{v}r) & =-4\pi r\text{\textgreek{W}}^{-2}(\partial_{v}\text{\textgreek{f}})^{2},\\
\partial_{u}(\text{\textgreek{W}}^{-2}\partial_{u}r) & =-4\pi r\text{\textgreek{W}}^{-2}(\partial_{u}\text{\textgreek{f}})^{2},\\
\partial_{u}\partial_{v}(r\text{\textgreek{f}}) & =-\frac{\text{\textgreek{W}}^{2}-4\partial_{u}r\partial_{v}r}{4r^{2}}\cdot r\text{\textgreek{f}}.
\end{cases}\label{eq:EinsteinScalarFieldInDoubleNull}
\end{equation}
The well-posedness of the asymptotically AdS initial-boundary value
problem for the system (\ref{eq:EinsteinScalarFieldInDoubleNull})
with reflecting boundary conditions on $\mathcal{I}$ was established
by Holzegel and Smulevici in \cite{HolzSmul2012}.

Numerical results in the direction of establishing the AdS instability
conjecture were first obtained in 2011 by Bizon and Rostworowski in
\cite{bizon2011weakly}, who studied the evolution of spherically
symmetric perturbations of $(\mathcal{M}_{AdS},g_{AdS})$ for (\ref{eq:EinsteinScalarField})
in Schwarzschild-type coordinates. More precisely, \cite{bizon2011weakly}
numerically simulated the evolution of initial data for (\ref{eq:EinsteinScalarField})
with $\text{\textgreek{f}}$ initially arranged into small amplitude
wave packets. It was found that, for certain families of initial arrangements
of this form (of ``size'' $\text{\textgreek{e}}$), after a finite
number of reflections on $\mathcal{I}$ (proportional to $\text{\textgreek{e}}^{-2}$),
the energy of the wave packets becomes substantially concentrated,
leading to a break down of the coordinate system associated with the
threshold of trapped surface formation. 

Following \cite{bizon2011weakly}, a vast amount of numerical and
heuristic works have been dedicated to the understanding of the global
dynamics of perturbations of $(\mathcal{M}_{AdS},g_{AdS})$ for the
system (\ref{eq:EinsteinScalarField}) (see, e.\,g.,~\cite{DiasHorSantos,BuchelEtAl2012,DiasEtAl,MaliborskiEtAl,BalasubramanianEtAl,CrapsEtAl2014,CrapsEtAl2015,BizonMalib,DimitrakopoulosEtAl,GreenMailardLehnerLieb,HorowitzSantos,DimitrakopoulosEtAl2015,DimitrakopoulosEtAl2016}).
In these works, the picture that arises regarding the long time dynamics
of \emph{generic} spherically symmetric perturbations is rather complicated:
Apart from perturbations that lead to instability and trapped surface
formation (\cite{DiasHorSantos,BuchelEtAl2012}), it appears that
there exist certain types of perturbations (dubbed ``islands of stability'')
which remain close to $(\mathcal{M}_{AdS},g_{AdS})$ for long times;
see \cite{DiasEtAl,MaliborskiEtAl,BalasubramanianEtAl,DimitrakopoulosEtAl2015}.
Perturbations of the latter type might in fact occupy an open set
in the moduli space of spherically symmetric initial data for (\ref{eq:EinsteinScalarField})
(see \cite{BalasubramanianEtAl,DimitrakopoulosEtAl2015}). The question
of existence of open ``corners'' of initial data around $(\mathcal{M}_{AdS},g_{AdS})$
leading to trapped surface formation has also been studied (see, e.\,g.~\cite{DimitrakopoulosEtAl}). 

Another interesting problem in this context is the characterization
of the possible end states of the evolution of unstable perturbations
of $(\mathcal{M}_{AdS},g_{AdS})$. In \cite{HolSmul2013}, Holzegel--Smulevici
established that the Schwarzschild--AdS spacetime $(\mathcal{M}_{Sch},g_{Sch})$
is an asymptotically stable solution of the system (\ref{eq:EinsteinScalarField})
in spherical symmetry, with perturbations decaying at an exponential
rate.%
\footnote{A similar result can presumably also be deduced for the vacuum Einstein
equations (\ref{eq:VacuumEinsteinEquations}) reduced under the biaxial
Bianchi IX symmetry in $4+1$ dimensions, following by an amalgamation
of the proofs of \cite{Holzegel5d2010} and \cite{HolSmul2013}.%
} This result supports the expectation that all spherically symmetric
perturbations of $(\mathcal{M}_{AdS},g_{AdS})$ for the system (\ref{eq:EinsteinScalarField})
leading to the formation of a trapped surface eventually settle down
to a member of the Schwarzschild--AdS family (see \cite{DafHol,DafHolStability2006}).
However, beyond spherical symmetry, Holzegel--Smulevici \cite{Holzegel2013,Holzegel2013a}
showed that solutions to the linear scalar wave equation 
\begin{equation}
\square_{g_{Sch}}\text{\textgreek{f}}=0\label{eq:ScalarWaveEquation}
\end{equation}
 on $(\mathcal{M}_{Sch},g_{Sch})$ (and, more generally, on Kerr--AdS)
decay at a slow (logarithmic) rate, which is insufficient in itself
to yield the non-linear stability of $(\mathcal{M}_{Sch},g_{Sch})$
(cf.~our remark below the statement of the AdS instability conjecture).
Thus, \cite{Holzegel2013a} conjectured that $(\mathcal{M}_{Sch},g_{Sch})$
is non-linearly unstable. On the other hand, based on a detailed analysis
of quasinormal modes on $(\mathcal{M}_{Sch},g_{Sch})$, Dias--Horowitz--Marolf--Santos
\cite{DiasEtAl} suggested that sufficiently regular, non-linear perturbations
of $(\mathcal{M}_{Sch},g_{Sch})$ still remain small, at least for
long times. As a result, the picture regarding the end state of the
evolution of generic perturbations of $(\mathcal{M}_{AdS},g_{AdS})$
outside spherical symmetry remains unclear (see also \cite{HorowitzSantos,DiasSantos2016,Rostworowski2017}).

Following \cite{bizon2011weakly}, the bulk of heuristic works have
implemented a frequency space analysis in the study of the AdS instability
conjecture. A notable exception is the work \cite{DimitrakopoulosEtAl}
of Dimitrakopoulos--Freivogel--Lippert--Yang, where a physical space
mechanism possibly leading to instability for the system (\ref{eq:EinsteinScalarFieldInDoubleNull})
is suggested. We will revisit the mechanism of \cite{DimitrakopoulosEtAl}
and compare it with the results of this paper at the and of Section
\ref{sub:Sketch-of-the-proof}.

\subsection{\label{sub:NeedOfAMirror}The Einstein--null dust system in spherical
symmetry}

A spherically symmetric model for (\ref{eq:VacuumEinsteinEquations})
which is even simpler than (\ref{eq:EinsteinScalarField}) is the
\emph{Einstein--massless Vlasov} system (see \cite{Andreasson2011,Rein1995}).
The case where the Vlasov field is supported only on radial geodesics
is a singular reduction of this system which is equivalent to the
\emph{Einstein--null dust }system, allowing for both ingoing and outgoing
dust (see \cite{Rendall1997}). This system was studied by Poisson
and Israel in their seminal work on mass inflation \cite{PoissonIsrael1990}.
In $3+1$ dimensions, it takes the form (in double null coordinates
$(u,v)$): 
\begin{equation}
\begin{cases}
\partial_{u}\partial_{v}(r^{2}) & =-\frac{1}{2}(1-\Lambda r^{2})\text{\textgreek{W}}^{2},\\
\partial_{u}\partial_{v}\log(\text{\textgreek{W}}^{2}) & =\frac{\text{\textgreek{W}}^{2}}{2r^{2}}\big(1+4\text{\textgreek{W}}^{-2}\partial_{u}r\partial_{v}r\big),\\
\partial_{v}(\text{\textgreek{W}}^{-2}\partial_{v}r) & =-4\pi r^{-1}\text{\textgreek{W}}^{-2}\bar{\text{\textgreek{t}}},\\
\partial_{u}(\text{\textgreek{W}}^{-2}\partial_{u}r) & =-4\pi r^{-1}\text{\textgreek{W}}^{-2}\text{\textgreek{t}},\\
\partial_{u}\bar{\text{\textgreek{t}}} & =0,\\
\partial_{v}\text{\textgreek{t}} & =0.
\end{cases}\label{eq:EinsteinNullDust}
\end{equation}

In certain cases, the Einstein--null dust system (\ref{eq:EinsteinNullDust})
can be formally viewed as a high frequency limit of the Einstein--scalar
field system (\ref{eq:EinsteinScalarFieldInDoubleNull}) (as was already
discussed in \cite{PoissonIsrael1990}): Setting 
\[
\text{\textgreek{t}}\doteq r^{2}(\partial_{u}\text{\textgreek{f}})^{2},\mbox{ }\bar{\text{\textgreek{t}}}\doteq r^{2}(\partial_{v}\text{\textgreek{f}})^{2}
\]
in (\ref{eq:EinsteinScalarField}) and dropping all lower order terms
from the wave equation for $\text{\textgreek{f}}$, one formally obtains
(\ref{eq:EinsteinNullDust}) in the region where $\partial_{u}\text{\textgreek{f}}\partial_{v}\text{\textgreek{f}}$
is negligible, i.\,e\@. outside the intersection of the supports
of $\text{\textgreek{t}},\bar{\text{\textgreek{t}}}$. While this
formal limiting procedure can be rigorously justified away from $r=0$,
the dynamical similarities between (\ref{eq:EinsteinScalarFieldInDoubleNull})
and (\ref{eq:EinsteinNullDust}) break down close to $r=0$. A fundamental
difference between these systems is the fact that, while small data
asymptotically AdS solutions to (\ref{eq:EinsteinScalarFieldInDoubleNull})
satisfying a reflecting boundary condition at $\mathcal{I}$ remain
regular (and ``small'') for large times, all non-trivial solutions
to the system (\ref{eq:EinsteinNullDust}) break down once the support
of $\bar{\text{\textgreek{t}}}$ reaches the axis $\text{\textgreek{g}}$
(i.\,e.~the timelike portion of $\{r=0\}$), independently of the
boundary conditions imposed at $\mathcal{I}$. This is an ill-posedness
statement for (\ref{eq:EinsteinNullDust}), which needs to be addressed
before any attempt to study the AdS instability conjecture in the
setting of (\ref{eq:EinsteinNullDust}).

We will now proceed to discuss this difference of (\ref{eq:EinsteinScalarFieldInDoubleNull})
and (\ref{eq:EinsteinNullDust}) in more detail.

\subsubsection*{Cauchy stability for the Einstein--scalar field system}

The following Cauchy stability result holds for the system (\ref{eq:EinsteinScalarFieldInDoubleNull}):

\begin{customprop}{1}[Cauchy stability for \eqref{eq:EinsteinScalarFieldInDoubleNull}; see \cite{HolzSmul2012}]\label{prop:CauchyStabilityIntro}
For a suitable initial data norm $||\cdot||_{initial}$, $(\mathcal{M}_{AdS},g_{AdS})$
is Cauchy stable as a solution of the system (\ref{eq:EinsteinScalarFieldInDoubleNull})
with reflecting boundary conditions on $\mathcal{I}$. That is to
say, for all fixed times $T_{*}>0$, any perturbation of the initial
data of $(\mathcal{M}_{AdS},g_{AdS})$ which is small enough (when
measured in terms of $||\cdot||_{initial}$) with respect to $T_{*}$
gives rise to a solution of (\ref{eq:EinsteinScalarFieldInDoubleNull})
which is regular and close to $(\mathcal{M}_{AdS},g_{AdS})$ for times
up to $T_{*}$.

\end{customprop}
\begin{rem*}
In the statement of Proposition \ref{prop:CauchyStabilityIntro},
Cauchy stability of $(\mathcal{M}_{AdS},g_{AdS})$ refers to stability
over fixed compact subsets of the \emph{conformal compactification}
of $(\mathcal{M}_{AdS},g_{AdS})$, such as subsets of the form $\{0\le t\le T_{*}\}$
in the $(t,r,\text{\textgreek{j}},\text{\textgreek{f}})$ coordinate
chart. Any such subset contains, in particular, a compact subset of
the timelike boundary $\mathcal{I}$. 

The initial data norm $||\cdot||_{initial}$, for which the Cauchy
stability of $(\mathcal{M}_{AdS},g_{AdS})$ follows from \cite{HolzSmul2012},
is a higher order, suitably weighted $C^{k}$ norm. However, this
is not the only norm for which $(\mathcal{M}_{AdS},g_{AdS})$ can
be shown to be Cauchy stable: An additional, highly non-trivial example
of such a norm is the bounded variation norm of Christodoulou \cite{ChristodoulouBoundedVariation}
(modified with suitable $r$-weights near $r=\infty$). Similar low-regularity
norms will also play an important role in this paper (see Section
\ref{sub:NeedOfAMirror}).

Assuming, for simplicity, that initial data are prescribed on the
outgoing null hypersurface corresponding to $u=0$, for $0\le v\le v_{*}$,
a necessary condition for Cauchy stability of $(\mathcal{M}_{AdS},g_{AdS})$
for the system (\ref{eq:EinsteinScalarFieldInDoubleNull}) with respect
to an initial data norm $||\cdot||_{initial}$ is that, for any given
$R_{0}>0$, $||\cdot||_{initial}$ controls the quantity 
\begin{equation}
\mathscr{M}\doteq\sup_{\substack{0\le v_{1}<v_{2}\le v_{*},\\
\frac{r(0,v_{2})}{r(0,v_{1})}<\frac{3}{2},\mbox{ }r(0,v_{2})\le R_{0}
}
}\frac{\tilde{m}(0,v_{2})-\tilde{m}(0,v_{1})}{(r(0,v_{2})-r(0,v_{1}))\big|\log\big(\frac{r(0,v_{2})}{r(0,v_{1})}-1\big)\big|},\label{eq:EnergyConcentration}
\end{equation}
where $\tilde{m}$ is the \emph{renormalised Hawking mass}, defined
in terms of the Hawking mass $m$, 
\begin{equation}
m\doteq\frac{r}{2}\Big(1-4\text{\textgreek{W}}^{-2}\partial_{u}r\partial_{v}r\Big),\label{eq:HawkingMassIntro}
\end{equation}
by the relation 
\begin{equation}
\tilde{m}\doteq m-\frac{1}{6}\Lambda r^{3}.\label{eq:RenormalisedhawkingMassIntro}
\end{equation}
This is a consequence of the fact that, when $\mathscr{M}$ exceeds
a certain threshold (depending on $R_{0}$), there exists a $u_{\dagger}\in(0,v_{*})$
and a point $p=(u_{\dagger},v_{\dagger})$ in the development of the
initial data such that 
\begin{equation}
\frac{2m}{r}(u_{\dagger},v_{\dagger})>1,\label{eq:IndicatorBlackHole}
\end{equation}
 a result proven by Christodoulou in \cite{Christodoulou1991}.%
\footnote{The result of \cite{Christodoulou1991} was restricted to the case
$\Lambda=0$, but the proof can be readily modified to include the
case $\Lambda<0$.%
} The bound (\ref{eq:IndicatorBlackHole}) implies that 
\begin{equation}
\partial_{u}r(u_{\dagger},v_{\dagger})<0\mbox{, }\partial_{v}r(u_{\dagger},v_{\dagger})<0,\label{eq:TrappedSurface}
\end{equation}
i.\,e.~that the symmetry sphere associated to $(u_{\dagger},v_{\dagger})$
is a \emph{trapped surface}. In particular, $(u_{\dagger},v_{\dagger})$
is contained in a black hole.%
\footnote{We should remark that (\ref{eq:TrappedSurface}) follows from (\ref{eq:IndicatorBlackHole})
under the assumption that $\partial_{u}r<0$ (which always holds provided,
initially, $\partial_{u}r|_{u=0}<0$; see \cite{ChristodoulouBoundedVariation}). %
} 

As a corollary, it follows that the total ADM mass of the initial
data, though expressible as a coercive functional on the space of
initial data of (\ref{eq:EinsteinScalarFieldInDoubleNull}), does
not yield a norm for which $(\mathcal{M}_{AdS},g_{AdS})$ is Cauchy
stable for (\ref{eq:EinsteinScalarFieldInDoubleNull}), since the
ADM mass manifestly fails to control (\ref{eq:EnergyConcentration}).
\end{rem*}

\subsubsection*{Break down at $r=0$ and ``trivial'' Cauchy instability for the
Einstein--null dust system }

The following instability result holds for the system (\ref{eq:EinsteinNullDust})
(see \cite{MoschidisMaximalDevelopment}):

\begin{customprop}{2}[Cauchy instability for \eqref{eq:EinsteinNullDust}]\label{prop:C0InextendibilityIntroduction}
Any globally hyperbolic spherically symmetric solution $(\mathcal{M},g;\text{\textgreek{t}},\bar{\text{\textgreek{t}}})$
of (\ref{eq:EinsteinNullDust}) with non-empty axis $\text{\textgreek{g}}$
``breaks down'' at the first point when a radial geodesic in the
support of $\bar{\text{\textgreek{t}}}$ reaches $\text{\textgreek{g}}$:
Beyond that point, $(\mathcal{M},g;\text{\textgreek{t}},\bar{\text{\textgreek{t}}})$
is $C^{0}$ inextendible as a  spherically symmetric solution to (\ref{eq:EinsteinNullDust}).
As a result, $(\mathcal{M}_{AdS},g_{AdS})$ is not a Cauchy stable
solution of (\ref{eq:EinsteinNullDust}) for any ``reasonable''
initial data topology.

\end{customprop}

For the precise definition of the notion of $C^{0}$ inextendibility
as a  spherically symmetric solution to (\ref{eq:EinsteinNullDust}),
see \cite{MoschidisMaximalDevelopment}. Note that this a \emph{stronger}
statement than $(\mathcal{M},g;\text{\textgreek{t}},\bar{\text{\textgreek{t}}})$
breaking down as a smooth solution of (\ref{eq:EinsteinNullDust}).
We should also remark the following regarding Proposition \ref{prop:C0InextendibilityIntroduction}:
\begin{itemize}
\item Proposition \ref{prop:C0InextendibilityIntroduction} holds independently
of the value of the cosmological constant $\Lambda$. In particular,
Minkowski spacetime $(\mathbb{R}^{3+1},\text{\textgreek{h}})$ is
not Cauchy stable for (\ref{eq:EinsteinNullDust}) with $\Lambda=0$
for any ``reasonable'' initial data topology.
\item \noindent Proposition \ref{prop:C0InextendibilityIntroduction} yields
a uniform upper bound on the time of existence of solutions $(\mathcal{M},g)$
to (\ref{eq:EinsteinNullDust}) for any initial data set for which
$\bar{\text{\textgreek{t}}}$ is not identically equal to $0$, depending
only on the distance of the initial support of $\bar{\text{\textgreek{t}}}$
from the axis and, thus, independently of the proximity of the initial
data to the trivial data (in any reasonable initial data norm). We
should also highlight that the instability of Proposition \ref{prop:C0InextendibilityIntroduction}
has nothing to do with trapped surface formation: Up to the first
retarded time when a radial geodesic in the support of $\bar{\text{\textgreek{t}}}$
reaches $\text{\textgreek{g}}$, any solution $(\mathcal{M},g)$ to
(\ref{eq:EinsteinNullDust}) arising from smooth initial data close
to $(\mathcal{M}_{AdS},g_{AdS})$ remains smooth and close to $(\mathcal{M}_{AdS},g_{AdS})$,
and $(\mathcal{M},g)$ \emph{contains no trapped surface}. In fact,
in this case, despite being $C^{0}$ inextendible as a globally hyperbolic
spherically symmetric solution to (\ref{eq:EinsteinNullDust}), $(\mathcal{M},g)$
is globally $C^{\infty}$-extendible as a spherically symmetric Lorentzian
manifold; see \cite{MoschidisMaximalDevelopment}.
\item The Cauchy stability statement for $(\mathcal{M}_{AdS},g_{AdS})$
for the system (\ref{eq:EinsteinScalarFieldInDoubleNull}) stated
in Proposition \ref{prop:CauchyStabilityIntro} can be informally
interpreted as the result of a natural dispersive mechanism close
to the axis $\text{\textgreek{g}}$ displayed by the system (\ref{eq:EinsteinScalarFieldInDoubleNull}),
which does not allow the energy of $\text{\textgreek{f}}$ to concentrate
on scales smaller than $\tilde{m}$ in $O(1)$ time, provided a suitable
initial norm of $\text{\textgreek{f}}$ (controlling at least (\ref{eq:EnergyConcentration}))
is small enough. No such mechanism is present for the system (\ref{eq:EinsteinNullDust}),
as is illustrated by Proposition \ref{prop:C0InextendibilityIntroduction}. 
\end{itemize}

\subsubsection*{Resolution of the ``trivial'' instability of (\ref{eq:EinsteinNullDust})
through an inner mirror}

In order to turn the spherically symmetric Einstein--null dust system
(\ref{eq:EinsteinNullDust}) into a well-posed, Cauchy-stable system
(a necessary step for converting (\ref{eq:EinsteinNullDust}) into
an effective model of the vacuum Einstein equations (\ref{eq:VacuumEinsteinEquations})),
it is necessary to explicitly add to (\ref{eq:EinsteinNullDust})
a mechanism that prevents the break down at $r=0$ described by Proposition
\ref{prop:C0InextendibilityIntroduction}, so that, moreover, an analogue
of Proposition \ref{prop:CauchyStabilityIntro} holds for (\ref{eq:EinsteinNullDust}).
This can be achieved by by placing an \emph{inner mirror} at $r=r_{0}>0$,
i.\,e.~by restricting (\ref{eq:EinsteinNullDust}) on $\{r\ge r_{0}\}$,
for some $r_{0}>0$, and imposing a reflecting boundary condition
on the portion $\text{\textgreek{g}}_{0}$ of the set $\{r=r_{0}\}$
which is timelike.
\begin{rem*}
The reflecting boundary condition on $\text{\textgreek{g}}_{0}$ can
be motivated by the fact that, for smooth spherically symmetric solutions
$(\mathcal{M},g;\text{\textgreek{f}})$ to (\ref{eq:EinsteinScalarField}),
the function $\text{\textgreek{f}}$, viewed as a function on the
quotient of $(\mathcal{M},g)$ by the spheres of symmetry, satisfies
a reflecting boundary condition on the axis.
\end{rem*}
The well-posedness and the properties of the maximal development for
the system (\ref{eq:EinsteinNullDust}) with reflecting boundary conditions
on $\mathcal{I}$ and $\text{\textgreek{g}}_{0}$ are addressed in
the companion paper \cite{MoschidisMaximalDevelopment}. The following
result is established in \cite{MoschidisMaximalDevelopment}:

\begin{figure}[h] 
\centering 
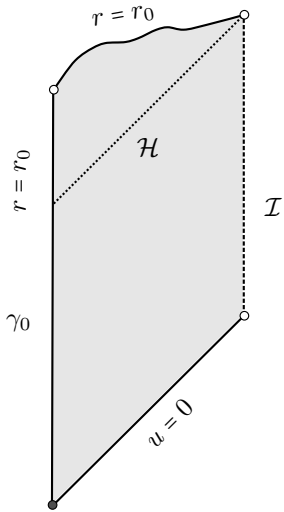 
\caption{Schematic depiction of the domain on which the maximal future development $(r,\Omega^2,\tau,\bar{\tau})$ of a smooth initial data set on $u=0$ (with reflecting boundary conditions on $\mathcal{I}$ and $\gamma_{0}$) is defined. A gauge conditions ensures that $\mathcal{I}$ and $\gamma_{0}$ are straight vertical lines. Conformal infinity $\mathcal{I}$ is always complete in this setting. In the case when the future event horizon $\mathcal{H}^{+}$ is non-empty,  it is smooth and has infinite affine length. In this case, apart from the mirror $\gamma_{0}$, the boundary of the domain has a spacelike portion on which $\{ r=r_{0}\}$.}
\end{figure}

\begin{customthm}{2}[Well posedness for \eqref{eq:EinsteinNullDust} with an inner mirror]\label{thm:MaximalDevelopmentIntro}
For any $r_{0}>0$ and any smooth asymptotically AdS initial data
set $(r,\text{\textgreek{W}}^{2},\text{\textgreek{t}},\bar{\text{\textgreek{t}}})|_{u=0}$
on $u=0$, there exists a unique smooth maximal future development
$(r,\text{\textgreek{W}}^{2},\text{\textgreek{t}},\bar{\text{\textgreek{t}}})$
on $\{r\ge r_{0}\}$, solving (\ref{eq:EinsteinNullDust}) with reflecting
boundary conditions on $\mathcal{I}$ and $\text{\textgreek{g}}_{0}$,
where $r|_{\text{\textgreek{g}}_{0}}=r_{0}$ and $\text{\textgreek{g}}_{0}$
coincides with the portion of the curve $\{r=r_{0}\}$ which is timelike
(fixing the gauge freedom by imposing a reflecting gauge condition
on both $\mathcal{I}$ and $\text{\textgreek{g}}_{0}$). For this
development, $\mathcal{I}$ is complete and $\{r=r_{0}\}$ is timelike
in the past of $\mathcal{I}$ (see Figure 1.2). 

In the case when the future event horizon $\mathcal{H}^{+}$ is non-empty,
it is smooth and future complete. A necessary condition for $\mathcal{H}^{+}$
to be non-empty is the existence of a point $(u_{\dagger},v_{\dagger})$
where (\ref{eq:IndicatorBlackHole}) holds. If the total mass $\tilde{m}|_{\mathcal{I}}$
and the mirror radius $r_{0}$ satisfy 
\begin{equation}
\frac{2\tilde{m}|_{\mathcal{I}}}{r_{0}}\le1-\frac{1}{3}\Lambda r_{0}^{2},\label{eq:BoundForTotalMassHorizon}
\end{equation}
then necessarily $\mathcal{H}^{+}=\emptyset$.

\end{customthm}

For a more detailed statement of Theorem \ref{thm:MaximalDevelopmentIntro},
see Section \ref{sec:ResultsFromTheOtherPaper} and \cite{MoschidisMaximalDevelopment}. 

In view of the fact that $\mathcal{H}^{+}=\emptyset$ in the case
when the total mass $\tilde{m}|_{\mathcal{I}}$ and the mirror radius
$r_{0}$ satisfy (\ref{eq:BoundForTotalMassHorizon}), in order to
address the AdS instability conjecture for the system (\ref{eq:EinsteinNullDust})
with reflecting boundary conditions on $\mathcal{I}$ and $\text{\textgreek{g}}_{0}$,
it is necessary to allow $r_{0}$ to shrink to $0$ with the size
of the data. Thus, addressing the AdS instability conjecture in this
setting requires establishing a Cauchy stability statement for $(\mathcal{M}_{AdS},g_{AdS})$
which is\emph{ independent of the precise value of the mirror radius
$r_{0}$}. This is the statement of the following result, proved in
our companion paper \cite{MoschidisMaximalDevelopment}:

\begin{customthm}{3}[Cauchy stability for \eqref{eq:EinsteinNullDust} uniformly in $r_0$]\label{thm:CauchyStabilityIntro}

Given $\text{\textgreek{e}}>0$, $u_{*}>0$, there exists a $\text{\textgreek{d}}>0$
such that the following statement holds: For \underline{any} $r_{0}>0$
and \underline{any} initial data set $(r,\text{\textgreek{W}}^{2},\text{\textgreek{t}},\bar{\text{\textgreek{t}}})|_{u=0}$
satisfying 
\begin{equation}
||(r,\text{\textgreek{W}}^{2},\text{\textgreek{t}},\bar{\text{\textgreek{t}}})||_{u=0}\doteq\sup_{\bar{v}}\int_{u=0}\frac{\bar{\text{\textgreek{t}}}(0,v)}{|\text{\textgreek{r}}(0,v)-\text{\textgreek{r}}(0,\bar{v})|+\tan^{-1}(\sqrt{-\Lambda}r_{0})}\,\frac{\sqrt{-\Lambda}dv}{\partial_{v}\text{\textgreek{r}}(0,v)}+\sup_{u=0}\Big(\Big|\big(1-\frac{2\tilde{m}}{r}\big)^{-1}-1\Big|+\sqrt{-\Lambda}\tilde{m}\Big)\le\text{\textgreek{e}},\label{eq:NormCauchyStabilityIntro}
\end{equation}
where 
\begin{equation}
\text{\textgreek{r}}(0,v)\doteq\tan^{-1}(\sqrt{-\Lambda}r)(0,v),
\end{equation}
the corresponding solution $(r,\text{\textgreek{W}}^{2},\text{\textgreek{t}},\bar{\text{\textgreek{t}}})$
to (\ref{eq:EinsteinNullDust}) with reflecting boundary conditions
on $\mathcal{I}$ and $\text{\textgreek{g}}_{0}$ will satisfy 
\begin{equation}
\sup_{0\le\bar{u}\le u_{*}}||(r,\text{\textgreek{W}}^{2},\text{\textgreek{t}},\bar{\text{\textgreek{t}}})||_{u=\bar{u}}\le\text{\textgreek{d}}.
\end{equation}

\end{customthm}

For a more detailed statement of Theorem \ref{thm:CauchyStabilityIntro},
see Section \ref{sec:ResultsFromTheOtherPaper} and \cite{MoschidisMaximalDevelopment}. 

Notice that the norm (\ref{eq:NormCauchyStabilityIntro}) vanishes
only for the trivial initial data $(r,\text{\textgreek{W}}^{2},0,0)$.
Informally, Theorem \ref{thm:CauchyStabilityIntro} implies that,
if the energy of the initial data concentrated on scales proportional
to the mirror radius $r_{0}$ is small enough, then the energy of
the solution to (\ref{eq:EinsteinNullDust}) (with reflecting boundary
conditions on $\mathcal{I}$ and $\text{\textgreek{g}}_{0}$) will
remain similarly dispersed for times less than any given constant.
In particular, no trapped surface can form in this timescale if $\text{\textgreek{d}}$
is chosen sufficiently small. 

In Section \ref{sec:ResultsFromTheOtherPaper}, we will also present
a Cauchy stability statement for general solutions of (\ref{eq:EinsteinNullDust})
with reflecting boundary conditions on $\mathcal{I}$ and $\text{\textgreek{g}}_{0}$,
which will be used in the proof of Theorem \ref{Theorem_Intro_Rough}
(see Theorem \ref{prop:CauchyStability}).

\subsection{\label{sub:The-main-result:Intro}Statement of Theorem \ref{Theorem_Intro_Rough}:
the non-linear instability of AdS}

According to Theorem \ref{thm:CauchyStabilityIntro}, a Cauchy stability
statement holds for $(\mathcal{M}_{AdS},g_{AdS})$ for time intervals
which are independent of the precise value of the mirror radius $r_{0}$,
depending only on the smallness of the initial data norm (\ref{eq:NormCauchyStabilityIntro}).
As a result, it is possible to study the AdS instability conjecture
for the system (\ref{eq:EinsteinNullDust}) with reflecting boundary
conditions on $\mathcal{I}$ and $\text{\textgreek{g}}_{0}$, for
perturbations which are small with respect to (\ref{eq:NormCauchyStabilityIntro}),
\emph{allowing the mirror radius $r_{0}$ to shrink with the size
of the data}. In this paper, we will prove the following result:

\begin{customthm}{1}[more precise version]\label{thm:TheoremDetailedIntro}
There exists a family of positive numbers $r_{0\text{\textgreek{e}}}$
(satisfying $r_{0\text{\textgreek{e}}}\xrightarrow{\text{\textgreek{e}}\rightarrow0}0$)
and smooth initial data $(r,\text{\textgreek{W}}^{2},\text{\textgreek{t}},\bar{\text{\textgreek{t}}})^{(\text{\textgreek{e}})}|_{u=0}$
for the system (\ref{eq:EinsteinNullDust}) satisfying the following
properties:

\begin{enumerate}

\item In the norm $||\cdot||_{u=0}$ defined by (\ref{eq:NormCauchyStabilityIntro}):
\begin{equation}
||(r,\text{\textgreek{W}}^{2},\text{\textgreek{t}},\bar{\text{\textgreek{t}}})^{(\text{\textgreek{e}})}||_{u=0}\xrightarrow{\text{\textgreek{e}}\rightarrow0}0.\label{eq:NormToZeroIntro}
\end{equation}

\item For any $\text{\textgreek{e}}>0$, the maximal development
$(r,\text{\textgreek{W}}^{2},\text{\textgreek{t}},\bar{\text{\textgreek{t}}})^{(\text{\textgreek{e}})}$
of $(r,\text{\textgreek{W}}^{2},\text{\textgreek{t}},\bar{\text{\textgreek{t}}})^{(\text{\textgreek{e}})}|_{u=0}$
for the system (\ref{eq:EinsteinNullDust}) with reflecting boundary
conditions on $\mathcal{I}$ and $\text{\textgreek{g}}_{0}$, $r|_{\text{\textgreek{g}}_{0}}=r_{0\text{\textgreek{e}}}$,
contains a trapped sphere, i.\,e.~there exists a point $(u_{\text{\textgreek{e}}},v_{\text{\textgreek{e}}})$
such that: 
\begin{equation}
\frac{2m^{(\text{\textgreek{e}})}}{r^{(\text{\textgreek{e}})}}(u_{\text{\textgreek{e}}},v_{\text{\textgreek{e}}})>1.\label{eq:TrappedSurfaceIntro}
\end{equation}
Thus, in view of Theorem \ref{thm:MaximalDevelopmentIntro}, $(r,\text{\textgreek{W}}^{2},\text{\textgreek{t}},\bar{\text{\textgreek{t}}})^{(\text{\textgreek{e}})}$
contains a non-empty, smooth and future complete event horizon $\mathcal{H}^{+}$
and a complete conformal infinity $\mathcal{I}$. 

\end{enumerate}

\end{customthm}

For the definitive statement of Theorem \ref{thm:TheoremDetailedIntro},
see Section \ref{sec:The-main-result:Details}. The following remarks
should be made concerning Theorem \ref{thm:TheoremDetailedIntro}:
\begin{itemize}
\item In view of the Cauchy stability of $(\mathcal{M}_{AdS},g_{AdS})$
with respect to (\ref{eq:NormCauchyStabilityIntro}) (see Theorem
\ref{thm:CauchyStabilityIntro}), the time%
\footnote{where time is measured with respect to the (dimensionless) coordinate
function $\bar{t}=\sqrt{-\Lambda}(u+v)$.%
} required to elapse before $\big(1-\frac{2m^{(\text{\textgreek{e}})}}{r^{(\text{\textgreek{e}})}}\big)$
becomes negative necessarily tends to $+\infty$ as $\text{\textgreek{e}}\rightarrow0$.
\item In view of the fact that $\mathcal{H}^{+}=\emptyset$ when (\ref{eq:BoundForTotalMassHorizon})
holds, in order for $(r,\text{\textgreek{W}}^{2},\text{\textgreek{t}},\bar{\text{\textgreek{t}}})^{(\text{\textgreek{e}})}$
to satisfy both (\ref{eq:NormToZeroIntro}) and (\ref{eq:TrappedSurfaceIntro}),
it is necessary that $r_{0\text{\textgreek{e}}}\rightarrow0$ as $\text{\textgreek{e}}\rightarrow0$,
at a rate which is at least as fast as that of $2\tilde{m}^{(\text{\textgreek{e}})}|_{\mathcal{I}}$,%
\footnote{Note that the renormalised Hawking mass $\tilde{m}^{(\text{\textgreek{e}})}$
is constant on $\mathcal{I}$ when imposing a reflecting boundary
condition.%
} i.\,e.: 
\begin{equation}
r_{0\text{\textgreek{e}}}\le2\tilde{m}^{(\text{\textgreek{e}})}|_{\mathcal{I}}.\label{eq:BoundForMirroRadius}
\end{equation}
In fact, we choose the family $r_{0\text{\textgreek{e}}}$, $(r,\text{\textgreek{W}}^{2},\text{\textgreek{t}},\bar{\text{\textgreek{t}}})^{(\text{\textgreek{e}})}|_{u=0}$
of Theorem \ref{thm:TheoremDetailedIntro} to saturate the bound (\ref{eq:BoundForMirroRadius})
in the limit $\text{\textgreek{e}}\rightarrow0$, i.\,e. 
\begin{equation}
\lim_{\text{\textgreek{e}}\rightarrow0}\frac{r_{0\text{\textgreek{e}}}}{2\tilde{m}^{(\text{\textgreek{e}})}|_{\mathcal{I}}}=1.\label{eq:SaturatedBound}
\end{equation}
 For the proof of Theorem \ref{thm:TheoremDetailedIntro}, (\ref{eq:SaturatedBound})
is not essential and can be relaxed; however, it is fundamental for
our proof that $r_{0\text{\textgreek{e}}}$ is bounded from below
by some small multiple of $\tilde{m}^{(\text{\textgreek{e}})}|_{\mathcal{I}}$.
\item It follows from the proof of Theorem \ref{thm:MaximalDevelopmentIntro}
that, in the case $\Lambda=0$, Minkowski spacetime $(\mathbb{R}^{3+1},\text{\textgreek{h}})$
is \emph{globally stable} (for the system (\ref{eq:EinsteinNullDust})
with reflecting boundary conditions on the inner mirror $\{r=r_{0}\}$)
to initial data perturbations which are small with respect to the
norm (\ref{eq:NormCauchyStabilityIntro}), independently of the precise
choice of $r_{0}$. This fact further justifies the choice of the
matter model and the norm (\ref{eq:NormCauchyStabilityIntro}) as
a setting for establishing the AdS instability conjecture.
\item The proof of Theorem \ref{thm:TheoremDetailedIntro} also applies
in the case $\Lambda=0$ when placing an \emph{outer mirror} at $r=R_{0}\gg r_{0}$
(in addition to the inner mirror at $r=r_{0}$), i.\,e.~restricting
the solutions of (\ref{eq:EinsteinNullDust}) in the region $\{r_{0}\le r\le R_{0}\}$
and imposing reflecting boundary conditions on both $\{r=r_{0}\}$
and $\{r=R_{0}\}$. This is in accordance with the numerical results
of \cite{BuchelEtAl2012} for the system (\ref{eq:EinsteinScalarField}).
\item It can be readily inferred by Cauchy stability (see Theorem \ref{prop:CauchyStability}
in Section \ref{sub:Cauchy-stability-inCauchyStability}) that, for
any $y_{\text{\textgreek{e}}}\doteq(r,\text{\textgreek{W}}^{2},\text{\textgreek{t}},\bar{\text{\textgreek{t}}})^{(\text{\textgreek{e}})}|_{u=0}$
in the family of initial data of Theorem \ref{thm:TheoremDetailedIntro},
there exists an open neighborhood $\mathcal{W}_{\text{\textgreek{e}}}$
of initial data around $y_{\text{\textgreek{e}}}$ such that, for
all $y\in\mathcal{W}_{\text{\textgreek{e}}}$, the maximal future
development of $y$ also contains a trapped surface. In particular,
the set of initial data leading to trapped surface formation is open.
An even stronger genericity statement would be the existence of an
\emph{open instability corner} in the space of initial data around
$(\mathcal{M}_{AdS},g_{AdS})$ (see \cite{DimitrakopoulosEtAl}),
i.\,e.~the existence of a $c_{1}>0$ such that $\big\{ y:\, dist(y,y_{\text{\textgreek{e}}})\le c_{1}||y_{\text{\textgreek{e}}}||_{u=0}\big\}\subset\mathcal{W}_{\text{\textgreek{e}}}$
for all $\text{\textgreek{e}}>0$ (with $dist(\cdot,\cdot)$ being
the distance function associated to (\ref{eq:NormCauchyStabilityIntro})
for $r_{0}=r_{0\text{\textgreek{e}}}$). While we have not addressed
the issue of genericity of the unstable initial data in this paper,
we expect that the proof of Theorem \ref{thm:TheoremDetailedIntro}
can be adapted to yield the existence of an instability corner around
$(\mathcal{M}_{AdS},g_{AdS})$. 
\item A plethora of numerical works (see, e.\,g.~\cite{bizon2011weakly,BuchelEtAl2012,BizonMalib,DimitrakopoulosEtAl})
suggest that, in the case of the Einstein--scalar field system (\ref{eq:EinsteinScalarField}),
for families of initial data $(\text{\textgreek{f}}_{\text{\textgreek{e}}}^{(0)},\text{\textgreek{f}}_{\text{\textgreek{e}}}^{(1)})$
for the scalar field $\text{\textgreek{f}}$ of the form 
\begin{equation}
(\text{\textgreek{f}}_{\text{\textgreek{e}}}^{(0)},\text{\textgreek{f}}_{\text{\textgreek{e}}}^{(1)})=(\text{\textgreek{e}}\text{\textgreek{f}}^{(0)},\text{\textgreek{e}}\text{\textgreek{f}}^{(1)})\label{eq:RescaledInitialData}
\end{equation}
(where $(\text{\textgreek{f}}^{(0)},\text{\textgreek{f}}^{(1)})$
is a fixed initial profile), trapped surface formation occurs at time
$\sim\text{\textgreek{e}}^{-2}$. However, any rigorous formulation
of this statement for general families of initial data, not necessarily
of the form (\ref{eq:RescaledInitialData}), requires fixing an initial
data norm for which the initial data size is measured, with different
choices of (scale invariant) norms possibly leading to different timescales
of trapped surface formation for initial data of size $\sim\text{\textgreek{e}}$.
For this reason, given that the initial data $(\text{\textgreek{t}},\bar{\text{\textgreek{t}}})^{(\text{\textgreek{e}})}|_{u=0}$
in Theorem \ref{thm:TheoremDetailedIntro} can not be viewed as a
rescaling of a fixed profile of the form (\ref{eq:RescaledInitialData}),
we have not tried to optimize the time required for trapped surface
formation in Theorem \ref{thm:TheoremDetailedIntro} in terms of the
initial norm (\ref{eq:NormCauchyStabilityIntro}).
\end{itemize}

\subsection{\label{sub:Sketch-of-the-proof}Sketch of the proof and remarks on
Theorem \ref{thm:TheoremDetailedIntro} }

We will now proceed to sketch the main arguments involved in the proof
of Theorem \ref{thm:TheoremDetailedIntro}.

\subsubsection*{Construction of the initial data}

\begin{figure}[t] 
\centering 
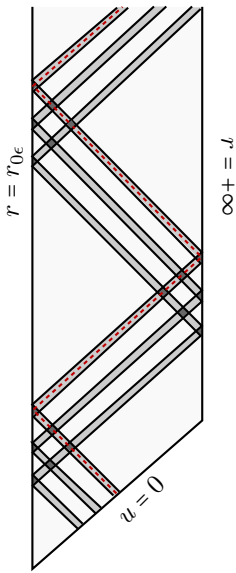 
\caption{The initial data $(r,\Omega^{2},\tau,\bar{\tau})^{(\epsilon)}|_{u=0}$ give rise to a bundle of ingoing beams which are  successively reflected off $\{r=r_{0\epsilon}\}$ and $\mathcal{I}=\{r=+\infty\}$. While the number of beams goes to infinity as $\epsilon\rightarrow 0$, for simplicity, we only depict here a bundle of three beams. As long as the total width of the bundle of beams remains small, the interaction set naturally splits into a part which lies close to $\{r=r_{0\epsilon}\}$ and a part near $\mathcal{I}$. We have also marked with a red dashed line the beam lying (initially) to the future of the rest.}
\end{figure}

The family of initial data $(r,\text{\textgreek{W}}^{2},\text{\textgreek{t}},\bar{\text{\textgreek{t}}})^{(\text{\textgreek{e}})}|_{u=0}$
in Theorem \ref{thm:TheoremDetailedIntro} is chosen so that its total
ADM mass $\tilde{m}^{(\text{\textgreek{e}})}|_{\mathcal{I}}$ and
the mirror radius $r_{0\text{\textgreek{e}}}$ satisfy (for $\text{\textgreek{e}}\ll1$)
\begin{equation}
r_{0\text{\textgreek{e}}},\tilde{m}^{(\text{\textgreek{e}})}|_{\mathcal{I}}\sim\text{\textgreek{e}}(-\Lambda)^{-\frac{1}{2}}.
\end{equation}
In particular, fixing a function $h(\text{\textgreek{e}})$ in terms
of $\text{\textgreek{e}}$ such that 
\[
\text{\textgreek{e}}\ll h(\text{\textgreek{e}})\ll1,
\]
the initial data $(r,\text{\textgreek{W}}^{2},\text{\textgreek{t}},\bar{\text{\textgreek{t}}})^{(\text{\textgreek{e}})}|_{u=0}$
are constructed so that the null dust initially forms a bundle of
narrow ingoing beams emanating from the region $r\sim1$; see Figure
1.3. The number of the beams is chosen to be large, i.\,e.~of order
$\sim(h(\text{\textgreek{e}}))^{-1}$, and the beams are initially
separated by gaps of $r$-width $\sim(h(\text{\textgreek{e}}))^{-1}\text{\textgreek{e}}(-\Lambda)^{-\frac{1}{2}}$.
The large number of beams and their initial separation are chosen
so that 
\begin{equation}
||(r,\text{\textgreek{W}}^{2},\text{\textgreek{t}},\bar{\text{\textgreek{t}}})^{(\text{\textgreek{e}})}||_{u=0}\sim h(\text{\textgreek{e}})\xrightarrow{\text{\textgreek{e}}\rightarrow0}0.\label{eq:InitialDataNorm}
\end{equation}

\subsubsection*{Remarks on the configuration of the null dust beams}

As the solution $(r,\text{\textgreek{W}}^{2},\text{\textgreek{t}},\bar{\text{\textgreek{t}}})^{(\text{\textgreek{e}})}$
arising fom the initial data set $(r,\text{\textgreek{W}}^{2},\text{\textgreek{t}},\bar{\text{\textgreek{t}}})^{(\text{\textgreek{e}})}|_{u=0}$
evolves according to equations (\ref{eq:EinsteinNullDust}), the null
dust beams are reflected successively off $\text{\textgreek{g}}_{0}=\{r=r_{0\text{\textgreek{e}}}\}$
and $\mathcal{I}$, as depicted in Figure 1.3. The beams separate
the spacetime into vacuum regions (the larger rectangular regions
between the beams in Figure 1.3), where the renormalised Hawking mass
$\tilde{m}^{(\text{\textgreek{e}})}$ is constant (recall the definition
of the Hawking mass and the renormalised Hawking mass by (\ref{eq:HawkingMassIntro})
and (\ref{eq:RenormalisedhawkingMassIntro}), respectively). The \emph{interaction
set} of the beams consists of all the points in the spacetime where
two different beams intersect (depicted in Figure 1.3 as the union
of all the smaller dark rectangles, lying in the intersection of any
two beams). As long as the total width of the bundle of beams remains
small, the interaction set can be split into two sets, one consisting
of the intersections occuring close to the mirror $\text{\textgreek{g}}_{0}$
and one consisting of the intersections near $\mathcal{I}$ (see Figure
1.3). 

Every beam is separated by the interaction set into several components.
To each such component, we can associate the \emph{mass difference}
$\mathfrak{D}\tilde{m}$ between the two vacuum regions which are
themselves separated by that beam component. The mass difference $\mathfrak{D}\tilde{m}$
measures the energy content of each beam component and, in view of
the non-linearity of the system (\ref{eq:EinsteinNullDust}), it is
not necessarily conserved along the beam after an intersection with
another beam. Precisely determining the resulting change in the mass
difference after the interaction of two beams will be the crux of
the proof of Theorem \ref{thm:TheoremDetailedIntro}.

\subsubsection*{Beam interactions and change in mass difference}

In Figure 1.4, the region around the intersection of an incoming null
dust beam $\text{\textgreek{z}}_{in}$ and an outgoing null dust beam
$\text{\textgreek{z}}_{out}$ is depicted. This region is separated
by the beams into 4 vacuum subregions $\mathcal{R}_{1},\ldots,\mathcal{R}_{4}$
with associated renormalised Hawking masses $\tilde{m}_{1},\ldots,\tilde{m}_{4}$
(see Figure 1.4). Before the intersection of the two beams, the mass
difference of the incoming beam $\text{\textgreek{z}}_{in}$ is 
\begin{equation}
\overline{\mathfrak{D}}_{-}\tilde{m}=\tilde{m}_{3}-\tilde{m}_{4},
\end{equation}
while the mass difference of the outgoing beam $\text{\textgreek{z}}_{out}$
is 
\begin{equation}
\mathfrak{D}_{-}\tilde{m}=\tilde{m}_{4}-\tilde{m}_{2}.
\end{equation}
After the intersection of the beams, the mass differences associated
to $\text{\textgreek{z}}_{in}$ and $\text{\textgreek{z}}_{out}$
become
\begin{equation}
\overline{\mathfrak{D}}_{+}\tilde{m}=\tilde{m}_{1}-\tilde{m}_{2}
\end{equation}
and 
\begin{equation}
\mathfrak{D}_{+}\tilde{m}=\tilde{m}_{3}-\tilde{m}_{1},
\end{equation}
respectively.

\begin{figure}[h] 
\centering 
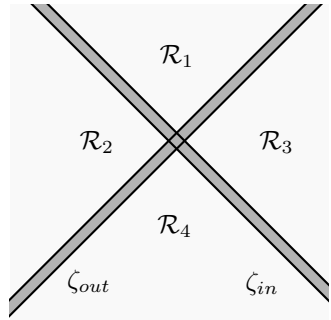 
\caption{The region in the $(u,v)$ plane around the intersection of an incoming beam $\zeta_{in}$ and an outgoing beam $\zeta_{out}$. The regions $\mathcal{R}_{i}$, $i=1,\ldots ,4$ are vacuum and the renormalised Hawking mass $\tilde{m}$ is constant (and equal to $\tilde{m}_{i}$) on each of the $\mathcal{R}_{i}$'s.}
\end{figure}

Assuming that 
\begin{equation}
\frac{2m}{r}<1
\end{equation}
and
\begin{equation}
\partial_{u}r<0<\partial_{v}r,
\end{equation}
 we can readily obtain the following differential relations for $r$
and $\tilde{m}$ from (\ref{eq:EinsteinNullDust}): 
\begin{align}
\partial_{u}\log\big(\frac{\partial_{v}r}{1-\frac{2m}{r}}\big)= & -\frac{4\pi}{r}\frac{\text{\textgreek{t}}}{-\partial_{u}r},\label{eq:IntroConstraints}\\
\partial_{v}\log\big(\frac{-\partial_{u}r}{1-\frac{2m}{r}}\big)= & \frac{4\pi}{r}\frac{\bar{\text{\textgreek{t}}}}{\partial_{v}r}\nonumber 
\end{align}
and 
\begin{align}
\partial_{u}\tilde{m}= & -2\pi\Big(\frac{1-\frac{2m}{r}}{-\partial_{u}r}\Big)\text{\textgreek{t}},\label{eq:IntroMassEquation}\\
\partial_{v}\tilde{m}= & 2\pi\Big(\frac{1-\frac{2m}{r}}{\partial_{v}r}\Big)\bar{\text{\textgreek{t} }}.\nonumber 
\end{align}
We will also assume that:

\begin{itemize}

\item{ The null dust beams $\text{\textgreek{z}}_{in}$ and $\text{\textgreek{z}}_{out}$
are sufficiently narrow so that, on their intersection $\text{\textgreek{z}}_{in}\cap\text{\textgreek{z}}_{out}$,
$r$ can be considered nearly constant:%
\footnote{This is possible in view of the fact that, for solutions $(r,\text{\textgreek{W}}^{2},\text{\textgreek{t}},\bar{\text{\textgreek{t}}})$
to (\ref{eq:EinsteinNullDust}), $r$ remains uniformly continuous
in the limit when $\text{\textgreek{t}},\bar{\text{\textgreek{t}}}$
tend to $\text{\textgreek{d}}$-functions in the $u,v$ variables,
respectively.%
}
\begin{equation}
\sup_{\text{\textgreek{z}}_{in}\cap\text{\textgreek{z}}_{out}}r-\inf_{\text{\textgreek{z}}_{in}\cap\text{\textgreek{z}}_{out}}r\ll\text{\textgreek{e}}(-\Lambda)^{-\frac{1}{2}},
\end{equation}
}

\item{ $\overline{\mathfrak{D}}_{+}\tilde{m}-\overline{\mathfrak{D}}_{-}\tilde{m}$
and $\mathfrak{D}_{+}\tilde{m}-\mathfrak{D}_{-}\tilde{m}$ are relatively
small.%
\footnote{Note that, necessarily, $\overline{\mathfrak{D}}_{+}\tilde{m}-\overline{\mathfrak{D}}_{-}\tilde{m}=-(\mathfrak{D}_{+}\tilde{m}-\mathfrak{D}_{-}\tilde{m})$.%
} }

\end{itemize}

Then, equations (\ref{eq:IntroConstraints})--(\ref{eq:IntroMassEquation}),
combined with the conservation laws 
\begin{align*}
\partial_{u}\bar{\text{\textgreek{t}}} & =0,\\
\partial_{v}\text{\textgreek{t}} & =0,
\end{align*}
 yield the following relations for the change in the mass difference
associated to $\text{\textgreek{z}}_{in}$ and $\text{\textgreek{z}}_{out}$
after their intersection: 
\begin{equation}
\overline{\mathfrak{D}}_{+}\tilde{m}=\overline{\mathfrak{D}}_{-}\tilde{m}\cdot\exp\Big(\frac{2}{r}\frac{\mathfrak{D}_{-}\tilde{m}}{1-\frac{2m_{2}}{r}}+\mathfrak{Err}_{in}\Big)\label{eq:MassDifferenceIncreaseIngoing}
\end{equation}
and 
\begin{equation}
\mathfrak{D}_{+}\tilde{m}=\mathfrak{D}_{-}\tilde{m}\cdot\exp\Big(-\frac{2}{r}\frac{\overline{\mathfrak{D}}_{-}\tilde{m}}{1-\frac{2m_{2}}{r}}+\mathfrak{Err}_{out}\Big),\label{eq:MassDecreaseOutgoing}
\end{equation}
where the error terms $\mathfrak{Err}_{in},\mathfrak{Err}_{out}$
are negligible compared to the other terms in (\ref{eq:MassDifferenceIncreaseIngoing}),
(\ref{eq:MassDecreaseOutgoing}) (see also the relations (\ref{eq:MassDifferenceIncreaseInUDirection})
and (\ref{eq:MassDifferenceDecreaseInVDirection}) in Section \ref{sub:Proof-of-Proposition}).
In particular, whenever an ingoing and an outgoing null dust beam
intersect, \emph{the mass difference of the ingoing beam increases,
while that of the outgoing beam decreases}. 
\begin{rem*}
Notice that, according to (\ref{eq:MassDifferenceIncreaseIngoing})
and (\ref{eq:MassDecreaseOutgoing}), the change in the mass difference
of each of the beams $\text{\textgreek{z}}_{in},\text{\textgreek{z}}_{out}$
after their intersection can be estimated in terms of the mass difference
of the other beam and the value of $r$ and $\inf(1-\frac{2m}{r})$
in the region of intersection. A relation for the change of the mass
difference of two infinitely thin, intersecting null dust beams was
also obtained in \cite{PoissonIsrael1990}.
\end{rem*}

\subsubsection*{The instability mechanism}

Let us now consider, among the null dust beams arising from the initial
data $(r,\text{\textgreek{W}}^{2},\text{\textgreek{t}},\bar{\text{\textgreek{t}}})^{(\text{\textgreek{e}})}|_{u=0}$,
the beam $\text{\textgreek{z}}_{0}$ which initially lies to the future
of the rest (this is the beam marked with a red dashed line in Figure
1.3). Denoting 
\begin{equation}
\mathcal{E}_{\text{\textgreek{z}}_{0}}[t_{*}]\doteq\mbox{ mass difference associated to }\text{\textgreek{z}}_{0}\mbox{ at }\text{\textgreek{z}}_{0}\cap\{u+v=t_{*}\},\label{eq:MassDifferenceIntro}
\end{equation}
 we will examine how $\mathcal{E}_{\text{\textgreek{z}}_{0}}$ changes
along $\text{\textgreek{z}}_{0}$, after each successive intersection
of $\text{\textgreek{z}}_{0}$ with the rest of the beams:

\begin{enumerate}

\item Starting from $u=0$ up to the first reflection of $\text{\textgreek{z}}_{0}$
off the inner mirror $\text{\textgreek{g}}_{0}$, the beam $\text{\textgreek{z}}_{0}$
is ingoing and intersects all the other beams \emph{after} they are
reflected off $\text{\textgreek{g}}_{0}$. Thus, applying (\ref{eq:MassDifferenceIncreaseIngoing})
successively at each intersection of $\text{\textgreek{z}}_{0}$ with
an outgoing beam, we infer that $\mathcal{E}_{\text{\textgreek{z}}_{0}}$
\underline{increases} at this step by a multiplicative factor 
\begin{equation}
A_{\text{\textgreek{g}}_{0}}\ge\exp\Bigg(\frac{2\big(\tilde{m}^{(\text{\textgreek{e}})}|_{\mathcal{I}}-\mathcal{E}_{\text{\textgreek{z}}_{0}}|_{u=0}\big)}{r_{\text{\textgreek{g}}_{0}}}(1-\text{\textgreek{e}})\Bigg),\label{eq:IncreaseFactor}
\end{equation}
where $r_{\text{\textgreek{g}}_{0}}$ is the value of $r$ at the
region of intersection of $\text{\textgreek{z}}_{0}$ with the first
beam which is reflected off $\{r=r_{0\text{\textgreek{e}}}\}$ (note
that $r_{\text{\textgreek{g}}_{0}}$ is also the $r$-width of the
bundle of beams when $\text{\textgreek{z}}_{0}$ first reaches the
mirror $\text{\textgreek{g}}_{0}$). In obtaining (\ref{eq:IncreaseFactor}),
we have assumed that $r_{0\text{\textgreek{e}}}\ll r_{\text{\textgreek{g}}_{0}}\ll(-\Lambda)^{-\frac{1}{2}}$,
$\tilde{m}^{(\text{\textgreek{e}})}|_{\mathcal{I}}\sim r_{0\text{\textgreek{e}}}$
and $\mathcal{E}_{\text{\textgreek{z}}_{0}}|_{u=0}\ll\tilde{m}^{(\text{\textgreek{e}})}|_{\mathcal{I}}$
(which holds in view of the way the initial data where chosen). 

\item The mass difference $\mathcal{E}_{\text{\textgreek{z}}_{0}}$
right before and right after the reflection of $\text{\textgreek{z}}_{0}$
off $\text{\textgreek{g}}_{0}$ is the same, in view of the reflecting
boundary conditions on $\text{\textgreek{g}}_{0}$.

\item From its first reflection off $\text{\textgreek{g}}_{0}$ up
to its first reflection off $\mathcal{I}$, the beam $\text{\textgreek{z}}_{0}$
is outgoing and intersects (again) the rest of the beams in the region
close to $\mathcal{I}$ (\emph{after} these beams are reflected off
$\mathcal{I}$). Applying (\ref{eq:MassDecreaseOutgoing}) successively
at each intersection, we infer that $\mathcal{E}_{\text{\textgreek{z}}_{0}}$
\underline{decreases} at this step, being multiplied by a factor
\begin{equation}
1>A_{out}\ge\exp\Bigg(-\frac{2\big(\tilde{m}^{(\text{\textgreek{e}})}|_{\mathcal{I}}-\mathcal{E}_{\text{\textgreek{z}}_{0}}|_{u=0}\big)}{r_{\mathcal{I}}}\Big(\frac{1}{(1-\text{\textgreek{e}}-\frac{1}{3}\Lambda r_{\mathcal{I}}^{2})}+\text{\textgreek{e}}\Big)\Bigg),\label{eq:DecreaseFactor}
\end{equation}
where $r_{\mathcal{I}}$ is the value of $r$ at the region of intersection
of $\text{\textgreek{z}}_{0}$ with the first beam which is reflected
off $\mathcal{I}$. In obtaining (\ref{eq:DecreaseFactor}), we have
asumed that $r_{\mathcal{I}}\gg(-\Lambda)^{-\frac{1}{2}}$ (which
holds in view of the way the initial data where chosen). 

\item  The mass difference $\mathcal{E}_{\text{\textgreek{z}}_{0}}$
right before and right after the reflection of $\text{\textgreek{z}}_{0}$
off $\mathcal{I}$ is the same, in view of the reflecting boundary
conditions on $\mathcal{I}$.

\end{enumerate}

Therefore, provided $r_{\text{\textgreek{g}}_{0}}\ll(-\Lambda)^{-\frac{1}{2}}\ll r_{\mathcal{I}}$,
we infer that, after the first reflection of $\text{\textgreek{z}}_{0}$
off $\text{\textgreek{g}}_{0}$ and $\mathcal{I}$, the mass difference
$\mathcal{E}_{\text{\textgreek{z}}_{0}}$ increases by a factor 
\begin{align}
A_{tot}=A_{in}\cdot A_{out}\ge & \exp\Bigg(\frac{2\big(\tilde{m}^{(\text{\textgreek{e}})}|_{\mathcal{I}}-\mathcal{E}_{\text{\textgreek{z}}_{0}}|_{u=0}\big)}{r_{\text{\textgreek{g}}_{0}}}(1-\text{\textgreek{e}})-\frac{2\big(\tilde{m}^{(\text{\textgreek{e}})}|_{\mathcal{I}}-\mathcal{E}_{\text{\textgreek{z}}_{0}}|_{u=0}\big)}{r_{\mathcal{I}}}\Big(\frac{1}{(1-\text{\textgreek{e}}-\frac{1}{3}\Lambda r_{\mathcal{I}}^{2})}+\text{\textgreek{e}}\Big)\Bigg)\label{eq:TotalInceaseFactorIntro}\\
 & \ge\exp\Big(\frac{\tilde{m}^{(\text{\textgreek{e}})}|_{\mathcal{I}}-\mathcal{E}_{\text{\textgreek{z}}_{0}}|_{u=0}}{r_{\text{\textgreek{g}}_{0}}}\Big).\nonumber 
\end{align}
The steps 1--4 in the above procedure can then be repeated for each
successive reflection of $\text{\textgreek{z}}_{0}$ off $\text{\textgreek{g}}_{0}$
and $\mathcal{I}$, as long as 
\begin{equation}
r_{0\text{\textgreek{e}}}\ll r_{\text{\textgreek{g}}_{0};n}\ll(-\Lambda)^{-\frac{1}{2}}\ll r_{\mathcal{I};n},\label{eq:BoundForNearAndAwayRegionIntro}
\end{equation}
where $r_{\text{\textgreek{g}}_{0};n},r_{\mathcal{I};n}$ are the
values of $r_{\text{\textgreek{g}}_{0}},r_{\mathcal{I};n}$ after
the $n$-th reflection of $\text{\textgreek{z}}_{0}$ on $\text{\textgreek{g}}_{0}$
and $\mathcal{I}$ (note that $r_{\text{\textgreek{g}}_{0};n}$ is
also the $r$-width of the bundle of beams at the $n$-th reflection
of $\text{\textgreek{z}}_{0}$ off $\text{\textgreek{g}}_{0}$). Thus,
as long as (\ref{eq:BoundForNearAndAwayRegionIntro}) holds, denoting
with $\mathcal{E}_{\text{\textgreek{z}}_{0};n}$ the value of $\mathcal{E}_{\text{\textgreek{z}}_{0}}$
at the $n$-th reflection of $\text{\textgreek{z}}_{0}$ off $\text{\textgreek{g}}_{0}$,
the following inductive bound holds: 
\begin{equation}
\mathcal{E}_{\text{\textgreek{z}}_{0};n}\ge A_{tot;n}\cdot\mathcal{E}_{\text{\textgreek{z}}_{0};n-1},\label{eq:InductiveEnIntro}
\end{equation}
where the multiplicative factor 
\begin{equation}
A_{tot;n}\doteq\exp\Big(\frac{\tilde{m}^{(\text{\textgreek{e}})}|_{\mathcal{I}}-\mathcal{E}_{\text{\textgreek{z}}_{0};n}}{r_{\text{\textgreek{g}}_{0};n}}\Big)\label{eq:MultiplicativeFactor}
\end{equation}
is always greater than $1$, since $\mathcal{E}_{\text{\textgreek{z}}_{0};n}<\tilde{m}^{(\text{\textgreek{e}})}|_{\mathcal{I}}$
(see also the relation (\ref{eq:BoundForMassIncrease}) in Section
\ref{sub:Inductive-bounds}).\emph{ This is the main mechanism driving
the instability, }and the proof of Theorem \ref{thm:TheoremDetailedIntro}
is aimed at showing that, for some large enough $n(\text{\textgreek{e}})$
depending on $\text{\textgreek{e}}$, 
\begin{equation}
\prod_{n=0}^{n(\text{\textgreek{e}})}A_{tot;n}>\frac{r_{0\text{\textgreek{e}}}}{2\mathcal{E}_{\text{\textgreek{z}}_{0}}|_{u=0}}.\label{eq:BigEnoughFactor}
\end{equation}
Inequality (\ref{eq:BigEnoughFactor}) implies (in view of (\ref{eq:InductiveEnIntro}))
that 
\begin{equation}
\frac{2\mathcal{E}_{\text{\textgreek{z}}_{0};n(\text{\textgreek{e}})}}{r_{0\text{\textgreek{e}}}}>1,\label{eq:BoundForTrappedSurfaceFormationIntro}
\end{equation}
 i.\,e.~that, after the $n(\text{\textgreek{e}})$'th successive
reflection of $\text{\textgreek{z}}_{0}$ on $\text{\textgreek{g}}_{0}$
and $\mathcal{I}$, the mass difference $\mathcal{E}_{\text{\textgreek{z}}_{0}}$
has become so large that a trapped surface (in particular, a point
where $\frac{2m}{r}>1$) necessarily forms before $\text{\textgreek{z}}_{0}$
reaches the mirror $\text{\textgreek{g}}_{0}=\{r=r_{0\text{\textgreek{e}}}\}$
for the $n(\text{\textgreek{e}})+1$-th time (provided $\text{\textgreek{z}}_{0}$
was initially chosen sufficiently ``narrow'').%
\footnote{We should remark that, once a trapped surface $\mathcal{S}$ has formed,
$\{r=r_{0\text{\textgreek{e}}}\}\cap J^{+}(\mathcal{S})$ (where $J^{+}(\mathcal{S})$
is the future of $\mathcal{S}$) will be spacelike and we will not
study the evolution of the spacetime beyond $\{r=r_{0\text{\textgreek{e}}}\}\cap J^{+}(\mathcal{S})$.
In particular, no more reflections of the beams will occur in the
future of $\mathcal{S}$. See Theorem  \ref{thm:MaximalDevelopmentIntro}.%
}

\subsubsection*{Control of $r_{\text{\textgreek{g}}_{0};n}$ and the final step before
trapped surface formation}

The main obstacle to establishing (\ref{eq:BigEnoughFactor}) (and,
thus, Theorem \ref{thm:TheoremDetailedIntro}) is the following: Once
$\mathcal{E}_{\text{\textgreek{z}}_{0}}$ exceeds $c\cdot r_{0\text{\textgreek{e}}}$
for some fixed (small) $c>0$, the total $r$-width of the bundle
of beams close to $\text{\textgreek{g}}_{0}$, i.\,e.~$r_{\text{\textgreek{g}}_{0};n}$
in (\ref{eq:MultiplicativeFactor}), increases after each successive
reflection off $\text{\textgreek{g}}_{0}$ and $\mathcal{I}$. Thus,
the multiplicative factor (\ref{eq:MultiplicativeFactor}) decreases
as $n$ grows. The increase in $r_{\text{\textgreek{g}}_{0};n}$ is
more dramatic when the spacetime is close to having a trapped surface,
i.\,e.~when $\frac{2m}{r}$ is close to $1$.%
\footnote{The example of two outgoing null rays in the exterior of Schwarzschild--AdS,
with mass $M\ll(-\Lambda)^{-\frac{1}{2}}$, serves to illustrate this
phenomenon: The $r$-separation of two rays emanating from the region
close to the future event horizon $\mathcal{H}^{+}$, where $\frac{2m}{r}\sim1$,
increases dramatically by the time they reach the region $r\sim(-\Lambda)^{-\frac{1}{2}}$.
This is, of course, nothing other than the celebrated \emph{red-shift}
effect.%
}

Controlling the growth of $r_{\text{\textgreek{g}}_{0};n}$ is achieved
by establishing an inductive bound of the following form: 
\begin{equation}
r_{\text{\textgreek{g}}_{0};n}\le r_{\text{\textgreek{g}}_{0};n-1}\cdot\Big(1+C_{0}\frac{r_{0\text{\textgreek{e}}}}{r_{\text{\textgreek{g}}_{0};n-1}}\big(\Big|\log\big(1-\frac{2\mathcal{E}_{\text{\textgreek{z}}_{0};n-1}}{r_{0\text{\textgreek{e}}}}\big)\Big|+1\big)\Big)\label{eq:InductiveBoundRinIntro}
\end{equation}
(see also the relation (\ref{eq:BoundForMaxBeamSeparation}) in Section
\ref{sub:Inductive-bounds}). Obtaining the bound (\ref{eq:InductiveBoundRinIntro})
is one of the most demanding parts in the proof of Theorem \ref{thm:TheoremDetailedIntro}
and requires controlling the $r$-distance $r_{\text{\textgreek{g}}_{0};n}^{(1)}$
of $\text{\textgreek{z}}_{0}$ from the second-to-top beam $\text{\textgreek{z}}_{1}$
at the $n$-th reflection off $\text{\textgreek{g}}_{0}$ for all
$n\le n(\text{\textgreek{e}})$, i.\,e.~establish a bound of the
form 
\begin{equation}
\frac{r_{\text{\textgreek{g}}_{0};n}^{(1)}}{r_{0\text{\textgreek{e}}}}\ge1+c_{0}(\mathcal{E}_{\text{\textgreek{z}}_{0};0}/r_{0\text{\textgreek{e}}}).\label{eq:LowerBoundSecondBeamIntro}
\end{equation}
 (see (\ref{eq:BoundSecondBeamchanged}) in Section \ref{sub:Inductive-bounds}).
The bound (\ref{eq:LowerBoundSecondBeamIntro}) is in turn obtained
by establishing an inductive bound of the form 
\begin{equation}
\log\Big(\frac{r_{\text{\textgreek{g}}_{0};n-1}^{(1)}}{r_{\text{\textgreek{g}}_{0};n}^{(1)}}\Big)\le C_{0}\log\Big(\frac{\mathcal{E}_{\text{\textgreek{z}}_{0};n}}{\mathcal{E}_{\text{\textgreek{z}}_{0};n-1}}\Big),\label{eq:InductiveBoundSecondBeam}
\end{equation}
estimating the decrease of $r_{\text{\textgreek{g}}_{0};n}^{(1)}$
by the increase of $\mathcal{E}_{\text{\textgreek{z}}_{0};n}$ at
each reflection (see (\ref{eq:UsefulBoundAlmostThere}) in Section
\ref{sub:Proof-of-Proposition}). The bound (\ref{eq:LowerBoundSecondBeamIntro})
is inferred from (\ref{eq:InductiveBoundSecondBeam}), in view of
the fact that $\mathcal{E}_{\text{\textgreek{z}}_{0};n}\ge\mathcal{E}_{\text{\textgreek{z}}_{0};n-1}$
and 
\begin{equation}
\sum_{n=1}^{n(\text{\textgreek{e}})}\log\Big(\frac{\mathcal{E}_{\text{\textgreek{z}}_{0};n}}{\mathcal{E}_{\text{\textgreek{z}}_{0};n-1}}\Big)=\log\Big(\frac{\mathcal{E}_{\text{\textgreek{z}}_{0};n(\text{\textgreek{e}})}}{\mathcal{E}_{\text{\textgreek{z}}_{0};0}}\Big)\le\log\Big(\frac{r_{0\text{\textgreek{e}}}}{2\mathcal{E}_{\text{\textgreek{z}}_{0};0}}\Big).
\end{equation}
At the level of the initial data, obtaining (\ref{eq:InductiveBoundRinIntro})
and (\ref{eq:InductiveBoundSecondBeam}) requires introducing a certain
hierarchy for the scales of the $r$-distances and mass differences
associated to the beams initially (see (\ref{eq:h_1_h_0_definition})
and (\ref{eq:h_2definition}) in Section \ref{sub:Parameters-and-auxiliary}).

Combining (\ref{eq:InductiveEnIntro}) and (\ref{eq:InductiveBoundRinIntro}),
we can show that there exists a large $n(\text{\textgreek{e}})$ such
that, after $n(\text{\textgreek{e}})$ reflections of $\text{\textgreek{z}}_{0}$
off $\text{\textgreek{g}}_{0}$ (but not earlier!), we have 
\begin{equation}
\frac{2\mathcal{E}_{\text{\textgreek{z}}_{0};n(\text{\textgreek{e}})}}{r_{0\text{\textgreek{e}}}}>1-c(\text{\textgreek{e}}),\label{eq:AlmostTrappdIntro}
\end{equation}
where $c(\text{\textgreek{e}})\ll h(\text{\textgreek{e}})$ is a fixed
function of $\text{\textgreek{e}}$. Note that, compared to (\ref{eq:BoundForTrappedSurfaceFormationIntro}),
(\ref{eq:AlmostTrappdIntro}) is a slightly weaker bound, which just
stops short of implying that a trapped surface is formed. In order
to complete the proof of Theorem \ref{thm:TheoremDetailedIntro},
we therefore have to consider two different scenarios for $\mathcal{E}_{\text{\textgreek{z}}_{0};n(\text{\textgreek{e}})}$:
\begin{casenv}
\item In the case when (\ref{eq:BoundForTrappedSurfaceFormationIntro})
holds, the proof of Theorem \ref{thm:TheoremDetailedIntro} follows
readily, since (\ref{eq:BoundForTrappedSurfaceFormationIntro}) implies
that, before $\text{\textgreek{z}}_{0}$ reaches $\{r=r_{0\text{\textgreek{e}}}\}$
for the $n(\text{\textgreek{e}})+1$-th time, a point arises where
$\frac{2m}{r}>1$.
\item In the case when (\ref{eq:AlmostTrappdIntro}) holds but (\ref{eq:BoundForTrappedSurfaceFormationIntro})
is violated, we can bound 
\begin{equation}
1-c(\text{\textgreek{e}})<\frac{2\mathcal{E}_{\text{\textgreek{z}}_{0};n(\text{\textgreek{e}})}}{r_{0\text{\textgreek{e}}}}\le1.\label{eq:SecondScenario}
\end{equation}
In this case, $\text{\textgreek{z}}_{0}$ reaches $\{r=r_{0\text{\textgreek{e}}}\}$
for the $n(\text{\textgreek{e}})+1$-th time before a trapped surface
has formed. One would be tempted to repeat the above procedure for
one more reflection, in an attempt to establish that a trapped surface
has formed before the $n(\text{\textgreek{e}})+2$-th reflection of
$\text{\textgreek{z}}_{0}$ off $\text{\textgreek{g}}_{0}$. However,
the bound (\ref{eq:SecondScenario}) implies that most of the bootstrap
assumptions needed for the proof of Theorem \ref{thm:TheoremDetailedIntro}
(which we have supressed in this sketch for the sake of simplicity)
are violated beyond the $n(\text{\textgreek{e}})+1$-th reflection
and, thus, the above procedure can not be repeated. For this reason,
we choose a different path: Applying a Cauchy stability statement
backwards in time (see Theorem \ref{prop:CauchyStability}), we show
that there exists a small perturbation $(r',(\text{\textgreek{W}}^{\prime})^{2},\text{\textgreek{t}}',\bar{\text{\textgreek{t}}}')^{(\text{\textgreek{e}})}|_{u=0}$
of the initial data $(r,\text{\textgreek{W}}^{2},\text{\textgreek{t}},\bar{\text{\textgreek{t}}})^{(\text{\textgreek{e}})}|_{u=0}$(satisfying
(\ref{eq:InitialDataNorm})), such that the perturbed solution $(r',(\text{\textgreek{W}}^{\prime})^{2},\text{\textgreek{t}}',\bar{\text{\textgreek{t}}}')^{(\text{\textgreek{e}})}$
to (\ref{eq:EinsteinNullDust}) satisfies (\ref{eq:BoundForTrappedSurfaceFormationIntro})
and, furthermore, 
\begin{equation}
\frac{2\mathcal{E}_{\text{\textgreek{z}}_{0};n(\text{\textgreek{e}})}^{\prime}}{r_{0\text{\textgreek{e}}}}>1
\end{equation}
(where $\mathcal{E}{}_{\text{\textgreek{z}}_{0}}^{\prime}$ is similarly
defined by the relation \ref{eq:MassDifferenceIntro} for $(r',(\text{\textgreek{W}}^{\prime})^{2},\text{\textgreek{t}}',\bar{\text{\textgreek{t}}}')^{(\text{\textgreek{e}})}$
in place of $(r,\text{\textgreek{W}}^{2},\text{\textgreek{t}},\bar{\text{\textgreek{t}}})^{(\text{\textgreek{e}})}$).
Thus, we end up in the scenario of Case 1, and the proof of Theorem
\ref{thm:TheoremDetailedIntro} follows readily.
\end{casenv}

\subsubsection*{Further remarks on the proof of Theorem \ref{thm:TheoremDetailedIntro}}

The proof of Theorem \ref{thm:TheoremDetailedIntro} involves many
technical issues related to the final step of the evolution before
a trapped surface is formed. Most of these technical issues simplify
considerably in the case when one restricts to showing a weaker instability
statement for $(\mathcal{M}_{AdS},g_{AdS})$, e.\,g.~by replacing
(\ref{eq:TrappedSurfaceIntro}) with 
\begin{equation}
\big(1-\frac{2m}{r}\big)^{(\text{\textgreek{e}})}\big|_{(u_{\text{\textgreek{e}}},v_{\text{\textgreek{e}}})}<\frac{1}{2}.
\end{equation}
See Sections \ref{sec:The-main-result:Details} and \ref{sub:Remark-on-Proposition}
for more details.

The mechanism leading to trapped surface formation in the proof of
Theorem \ref{thm:TheoremDetailedIntro} only made use of the fact
that we chose the initial data $(r,\text{\textgreek{W}}^{2},\text{\textgreek{t}},\bar{\text{\textgreek{t}}})^{(\text{\textgreek{e}})}|_{u=0}$
so that the matter was supported in narrow null beams, successively
reflected off $\text{\textgreek{g}}_{0}$ and $\mathcal{I}$, while
the matter model satisfied the condition 
\begin{equation}
T_{uv}=\text{\textgreek{W}}^{2}g^{AB}T_{AB}=0.
\end{equation}
Thus, we expect that the same mechanism can be adapted to the case
of more general matter fields, which allow for matter to be arranged
into narrow and suficiently localised null beams, satisfying (in a
region around the set of intersection of the beams) 
\begin{equation}
T_{uv},|\text{\textgreek{W}}^{2}g^{AB}T_{AB}|\ll T_{uu}+T_{vv},
\end{equation}
with such a configuration arising moreover from inital data which
are small in a norm for which $(\mathcal{M}_{AdS},g_{AdS})$ is Cauchy
stable. For an application of this mechanism in the case of the spherically
symmetric Einstein--massless Vlasov system (without reducing to the
radial case and without an inner mirror), see our forthcoming \cite{MoschidisVlasov}.

Finally, let us remark that the general mechanism of instability suggested
by the proof of Theorem \ref{thm:TheoremDetailedIntro} can be summarized
as follows: In a configuration consisting of a relatively narrow bundle
of nearly-null beams of matter that are successively reflected on
$\mathcal{I}$ and $r=0$ (on an approximately $(\mathcal{M}_{AdS},g_{AdS})$
background), the energy content of the ``top'' beam will increase
after each pair of reflections. A similar physical space mechanism
was described for the Einstein--scalar field system (\ref{eq:EinsteinScalarFieldInDoubleNull})
in \cite{DimitrakopoulosEtAl}, where it was suggested that, on a
nearly-null scalar field beam successively reflected off $\mathcal{I}$
and the center $r=0$, the energy density on the top part of the beam
tends to increase.

\subsection{Outline of the paper}

This paper is organised as follows:

In Section \ref{sec:The-Einstein--Vlasov-system}, we will introduce
the spherically symmetric Einstein--radial massless Vlasov system
in double null coordinates. We will also formulate the notion of reflecting
boundary conditions for this system on $\mathcal{I}$ and on timelike
hypersurfaces of the form $\{r=r_{0}\}$.

In Section \ref{sec:ResultsFromTheOtherPaper}, we will formulate
the asymptotically AdS characteristic initial-boundary value problem
for the spherically symmetric Einstein--radial massless Vlasov system.
We will then recall the main results established in \cite{MoschidisMaximalDevelopment}
regarding the structure of the maximal development and the Cauchy
stability properties for this system.

In Section \ref{sec:The-main-result:Details}, we will provide a technical
statement of the main result of this paper, namely the instability
of AdS for the Einstein--radial massless Vlasov system with reflecting
boundary conditions on $\{r=r_{0}\}$ and $\mathcal{I}$. The proof
of this result will occupy Sections \ref{sec:Preliminary-constructions}
and \ref{sec:Proof}.

\subsection{Acknowledgements}

I would like to thank my advisor Mihalis Dafermos for suggesting this
problem to me, as well for numerous fruitful discussions and crucial
suggestions. I would also like to thank Igor Rodnianski for many additional
comments and suggestions. This work was completed while the author
was a visitor at DPMMS and King's College, Cambridge.

\section{\label{sec:The-Einstein--Vlasov-system}The Einstein--massless Vlasov
system in spherical symmetry}

In this Section, we will review the basic properties of the spherically
symmetric Einstein--massless Vlasov system in $3+1$ dimensions, expressed
in double null coordinates, following the conventions introduced in
\cite{DafermosRendall}. We will also introduce the notion of the
reflecting boundary condition on timelike hypersurfaces for the radial
massless Vlasov equation. To this end, we will follow the conventions
adopted in our companion paper \cite{MoschidisMaximalDevelopment}.

\subsection{\label{sub:Spherically-symmetric-spacetimes}Spherically symmetric
spacetimes in double null coordinates}

Let $(\mathcal{M}^{3+1},g)$ be a smooth Lorentzian manifold, such
that $\mathcal{M}$ is of the form 
\begin{equation}
\mathcal{M}\simeq\mathcal{U}\times\mathbb{S}^{2}\label{eq:SphericallySymmetricmanifold}
\end{equation}
where $\mathcal{U}$ is an open domain of $\mathbb{R}^{2}$ with piecewise
Lipschitz boundary $\partial\mathcal{U}$ and, in the standard $(u,v)$
coordinates on $\mathcal{U}$, $g$ takes the form 
\begin{equation}
g=-\text{\textgreek{W}}^{2}(u,v)dudv+r^{2}(u,v)g_{\mathbb{S}^{2}},\label{eq:SphericallySymmetricMetric}
\end{equation}
where $g_{\mathbb{S}^{2}}$ is the standard round metric on $\mathbb{S}^{2}$
and $\text{\textgreek{W}},r:\mathcal{U}\rightarrow(0,+\infty)$ are
smooth functions. In addition, we will assume that 
\begin{equation}
\inf_{\mathcal{U}}r>0.
\end{equation}
We will also fix a time orientation on $\mathcal{M}$ by requiring
that the timelike vector field $N=\partial_{u}+\partial_{v}$ is future
directed.
\begin{rem*}
Notice that the action of $SO(3)$ on $(\mathcal{M},g)$ through rotations
of the $\mathbb{S}^{2}$ factor of (\ref{eq:SphericallySymmetricmanifold})
is an isometric action.
\end{rem*}
We will also define the \emph{Hawking mass} $m:\mathcal{M}\rightarrow\mathbb{R}$
by the expression 
\begin{equation}
m=\frac{r}{2}\big(1-g(\nabla r,\nabla r)\big).
\end{equation}
 Viewed as a function on $\mathcal{U}$, $m$ takes the form: 
\begin{equation}
m=\frac{r}{2}\big(1+4\text{\textgreek{W}}^{-2}\partial_{u}r\partial_{v}r\big).\label{eq:DefinitionHawkingMass}
\end{equation}
Equivalently, we have 
\begin{equation}
\text{\textgreek{W}}^{2}=4\frac{(-\partial_{u}r)\partial_{v}r}{1-\frac{2m}{r}}.\label{eq:RelationHawkingMass}
\end{equation}

In any local coordinate chart $(y^{1},y^{2})$ on $\mathbb{S}^{2}$,
the non-zero Christoffel symbols of (\ref{eq:SphericallySymmetricMetric})
in the $(u,v,y^{1},y^{2})$ local coordinate chart on $\mathcal{M}$
are computed as follows: 
\begin{gather}
\text{\textgreek{G}}_{uu}^{u}=\partial_{u}\log(\text{\textgreek{W}}^{2}),\hphantom{A}\text{\textgreek{G}}_{vv}^{v}=\partial_{u}\log(\text{\textgreek{W}}^{2}),\label{eq:ChristoffelSymbols}\\
\text{\textgreek{G}}_{AB}^{u}=\text{\textgreek{W}}^{-2}\partial_{v}(r^{2})(g_{\mathbb{S}^{2}})_{AB},\hphantom{\,}\text{\textgreek{G}}_{AB}^{v}=\text{\textgreek{W}}^{-2}\partial_{u}(r^{2})(g_{\mathbb{S}^{2}})_{AB},\nonumber \\
\text{\textgreek{G}}_{uB}^{A}=r^{-1}\partial_{u}r\text{\textgreek{d}}_{B}^{A},\hphantom{\,}\text{\textgreek{G}}_{vB}^{A}=r^{-1}\partial_{v}r\text{\textgreek{d}}_{B}^{A},\nonumber \\
\text{\textgreek{G}}_{BC}^{A}=(\text{\textgreek{G}}_{\mathbb{S}^{2}})_{BC}^{A},\nonumber 
\end{gather}
where the latin indices $A,B,C$ are associated to the spherical coordinates
$y^{1},y^{2}$, $\text{\textgreek{d}}_{B}^{A}$ is Kronecker delta
and $\text{\textgreek{G}}_{\mathbb{S}^{2}}$ are the Christoffel symbols
of the round sphere in the $(y^{1},y^{2})$ coordinate chart.

For any pair of smooth functions $f_{1},f_{2}:\mathbb{R}\rightarrow\mathbb{R}$
with $f_{1}^{\prime},f_{2}^{\prime}\neq0$, the coordinate transformation
\begin{equation}
(\bar{u},\bar{v})=(f_{1}(u),f_{2}(v)),\label{eq:GeneralCoordinateTransformation}
\end{equation}
mapping $\mathcal{U}$ to $\bar{\mathcal{U}}\subset\mathbb{R}^{2}$,
can be used to diffeomorphically identify $\mathcal{M}$ with $\bar{\mathcal{U}}\times\mathbb{S}^{2}$.
In these new coordinates, the metric $g$ takes the form 
\begin{equation}
g=-\bar{\text{\textgreek{W}}}^{2}(\bar{u},\bar{v})d\bar{u}d\bar{v}+r^{2}(\bar{u},\bar{v})g_{\mathbb{S}^{2}},\label{eq:SphericallySymmetricMetricNewGauge}
\end{equation}
where 
\begin{gather}
\bar{\text{\textgreek{W}}}^{2}(\bar{u},\bar{v})=\frac{1}{f_{1}^{\prime}f_{2}^{\prime}}\text{\textgreek{W}}^{2}(f_{1}^{-1}(\bar{u}),f_{2}^{-1}(\bar{v})),\label{eq:NewOmega}\\
r(\bar{u},\bar{v})=r(f_{1}^{-1}(\bar{u}),f_{2}^{-1}(\bar{v})).\label{eq:NewR}
\end{gather}
We will frequently make use of such coordinate transformations, without
renaming the coordinates each time. 

Note that $m$ is invariant under coordinate transformations of the
form $(u,v)\rightarrow(f_{1}(u),f_{2}(v))$, i.\,e.
\begin{equation}
m(\bar{u},\bar{v})=m(f_{1}^{-1}(\bar{u}),f_{2}^{-1}(\bar{v})).
\end{equation}

\subsection{\label{sub:VlasovEquations}The radial massless Vlasov equation}

Let $(\mathcal{M},g)$ be as in Section \ref{sub:Spherically-symmetric-spacetimes}.
Let $f\ge0$ be a measure on $T\mathcal{M}$ which is constant along
the geodesic flow, that is to say, in any local coordinate chart $(x^{0},x^{1},x^{2},x^{3})$
on $\mathcal{M}$ with associated momentum coordinates $(p^{0},p^{1},p^{2},p^{3})$
on the fibers of $T\mathcal{M}$, $f$ satisfies (as a distribution)
the first order equation 
\begin{equation}
p^{\text{\textgreek{a}}}\partial_{x^{\text{\textgreek{a}}}}f-\text{\textgreek{G}}_{\text{\textgreek{b}\textgreek{g}}}^{\text{\textgreek{a}}}p^{\text{\textgreek{b}}}p^{\text{\textgreek{g}}}\partial_{p^{\text{\textgreek{a}}}}f=0,\label{eq:VlasovEquation}
\end{equation}
where $\text{\textgreek{G}}_{\text{\textgreek{b}\textgreek{g}}}^{\text{\textgreek{a}}}$
are the Christoffel symbols of $g$ in the chart $(x^{0},x^{1},x^{2},x^{3})$.
We will call $f$ a \emph{massless Vlasov field} if it is supported
on the set $P\subset T\mathcal{M}$ of null vectors, i.\,e.~on the
set 
\begin{equation}
g_{\text{\textgreek{a}\textgreek{b}}}(x)p^{\text{\textgreek{a}}}p^{\text{\textgreek{b}}}=0.
\end{equation}

Associated to $f$ is a symmetric $(0,2)$-form on $\mathcal{M}$
(possibly defined only in the sense of distributions), the \emph{energy
momentum }tensor of $f$, given by the expression 
\begin{equation}
T_{\text{\textgreek{a}\textgreek{b}}}(x)=\int_{\text{\textgreek{p}}^{-1}(x)}p_{\text{\textgreek{a}}}p_{\text{\textgreek{b}}}f,\label{eq:EnergyMomentumTensor}
\end{equation}
where $\text{\textgreek{p}}^{-1}(x)$ denotes the fiber of $T\mathcal{M}$
over $x\in\mathcal{M}$ and the indices of the momentum coordinates
are lowered with the use of the metric $g$, i.\,e.
\begin{equation}
p_{\text{\textgreek{g}}}=g_{\text{\textgreek{g}}\text{\textgreek{d}}}(x)p^{\text{\textgreek{d}}}.
\end{equation}

\begin{rem*}
In this paper, we will only consider distributions $f$ for which
the expression (\ref{eq:EnergyMomentumTensor}) is finite for all
$x\in\mathcal{M}$ and depends smoothly on $x\in\mathcal{M}$. 
\end{rem*}
We will consider only distributions $f$ which are spherically symmetric,
i.\,e.~invariant under the action of $SO(3)$ on $\mathcal{M}$.
In that case, in any $(u,v,y^{1},y^{2})$ local coordinate chart as
in Section \ref{sub:Spherically-symmetric-spacetimes}, the energy-momentum
tensor $T$ is of the form 
\begin{equation}
T=T_{uu}(u,v)du^{2}+2T_{uv}(u,v)dudv+T_{vv}(u,v)dv^{2}+T_{AB}(u,v)dy^{A}dy^{B}.\label{eq:SphericallySymmetricTensor}
\end{equation}
 Furthermore, we will restrict to \emph{radial }Vlasov fields $f$,
i.\,e.~fields supported only on radial null vectors which are normal
to the orbits of the action of $SO(3)$ on $\mathcal{M}$. In any
$(u,v,y^{1},y^{2})$ local coordinate chart as in Section \ref{sub:Spherically-symmetric-spacetimes}
(with associated momentum coordinates $(p^{u},p^{v},p^{1},p^{2})$),
a spherically symmetric, radial massless Vlasov field $f$ has the
form 
\begin{equation}
f(u,v,y^{1},y^{2};p^{u},p^{v},p^{1},p^{2})=\big(\bar{f}_{in}(u,v;p^{u})+\bar{f}_{out}(u,v;p^{v})\big)\text{\textgreek{d}}\big(\sqrt{(g_{\mathbb{S}^{2}})_{AB}p^{A}p^{B}}\big)\text{\textgreek{d}}(\text{\textgreek{W}}^{2}p^{u}p^{v}),\label{eq:RadialVlasovField}
\end{equation}
where $\bar{f}_{in},\bar{f}_{out}\ge0$ and $\text{\textgreek{d}}$
is the Dirac delta funcion on $\mathbb{R}$. In this case, the only
non-zero components of the energy momentum tensor (\ref{eq:EnergyMomentumTensor})
are the $T_{uu}$ and $T_{vv}$ components. In particular, in terms
of $\bar{f}_{in},\bar{f}_{out}$, we (formally) compute that 
\begin{gather}
T_{uu}(u,v)=\int_{0}^{+\infty}\text{\textgreek{W}}^{4}(p^{v})^{2}\bar{f}_{out}(u,v;p^{v})\, r^{2}\frac{dp^{v}}{p^{v}},\label{eq:T_uuComponent}\\
T_{vv}(u,v)=\int_{0}^{+\infty}\text{\textgreek{W}}^{4}(p^{u})^{2}\bar{f}_{in}(u,v;p^{u})\, r^{2}\frac{dp^{u}}{p^{u}}.\label{eq:T_vvComponent}
\end{gather}

\begin{rem*}
In this paper, we will only consider the case when $\bar{f}_{in},\bar{f}_{out}$
are smooth and compactly supported in the $p^{u},p^{v}$ variables,
respectively.
\end{rem*}
In the case when $f$ is of the form (\ref{eq:RadialVlasovField}),
equation (\ref{eq:VlasovEquation}) is equivalent to the following
system for $\bar{f}_{in}$ and $\bar{f}_{out}$:
\begin{gather}
\partial_{u}(\text{\textgreek{W}}^{4}r^{4}p^{u}\bar{f}_{in})+p^{u}\partial_{p^{u}}(\text{\textgreek{W}}^{4}r^{4}p^{u}\bar{f}_{in})=0,\label{eq:IngoingEquation}\\
\partial_{v}(\text{\textgreek{W}}^{4}r^{4}p^{v}\bar{f}_{out})+p^{v}\partial_{p^{v}}(\text{\textgreek{W}}^{4}r^{4}p^{v}\bar{f}_{out})=0.\label{eq:OutgoingEquation}
\end{gather}
The equations (\ref{eq:IngoingEquation})--(\ref{eq:OutgoingEquation})
readily yield the following transport equations for $T_{uu}$, $T_{vv}$:
\begin{gather}
\partial_{v}(r^{2}T_{uu})=0,\label{eq:EquationT_uu}\\
\partial_{u}(r^{2}T_{vv})=0.\label{eq:EquationT_vv}
\end{gather}

\begin{rem*}
Under a coordinate transformation of the form (\ref{eq:GeneralCoordinateTransformation}),
$\bar{f}_{in},\bar{f}_{out}$ transform as 
\begin{equation}
\bar{f}_{in}^{(new)}(f_{1}(u),f_{2}(v);f_{1}^{\prime}(u)p)=\bar{f}_{in}\big(u,v;p\big)\label{eq:NewIngoingVlasov}
\end{equation}
and 
\begin{equation}
\bar{f}_{out}^{(new)}(f_{1}(u),f_{2}(v);f_{2}^{\prime}(v)p)=\bar{f}_{out}\big(u,v;p\big).\label{eq:NewOutgoingVlasov}
\end{equation}

\end{rem*}

\subsection{\label{sub:The-Einstein-equations}The spherically symmetric Einstein--radial
massless Vlasov system}

Let $(\mathcal{M},g)$ be a smooth Lorentzian manifold and let $\text{\textgreek{L}}<0$.
Let also $f$ be a non-negative measure on $T\mathcal{M}$. The \emph{Einstein--Vlasov}
system for $(\mathcal{M},g;f)$ with cosmological constant $\Lambda$
is 
\begin{equation}
\begin{cases}
Ric_{\text{\textgreek{m}\textgreek{n}}}(g)-\frac{1}{2}R(g)g_{\text{\textgreek{m}\textgreek{n}}}+\text{\textgreek{L}}g_{\text{\textgreek{m}\textgreek{n}}}=8\text{\textgreek{p}}T_{\text{\textgreek{m}\textgreek{n}}},\\
p^{\text{\textgreek{a}}}\partial_{x^{\text{\textgreek{a}}}}f-\text{\textgreek{G}}_{\text{\textgreek{b}\textgreek{g}}}^{\text{\textgreek{a}}}p^{\text{\textgreek{b}}}p^{\text{\textgreek{g}}}\partial_{p^{\text{\textgreek{a}}}}f=0,
\end{cases}\label{eq:EinsteinVlasovEquations}
\end{equation}
where $T_{\text{\textgreek{m}\textgreek{n}}}$ is expressed in terms
of $f$ by (\ref{eq:EnergyMomentumTensor}).

Restricting to the case where $(\mathcal{M},g)$ is a spherically
symmetric spacetime as in Section \ref{sub:Spherically-symmetric-spacetimes}
and $f$ is a radial massless Vlasov field (i.\,e.~has the form
(\ref{eq:RadialVlasovField})), the system (\ref{eq:EinsteinVlasovEquations})
is equivalent to the following system for $(r,\text{\textgreek{W}}^{2},\bar{f}_{in},\bar{f}_{out})$:
\begin{align}
\partial_{u}\partial_{v}(r^{2})= & -\frac{1}{2}(1-\Lambda r^{2})\text{\textgreek{W}}^{2},\label{eq:RequationFinal}\\
\partial_{u}\partial_{v}\log(\text{\textgreek{W}}^{2})= & \frac{\text{\textgreek{W}}^{2}}{2r^{2}}\big(1+4\text{\textgreek{W}}^{-2}\partial_{u}r\partial_{v}r\big),\label{eq:OmegaEquationFinal}\\
\partial_{v}(\text{\textgreek{W}}^{-2}\partial_{v}r)= & -4\pi rT_{vv}\text{\textgreek{W}}^{-2},\label{eq:ConstrainVFinal}\\
\partial_{u}(\text{\textgreek{W}}^{-2}\partial_{u}r)= & -4\pi rT_{uu}\text{\textgreek{W}}^{-2},\label{eq:ConstraintUFinal}\\
\partial_{u}(\text{\textgreek{W}}^{4}r^{4}p^{u}\bar{f}_{in})= & -p^{u}\partial_{p^{u}}(\text{\textgreek{W}}^{4}r^{4}p^{u}\bar{f}_{in}),\label{eq:IngoingVlasovFinal}\\
\partial_{v}(\text{\textgreek{W}}^{4}r^{4}p^{v}\bar{f}_{out})= & -p^{v}\partial_{p^{v}}(\text{\textgreek{W}}^{4}r^{4}p^{v}\bar{f}_{out}),\label{eq:OutgoingVlasovFinal}
\end{align}
where $T_{uu},T_{vv}$ are expressed in terms of $\bar{f}_{out},\bar{f}_{in}$
by (\ref{eq:T_uuComponent}), (\ref{eq:T_vvComponent}), respectively.
Notice that the system (\ref{eq:RequationFinal})--(\ref{eq:OutgoingVlasovFinal})
reduces to the following system for $(r,\text{\textgreek{W}}^{2},T_{uu},T_{vv})$:
\begin{align}
\partial_{u}\partial_{v}(r^{2})= & -\frac{1}{2}(1-\Lambda r^{2})\text{\textgreek{W}}^{2},\label{eq:RequationFinal-2}\\
\partial_{u}\partial_{v}\log(\text{\textgreek{W}}^{2})= & \frac{\text{\textgreek{W}}^{2}}{2r^{2}}\big(1+4\text{\textgreek{W}}^{-2}\partial_{u}r\partial_{v}r\big),\label{eq:OmegaEquationFinal-2}\\
\partial_{v}(\text{\textgreek{W}}^{-2}\partial_{v}r)= & -4\pi rT_{vv}\text{\textgreek{W}}^{-2},\label{eq:ConstrainVFinal-1}\\
\partial_{u}(\text{\textgreek{W}}^{-2}\partial_{u}r)= & -4\pi rT_{uu}\text{\textgreek{W}}^{-2},\label{eq:ConstraintUFinal-1}\\
\partial_{u}(r^{2}T_{vv})= & 0,\label{eq:IngoingConservationClosed}\\
\partial_{v}(r^{2}T_{uu})= & 0.\label{eq:OutgoingConservationClosed}
\end{align}

\begin{rem*}
The system (\ref{eq:RequationFinal-2})--(\ref{eq:OutgoingConservationClosed})
is the Einstein--null dust system with both ingoing and outgoing dust
(used as a model for self-gravitating radiation already in \cite{PoissonIsrael1990}).
In the notation of Section \ref{sub:NeedOfAMirror} of the introduction,
\begin{equation}
r^{2}T_{vv}=\bar{\text{\textgreek{t}}}
\end{equation}
 and 
\begin{equation}
r^{2}T_{uu}=\text{\textgreek{t}}.
\end{equation}

\end{rem*}
Defining the renormalised Hawking mass as 
\begin{equation}
\tilde{m}\doteq m-\frac{1}{6}\Lambda r^{3},\label{eq:RenormalisedHawkingMass}
\end{equation}
 and using the relation (\ref{eq:DefinitionHawkingMass}), equations
(\ref{eq:RequationFinal-2})--(\ref{eq:OutgoingConservationClosed})
formally give rise to the following system for $(r,\tilde{m},T_{uu},T_{vv})$
(valid in the region of $\mathcal{U}$ where $\partial_{v}r>0$, $\partial_{u}r<0$
and $1-\frac{2m}{r}>0$):

\begin{align}
\partial_{u}\log\big(\frac{\partial_{v}r}{1-\frac{2m}{r}}\big)= & -4\pi r^{-1}\frac{r^{2}T_{uu}}{-\partial_{u}r},\label{eq:DerivativeInUDirectionKappa}\\
\partial_{v}\log\big(\frac{-\partial_{u}r}{1-\frac{2m}{r}}\big)= & 4\pi r^{-1}\frac{r^{2}T_{vv}}{\partial_{v}r},\label{eq:DerivativeInVDirectionKappaBar}\\
\partial_{u}\partial_{v}r= & -\frac{2\tilde{m}-\frac{2}{3}\Lambda r^{3}}{r^{2}}\frac{(-\partial_{u}r)\partial_{v}r}{1-\frac{2m}{r}},\label{eq:EquationRForProof}\\
\partial_{u}\tilde{m}= & -2\pi\frac{\big(1-\frac{2m}{r}\big)}{-\partial_{u}r}r^{2}T_{uu},\label{eq:DerivativeTildeUMass}\\
\partial_{v}\tilde{m}= & 2\pi\frac{\big(1-\frac{2m}{r}\big)}{\partial_{v}r}r^{2}T_{vv}\label{eq:DerivativeTildeVMass}\\
\partial_{u}(r^{2}T_{vv})= & 0,\label{eq:ConservationT_vv}\\
\partial_{v}(r^{2}T_{uu})= & 0.\label{eq:ConservationT_uu}
\end{align}

\subsection{\label{sub:Reflective-boundary-conditions}The reflecting boundary
condition for the Vlasov equation}

Let $(\mathcal{M},g)$ be as in Section \ref{sub:Spherically-symmetric-spacetimes}.
Recall that $\mathcal{M}$ splits topologically as the product 
\[
\mathcal{M}\simeq\mathcal{U}\times\mathbb{S}^{2}.
\]
 Let $\partial_{tim}\mathcal{U}$ be the subset of the boundary $\partial\mathcal{U}$
of $\mathcal{U}\subset\mathbb{R}^{2}$ consisting of a union of connected,
timelike Lipschitz curves with respect to the comparison metric 
\begin{equation}
g_{comp}=-dudv\label{eq:ComparisonUVMetric}
\end{equation}
 on $\mathbb{R}^{2}$. Recall that a connected Lipschitz curve $\text{\textgreek{g}}$
in $\mathbb{R}^{2}$ is said to be timelike with respect to (\ref{eq:ComparisonUVMetric})
if, for every point $p=(u_{*},v_{*})\in\text{\textgreek{g}}$, we
have 
\begin{equation}
\text{\textgreek{g}}\backslash p\subset I^{+}(p)\cup I^{-}(p)\doteq\big(\{u>u_{*}\}\cap\{v>v_{*}\}\big)\cup\big(\{u<u_{*}\}\cap\{v<v_{*}\}\big).
\end{equation}

\begin{figure}[h] 
\centering 
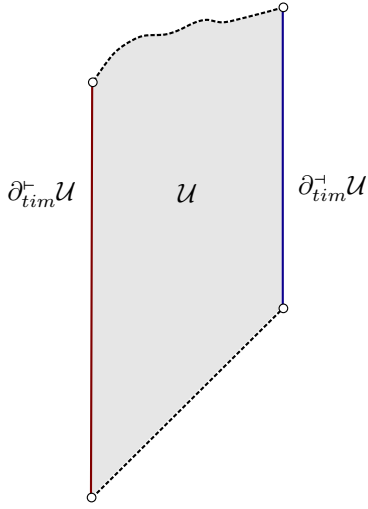 
\caption{For a domain $\mathcal{U}\subset \mathbb{R}^{2}$ as depicted above, the timelike portion $\partial_{tim}\mathcal{U}$  of the boundary $\partial \mathcal{U}$ splits as the union of a "left" component $\partial^{\vdash}_{tim}\mathcal{U}$ and a "right" component $\partial^{\dashv}_{tim}\mathcal{U}$. In general, $\partial^{\vdash}_{tim}\mathcal{U}$ and $\partial^{\dashv}_{tim}\mathcal{U}$ need not necessarily be straight line segments as depicted above. However, in the following sections, we will impose a gauge condition on the domains under consideration that will indeed fix $\partial^{\vdash}_{tim}\mathcal{U}$ and $\partial^{\dashv}_{tim}\mathcal{U}$ to be vertical line segments (see Definitions \ref{def:DevelopmentSets} and \ref{def:Development}).}
\end{figure}

Let us fix $w:\mathcal{U}\cup\partial_{tim}\mathcal{U}\rightarrow\mathbb{R}$
to be a smooth boundary defining function of $\partial_{tim}\mathcal{U}$,
i.\,e.~
\[
w|_{\partial_{tim}\mathcal{U}}=0,
\]
\[
dw|_{\partial_{tim}\mathcal{U}}\neq0
\]
 and 
\[
w|_{\mathcal{U}}>0.
\]
 We can split $\partial_{tim}\mathcal{U}$ into its ``left'' and
``right'' components as 
\begin{equation}
\partial_{tim}\mathcal{U}=\partial_{tim}^{\vdash}\mathcal{U}\cup\partial_{tim}^{\dashv}\mathcal{U},
\end{equation}
where 
\begin{gather*}
\partial_{tim}^{\vdash}\mathcal{U}=\big\{(u_{0},v_{0})\in\partial_{tim}\mathcal{U}\,:\,\partial_{v}w(u_{0},v_{0})>0\big\},\\
\partial_{tim}^{\dashv}\mathcal{U}=\big\{(u_{0},v_{0})\in\partial_{tim}\mathcal{U}\,:\,\partial_{v}w(u_{0},v_{0})<0\big\}.
\end{gather*}

\begin{rem*}
Notice that any future directed radial null geodesic of $\mathcal{M}=\mathcal{U}\times\mathbb{S}^{2}$
with a future limiting point on $\partial_{tim}^{\vdash}\mathcal{U}\times\mathbb{S}^{2}$
(in the ambient $\mathbb{R}^{2}\times\mathbb{S}^{2}$ topology of
$\bar{\mathcal{U}}\times\mathbb{S}^{2}$) is necessarily ingoing.
Similarly, future directed radial null geodesics ``terminating''
at $\partial_{tim}^{\dashv}\mathcal{U}\times\mathbb{S}^{2}$ are necessarily
outgoing.

In the next sections, we will only consider the reflection of radial
null geodesics on parts of $\partial_{tim}\mathcal{U}$ for which
either $r-r_{0}$ (for some constant $r_{0}>0$) or $1/r$ is a boundary
defining function.
\end{rem*}
Following \cite{MoschidisMaximalDevelopment}, we will define the
reflecting boundary condition on $\partial_{tim}\mathcal{U}$ for
the radial massless Vlasov equation as follows:
\begin{defn*}
A radial massless Vlasov field $f$ on $T\mathcal{M}$ will be said
to satisfy the \emph{reflecting boundary condition} on $\partial_{tim}\mathcal{U}\times\mathbb{S}^{2}$
if and only if 

\begin{itemize}

\item{For any $(u_{0},v_{0})\in\partial_{tim}^{\vdash}\mathcal{U}$
and any $p>0$:
\begin{equation}
\lim_{h\rightarrow0^{+}}\Bigg(\frac{\bar{f}_{out}\big(u_{0},v_{0}+h;\,\frac{-\partial_{u}w}{\partial_{v}w}(u_{0},v_{0})\cdot\text{\textgreek{W}}^{-2}(u_{0},v_{0}+h)\cdot p\big)}{\bar{f}_{in}\big(u_{0}-h,v_{0};\,\text{\textgreek{W}}^{-2}(u_{0}-h,v_{0})\cdot p\big)}\Bigg)=1.\label{eq:LeftBoundaryCondition}
\end{equation}
}

\item{For any $(u_{1},v_{1})\in\partial_{tim}^{\dashv}\mathcal{U}$
and any $p>0$:
\begin{equation}
\lim_{h\rightarrow0^{+}}\Bigg(\frac{\bar{f}_{in}\big(u_{1}+h,v_{1};\,\frac{-\partial_{v}w}{\partial_{u}w}(u_{1},v_{1})\cdot\text{\textgreek{W}}^{-2}(u_{1}+h,v_{1})\cdot p\big)}{\bar{f}_{out}\big(u_{1},v_{1}-h;\,\text{\textgreek{W}}^{-2}(u_{1},v_{1}-h)\cdot p\big)}\Bigg)=1.\label{eq:RightBoundaryCondition}
\end{equation}
}

\end{itemize}
\end{defn*}
Note that the relations (\ref{eq:LeftBoundaryCondition}) and (\ref{eq:RightBoundaryCondition})
for $\bar{f}_{in},\bar{f}_{out}$ imply the following boundary relations
for the components (\ref{eq:T_uuComponent})--(\ref{eq:T_vvComponent})
of the energy momentum tensor $T$: 
\begin{itemize}
\item For any $(u_{0},v_{0})\in\partial_{tim}^{\vdash}\mathcal{U}$:
\begin{equation}
\lim_{h\rightarrow0^{+}}\frac{r^{2}T_{uu}(u_{0},v_{0}+h)}{r^{2}T_{vv}(u_{0}-h,v_{0})}=\Big(\frac{-\partial_{u}w}{\partial_{v}w}(u_{0},v_{0})\Big)^{2}.\label{eq:LeftBoundaryConditionT}
\end{equation}

\item For any $(u_{1},v_{1})\in\partial_{tim}^{\dashv}\mathcal{U}$:
\begin{equation}
\lim_{h\rightarrow0^{+}}\frac{r^{2}T_{vv}(u_{1}+h,v_{1})}{r^{2}T_{uu}(u_{1},v_{1}-h)}=\Big(\frac{-\partial_{v}w}{\partial_{u}w}(u_{0},v_{0})\Big)^{2}.\label{eq:RightBoundaryConditionT}
\end{equation}

\end{itemize}

\section{\label{sec:ResultsFromTheOtherPaper}The boundary--characteristic
initial value problem: well-posedness and Cauchy stability }

In this Section, we will formulate the asymptotically AdS initial
value problem for the system (\ref{eq:RequationFinal})--(\ref{eq:OutgoingVlasovFinal})
with reflecting boundary conditions on $\{r=r_{0}\}$ and $\mathcal{I}$,
for some $r_{0}>0$. We will then recall the main results established
in \cite{MoschidisMaximalDevelopment}, regarding the well-posedness
and the structure of the maximal development for this system.

\subsection{Asymptotically AdS characteristic initial data}

The following definition was introduced in \cite{MoschidisMaximalDevelopment}:

\begin{customdef}{3.1}[Definition 3.1 in \cite{MoschidisMaximalDevelopment}]\label{def:TypeII}

For any $v_{1}<v_{2}$ and any $r_{0}>0$, let $r_{/}:[v_{1},v_{2})\rightarrow[r_{0},+\infty)$,
$\text{\textgreek{W}}_{/}:[v_{1},v_{2})\rightarrow(0,+\infty)$ and
$\bar{f}_{in/},\bar{f}_{out/}:[v_{1},v_{2})\times(0,+\infty)\rightarrow[0,+\infty)$
be $C^{\infty}$ functions, such that 
\begin{equation}
r_{/}(v_{1})=r_{0}
\end{equation}
and 
\begin{equation}
\lim_{v\rightarrow v_{2}}r_{/}(v)=+\infty.\label{eq:RGoesToInfinity}
\end{equation}
Let us define $(\partial_{u}r)_{/}:[v_{1},v_{2})\rightarrow(-\infty,0)$
by the relation 
\begin{equation}
(\partial_{u}r)_{/}(v)=\frac{1}{r_{/}(v)}\Big(-r_{/}\partial_{v}r_{/}(v_{1})-\frac{1}{4}\int_{v1}^{v}(1-\Lambda r_{/}^{2}(\bar{v}))\text{\textgreek{W}}_{/}^{2}(\bar{v})\, d\bar{v}\Big).\label{eq:TransversalDerivativeU-1}
\end{equation}
 We will call $(r_{/},\text{\textgreek{W}}_{/}^{2},\bar{f}_{in/},\bar{f}_{out/})$
an \emph{asymptotically AdS boundary-characteristic initial data set}
on $[v_{1},v_{2})$ for the system (\ref{eq:RequationFinal})--(\ref{eq:OutgoingVlasovFinal})
satisfying the reflecting gauge condition at $r=r_{0},+\infty$ if:

\begin{itemize}

\item{ $(r_{/},\text{\textgreek{W}}_{/})$ satisfies the constraint
equation 
\begin{equation}
\partial_{v}(\text{\textgreek{W}}_{/}^{-2}\partial_{v}r_{/})=-4\pi r_{/}(T_{vv})_{/}\text{\textgreek{W}}_{/}^{-2},\label{eq:ConstraintVDef}
\end{equation}
 where 
\begin{equation}
(T_{vv})_{/}(v)\doteq\int_{0}^{+\infty}\text{\textgreek{W}}_{/}^{4}(v)(p^{u})^{2}\bar{f}_{in/}(v;p^{u})\, r_{/}^{2}(v)\frac{dp^{u}}{p^{u}}.\label{eq:EnergyMomentumIntialRight}
\end{equation}
}

\item{ $\bar{f}_{out/}$ solves the massless radial Vlasov equation
\begin{gather}
\partial_{v}\big(\text{\textgreek{W}}_{/}^{4}(v)r_{/}^{4}(v)p^{v}\bar{f}_{out/}(v,p^{v})\big)+p^{v}\partial_{p^{v}}\big(\text{\textgreek{W}}_{/}^{4}(v)r_{/}^{4}(v)p^{v}\bar{f}_{out/}(v,p^{v})\big)=0.\label{eq:OutgoingEquationCombatibility}
\end{gather}
}

\item{ $(\partial_{u}r)_{/}$ satisfies 
\begin{equation}
\lim_{v\rightarrow v_{2}^{-}}\frac{(\partial_{u}r)_{/}}{\partial_{v}r_{/}}=1.\label{eq:GaugeInfinityInitialData}
\end{equation}
}

\item{ $\bar{f}_{out/},\bar{f}_{in/}$ satisfy the following compatibility
conditions at $v=v_{1},v_{2}$ for any $p>0$:
\begin{equation}
\frac{\bar{f}_{out/}\big(v_{1};\,\frac{-(\partial_{u}r)_{/}}{\partial_{v}r_{/}}(v_{1})\cdot\text{\textgreek{W}}_{/}^{-2}(v_{1})\cdot p\big)}{\bar{f}_{in/}\big(v_{1};\,\text{\textgreek{W}}_{/}^{-2}(v_{1})\cdot p\big)}=1\label{eq:LeftBoundaryConditionInitialData}
\end{equation}
and
\begin{equation}
\lim_{h\rightarrow0^{+}}\Bigg(\frac{\bar{f}_{in/}\big(v_{2}-h;\,\frac{\partial_{v}r_{/}}{-(\partial_{u}r)_{/}}(v_{2}-h)\cdot\text{\textgreek{W}}_{/}^{-2}(v_{2}-h)\cdot p\big)}{\bar{f}_{out/}\big(v_{2}-h;\,\text{\textgreek{W}}_{/}^{-2}(v_{2}-h)\cdot p\big)}\Bigg)=1.\label{eq:RightBoundaryConditionInitialData}
\end{equation}
}

\end{itemize}

\end{customdef}
\begin{rem*}
Notice that the constraint equation (\ref{eq:ConstraintVDef}) implies
\begin{equation}
\partial_{v}(\text{\textgreek{W}}_{/}^{-2}\partial_{v}r_{/})\le0.
\end{equation}
Thus, (\ref{eq:RGoesToInfinity}) yields 
\begin{equation}
\partial_{v}r_{/}>0\label{eq:NonTrappedInitialData}
\end{equation}
everywhere on $[v_{1},v_{2})$.
\end{rem*}
Given any asymptotically AdS boundary-characteristic initial data
set $(r_{/},\text{\textgreek{W}}_{/}^{2},\bar{f}_{in/},\bar{f}_{out/})$
on $[v_{1},v_{2})$ with reflecting gauge conditions at $r=r_{0},+\infty$,
we will also define the initial Hawking mass $m_{/}$ and initial
renormalised Hawking mass $\tilde{m}_{/}$ on $[v_{1}.v_{2})$ by
the relations 
\begin{equation}
m_{/}\doteq\frac{r_{/}}{2}\big(1-4\text{\textgreek{W}}_{/}^{-2}(\partial_{u}r)_{/}\partial_{v}r_{/}\big),\label{eq:DefinitionHawkingMassCharacteristic}
\end{equation}
and 
\begin{equation}
\tilde{m}_{/}\doteq m_{/}-\frac{1}{6}\Lambda r_{/}^{3},
\end{equation}
in accordance with (\ref{eq:RelationHawkingMass}), (\ref{eq:RenormalisedHawkingMass}).

\subsection{Developments with reflecting boundary conditions on $r=r_{0},+\infty$}

We will only consider solutions $(r,\text{\textgreek{W}}^{2},\bar{f}_{in},\bar{f}_{out})$
to (\ref{eq:RequationFinal})--(\ref{eq:OutgoingVlasovFinal}) satisfying
a reflecting gauge condition on $\partial_{tim}\mathcal{U}$, which
fixes $\partial_{tim}\mathcal{U}$ to be a union of vertical straight
lines in the $(u,v)$-plane. This motivates defining the following
class of domains $\mathcal{U}$ in the plane (see \cite{MoschidisMaximalDevelopment}):

\begin{figure}[h] 
\centering 
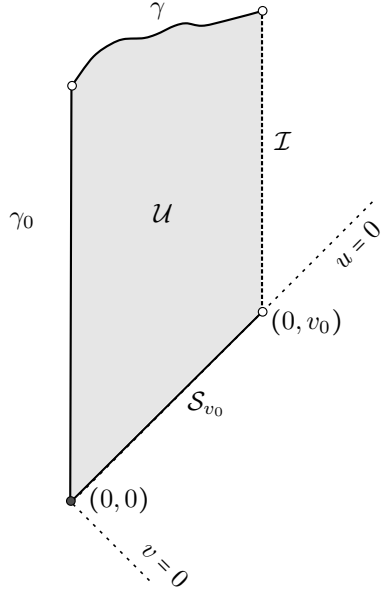 
\caption{A typical domain $\mathcal{U}\in\mathscr{U}_{v_{0}}$ would be as depicted above. In the case when the boundary set $\gamma$ is empty, it is necessary that both $\gamma_{0}$ and $\mathcal{I}$ are unbounded (i.\,e.~extend all the way to $u+v=\infty$).}
\end{figure}

\begin{customdef}{3.2}[Definition 3.3 in \cite{MoschidisMaximalDevelopment}]\label{def:DevelopmentSets}

For any $v_{0}>0$, let $\mathscr{U}_{v_{0}}$ be the set of all connected
open domains $\mathcal{U}$ of the $(u,v)$-plane with piecewise Lipschitz
boundary $\partial\mathcal{U}$, with the property that $\partial\mathcal{U}$
splits as the union 
\begin{equation}
\partial\mathcal{U}=\text{\textgreek{g}}_{0}\cup\mathcal{I}\cup\mathcal{S}_{v_{0}}\cup clos(\text{\textgreek{g}}),\label{eq:BoundaryOfU}
\end{equation}
where, for some $0<u_{\text{\textgreek{g}}_{0}},u_{\mathcal{I}}\le+\infty$,
\begin{equation}
\text{\textgreek{g}}_{0}=\{u=v\}\cap\{0\le u<u_{\text{\textgreek{g}}_{0}}\},\label{eq:AxisForm}
\end{equation}
\begin{equation}
\mathcal{I}=\{u=v-v_{0}\}\cap\{0\le u<u_{\mathcal{I}}\},\label{eq:InfinityForm}
\end{equation}
\begin{equation}
\mathcal{S}_{v_{0}}=\{0\}\times[0,v_{0}],
\end{equation}
 and $\text{\textgreek{g}}:(x_{1},x_{2})\rightarrow\mathbb{R}^{2}$
is a Lipschitz, achronal (with respect to the reference Lorentzian
metric (\ref{eq:ComparisonUVMetric})) curve, which is allowed to
be empty (the closure $clos(\text{\textgreek{g}})$ of $\text{\textgreek{g}}$
in (\ref{eq:BoundaryOfU}) is considered with respect to the standard
topology of $\mathbb{R}^{2}$).

\end{customdef}
\begin{rem*}
It follows readily from Definition \ref{def:DevelopmentSets} that
$\mathcal{U}$ is necessarily contained in the future domain of dependence
of $\mathcal{S}_{v_{0}}\cup\text{\textgreek{g}}_{0}\cup\mathcal{I}$
(with respect to the comparison metric (\ref{eq:ComparisonUVMetric})).
In the case when $\text{\textgreek{g}}=\emptyset$ in (\ref{eq:BoundaryOfU}),
it is necessary that both $\text{\textgreek{g}}_{0}$ and $\mathcal{I}$
extend all the way to $u+v=+\infty$.
\end{rem*}
A development of an asymptotically AdS boundary-characteristic initial
data set for the system (\ref{eq:RequationFinal})--(\ref{eq:OutgoingVlasovFinal})
with reflecting boundary conditions on $r=r_{0},+\infty$ is defined
as follows (see \cite{MoschidisMaximalDevelopment}):

\begin{customdef}{3.3}[Definition 3.4 in \cite{MoschidisMaximalDevelopment}]\label{def:Development}

For any $v_{0}>0$ and $r_{0}>0$, let $(r_{/},\text{\textgreek{W}}_{/}^{2},\bar{f}_{in/},\bar{f}_{out/})$
be a smooth asymptotically AdS boundary-characteristic initial data
set on $[0,v_{0})$ for the system (\ref{eq:RequationFinal})--(\ref{eq:OutgoingVlasovFinal})
satisfying the reflecting gauge condition at $r=r_{0},+\infty$, according
to Definition \ref{def:TypeII}. A \underline{future development}
of $(r_{/},\text{\textgreek{W}}_{/}^{2},\bar{f}_{in/},\bar{f}_{out/})$
will consist of an open set $\mathcal{U}\in\mathscr{U}_{v_{0}}$ (see
Definition \ref{def:DevelopmentSets}) and smooth functions $r:\mathcal{U}\rightarrow(r_{0},+\infty)$,
$\text{\textgreek{W}}^{2}:\mathcal{U}\rightarrow(0,+\infty)$ and
$\bar{f}_{in},\bar{f}_{out}:\mathcal{U}\times(0,+\infty)\rightarrow[0,+\infty)$
satisfying the following properties:

\begin{enumerate}

\item The functions $r,\text{\textgreek{W}}^{2},\bar{f}_{in},\bar{f}_{out}$
solve the system (\ref{eq:RequationFinal})--(\ref{eq:OutgoingVlasovFinal})
on $\mathcal{U}$.

\item The functions $r,\text{\textgreek{W}}^{2},\bar{f}_{in},\bar{f}_{out}$
satisfy the given initial conditions on $\mathcal{S}_{v_{0}}=\{0\}\times[0,v_{0})$,
i.\,e.: 
\begin{equation}
(r,\text{\textgreek{W}}^{2},\bar{f}_{in},\bar{f}_{out})|_{\mathcal{S}_{v_{0}}}=(r_{/},\text{\textgreek{W}}_{/}^{2},\bar{f}_{in/},\bar{f}_{out/}).\label{eq:InitialDataRightMaximal}
\end{equation}

\item The functions $(r,\bar{f}_{in},\bar{f}_{out})$ satisfy on
$\text{\textgreek{g}}_{0}$ the boundary conditions 
\begin{equation}
r|_{\text{\textgreek{g}}_{0}}=r_{0}\label{eq:MirrorRMaximal}
\end{equation}
 and 
\begin{equation}
\bar{f}_{out}\big(u_{*},v_{*};\, p\big)=\bar{f}_{in}\big(u_{*},v_{*};\, p\big),\label{eq:ReflectionMirrorMaximal}
\end{equation}
for all $(u_{*},v_{*})\in\text{\textgreek{g}}_{0}$ and $p>0$, and
on $\mathcal{I}$ the boundary conditions 
\begin{equation}
(1/r)|_{\mathcal{I}}=0\label{eq:InfinityRMaximal}
\end{equation}
 and 
\begin{equation}
\lim_{h\rightarrow0^{+}}\Bigg(\frac{\bar{f}_{in}\big(u_{*}+h,v_{*};\,\text{\textgreek{W}}^{-2}(u_{*}+h,v_{*})\cdot p\big)}{\bar{f}_{out}\big(u_{*},v_{*}-h;\,\text{\textgreek{W}}^{-2}(u_{*},v_{*}-h)\cdot p\big)}\Bigg)=1,\label{eq:ReflectionInfinityMaximal}
\end{equation}
for all $(u_{*},v_{*})\in\mathcal{I}$ and $p>0$.

\item The following are satisfied on $\text{\textgreek{g}}_{0}$
and $\mathcal{I}$: 
\begin{equation}
\partial_{u}r|_{\text{\textgreek{g}}_{0}}=-\partial_{v}r|_{\text{\textgreek{g}}_{0}}\label{eq:GaugeMirrorMaximal}
\end{equation}
 and 
\begin{equation}
\partial_{u}(1/r)|_{\mathcal{I}}=-\partial_{v}(1/r)|_{\mathcal{I}}.\label{eq:GaugeInfinityMaximal}
\end{equation}

\end{enumerate}

\end{customdef}
\begin{rem*}
Notice that the boundary conditions (\ref{eq:MirrorRMaximal}) and
(\ref{eq:InfinityRMaximal}), combined with the form (\ref{eq:AxisForm})
and (\ref{eq:InfinityForm}) of $\text{\textgreek{g}}_{0}$ and $\mathcal{I}$,
respectively, imply the relations (\ref{eq:GaugeMirrorMaximal}) and
(\ref{eq:GaugeInfinityMaximal}). However, the relations (\ref{eq:GaugeMirrorMaximal})
and (\ref{eq:GaugeInfinityMaximal}) should be viewed as \emph{gauge
conditions} fixing, in conjuction with (\ref{eq:MirrorRMaximal})
and (\ref{eq:InfinityRMaximal}), the form (\ref{eq:AxisForm}) and
(\ref{eq:InfinityForm}) of $\text{\textgreek{g}}_{0}$ and $\mathcal{I}$.
\end{rem*}
If $\mathscr{D}=(\mathcal{U};r,\text{\textgreek{W}}^{2},\bar{f}_{in},\bar{f}_{out})$
and $\mathscr{D}^{\prime}=(\mathcal{U}^{\prime};r^{\prime},(\text{\textgreek{W}}^{\prime})^{2},\bar{f}_{in}^{\prime},\bar{f}_{out}^{\prime})$
are two future developments of the same initial data $(r_{/},\text{\textgreek{W}}_{/}^{2},\bar{f}_{in/},\bar{f}_{out/})$,
we will say that $\mathscr{D}^{\prime}$ is an extension of $\mathscr{D}$,
writing $\mathscr{D}\subseteq\mathscr{D}^{\prime}$, if $\mathcal{U}\subseteq\mathcal{U}^{\prime}$
and the restriction of $(r^{\prime},(\text{\textgreek{W}}^{\prime})^{2},\bar{f}_{in}^{\prime},\bar{f}_{out}^{\prime})$
on $\mathcal{U}$ coincides with $(r,\text{\textgreek{W}}^{2},\bar{f}_{in},\bar{f}_{out})$.
\begin{rem*}
If $\mathscr{D}=(\mathcal{U};r,\text{\textgreek{W}}^{2},\bar{f}_{in},\bar{f}_{out})$
and $\mathscr{D}^{\prime}=(\mathcal{U}^{\prime};r^{\prime},(\text{\textgreek{W}}^{\prime})^{2},\bar{f}_{in}^{\prime},\bar{f}_{out}^{\prime})$
are two future developments of the same initial data $(r_{/},\text{\textgreek{W}}_{/}^{2},\bar{f}_{in/},\bar{f}_{out/})$,
then 
\begin{equation}
(r,\text{\textgreek{W}}^{2},\bar{f}_{in},\bar{f}_{out})|_{\mathcal{U}\cap\mathcal{U}^{\prime}}=(r^{\prime},(\text{\textgreek{W}}^{\prime})^{2},\bar{f}_{in}^{\prime},\bar{f}_{out}^{\prime})|_{\mathcal{U}\cap\mathcal{U}^{\prime}}
\end{equation}
(see \cite{MoschidisMaximalDevelopment}).
\end{rem*}

\subsection{The maximal development}

The following result was established in \cite{MoschidisMaximalDevelopment}:

\begin{customthm}{3.1}[Theorem 1 in \cite{MoschidisMaximalDevelopment}]\label{thm:maximalExtension}

For any $v_{0}>0$ and $r_{0}>0$, let $(r_{/},\text{\textgreek{W}}_{/}^{2},\bar{f}_{in/},\bar{f}_{out/})$
be a smooth asymptotically AdS boundary-characteristic initial data
set on $[0,v_{0})$ for the system (\ref{eq:RequationFinal})--(\ref{eq:OutgoingVlasovFinal})
satisfying the reflecting gauge condition at $r=r_{0},+\infty$, according
to Definition \ref{def:TypeII}, such that the quantities $\frac{\text{\textgreek{W}}_{/}^{2}}{1-\frac{1}{3}\Lambda r_{/}^{2}},r_{/}^{2}(T_{vv})_{/}$
and $\tan^{-1}r_{/}$ extend smoothly on $v=v_{0}$. Then, there exists
a unique, smooth future development $(\mathcal{U};r,\text{\textgreek{W}}^{2},\bar{f}_{in},\bar{f}_{out})$
of $(r_{/},\text{\textgreek{W}}_{/}^{2},\bar{f}_{in/},\bar{f}_{out/})$
which is \underline{maximal}, i.\,e.~any other future development
$(\mathcal{U}^{\prime};r^{\prime},(\text{\textgreek{W}}^{\prime})^{2},\bar{f}_{in}^{\prime},\bar{f}_{out}^{\prime})$
of $(r_{/},\text{\textgreek{W}}_{/}^{2},\bar{f}_{in/},\bar{f}_{out/})$
with $r^{\prime}\ge r_{0}$ everywhere on $\mathcal{U}^{\prime}$satisfies
$\mathcal{U}^{\prime}\subseteq\mathcal{U}$ and $r^{\prime},(\text{\textgreek{W}}^{\prime})^{2},\bar{f}_{in}^{\prime},\bar{f}_{out}^{\prime}$
are the restrictions of $r,\text{\textgreek{W}}^{2},\bar{f}_{in},\bar{f}_{out}$
on $\mathcal{U}^{\prime}$.

The maximal future development $(\mathcal{U};r,\text{\textgreek{W}}^{2},\bar{f}_{in},\bar{f}_{out})$
satisfies the following properties (for the definition of the curves
$\text{\textgreek{g}}_{0},\mathcal{I},\text{\textgreek{g}}$, see
Definition \ref{def:DevelopmentSets}):

\begin{enumerate}

\item The renormalised Hawking mass $\tilde{m}$ is conserved on
$\text{\textgreek{g}}_{0}$ and $\mathcal{I}$, i.\,e.: 
\begin{equation}
\tilde{m}|_{\text{\textgreek{g}}_{0}}=\tilde{m}|_{\text{\textgreek{g}}_{0}\cap\{u=0\}}\label{eq:ConstantMassMirror}
\end{equation}
and 
\begin{equation}
\tilde{m}|_{\mathcal{I}}=\tilde{m}|_{\mathcal{I}\cap\{u=0\}}.\label{eq:ConstantMassInfinity}
\end{equation}

\item The curve $\mathcal{I}$ is conformally complete, i.\,e. $\text{\textgreek{W}}^{2}/(1-\frac{1}{3}\Lambda r^{2})$
has a finite limit on $\mathcal{I}$ and: 
\begin{equation}
\int_{\mathcal{I}}\sqrt{\frac{\text{\textgreek{W}}^{2}}{1-\frac{1}{3}\Lambda r^{2}}}\Big|_{\mathcal{I}}\, du=+\infty.\label{eq:CompleConformalInfinity}
\end{equation}

\item We have 
\begin{equation}
\partial_{u}r<0,\label{eq:NegativeDerivativeRMaximal}
\end{equation}
\begin{equation}
\big(1-\frac{2m}{r}\big)\big|_{J^{-}(\mathcal{I})\cup J^{-}(\text{\textgreek{g}}_{0})}>0\label{eq:NonTrappigMaximal}
\end{equation}
 and 
\begin{equation}
\partial_{v}r|_{J^{-}(\mathcal{I})\cup J^{-}(\text{\textgreek{g}}_{0})}>0,\label{eq:AlternativeNonTrappingMaximal}
\end{equation}
where 
\begin{equation}
J^{-}(\mathcal{I})=\big\{0\le u<\sup_{\mathcal{I}}u\big\}\cap\mathcal{U}\label{eq:PastOfInfinity}
\end{equation}
is the causal past of $\mathcal{I}$ and 
\begin{equation}
J^{-}(\text{\textgreek{g}}_{0})=\big\{0\le v<\sup_{\text{\textgreek{g}}_{0}}v\big\}\cap\mathcal{U}
\end{equation}
is the causal past of $\text{\textgreek{g}}_{0}$ (with respect to
the reference Lorenztian metric (\ref{eq:ComparisonUVMetric})).

\item In the case $\mathcal{U}\backslash J^{-}(\mathcal{I})\neq\emptyset$,
the future event horizon 
\begin{equation}
\mathcal{H}^{+}\doteq\mathcal{U}\cap\partial J^{-}(\mathcal{I})=\big\{ u=\sup_{\mathcal{I}}u\big\}\cap\mathcal{U}\label{eq:DefinitionHorizon}
\end{equation}
has the following properties:

\begin{enumerate}

\item $\mathcal{H}^{+}$ has infinite affine length, i.\,e.: 
\begin{equation}
\int_{\mathcal{H}^{+}}\text{\textgreek{W}}^{2}\, dv=+\infty.\label{eq:InfiniteLengthHorizon}
\end{equation}

\item We have 
\begin{equation}
\sup_{\mathcal{H}^{+}}r=r_{S}\label{eq:UpperBoundRHorizon}
\end{equation}
and
\begin{equation}
\inf_{\mathcal{H}^{+}}\Big(1-\frac{2m}{r}\Big)=0,\label{eq:TrappingAsymptoticallyHorizon}
\end{equation}
where $r_{S}$ defined by the relation 
\begin{equation}
1-2\frac{\lim_{v\rightarrow v_{0}^{-}}\tilde{m}_{/}(v)}{r_{S}}-\frac{1}{3}\Lambda r_{S}^{2}=0.\label{eq:DefinitionRs}
\end{equation}

\end{enumerate}

\item In the case $\mathcal{H}^{+}\neq\emptyset$, the curve $\text{\textgreek{g}}_{0}$
is bounded and satisfies 
\begin{equation}
\text{\textgreek{g}}_{0}\nsubseteq J^{-}(\mathcal{I}),\label{eq:MirrorExtendsBeyondHorizon}
\end{equation}
i.\,e.~$\text{\textgreek{g}}_{0}$ contains points lying to the
future of $\mathcal{H}^{+}$. 

\item In the case $\mathcal{H}^{+}\neq\emptyset$, the curve $\text{\textgreek{g}}$
is non-empty, piecewise smooth and $r$ extends continuously on $\text{\textgreek{g}}$
with $r|_{\text{\textgreek{g}}_{0}}=r_{0}$. Furthermore, for any
point $(u_{1},v_{1})\in\text{\textgreek{g}}$, the line $\{v=v_{1}\}$
intersects $\mathcal{I}$.%
\footnote{In other words, there is no point in $\text{\textgreek{g}}$ which
lies on the curve $\{v=v_{\mathcal{I}}\}$, where $(u_{\mathcal{I}},v_{\mathcal{I}})$
is the future limit point of $\mathcal{I}$.%
}

\end{enumerate}

\end{customthm}
\begin{rem*}
In the case when $\mathcal{U}\backslash J^{-}(\mathcal{I})\neq\emptyset$
(and thus $\mathcal{H}^{+}\neq\emptyset$), in view of (\ref{eq:UpperBoundRHorizon}),
(\ref{eq:DefinitionRs}) and the fact that $r>r_{0}$ on $\mathcal{U}$,
it is necessary that 
\begin{equation}
2\frac{\lim_{v\rightarrow v_{0}^{-}}\tilde{m}_{/}(v)}{r_{0}}>1-\frac{1}{3}\Lambda r_{0}^{2}.\label{eq:LowerBoundMass}
\end{equation}

\end{rem*}
In a similar way, we can uniquely define the \emph{maximal past development}
$(\mathcal{U};r,\text{\textgreek{W}}^{2},\bar{f}_{in},\bar{f}_{out})$
of $(r_{/},\text{\textgreek{W}}_{/}^{2},\bar{f}_{in/},\bar{f}_{out/})$,
satisfying the properties outlined by Theorem \ref{thm:maximalExtension}
after performing a ``time reversal'' transformation $(u,v)\rightarrow(-v,-u)$.
Notice that such a coordinate transformation turns an asymptotically
AdS boundary-characteristic initial data set on $u=0$ into an asymptotically
AdS boundary-characteristic initial data set on $v=0$. However, Theorem
\ref{thm:maximalExtension} also holds (with exactly the same proof)
for such initial data sets.

\subsection{\label{sub:Cauchy-stability-inCauchyStability}Cauchy stability in
a rough norm, uniformly in $r_{0}$}

In \cite{MoschidisMaximalDevelopment}, the following ``norm'' was
introduced for smooth asymptotically AdS boundary-characteristic initial
data sets $(r_{/},\text{\textgreek{W}}_{/}^{2},\bar{f}_{in/},\bar{f}_{out/})$
on $[0,v_{0})$ for the system (\ref{eq:RequationFinal})--(\ref{eq:OutgoingVlasovFinal}):
\begin{align}
||(r_{/},\text{\textgreek{W}}_{/}^{2},\bar{f}_{in/},\bar{f}_{out/})||_{\mathcal{C\mathcal{S}}}\doteq\sqrt{-\Lambda} & \sup_{0\le v<v_{0}}|\tilde{m}_{/}(v)|+(-\Lambda)\sup_{0\le v<v_{0}}\int_{0}^{v_{0}}\frac{1}{\text{\textgreek{r}}_{/}(v)-\text{\textgreek{r}}_{/}(\bar{v})+\text{\textgreek{r}}_{/}(0)}\Big(\frac{r_{/}^{2}(T_{vv})_{/}}{\partial_{v}\text{\textgreek{r}}_{/}}\Big)(\bar{v})\, d\bar{v}+\label{eq:GeometricNormForCauchyStability}\\
 & +\sup_{0\le v<v_{0}}\max\big\{\frac{2m_{/}}{r_{/}},0\big\},\nonumber 
\end{align}
where 
\begin{equation}
\text{\textgreek{r}}_{/}(v)\doteq\tan^{-1}\big(\sqrt{-\Lambda}r_{/}(v)\big).
\end{equation}

\begin{rem*}
Note that, in (\ref{eq:GeometricNormForCauchyStability}), 
\[
\text{\textgreek{r}}_{/}(0)=\tan^{-1}\big(\sqrt{-\Lambda}r_{0}\big).
\]

The expression (\ref{eq:GeometricNormForCauchyStability}) is invariant
under gauge transformations, as well as scale transformations of the
form $(u,v)\rightarrow(\text{\textgreek{l}}u,\text{\textgreek{l}}v)$,
$(r,\tilde{m},\Lambda)\rightarrow(\text{\textgreek{l}}r,\text{\textgreek{l}}\tilde{m},\text{\textgreek{l}}^{-2}\Lambda)$,
$r_{0}\rightarrow\text{\textgreek{l}}r_{0}$, $(\bar{f}_{in},\bar{f}_{out})\rightarrow(\text{\textgreek{l}}^{-4}\bar{f}_{in},\text{\textgreek{l}}^{-4}\bar{f}_{out})$.
Moreover, $||(r_{/},\text{\textgreek{W}}_{/}^{2},\bar{f}_{in/},\bar{f}_{out/})||_{\mathcal{C\mathcal{S}}}=0$
if and only if $(r_{/},\text{\textgreek{W}}_{/}^{2},\bar{f}_{in/},\bar{f}_{out/})$
are the initial data for pure AdS spacetime on $\{r\ge r_{0}\}$,
i.\,e.~if $\bar{f}_{in/}=\bar{f}_{out/}=0$ and $\tilde{m}=0$.
\end{rem*}
The following Cauchy stability result for the trivial initial data
was established in \cite{MoschidisMaximalDevelopment}:

\begin{customprop}{3.1}[Corollary 1 in \cite{MoschidisMaximalDevelopment}]\label{prop:CauchyStabilityOfAdS}

For any (possibly large) $l_{*}>0$, there exists a (small) $\text{\textgreek{e}}_{0}>0$
and a constant $C_{l_{*}}>0$ depending only on $l_{*}$, so that
the following statement holds: For any $v_{0}>0$ and $0<r_{0}<(-\Lambda)^{-\frac{1}{2}}$,
if $(r_{/},\text{\textgreek{W}}_{/}^{2},\bar{f}_{in/},\bar{f}_{out/})$
is a smooth asymptotically AdS boundary-characteristic initial data
set on $[0,v_{0})$ for the system (\ref{eq:RequationFinal})--(\ref{eq:OutgoingVlasovFinal})
satisfying the reflecting gauge condition at $r=r_{0},+\infty$, according
to Definition \ref{def:TypeII}, such that the quantities $\frac{\text{\textgreek{W}}_{/}^{2}}{1-\frac{1}{3}\Lambda r_{/}^{2}},r_{/}^{2}(T_{vv})_{/}$
and $\tan^{-1}r_{/}$ extend smoothly on $v=v_{0}$ and moreover 
\begin{equation}
||(r_{/},\text{\textgreek{W}}_{/}^{2},\bar{f}_{in/},\bar{f}_{out/})||_{\mathcal{C}\mathcal{S}}<\text{\textgreek{e}}\label{eq:SmallnessForCauchyStability}
\end{equation}
for some $0<\text{\textgreek{e}}\le\text{\textgreek{e}}_{0}$, then
the maximal development $(\mathcal{U};r,\text{\textgreek{W}}^{2},\bar{f}_{in},\bar{f}_{out})$
satisfies 
\begin{equation}
\mathcal{W}_{l_{*}}\doteq\{0<u\le l_{*}v_{0}\}\cap\{u<v<u+v_{0}\}\subset\mathcal{U}\label{eq:InclusionInMaximalDomain}
\end{equation}
and 
\begin{equation}
\sqrt{-\Lambda}\sup_{\mathcal{W}_{l_{*}}}|\tilde{m}|+\sup_{\mathcal{W}_{l_{*}}}\log\Bigg(\frac{1-\frac{1}{3}\Lambda r^{2}}{1-\max\{\frac{2m}{r},0\}}\Bigg)+\sup_{\bar{u}}\int_{\{u=\bar{u}\}\cap\mathcal{W}_{l_{*}}}\frac{rT_{vv}}{\partial_{v}r}\, dv+\sup_{\bar{v}}\int_{\{v=\bar{v}\}\cap\mathcal{W}_{l_{*}}}\frac{rT_{uu}}{(-\partial_{u}r)}\, du<C_{l_{*}}\text{\textgreek{e}}.\label{eq:SmallnessCauchyStability}
\end{equation}

\end{customprop}
\begin{rem*}
Proposition \ref{prop:CauchyStabilityOfAdS} should be interpreted
as a Cauchy stability statement for the pure AdS initial data set
with respect to the topology defined by (\ref{eq:GeometricNormForCauchyStability})
which is independent of the radius $r_{0}$ of the reflecting boundary. 

Considering the spherically symmetric Einstein--scalar field system
(\ref{eq:EinsteinScalarFieldInDoubleNull}) with an inner mirror placed
at $\{r=r_{0}\}$, the analogue of the initial data norm (\ref{eq:GeometricNormForCauchyStability})
(obtained using the substitution $(T_{vv})_{/}\rightarrow(\partial_{v}\text{\textgreek{f}})|_{u=0}$)
is rougher compared to the bounded variation norm of Christodoulou
(see \cite{ChristodoulouBoundedVariation}). It is not known whether
(\ref{eq:EinsteinScalarFieldInDoubleNull}), restricted to the exterior
of an inner mirror at $\{r=r_{0}\}$, satisfies a Cauchy stability
estimate with respect to the analogue of the initial data norm (\ref{eq:GeometricNormForCauchyStability})
which is independent of $r_{0}$ (although local existence and uniqueness
follow trivially in this case for fixed $r_{0}$). 
\end{rem*}
In fact, Proposition \ref{prop:CauchyStabilityOfAdS} is a special
case of the following Cauchy stability estimate established in \cite{MoschidisMaximalDevelopment}:

\begin{customthm}{3.2}[Theorem 2 in \cite{MoschidisMaximalDevelopment}]\label{prop:CauchyStability}

For any $v_{1}<v_{2}$ and $0<r_{0}<(-\Lambda)^{-1/2}$, let $(r_{/i},\text{\textgreek{W}}_{/i}^{2},\bar{f}_{in/i},\bar{f}_{out/i})$,
$i=1,2$, be two smooth  asymptotically AdS boundary-characteristic
initial data sets on $[v_{1},v_{2})$ for the system (\ref{eq:RequationFinal})--(\ref{eq:OutgoingVlasovFinal})
satisfying the reflective gauge condition at $r=r_{0},+\infty$, according
to Definition \ref{def:TypeII}, such that the quantities $\frac{\text{\textgreek{W}}_{/i}^{2}}{1-\frac{1}{3}\Lambda r_{/i}^{2}},r_{/i}^{2}(T_{vv})_{/i}$
and $\tan^{-1}r_{/i}$ extend smoothly on $v=v_{2}$. Assume, also,
the following conditions:

\begin{enumerate}

\item For some $u_{0}>0$, the maximal future development $(\mathcal{U}_{1};r_{1},\text{\textgreek{W}}_{1}^{2},\bar{f}_{in1},\bar{f}_{out1})$
of $(r_{/1},\text{\textgreek{W}}_{/1}^{2},\bar{f}_{in/1},\bar{f}_{out/1})$
satisfies 
\begin{equation}
\mathcal{W}_{u_{0}}\doteq\{0<u<u_{0}\}\cap\{u+v_{1}<v<u+v_{2}\}\subset\mathcal{U}_{1}
\end{equation}
and 
\begin{align}
\sup_{\mathcal{W}_{u_{0}}}\Bigg\{\Big|\log\big(\frac{\text{\textgreek{W}}_{1}^{2}}{1-\frac{1}{3}\Lambda r_{1}^{2}}\big)\Big|+\Big|\log\Big(\frac{2\partial_{v}r_{1}}{1-\frac{2m_{1}}{r_{1}}}\Big)\Big|+\Big|\log\Big(\frac{1-\frac{2m_{1}}{r_{1}}}{1-\frac{1}{3}\Lambda r_{1}^{2}}\Big)\Big|+\sqrt{-\Lambda}|\tilde{m}_{1}|\Bigg\}+\label{eq:UpperBoundNonTrappingForCauchyStability}\\
+\sup_{\bar{u}}\int_{\{u=\bar{u}\}\cap\mathcal{W}_{u_{0}}}r_{1}\frac{(T_{vv})_{1}}{\partial_{v}r_{1}}\, dv+\sup_{\bar{v}}\int_{\{v=\bar{v}\}\cap\mathcal{W}_{u_{0}}}r_{1}\frac{(T_{uu})_{1}}{-\partial_{u}r_{1}}\, du & =C_{0}<+\infty.\nonumber 
\end{align}

\item The $(r_{/i},\text{\textgreek{W}}_{/i}^{2},\bar{f}_{in/i},\bar{f}_{out/i})$,
$i=1,2$, are $\text{\textgreek{d}}$-close in the following sense:
\begin{align}
\sup_{v\in[v_{1},v_{2})}\Bigg\{\Big|\log\big(\frac{\text{\textgreek{W}}_{/1}^{2}}{1-\frac{1}{3}\Lambda r_{/1}^{2}}\big)-\log\big(\frac{\text{\textgreek{W}}_{/2}^{2}}{1-\frac{1}{3}\Lambda r_{/2}^{2}}\big)\Big|+\Big|\log\Big(\frac{2\partial_{v}r_{/1}}{1-\frac{2m_{/1}}{r_{/1}}}\Big)-\log\Big(\frac{2\partial_{v}r_{/2}}{1-\frac{2m_{/2}}{r_{/2}}}\Big)\Big|+\label{eq:GaugeDifferenceBoundCauchystability}\\
+\Big|\log\Big(\frac{1-\frac{2m_{/_{1}}}{r_{/1}}}{1-\frac{1}{3}\Lambda r_{/1}^{2}}\Big)-\log\Big(\frac{1-\frac{2m_{/_{2}}}{r_{/2}}}{1-\frac{1}{3}\Lambda r_{/2}^{2}}\Big)\Big|+\sqrt{-\Lambda}|\tilde{m}_{/1}-\tilde{m}_{/2}|\Bigg\}(v) & \le\text{\textgreek{d}}\nonumber 
\end{align}
and 
\begin{equation}
\sup_{v\in[v_{1},v_{2}]}(-\Lambda)\int_{v_{1}}^{v_{2}}\frac{\big|r_{/1}^{2}\frac{(T_{vv})_{/1}}{\partial_{v}\text{\textgreek{r}}_{1}}(\bar{v})-r_{/2}^{2}\frac{(T_{vv})_{/2}}{\partial_{v}\text{\textgreek{r}}_{2}}(\bar{v})\big|}{|\text{\textgreek{r}}_{/}(v)-\text{\textgreek{r}}_{/}(\bar{v})|+\text{\textgreek{r}}_{/}(v_{1})}\, d\bar{v}\le\text{\textgreek{d}},\label{eq:DifferenceBoundCauchyStability}
\end{equation}
where $C_{1}$ is a large fixed absolute constant, $\text{\textgreek{d}}$
satisfies 
\begin{equation}
0\le\text{\textgreek{d}}\le\text{\textgreek{d}}_{0}\doteq\exp\big(-\exp\big(C_{1}(1+C_{0})\frac{u_{0}}{v_{2}-v_{1}}\big)\big)\label{eq:SmallnessDeltaForCauchyStability}
\end{equation}
and $\text{\textgreek{r}}_{/}$ is defined by the relation 
\begin{equation}
\text{\textgreek{r}}_{/}(v)\doteq\tan^{-1}\big(\sqrt{-\frac{\Lambda}{3}}r_{/}(v)\big).
\end{equation}

\end{enumerate}

Then, the maximal development $(\mathcal{U}_{2};r_{2},\text{\textgreek{W}}_{2}^{2},\bar{f}_{in2},\bar{f}_{out2})$
of $(r_{/2},\text{\textgreek{W}}_{/2}^{2},\bar{f}_{in/2},\bar{f}_{out/2})$
satisfies 
\begin{equation}
\mathcal{W}_{u_{0}}\subset\mathcal{U}_{2}
\end{equation}
and 
\begin{align}
\sup_{\mathcal{W}_{u_{0}}}\Bigg\{\Big|\log\big(\frac{\text{\textgreek{W}}_{1}^{2}}{1-\frac{1}{3}\Lambda r_{1}^{2}}\big)-\log\big(\frac{\text{\textgreek{W}}_{2}^{2}}{1-\frac{1}{3}\Lambda r_{2}^{2}}\big)\Big|+\Big|\log\Big(\frac{2\partial_{v}r_{1}}{1-\frac{2m_{1}}{r_{1}}}\Big)-\log\Big(\frac{2\partial_{v}r_{2}}{1-\frac{2m_{2}}{r_{2}}}\Big)\Big|+\label{eq:UpperBoundNonTrappingForCauchyStability-1}\\
+\Big|\log\Big(\frac{1-\frac{2m_{1}}{r_{1}}}{1-\frac{1}{3}\Lambda r_{1}^{2}}\Big)-\log\Big(\frac{1-\frac{2m_{2}}{r_{2}}}{1-\frac{1}{3}\Lambda r_{2}^{2}}\Big)\Big|+\sqrt{-\Lambda}|\tilde{m}_{1}-\tilde{m}_{2}|\Bigg\}+\nonumber \\
+\sup_{\bar{u}}\int_{\{u=\bar{u}\}\cap\mathcal{W}_{u_{0}}}\big|r_{1}(T_{vv})_{1}-r_{2}(T_{vv})_{2}\big|\, dv+\sup_{\bar{v}}\int_{\{v=\bar{v}\}\cap\mathcal{W}_{u_{0}}}\big|r_{1}(T_{uu})_{1}-r_{2}(T_{uu})_{2}\big|\, du\nonumber \\
\le\exp\big(\exp\big( & C_{1}(1+C_{0})\big)\frac{u_{0}}{v_{2}-v_{1}}\big)\text{\textgreek{d}}.\nonumber 
\end{align}

\end{customthm}
\begin{rem*}
By repeating the proof of Theorem \ref{prop:CauchyStability}, the
Cauchy stability estimate (\ref{eq:UpperBoundNonTrappingForCauchyStability-1})
also holds in the case when $(\mathcal{U}_{i};r_{i},\text{\textgreek{W}}_{i}^{2},\bar{f}_{in;i},\bar{f}_{out;i})$,
$i=1,2$, are the maximal \underline{past} developments of $(r_{/i},\text{\textgreek{W}}_{/i}^{2},\bar{f}_{in/i},\bar{f}_{out/i})$,
i.\,e.~when $\mathcal{W}_{u_{0}}$ is replaced by 
\begin{equation}
\mathcal{W}_{u_{0}}^{(-)}\doteq\{-u_{0}\le u<0\}\cap\{u+v_{1}<v<u+v_{2}\}
\end{equation}
and (\ref{eq:UpperBoundNonTrappingForCauchyStability}) holds on $\mathcal{W}_{u_{0}}^{(-)}$
in place of $\mathcal{W}_{u_{0}}$. 
\end{rem*}

\section{\label{sec:The-main-result:Details}Final statement of Theorem \ref{thm:TheoremDetailedIntro}:
the non-linear instability of AdS}

The main result of this paper is the following: 

\begin{customthm}{1}[final version]\label{thm:TheTheorem} For any
$\text{\textgreek{e}}\in(0,1]$, there exist $r_{0\text{\textgreek{e}}}$,
$v_{0\text{\textgreek{e}}}$ depending smoothly on $\text{\textgreek{e}}$
such that 
\begin{equation}
r_{0\text{\textgreek{e}}}\xrightarrow{\text{\textgreek{e}}\rightarrow0}0
\end{equation}
and 
\begin{equation}
\sqrt{-\Lambda}v_{0\text{\textgreek{e}}}\xrightarrow{\text{\textgreek{e}}\rightarrow0}\frac{\pi}{\sqrt{3}},
\end{equation}
 as well as a family $(r_{/}^{(\text{\textgreek{e}})},(\text{\textgreek{W}}_{/}^{(\text{\textgreek{e}})})^{2},\bar{f}_{in/}^{(\text{\textgreek{e}})},\bar{f}_{out/}^{(\text{\textgreek{e}})})$
of smooth asymptotically AdS boundary-characteristic initial data
sets for the system (\ref{eq:RequationFinal})--(\ref{eq:OutgoingVlasovFinal})
satisfying the reflecting gauge condition at $r=r_{0},+\infty$, such
that the following hold:

\begin{enumerate}

\item The family $(r_{/}^{(\text{\textgreek{e}})},(\text{\textgreek{W}}_{/}^{(\text{\textgreek{e}})})^{2},\bar{f}_{in/}^{(\text{\textgreek{e}})},\bar{f}_{out/}^{(\text{\textgreek{e}})})$
satisfies
\begin{equation}
||(r_{/}^{(\text{\textgreek{e}})},(\text{\textgreek{W}}_{/}^{(\text{\textgreek{e}})})^{2},\bar{f}_{in/}^{(\text{\textgreek{e}})},\bar{f}_{out/}^{(\text{\textgreek{e}})})||_{\mathcal{C}\mathcal{S}}\xrightarrow{\text{\textgreek{e}}\rightarrow0}0,\label{eq:DecayInInitialDataNorm}
\end{equation}
where $||\cdot||_{CS}$ is the norm defined by (\ref{eq:GeometricNormForCauchyStability}). 

\item There exists a trapped sphere, i.\,e.~point $(u_{\dagger},v_{\dagger})$
in the maximal future development $(\mathcal{U}_{\text{\textgreek{e}}};r_{\text{\textgreek{e}}},\text{\textgreek{W}}_{\text{\textgreek{e}}}^{2},\bar{f}_{in\text{\textgreek{e}}},\bar{f}_{out\text{\textgreek{e}}})$
of $(r_{/}^{(\text{\textgreek{e}})},(\text{\textgreek{W}}_{/}^{(\text{\textgreek{e}})})^{2},\bar{f}_{in/}^{(\text{\textgreek{e}})},\bar{f}_{out/}^{(\text{\textgreek{e}})})$
such that 
\begin{equation}
\frac{2m}{r}(u_{\dagger},v_{\dagger})>1.\label{eq:TrappedSurfaceOccurs}
\end{equation}
In particular, in view of Theorem \ref{thm:maximalExtension}, $(\mathcal{U}_{\text{\textgreek{e}}};r_{\text{\textgreek{e}}},\text{\textgreek{W}}_{\text{\textgreek{e}}}^{2},\bar{f}_{in\text{\textgreek{e}}},\bar{f}_{out\text{\textgreek{e}}})$
has a non-empty future event horizon $\mathcal{H}^{+}$ (defined by
(\ref{eq:FutureEventHorizon})), satisfying the properties 4.a and
4.b of Theorem \ref{thm:maximalExtension}, and a complete conformal
infinity $\mathcal{I}$ (satisfying (\ref{eq:CompleConformalInfinity})).

\end{enumerate}

\end{customthm}
\begin{rem*}
If 
\begin{equation}
\text{\textgreek{d}}(\text{\textgreek{e}})\doteq||(r_{/}^{(\text{\textgreek{e}})},(\text{\textgreek{W}}_{/}^{(\text{\textgreek{e}})})^{2},\bar{f}_{in/}^{(\text{\textgreek{e}})},\bar{f}_{out/}^{(\text{\textgreek{e}})})||_{\mathcal{C}\mathcal{S}},
\end{equation}
the point $(u_{\dagger},v_{\dagger})$ satisfies the upper bound 
\begin{equation}
u_{\dagger}\le\exp\big(\exp(\text{\textgreek{d}}^{6}(\text{\textgreek{e}}))\big)v_{0}.\label{eq:UpperBoundTrappedSurface}
\end{equation}
On the other hand, in view of Proposition \ref{prop:CauchyStabilityOfAdS},
we necessarily have 
\begin{equation}
u_{\dagger}\xrightarrow{\text{\textgreek{e}}\rightarrow0}+\infty.
\end{equation}

In the simpler case when one is interested in a weaker instability
statement, such as the existence of a point $(u_{\dagger},v_{\dagger})$
where 
\begin{equation}
\frac{2m}{r}\Big|_{(u_{\dagger},v_{\dagger})}>\frac{1}{2}\label{eq:WeakerinstabilitySatatement}
\end{equation}
(instead of the stronger bound (\ref{eq:TrappedSurfaceOccurs})),
the proof of Theorem \ref{thm:TheTheorem} can be substantially simplified.
In the case of (\ref{eq:WeakerinstabilitySatatement}), the upper
bound (\ref{eq:UpperBoundTrappedSurface}) can be improved into a
polynomial bound 
\begin{equation}
u_{\dagger}\le(\text{\textgreek{d}}(\text{\textgreek{e}}))^{-C_{1}}v_{0},
\end{equation}
for some fixed $C_{1}>0$.
\end{rem*}

\section{\label{sec:Preliminary-constructions}Construction of the initial
data and notation}

As described already in Section \ref{sub:Sketch-of-the-proof} of
the introduction, the initial data family in Theorem\inputencoding{latin1}{~}\inputencoding{latin9}\ref{thm:TheTheorem}
will be such that their development consists of a large number of
initially ingoing Vlasov beams. In this section, we will construct
such a family $(r_{/\text{\textgreek{e}}},\text{\textgreek{W}}_{/\text{\textgreek{e}}}^{2},\bar{f}_{in/\text{\textgreek{e}}},\bar{f}_{out/\text{\textgreek{e}}})$
of asymptotically AdS boundary-characteristic initial data for (\ref{eq:RequationFinal})--(\ref{eq:OutgoingVlasovFinal}).
The family $(r_{/}^{(\text{\textgreek{e}})},(\text{\textgreek{W}}_{/}^{(\text{\textgreek{e}})})^{2},\bar{f}_{in/}^{(\text{\textgreek{e}})},\bar{f}_{out/}^{(\text{\textgreek{e}})})$
in the statement of Theorem\inputencoding{latin1}{~}\inputencoding{latin9}\ref{thm:TheTheorem}
will be eventually obtained from $(r_{/\text{\textgreek{e}}},\text{\textgreek{W}}_{/\text{\textgreek{e}}}^{2},\bar{f}_{in/\text{\textgreek{e}}},\bar{f}_{out/\text{\textgreek{e}}})$
after possibly adding a suitable perturbation (see Section \ref{sec:Proof}). 

This section is organised as follows: In Section \ref{sub:Parameters-and-auxiliary},
we will introduce a certain hierarchy of parameters that will be necessary
for the construction of $(r_{/\text{\textgreek{e}}},\text{\textgreek{W}}_{/\text{\textgreek{e}}}^{2},\bar{f}_{in/\text{\textgreek{e}}},\bar{f}_{out/\text{\textgreek{e}}})$
in Section \ref{sub:Construction-of-the-initial-data}. In Section
\ref{sub:Notational-conventions-andNotational5}, we will introduce
some basic notation related to the maximal future development $(\mathcal{U}_{\text{\textgreek{e}}};r_{\text{\textgreek{e}}},\text{\textgreek{W}}_{\text{\textgreek{e}}}^{2},\bar{f}_{in\text{\textgreek{e}}},\bar{f}_{out\text{\textgreek{e}}})$
of $(r_{/\text{\textgreek{e}}},\text{\textgreek{W}}_{/\text{\textgreek{e}}}^{2},\bar{f}_{in/\text{\textgreek{e}}},\bar{f}_{out/\text{\textgreek{e}}})$.
Finally, in Section \ref{sub:Some-geometric-constructions}, we will
perform some basic geometric constructions on $(\mathcal{U}_{\text{\textgreek{e}}};r_{\text{\textgreek{e}}},\text{\textgreek{W}}_{\text{\textgreek{e}}}^{2},\bar{f}_{in\text{\textgreek{e}}},\bar{f}_{out\text{\textgreek{e}}})$,
related to the separation of $\mathcal{U}_{\text{\textgreek{e}}}$
into various subregions by the Vlasov beams arising from the initial
data.

\subsection{\label{sub:Parameters-and-auxiliary}Parameters and auxiliary functions}

Let us fix some smooth and strictly increasing functions $h_{0},h_{1},h_{2}:(0,1)\rightarrow(0,1)$,
so that 
\begin{equation}
\lim_{\text{\textgreek{e}}\rightarrow0^{+}}h_{0}(\text{\textgreek{e}})=\lim_{\text{\textgreek{e}}\rightarrow0^{+}}h_{1}(\text{\textgreek{e}})=\lim_{\text{\textgreek{e}}\rightarrow0^{+}}h_{2}(\text{\textgreek{e}})=0,
\end{equation}
\begin{equation}
\lim_{\text{\textgreek{e}}\rightarrow0^{+}}\text{\textgreek{e}}\cdot\exp(\frac{1}{(h_{1}(\text{\textgreek{e}}))^{6}})=\lim_{\text{\textgreek{e}}\rightarrow0^{+}}h_{1}(\text{\textgreek{e}})\cdot\exp\Big(\exp(\frac{1}{(h_{0}(\text{\textgreek{e}}))^{6}})\Big)=0\label{eq:h_1_h_0_definition}
\end{equation}
and 
\begin{equation}
\lim_{\text{\textgreek{e}}\rightarrow0^{+}}h_{2}(\text{\textgreek{e}})\cdot\exp(\text{\textgreek{e}}^{-2})=0.\label{eq:h_2definition}
\end{equation}
In particular, the following relations hold for $\text{\textgreek{e}}\ll1$:
\begin{equation}
h_{2}(\text{\textgreek{e}})\ll\text{\textgreek{e}}\ll h_{1}(\text{\textgreek{e}})\ll h_{0}(\text{\textgreek{e}})\ll1.\label{eq:SchematicRelationParameters}
\end{equation}

Let $\text{\textgreek{q}}:\mathbb{R}\rightarrow[0,1]$ be a smooth
cut-off function, satisfying $\text{\textgreek{q}}|_{[-1,1]}=1$,
$\text{\textgreek{q}}|_{\mathbb{R}\backslash[-2,2]}=0$ and 
\begin{equation}
\text{\textgreek{q}}|_{(-2,2)}>0,\label{eq:SupportedCutOff}
\end{equation}
and let $\text{\textgreek{e}}_{0}\ll1$ be a small enough absolute
constant. For any $0<\text{\textgreek{e}}<\text{\textgreek{e}}_{0}$,
any $r_{0}>0$ satisfying 
\begin{equation}
1-\exp\big(-2(h_{0}(\text{\textgreek{e}}))^{-4}\big)<\frac{r_{0}}{\frac{2}{\sqrt{-\Lambda}}\text{\textgreek{e}}-\frac{1}{3}\Lambda r_{0}^{3}}<1-\frac{1}{2}\exp\big(-2(h_{0}(\text{\textgreek{e}}))^{-4}\big)\label{eq:BoundMirror}
\end{equation}
(note that (\ref{eq:BoundMirror}) implies that $\frac{r_{0}\sqrt{-\Lambda}}{2\text{\textgreek{e}}}=1+O\big(\exp\big(-2(h_{0}(\text{\textgreek{e}}))^{-4}\big)\big)$
as $\text{\textgreek{e}}\rightarrow0$), we will define the following
function on $[0,+\infty)\times(0,+\infty)$:
\begin{equation}
\bar{f}_{\text{\textgreek{e}}}(v,p^{u})\doteq C_{\text{\textgreek{e}}r_{0}}\sum_{j=0}^{\lceil1/h_{1}(\text{\textgreek{e}})\rceil}\text{\textgreek{q}}\Big(p^{u}-3\Big)\cdot\frac{1}{h_{2}(\text{\textgreek{e}})}\text{\textgreek{q}}\Big(\frac{(v-v^{(j)})\sqrt{-\Lambda}-2h_{2}(\text{\textgreek{e}})}{h_{2}(\text{\textgreek{e}})}\Big)\cdot h_{(j)}(\text{\textgreek{e}})\cdot\text{\textgreek{e}},\label{eq:TheIngoingVlasovInitially}
\end{equation}
for some constant $C_{\text{\textgreek{e}}r_{0}}$ to be specified
in terms of $\text{\textgreek{e}},r_{0}$ later, where $\lceil1/h_{1}(\text{\textgreek{e}})\rceil$
denotes the least integer greater than or equal to $1/h_{1}(\text{\textgreek{e}})$,
\begin{equation}
v^{(j)}=\frac{\pi}{\sqrt{-\Lambda}}-j\frac{\text{\textgreek{e}}}{h_{1}(\text{\textgreek{e}})\sqrt{-\Lambda}}\label{eq:DefinitionV_j}
\end{equation}
 for any $0\le j\le\lceil1/h_{1}(\text{\textgreek{e}})\rceil$, 
\begin{equation}
h_{(0)}=h_{0},\label{eq:TopBeamWeight}
\end{equation}
 and 
\begin{equation}
h_{(j)}=h_{1}
\end{equation}
 for all $1\le j\le\lceil1/h_{1}(\text{\textgreek{e}})\rceil$.

\subsection{\label{sub:Construction-of-the-initial-data}Construction of the
initial data family}

For any $0<\text{\textgreek{e}}<\text{\textgreek{e}}_{0}$, any $r_{0}$
satisfying (\ref{eq:BoundMirror}), we will define the following asymptotically
AdS boundary-characteristic initial data set according to Definition
\ref{def:TypeII}: 
\begin{defn}
\label{def:ThefamilyOfInitialData} For any $0<\text{\textgreek{e}}<\text{\textgreek{e}}_{0}$,
any $r_{0}$ satisfying (\ref{eq:BoundMirror}), we define $v_{0}=v_{0}(r_{0},\text{\textgreek{e}})>0$
and the set of smooth functions $r_{/\text{\textgreek{e}}}:[0,v_{0})\rightarrow[r_{0},+\infty)$,
$\text{\textgreek{W}}_{/\text{\textgreek{e}}}^{2}:[0,v_{0})\rightarrow(0,+\infty)$,
$\bar{f}_{in/\text{\textgreek{e}}}:[0,v_{0})\times(0,+\infty)\rightarrow[0,+\infty)$
and $\bar{f}_{out/\text{\textgreek{e}}}:[0,v_{0})\times(0,+\infty)\rightarrow[0,+\infty)$
by the requirement that $(r_{/\text{\textgreek{e}}},\text{\textgreek{W}}_{/\text{\textgreek{e}}}^{2},\bar{f}_{in/\text{\textgreek{e}}},\bar{f}_{out/\text{\textgreek{e}}})$
is an asymptotically AdS boundary-characteristic initial data set
on $[0,v_{0})$ for the system (\ref{eq:RequationFinal})--(\ref{eq:OutgoingVlasovFinal})
satisfying the reflecting gauge condition at $r=r_{0},+\infty$ so
that 
\begin{equation}
\frac{\partial_{v}r_{/\text{\textgreek{e}}}}{1-\frac{2m_{/\text{\textgreek{e}}}}{r_{/_{\text{\textgreek{e}}}}}}=\frac{1}{2}\label{eq:ConditionOnDvRInitiallyFamily}
\end{equation}
(where $m_{/\text{\textgreek{e}}}$ is defined in terms of $r_{/\text{\textgreek{e}}},\text{\textgreek{W}}_{/\text{\textgreek{e}}}^{2}$
by (\ref{eq:DefinitionHawkingMassCharacteristic})), 
\begin{equation}
\bar{f}_{out/\text{\textgreek{e}}}=0\label{eq:OutgoingInitialdataFamily}
\end{equation}
and 
\begin{equation}
\bar{f}_{in/\text{\textgreek{e}}}(v,p^{u})=\bar{f}_{\text{\textgreek{e}}}(v,p^{u})\label{eq:IngoingInitialDataFamily}
\end{equation}
for all $0\le v\le v_{0}$ and $p^{u}>0$. The constant $C_{\text{\textgreek{e}}r_{0}}$
in (\ref{eq:TheIngoingVlasovInitially}) is fixed in terms of $\text{\textgreek{e}},r_{0}$
by the requirement that 
\begin{equation}
\lim_{v\rightarrow_{v_{0}}}\tilde{m}_{/\text{\textgreek{e}}}=\frac{\text{\textgreek{e}}}{\sqrt{-\Lambda}}\label{eq:MassAtinfinityFixedInitially}
\end{equation}
(in particular, there exists some fixed (large) $C_{0}>1$, independent
of $\text{\textgreek{e}},r_{0}$, so that $C_{\text{\textgreek{e}}r_{0}}\in[C_{0}^{-1},C_{0}]$
for any $0<\text{\textgreek{e}}<\text{\textgreek{e}}_{0}$, any $r_{0}$
satisfying (\ref{eq:BoundMirror}).%
\footnote{In fact, it suffices to choose $C_{0}=50$.%
}
\end{defn}
\begin{figure}[h] 
\centering 
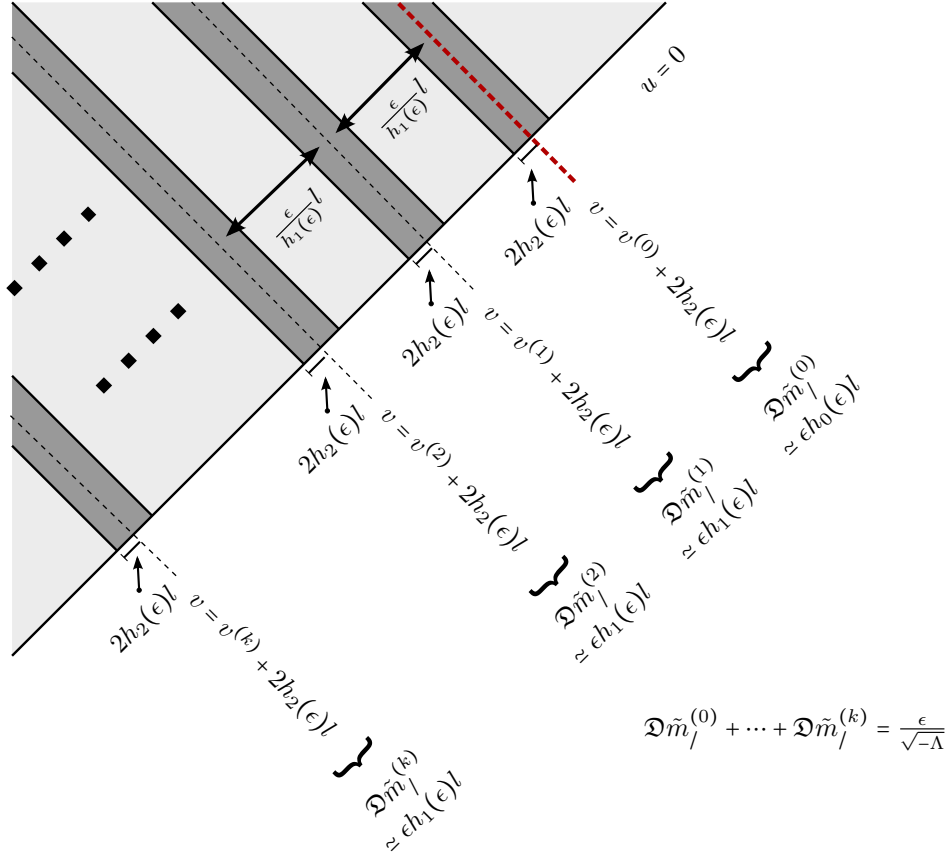 
\caption{The initial data $(r_{\epsilon /},\Omega^{2}_{\epsilon /},\bar{f}_{in \epsilon /},\bar{f}_{out \epsilon /})$ give rise to $k +1$ Vlasov beams, initially arranged as depicted above (using, for convenience, the abbreviation $k=\lceil 1/h_{1}(\epsilon ) \rceil $ and $l=(-\Lambda)^{-1/2}$). For any $0\le j \le k$, we can associate to the beam centered around $v=v^{(j)}+2h_{2}(\epsilon )l$ the initial mass difference $\mathfrak{D} \tilde{m}^{(j)}_{/}$, defined as the renormalised Hawking mass difference between the two vacuum regions surrounding the beam. The top beam (centered initially around $v=v^{0}+2h_{2}(\epsilon )l$) has a greater initial mass difference than the rest of the beams. }
\end{figure}
\begin{rem*}
The conditions (\ref{eq:ConditionOnDvRInitiallyFamily})--(\ref{eq:IngoingInitialDataFamily})
determine $v_{0}$ and $r_{/\text{\textgreek{e}}},\text{\textgreek{W}}_{/\text{\textgreek{e}}}^{2}$
uniquely in terms of $\text{\textgreek{e}},r_{0}$. While $(r_{/\text{\textgreek{e}}},\text{\textgreek{W}}_{/\text{\textgreek{e}}}^{2},\bar{f}_{in/\text{\textgreek{e}}},\bar{f}_{out/\text{\textgreek{e}}})$
depend on both $\text{\textgreek{e}}$ and $r_{0}$, we will only
use the subscript $\text{\textgreek{e}}$ in their notation, since
most of the estimates that we will later establish for their maximal
development will depend only on $\text{\textgreek{e}}$.

The initial data $(r_{/}^{(\text{\textgreek{e}})},(\text{\textgreek{W}}_{/}^{(\text{\textgreek{e}})})^{2},\bar{f}_{in/}^{(\text{\textgreek{e}})},\bar{f}_{out/}^{(\text{\textgreek{e}})})$
in the statement of Theorem \ref{thm:TheTheorem} will eventually
be chosen to be small perturbations of $(r_{/\text{\textgreek{e}}},\text{\textgreek{W}}_{/\text{\textgreek{e}}}^{2},\bar{f}_{in/\text{\textgreek{e}}},\bar{f}_{out/\text{\textgreek{e}}})$
(see Section \ref{sec:Proof}).
\end{rem*}

\subsection{\label{sub:Notational-conventions-andNotational5}Notational conventions
and basic computations}

For any $0<\text{\textgreek{e}}<\text{\textgreek{e}}_{0}$ and any
$r_{0}$ satisfying (\ref{eq:BoundMirror}), let $(r_{/\text{\textgreek{e}}},\text{\textgreek{W}}_{/\text{\textgreek{e}}}^{2},\bar{f}_{in/\text{\textgreek{e}}},\bar{f}_{out/\text{\textgreek{e}}})$
be the initial data set defined by Definition~\ref{def:ThefamilyOfInitialData}.
Assuming that $\text{\textgreek{e}}_{0}$ is fixed small enough, for
any $0<\text{\textgreek{e}}<\text{\textgreek{e}}_{0}$ and any $r_{0}$
satisfying (\ref{eq:BoundMirror}), the initial data set $(r_{/\text{\textgreek{e}}},\text{\textgreek{W}}_{/\text{\textgreek{e}}}^{2},\bar{f}_{in/\text{\textgreek{e}}},\bar{f}_{out/\text{\textgreek{e}}})$
satisfies the following estimate depending only on $\text{\textgreek{e}}$:
\begin{equation}
||(r_{/\text{\textgreek{e}}},\text{\textgreek{W}}_{/\text{\textgreek{e}}}^{2},\bar{f}_{in/\text{\textgreek{e}}},\bar{f}_{out/\text{\textgreek{e}}})||_{\mathcal{C}\mathcal{S}}\le Ch_{0}(\text{\textgreek{e}}),\label{eq:CSnormFamily}
\end{equation}
where $||\cdot||_{\mathcal{C}\mathcal{S}}$ is defined by (\ref{eq:GeometricNormForCauchyStability})
and $C>0$ is a fixed constant. 

Let $(\mathcal{U}_{\text{\textgreek{e}}};r_{\text{\textgreek{e}}},\text{\textgreek{W}}_{\text{\textgreek{e}}}^{2},\bar{f}_{in\text{\textgreek{e}}},\bar{f}_{out\text{\textgreek{e}}})$
be the maximal future development of $(r_{/\text{\textgreek{e}}},\text{\textgreek{W}}_{/\text{\textgreek{e}}}^{2},\bar{f}_{in/\text{\textgreek{e}}},\bar{f}_{out/\text{\textgreek{e}}})$
(see Theorem \ref{thm:maximalExtension}). In view of Proposition~\ref{prop:CauchyStabilityOfAdS},
the bound (\ref{eq:CSnormFamily}) implies that, for any fixed $u_{*}>0$
and any $\text{\textgreek{d}}>0$, there exists an $\text{\textgreek{e}}_{\text{\textgreek{d}}u_{*}}>0$
sufficiently small depending only on $\text{\textgreek{d}}$ and $u_{*}$
so that, for any $0\le\text{\textgreek{e}}<\text{\textgreek{e}}_{\text{\textgreek{d}}u_{*}}$:
\begin{equation}
\mathcal{W}_{u_{*}}\doteq\{0<u<u_{*}\}\cap\{u<v<u+v_{0}\}\subset\mathcal{U}\label{eq:InclusionInMaximalDomain-1}
\end{equation}
and 
\begin{equation}
\sqrt{-\Lambda}\sup_{\mathcal{W}_{u_{*}}}|\tilde{m}_{\text{\textgreek{e}}}|+\sup_{\mathcal{W}_{u_{*}}}\log\Bigg(\frac{1-\frac{1}{3}\Lambda r_{\text{\textgreek{e}}}^{2}}{1-\max\{\frac{2m_{\text{\textgreek{e}}}}{r_{\text{\textgreek{e}}}},0\}}\Bigg)+\sup_{\bar{u}}\int_{\{u=\bar{u}\}\cap\mathcal{W}_{u_{*}}}\frac{r_{\text{\textgreek{e}}}(T_{vv})_{\text{\textgreek{e}}}}{\partial_{v}r_{\text{\textgreek{e}}}}\, dv+\sup_{\bar{v}}\int_{\{v=\bar{v}\}\cap\mathcal{W}_{u_{*}}}\frac{r_{\text{\textgreek{e}}}(T_{uu})_{\text{\textgreek{e}}}}{(-\partial_{u}r_{\text{\textgreek{e}}})}\, du<\text{\textgreek{d}},\label{eq:SmallnessCauchyStability-1}
\end{equation}
where $m_{\text{\textgreek{e}}},\tilde{m}_{\text{\textgreek{e}}},(T_{uu})_{\text{\textgreek{e}}},(T_{vv})_{\text{\textgreek{e}}}$
are defined in terms of $r_{\text{\textgreek{e}}},\text{\textgreek{W}}_{\text{\textgreek{e}}}^{2}\bar{f}_{in\text{\textgreek{e}}},\bar{f}_{out\text{\textgreek{e}}}$
by (\ref{eq:DefinitionHawkingMass}), (\ref{eq:RenormalisedHawkingMass}),
(\ref{eq:T_uuComponent}) and (\ref{eq:T_vvComponent}). In particular,
if 
\begin{equation}
g_{AdS}=-\text{\textgreek{W}}_{AdS,r_{0},v_{0}}^{2}dudv+r_{AdS,r_{0},v_{0}}g_{\mathbb{S}^{2}}
\end{equation}
is the pure AdS metric in a spherically symmetric coordinate chart
$(u,v)$ such that $r_{AdS,r_{0},v_{0}}=r_{0}$ on $\{u=v\}$ and
$r_{AdS,r_{0},v_{0}}=+\infty$ on $\{u=v-v_{0}\}$,%
\footnote{Note that such a coordinate chart is not unique.%
} then $(\mathcal{U}_{\text{\textgreek{e}}};r_{\text{\textgreek{e}}},\text{\textgreek{W}}_{\text{\textgreek{e}}}^{2},\bar{f}_{in\text{\textgreek{e}}},\bar{f}_{out\text{\textgreek{e}}})$,
when restricted on $\mathcal{W}_{u_{*}}$, is $\text{\textgreek{d}}$-close
to $(\mathcal{W}_{u_{*}};r_{AdS,r_{0},v_{0}},\text{\textgreek{W}}_{AdS,r_{0},v_{0}}^{2},0,0)$
with respect to the (gauge invariant) distance defined by (\ref{eq:GeometricNormForCauchyStability}).
Notice also that (\ref{eq:SmallnessCauchyStability-1}) implies that,
provided $\text{\textgreek{d}}$ is small enough, the spacetime $(\mathcal{W}_{u_{*}}\times\mathbb{S}^{2},g_{\text{\textgreek{e}}})$
does not contain any trapped surface, where 
\begin{equation}
g_{\text{\textgreek{e}}}=-\text{\textgreek{W}}_{\text{\textgreek{e}}}^{2}dudv+r_{\text{\textgreek{e}}}^{2}g_{\mathbb{S}^{2}}.
\end{equation}

Notice that, in view of the conservation of $\tilde{m}$ on $\text{\textgreek{g}}_{0}$
and $\mathcal{I}$ (see (\ref{eq:ConstantMassMirror}) and (\ref{eq:ConstantMassInfinity})),
we have: 
\begin{equation}
\tilde{m}_{\text{\textgreek{e}}}|_{\text{\textgreek{g}}_{0}}=0\label{eq:MassAxis}
\end{equation}
and 
\begin{equation}
\tilde{m}_{\text{\textgreek{e}}}|_{\mathcal{I}}=\lim_{v\rightarrow v_{0}}\tilde{m}_{/\text{\textgreek{e}}}(v)=\frac{\text{\textgreek{e}}}{\sqrt{-\Lambda}}.\label{eq:MassInfinity}
\end{equation}

For each $0\le j\le\lceil1/h_{1}(\text{\textgreek{e}})\rceil$, we
can associate to the beam centered at $v=v^{(j)}+\frac{2}{\sqrt{-\Lambda}}h_{2}(\text{\textgreek{e}})$
the mass difference 
\begin{equation}
\mathfrak{D}\tilde{m}_{/}^{(j)}\doteq\tilde{m}_{/\text{\textgreek{e}}}\big(v^{(j)}+\frac{4}{\sqrt{-\Lambda}}h_{2}(\text{\textgreek{e}})\big)-\tilde{m}_{/\text{\textgreek{e}}}\big(v^{(j)}\big).\label{eq:InitialMassDifferenceBeams}
\end{equation}
Notice that 
\begin{equation}
\mathcal{D}\tilde{m}_{/}^{(0)}\simeq\frac{h_{0}(\text{\textgreek{e}})\text{\textgreek{e}}}{\sqrt{-\Lambda}}
\end{equation}
and, for all $1\le j\le\lceil1/h_{1}(\text{\textgreek{e}})\rceil$:
\begin{equation}
\mathcal{D}\tilde{m}_{/}^{(j)}\simeq\frac{h_{1}(\text{\textgreek{e}})\text{\textgreek{e}}}{\sqrt{-\Lambda}}.
\end{equation}
Furthermore: 
\begin{equation}
\sum_{j=0}^{\lceil1/h_{1}(\text{\textgreek{e}})\rceil}\mathcal{D}\tilde{m}_{/}^{(j)}=\lim_{v\rightarrow v_{0}}\tilde{m}_{/\text{\textgreek{e}}}(v)=\frac{\text{\textgreek{e}}}{\sqrt{-\Lambda}}.
\end{equation}

\subsection{\label{sub:Some-geometric-constructions}Some geometric constructions
on $\mathcal{U}_{\epsilon}$}

For any $0<\text{\textgreek{e}}<\text{\textgreek{e}}_{0}$ and any
$r_{0}$ satisfying (\ref{eq:BoundMirror}), we will define some special
subsets of the domain $\mathcal{U}_{\text{\textgreek{e}}}$ of the
maximal future development $(\mathcal{U}_{\text{\textgreek{e}}};r_{\text{\textgreek{e}}},\text{\textgreek{W}}_{\text{\textgreek{e}}}^{2},\bar{f}_{in\text{\textgreek{e}}},\bar{f}_{out\text{\textgreek{e}}})$
of the initial data set $(r_{/\text{\textgreek{e}}},\text{\textgreek{W}}_{/\text{\textgreek{e}}}^{2},\bar{f}_{in/\text{\textgreek{e}}},\bar{f}_{out/\text{\textgreek{e}}})$. 
\begin{rem*}
In the rest of this section, we will adopt the convention that the
boundary $\partial A$ of a subset $A\subseteq\mathcal{U}_{\text{\textgreek{e}}}$
is the boundary of $A$ as a subset of $\mathbb{R}^{2}$ (with respect
to the ambient topology of $\mathbb{R}^{2}$)
\end{rem*}
Let us define the domain of outer communications $\mathcal{D}_{\text{\textgreek{e}}}$
of $\mathcal{U}_{\text{\textgreek{e}}}$ as 
\begin{equation}
\mathcal{D}_{\text{\textgreek{e}}}\doteq J^{-}(\mathcal{I})\cap\mathcal{U}_{\text{\textgreek{e}}},\label{eq:DomainOfOuterCommunications}
\end{equation}
where $J^{-}(\mathcal{I})$ is the causal past of $\mathcal{I}$ with
respect to the reference metric (\ref{eq:ComparisonUVMetric}) (see
(\ref{eq:PastOfInfinity})). In accordance with Theorem \ref{thm:maximalExtension},
we will also define the future event horizon $\mathcal{H}_{\text{\textgreek{e}}}^{+}$
of $\mathcal{U}_{\text{\textgreek{e}}}$ as 
\begin{equation}
\mathcal{H}_{\text{\textgreek{e}}}^{+}\doteq\partial\mathcal{D}_{\text{\textgreek{e}}}\cap\mathcal{U}_{\text{\textgreek{e}}}.\label{eq:FutureEventHorizon}
\end{equation}
Note that we allow $\mathcal{H}_{\text{\textgreek{e}}}^{+}$ to be
empty. In view of Theorem \ref{thm:maximalExtension}, in the case
when $\mathcal{H}_{\text{\textgreek{e}}}^{+}$ is non-empty, it is
necessarily of the form 
\begin{equation}
\mathcal{H}_{\text{\textgreek{e}}}^{+}=\{u=u_{\mathcal{H}_{\text{\textgreek{e}}}^{+}}\}\cap\mathcal{U}_{\text{\textgreek{e}}}
\end{equation}
and has infinite affine length.

We will also define 
\begin{equation}
\mathcal{J}_{\text{\textgreek{e}}}\doteq J^{-}(\text{\textgreek{g}}_{0})\cap\mathcal{U}_{\text{\textgreek{e}}}.
\end{equation}
Notice that, as a consequence of Theorem \ref{thm:maximalExtension},
on $\mathcal{J}_{\text{\textgreek{e}}}\cup\mathcal{D}_{\text{\textgreek{e}}}$
we have 
\begin{equation}
1-\frac{2m}{r}>0,
\end{equation}
i.\,e.~trapped spheres can only appear in the region $\mathcal{U}_{\text{\textgreek{e}}}\backslash(\mathcal{J}_{\text{\textgreek{e}}}\cup\mathcal{D}_{\text{\textgreek{e}}})$.
In the case $\mathcal{H}_{\text{\textgreek{e}}}^{+}\neq\emptyset$,
Theorem \ref{thm:maximalExtension} also implies that $\mathcal{J}_{\text{\textgreek{e}}}\backslash\mathcal{D}_{\text{\textgreek{e}}}\neq\emptyset$. 

For any $v_{*}\in[0,v_{0}]$ and any integer $n\ge1$, we will define
\begin{equation}
U_{n}(v_{*})\doteq v_{*}+(n-1)v_{0}
\end{equation}
and 
\begin{equation}
V_{n}(v_{*})=v_{*}+nv_{0}.
\end{equation}
We will also set 
\begin{equation}
V_{0}(v_{*})\doteq v_{*}.
\end{equation}
Notice that the segment $\{u=U_{n}(v_{*})\}\cap\mathcal{U}_{\text{\textgreek{e}}}$
is the image of the ingoing null geodesic of $\mathcal{U}_{\text{\textgreek{e}}}$
emanating from the point $(0,v_{*})$ after $n$ reflections off $\text{\textgreek{g}}_{0}$
and $n-1$ reflections off $\mathcal{I}$, while the segment $\{v=V_{n}(v_{*})\}\cap\mathcal{U}_{\text{\textgreek{e}}}$
is the image of the same null geodesic after $n$ reflections off
$\text{\textgreek{g}}_{0}$ and $n$ reflections off $\mathcal{I}$. 

\begin{figure}[h] 
\centering 
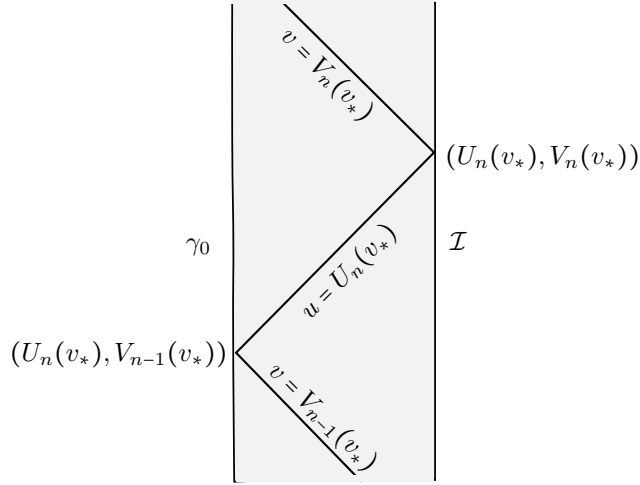 
\caption{Schematic depiction of the lines $v=V_{n-1}(v_{*})$, $u=U_{n}(v_{*})$ and $v=V_{n}(v_{*})$.}
\end{figure}

Let us define the domains $\mathcal{R}_{\text{\textgreek{e}}n}^{(i,j)}\subset\mathcal{U}_{\text{\textgreek{e}}}$
for any $n\in\mathbb{N}$, $0\le i\le\lceil1/h_{1}(\text{\textgreek{e}})\rceil$
and $i\le j\le\lceil1/h_{1}(\text{\textgreek{e}})\rceil+i+1$ by the
relation 
\begin{equation}
\mathcal{R}_{\text{\textgreek{e}}n}^{(i,j)}=\Big\{ U_{n}\big(v^{(i)}+\frac{4}{\sqrt{-\Lambda}}h_{2}(\text{\textgreek{e}})\big)<u<U_{n}\big(v^{(i-1)}\big)\Big\}\cap\Big\{ V_{n}\big(v^{(j)}+\frac{4}{\sqrt{-\Lambda}}h_{2}(\text{\textgreek{e}})\big)<v<V_{n}\big(v^{(j-1)}\big)\Big\}\cap\mathcal{U}_{\text{\textgreek{e}}},\label{eq:Domains_D}
\end{equation}
where we have used the following conventions in the expression (\ref{eq:Domains_D}):
\begin{enumerate}
\item $U_{n}\big(v^{(-1)}\big)\doteq U_{n+1}\big(v^{(\lceil1/h_{1}(\text{\textgreek{e}})\rceil)}\big)$. 
\item $V_{n}\big(v^{(\lceil1/h_{1}(\text{\textgreek{e}})\rceil+l)}+c\big)\doteq V_{n-1}\big(v^{(l-1)}+c\big)$
for any integer $1\le l\le\lceil1/h_{1}(\text{\textgreek{e}})\rceil$
and any $c\ge0$.
\end{enumerate}
\begin{figure}[h!] 
\centering 
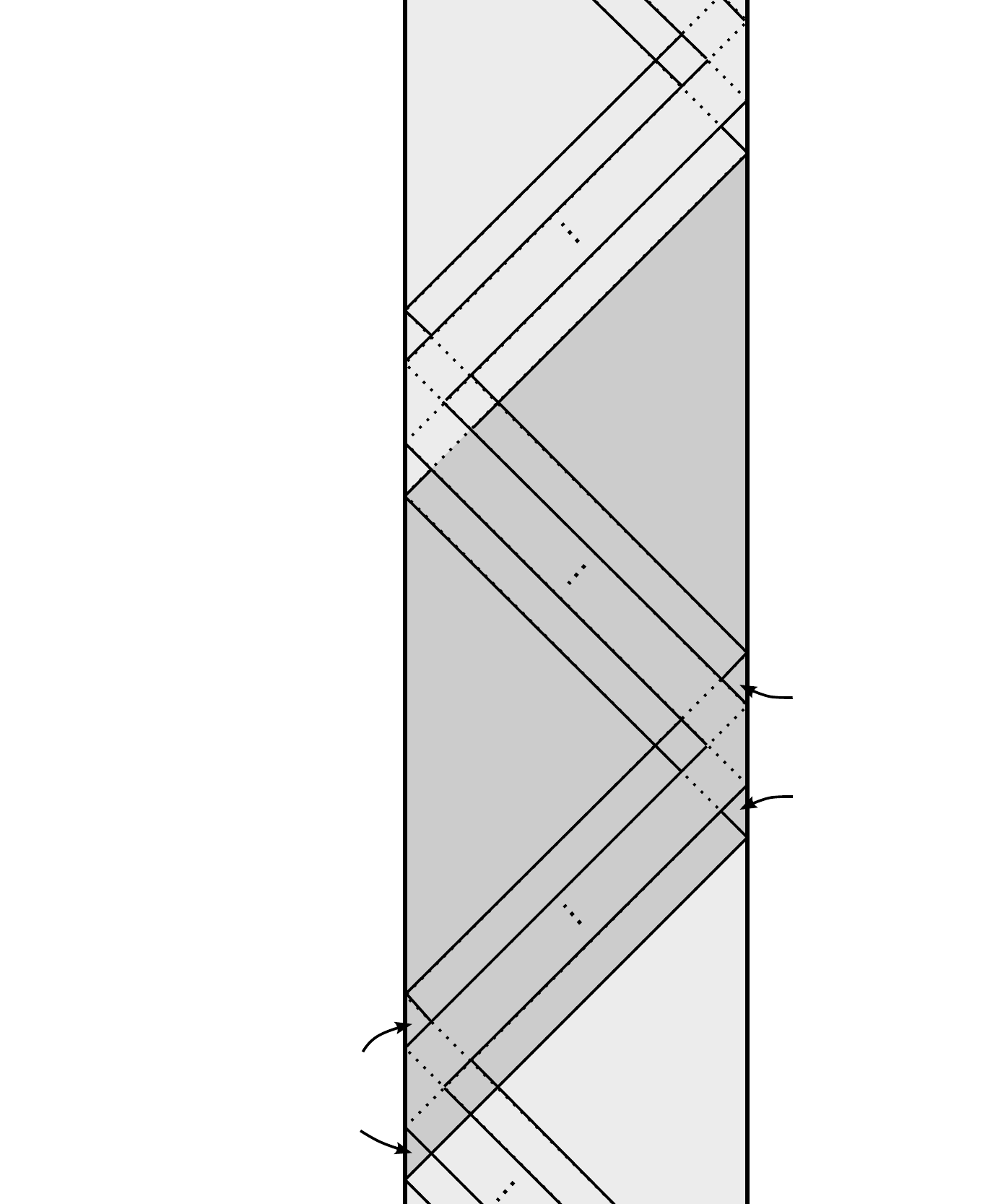 
\caption{Schematic depiction of the domains $\mathcal{R}_{\epsilon n}^{(i,j)}$ for $0\le i\le k$ and $i\le j \le k+i+1$ (where $k=\lceil1/h_{1}(\text{\textgreek{e}})\rceil$). We have used the shorthand notation $u_{n}^{(j)}=U_{n}\big( v^{(j)}\big) $ and $v_{n}^{(j)}=V_{n}\big( v^{(j)}\big) $. Having assumed that $h_{2}(\text{\textgreek{e}})\ll \text{\textgreek{e}})$, all the beams of the form $\Big\{U_{n}\big(v^{(i)}\big)\le u\le U_{n}\big(v^{(i)}+\frac{4}{\sqrt{-\Lambda}}h_{2}(\text{\textgreek{e}})\big) \Big\} $ and $\Big\{V_{n}\big(v^{(j)}\big)\le v\le V_{n}\big(v^{(j)}+\frac{4}{\sqrt{-\Lambda}}h_{2}(\text{\textgreek{e}})\big) \Big\} $, which separate the domains $\mathcal{R}_{\epsilon n}^{(i,j)}$, are depicted as straight line segments.}
\end{figure}
\begin{rem*}
The boundary of the domains $\mathcal{R}_{\text{\textgreek{e}}n}^{(i,i)}$,
$0\le i\le\lceil1/h_{1}(\text{\textgreek{e}})\rceil$, contains a
segment $\mathcal{I}$, while the boundary of the domains $\mathcal{R}_{\text{\textgreek{e}}n}^{(i,\lceil1/h_{1}(\text{\textgreek{e}})\rceil+1+i)}$,
$0\le i\le\lceil1/h_{1}(\text{\textgreek{e}})\rceil$, contains a
segment in $\text{\textgreek{g}}_{0}$. 
\end{rem*}
\begin{figure}[h] 
\centering 
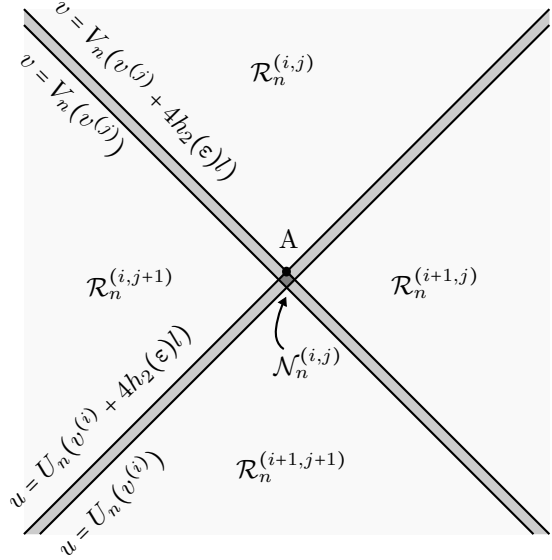 
\caption{Typical arrangement of neighboring vacuum domains not intersecting $\text{\textgreek{g}}_{0}$ or $\mathcal{I}$. The point $A$ at the lower corner of $\mathcal{R}_{n}^{(i,j)}$ satisfies $r(A)=r_{n}^{(i,j)}$. For simplicity, we have used the shorthand notation $l=(-\Lambda)^{-1/2}$.}
\end{figure}

Notice that $T_{uu}=T_{vv}=0$ in $\mathcal{R}_{\text{\textgreek{e}}n}^{(i,j)}$.
In particular, all the domains $(\mathcal{R}_{\text{\textgreek{e}}n}^{(i,j)}\times\mathbb{S}^{2},g_{\text{\textgreek{e}}})$
are isometric to a region of a member of the Schwarzschild-AdS family
(or to a region of pure AdS spacetime), and the renormalised mass
function $\tilde{m}_{\text{\textgreek{e}}}$ is constant on them.
We will define for any $0\le i\le\lceil1/h_{1}(\text{\textgreek{e}})\rceil$,
$i\le j\le\lceil1/h_{1}(\text{\textgreek{e}})\rceil+i+1$ and $n\in\mathbb{N}$
such that $\mathcal{R}_{\text{\textgreek{e}}n}^{(i,j)}\neq\emptyset$:
\begin{equation}
\tilde{m}_{\text{\textgreek{e}}n}^{(i,j)}\doteq\tilde{m}_{\text{\textgreek{e}}}|_{\mathcal{R}_{\text{\textgreek{e}}n}^{(i,j)}}.\label{eq:Renormalised_Mass_D_n}
\end{equation}
In view of (\ref{eq:MassAxis}) and (\ref{eq:MassInfinity}), we immediately
calculate that for all $0\le i\le\lceil1/h_{1}(\text{\textgreek{e}})\rceil$
and all $n\in\mathbb{N}$ such that $\mathcal{R}_{\text{\textgreek{e}}n}^{(i,i)},\mathcal{R}_{\text{\textgreek{e}}n}^{(i,\lceil1/h_{1}(\text{\textgreek{e}})\rceil+1+i)}\neq\emptyset$:
\begin{equation}
\tilde{m}_{\text{\textgreek{e}}n}^{(i,\lceil1/h_{1}(\text{\textgreek{e}})\rceil+1+i)}=0\label{eq:ZeroMassNearAxis}
\end{equation}
and 
\begin{equation}
\tilde{m}_{\text{\textgreek{e}}n}^{(i,i)}=\tilde{m}_{\text{\textgreek{e}}0}^{(i,i)}=\frac{\text{\textgreek{e}}}{\sqrt{-\Lambda}}.\label{eq:MassAtInfinity}
\end{equation}

For any $n\in\mathbb{N}$, $0\le i\le\lceil1/h_{1}(\text{\textgreek{e}})\rceil$
and $i+1\le j\le\lceil1/h_{1}(\text{\textgreek{e}})\rceil+i$, we
will define the interaction regions: 
\begin{equation}
\mathcal{N}_{\text{\textgreek{e}}n}^{(i,j)}\doteq\Big\{ U_{n}\big(v^{(i)}\big)\le u\le U_{n}\big(v^{(i)}+\frac{4}{\sqrt{-\Lambda}}h_{2}(\text{\textgreek{e}})\big)\Big\}\cap\Big\{ V_{n}\big(v^{(j)}\big)\le u\le V_{n}\big(v^{(j)}+\frac{4}{\sqrt{-\Lambda}}h_{2}(\text{\textgreek{e}})\big)\Big\}\cap\mathcal{U}_{\text{\textgreek{e}}},\label{eq:InteractionRegion}
\end{equation}
where the conventions stated below (\ref{eq:Domains_D}) hold regarding
indices smaller than $0$ or larger than $\lceil1/h_{1}(\text{\textgreek{e}})\rceil$. 

Let us define for any $0\le i\le\lceil1/h_{1}(\text{\textgreek{e}})\rceil$,
$i\le j\le\lceil1/h_{1}(\text{\textgreek{e}})\rceil+i+1$ and $n\in\mathbb{N}$
such that $\mathcal{R}_{\text{\textgreek{e}}n}^{(i,j)}\neq\emptyset$:
\begin{equation}
r_{\text{\textgreek{e}}n}^{(i,j)}\doteq r_{\text{\textgreek{e}}}\big(U_{n}(v^{(i)}+\frac{4}{\sqrt{-\Lambda}}h_{2}(\text{\textgreek{e}})),V_{n}(v^{(j)}+\frac{4}{\sqrt{-\Lambda}}h_{2}(\text{\textgreek{e}}))\big).\label{eq:r_n}
\end{equation}
Note that $r_{\text{\textgreek{e}}n}^{(i,i)}=+\infty$ and $r_{\text{\textgreek{e}}n}^{(i,\lceil1/h_{1}(\text{\textgreek{e}})\rceil+i+1)}=r_{0\text{\textgreek{e}}}$.

Finally, let us remark that, in view of property (\ref{eq:SupportOfCutOff})
of the cut-off used in the construction of the initial data and equations
(\ref{eq:ConservationT_vv})--(\ref{eq:ConservationT_uu}), for any
$1\le n\le n_{f}$, $0\le i\le\lceil1/h_{1}(\text{\textgreek{e}})\rceil$,
$i\le j\le\lceil1/h_{1}(\text{\textgreek{e}})\rceil+i+1$, we have
\begin{equation}
T_{uu}>0\mbox{ on }\Big\{ U_{n}(v^{(i)})<u<U_{n}(v^{(i)}+\frac{4}{\sqrt{-\Lambda}}h_{2}(\text{\textgreek{e}}))\Big\}\label{eq:TuuSupport}
\end{equation}
and 
\begin{equation}
T_{vv}>0\mbox{ on }\Big\{ V_{n}(v^{(j)})<u<V_{n}(v^{(j)}+\frac{4}{\sqrt{-\Lambda}}h_{2}(\text{\textgreek{e}}))\Big\}.\label{eq:TvvSupport}
\end{equation}

\section{\label{sec:Proof}Proof of Theorem \ref{thm:TheTheorem}}

In this Section, we will prove Theorem \ref{thm:TheTheorem}. In order
to simplify our notation, from now on, we will often drop the subscripts
$\text{\textgreek{e}}$ in notations related to the maximal future
development $(\mathcal{U}_{\text{\textgreek{e}}};r_{\text{\textgreek{e}}},\text{\textgreek{W}}_{\text{\textgreek{e}}}^{2},\bar{f}_{in\text{\textgreek{e}}},\bar{f}_{out\text{\textgreek{e}}})$
of the initial data $(r_{/\text{\textgreek{e}}},\text{\textgreek{W}}_{/\text{\textgreek{e}}}^{2},\bar{f}_{in/\text{\textgreek{e}}},\bar{f}_{out/\text{\textgreek{e}}})$
(see Definition \ref{def:ThefamilyOfInitialData}).

For any $0<\text{\textgreek{e}}<\text{\textgreek{e}}_{0}$ (provided
$\text{\textgreek{e}}_{0}$ is fixed sufficiently small), any $r_{0}>0$
satisfying (\ref{eq:BoundMirror}), let $(\mathcal{U}_{\text{\textgreek{e}}};r,\text{\textgreek{W}}^{2},\bar{f}_{in},\bar{f}_{out})$
be the maximal future development of $(r_{/\text{\textgreek{e}}},\text{\textgreek{W}}_{/\text{\textgreek{e}}}^{2},\bar{f}_{in/\text{\textgreek{e}}},\bar{f}_{out/\text{\textgreek{e}}})$,
and let us define
\begin{equation}
u_{+}\doteq\sup\big\{ u_{*}>0:\mbox{ }1-\frac{2m}{r}>h_{3}(\text{\textgreek{e}})\mbox{ on }\mathcal{U}_{\text{\textgreek{e}}}\cap\{u<u_{*}\}\big\}\label{eq:UpperUNonTrapping}
\end{equation}
and 
\begin{equation}
\mathcal{U}_{\text{\textgreek{e}}}^{+}=\mathcal{U}_{\text{\textgreek{e}}}\cap\big\{ u<\min\{u_{+},(h_{1}(\text{\textgreek{e}}))^{-2}v_{0\text{\textgreek{e}}}\}\big\},\label{eq:DefinitionUntrappedRegion}
\end{equation}
where 
\begin{equation}
h_{3}(\text{\textgreek{e}})=\exp\Big\{-\exp\Big((h_{1}(\text{\textgreek{e}}))^{-5}\exp\big(-2(h_{0}(\text{\textgreek{e}}))^{-4}\big)\Big)\Big\}.\label{eq:h_3definition}
\end{equation}
Let us also set 
\begin{equation}
k\doteq\lceil1/h_{1}(\text{\textgreek{e}})\rceil\label{eq:kappa}
\end{equation}
and 
\begin{equation}
n_{f}\doteq\lfloor(u_{+}-v^{(0)})/v_{0}\rfloor,\label{eq:defNf}
\end{equation}
where $\lceil x\rceil$ denotes the least integer greater than or
equal to $x$, while $\lfloor x\rfloor$ denotes the largest integer
less than or equal to $x$. 

The proof of Theorem \ref{thm:TheTheorem} will follow in two steps:
First, in Section \ref{sub:NearlyTrapped}, we will show that: 
\begin{equation}
\sup_{\mathcal{U}_{\text{\textgreek{e}}}^{+}}\big(1-\frac{2m}{r}\big)=h_{3}(\text{\textgreek{e}}),\label{eq:NearTrappingIsAchieved}
\end{equation}
i.\,e.~that $\mathcal{U}_{\text{\textgreek{e}}}^{+}$ contains a
nearly-trapped sphere. Then, in Section \ref{sub:FinalStep}, we will
show that, at the final step of the evolution, either a trapped sphere
is formed, or there exists a small perturbation of the initial data
$(r_{/\text{\textgreek{e}}},\text{\textgreek{W}}_{/\text{\textgreek{e}}}^{2},\bar{f}_{in/\text{\textgreek{e}}},\bar{f}_{out/\text{\textgreek{e}}})$
giving rise to a trapped sphere. 

Before proving (\ref{eq:NearTrappingIsAchieved}), we will need to
establish some necessary bounds for the evolution of $(r,\text{\textgreek{W}}^{2},\bar{f}_{in},\bar{f}_{out})$
in the region $\mathcal{U}_{\text{\textgreek{e}}}^{+}$. These bounds,
which will be obtained in Section \ref{sub:Inductive-bounds}, will
be used both in Section \ref{sub:NearlyTrapped} and in Section \ref{sub:FinalStep}.

\subsection{\label{sub:Inductive-bounds}Inductive bounds for the evolution in
the region $\mathcal{U}_{\text{\textgreek{e}}}^{+}$}

In this Section, we will establish a number of useful bounds for $(\mathcal{U}_{\text{\textgreek{e}}}^{+};r,\text{\textgreek{W}}^{2},\bar{f}_{in},\bar{f}_{out})$.
These bounds will include a number of inductive bounds for the quantities
$\tilde{m}_{n}^{(1,k+1)}$, $r_{n}^{(k,k+1)}$ and $r_{n}^{(1,k+1)}$
(with $k$ defined by (\ref{eq:kappa})), that will be of fundamental
significance in the proof of Theorem \ref{thm:TheTheorem}.

In particular, we will prove the following result:
\begin{prop}
\label{prop:TheMainBootstrapBeforeTrapping} For any $0<\text{\textgreek{e}}<\text{\textgreek{e}}_{0}$,
the following bounds hold for $(\mathcal{U}_{\text{\textgreek{e}}}^{+};r,\text{\textgreek{W}}^{2},\bar{f}_{in},\bar{f}_{out})$: 

\begin{enumerate}

\item On $\mathcal{U}_{\text{\textgreek{e}}}^{+}$, we can estimate:
\begin{equation}
\Big|\log\Big(\frac{-\partial_{u}r}{1-\frac{1}{3}\Lambda r^{2}}\Big)\Big|+\Big|\log\Big(\frac{\partial_{v}r}{1-\frac{2m}{r}}\Big)\Big|\le\big(h_{1}(\text{\textgreek{e}})\big)^{-4}\log\big((h_{3}(\text{\textgreek{e}}))^{-1}\big).\label{eq:RoughBoundGeometry}
\end{equation}

\item For any $1\le n\le n_{f}$: 
\begin{equation}
r_{n}^{(0,k)}\ge\frac{\text{\textgreek{e}}^{-\frac{1}{2}}}{\sqrt{-\Lambda}},\label{eq:BoundForRAwayInteractionProp}
\end{equation}
\begin{equation}
r_{n}^{(k,k+1)}\le\frac{\text{\textgreek{e}}^{\frac{1}{2}}}{\sqrt{-\Lambda}},\label{eq:UpperBoundForAxisInteractionProp}
\end{equation}
\begin{equation}
\frac{2(\tilde{m}|_{\mathcal{I}}-\tilde{m}_{n}^{(1,k+1)})}{r_{0}}\ge\exp\big(-2(h_{0}(\text{\textgreek{e}}))^{-4}\big),\label{eq:EnoughMassBehind}
\end{equation}
\begin{equation}
\frac{2(\tilde{m}|_{\mathcal{I}}-\tilde{m}_{n}^{(1,k+1)})}{r_{0}}\le1-\frac{1}{C_{0}}h_{0}(\text{\textgreek{e}})\label{eq:NotEnoughMassBehind}
\end{equation}
and 
\begin{equation}
\frac{r_{n}^{(1,k+1)}}{r_{0}}-1\ge\exp\Big(-(h_{0}(\text{\textgreek{e}}))^{-4}\Big),\label{eq:BoundSecondBeamchanged}
\end{equation}
where $C_{0}>1$ is a large fixed constant (independent of all the
parameters). 

\item For any $2\le n\le n_{f}$: 
\begin{equation}
\frac{\tilde{m}_{n}^{(1,k+1)}}{\tilde{m}_{n-1}^{(1,k+1)}}\ge1+\frac{1}{4}\exp\big(-2(h_{0}(\text{\textgreek{e}}))^{-4}\big)\frac{r_{0}}{r_{n}^{(k,k+1)}}\label{eq:BoundForMassIncrease}
\end{equation}
and 
\begin{equation}
\frac{r_{n}^{(k,k+1)}-r_{0}}{r_{n-1}^{(k,k+1)}-r_{0}}\le1+2C_{0}\frac{r_{0}}{r_{n-1}^{(k,k+1)}}\Big(\Big|\log\big(1-\frac{2\tilde{m}_{n-1}^{(1,k+1)}}{r_{0}}\big)\Big|+(h_{0}(\text{\textgreek{e}}))^{-4}\Big).\label{eq:BoundForMaxBeamSeparation}
\end{equation}

\end{enumerate}
\end{prop}
Before presenting the proof of Proposition \ref{prop:TheMainBootstrapBeforeTrapping}
(in Section \ref{sub:Proof-of-Proposition}), we will briefly comment
on the nature of the bounds (\ref{eq:RoughBoundGeometry})--(\ref{eq:BoundForMaxBeamSeparation})
and their relation with the specific choice of the parameters (\ref{eq:h_1_h_0_definition})--(\ref{eq:h_2definition}).

\subsubsection{\label{sub:Remark-on-Proposition}Remarks on Proposition \ref{prop:TheMainBootstrapBeforeTrapping}}

The bounds (\ref{eq:RoughBoundGeometry})--(\ref{eq:BoundForMaxBeamSeparation})
in Proposition \ref{prop:TheMainBootstrapBeforeTrapping} lie at the
heart of the proof of Theorem \ref{thm:TheTheorem}. The precise form
of the initial data (\ref{eq:TheIngoingVlasovInitially}), the range
(\ref{eq:BoundMirror}) for the mirror radius $r_{0}$ and the asymptotic
bounds (\ref{eq:h_1_h_0_definition})--(\ref{eq:h_2definition}) on
the parameters $h_{0},h_{1},h_{2}$ were carefuly chosen so that (\ref{eq:RoughBoundGeometry})--(\ref{eq:BoundForMaxBeamSeparation})
can be obtained. We will now proceed to briefly comment on the role
of the bounds (\ref{eq:RoughBoundGeometry})--(\ref{eq:BoundForMaxBeamSeparation})
in the proof of Theorem \ref{thm:TheTheorem}. The reader is advised
to review first the sketch of the proof in Section \ref{sub:Sketch-of-the-proof}
of the introduction. Let us remark that, in the notation of Section
\ref{sub:Sketch-of-the-proof}, 
\begin{equation}
\mathcal{E}_{\text{\textgreek{z}}_{0};n}=\tilde{m}_{n}^{(1,k+1)},
\end{equation}
 
\begin{equation}
r_{\text{\textgreek{g}}_{0};n}=r_{n}^{(k,k+1)}
\end{equation}
and
\begin{equation}
r_{\text{\textgreek{g}}_{0};n}^{(1)}=r_{n}^{(1,k+1)}.
\end{equation}

The bound (\ref{eq:RoughBoundGeometry}) is a ``trivial'' bound
controlling quantities related to the chosen gauge. The right hand
side of (\ref{eq:RoughBoundGeometry}), upon integration across any
specific beam (in a direction transversal to the beam), will yield
a small quantity, in view of the fact that the width of the null beams
emanating from $u=0$, $v\sim v^{(j)}$ was chosen to be $\sim h_{2}(\text{\textgreek{e}})$
and, moreover, $h_{2}(\text{\textgreek{e}})$ was chosen in (\ref{eq:h_2definition})
to be small compared to the right hand side of (\ref{eq:RoughBoundGeometry}).
This fact will prove convenient for the proof of Proposition \ref{prop:TheMainBootstrapBeforeTrapping}
and Theorem \ref{thm:TheTheorem}, as it will enable us to ``ignore''
the variation of certain quantities across the width of any specific
beam. That is to say, the bound (\ref{eq:RoughBoundGeometry}) will
enable us to frequently treat the null beams as line segments having
negligible width.

The bounds (\ref{eq:BoundForRAwayInteractionProp})--(\ref{eq:UpperBoundForAxisInteractionProp})
are quantitative expressions of the fact that the set of interactions
of the beams splits into two portions, one close to $r=r_{0}$ and
one close to $\mathcal{I}$. 

The lower bound (\ref{eq:EnoughMassBehind}) is necessary in order
to establish (\ref{eq:BoundForMassIncrease}). In order to obtain
(\ref{eq:EnoughMassBehind}), it is necessary that $r_{0}$ satisfies
the upper bound of (\ref{eq:BoundMirror}).

The upper bound (\ref{eq:NotEnoughMassBehind}) implies that a trapped
sphere (i.\,e.~a sphere where $\frac{2m}{r}>1$) can not be formed
at $\mathcal{R}_{\text{\textgreek{e}}n}^{(i,j)}$ for any $j>k+1$,
since one can also show that $\tilde{m}\le\tilde{m}|_{\mathcal{I}}-\tilde{m}_{n}^{(1,k+1)}$
in those regions. In order to obtain (\ref{eq:NotEnoughMassBehind}),
it is necessary that the mirror radius $r_{0}$ satisfies the lower
bound of (\ref{eq:BoundMirror}). 

In the language of Section \ref{sub:Sketch-of-the-proof} of the introduction,
the bound (\ref{eq:BoundSecondBeamchanged}) states that, when $\text{\textgreek{z}}_{0}$
reaches $\{r=r_{0}\}$ for the $n$-th time, the $r$-distance of
the top beam $\text{\textgreek{z}}_{0}$ from the second-to-top beam
$\text{\textgreek{z}}_{1}$, i.\,e.~$r_{n}^{(1,k+1)}-r_{0}$, can
be bounded from below by a small multiple of $r_{0}$ which is large
compared to $h_{1}(\text{\textgreek{e}})r_{0}$. As a consequence
of (\ref{eq:BoundSecondBeamchanged}) and the bound (\ref{eq:BoundMirror})
for $r_{0}$, for any $i\neq1$, $\mathcal{R}_{\text{\textgreek{e}}n}^{(i,k+1)}$
does not contain a trapped sphere. As a result, combining (\ref{eq:NotEnoughMassBehind}),
(\ref{eq:BoundForRAwayInteractionProp}) and (\ref{eq:BoundSecondBeamchanged}),
we infer that, among all regions $\mathcal{R}_{\text{\textgreek{e}}n}^{(i,j)}$,
a trapped sphere can only appear for $i=1$, $j=k+1$. This fact serves
to simplify the proof of Theorem \ref{thm:TheTheorem}, by avoiding
considering multiple scenarios of trapped surface formation. Furthermore,
it is crucial in obtaining (\ref{eq:BoundForMaxBeamSeparation}).

Establishing (\ref{eq:BoundSecondBeamchanged}) is the most demanding
part of the proof of Proposition \ref{prop:TheMainBootstrapBeforeTrapping}.
It requires obtaining a lower bound in the rate of decrease of $r_{n}^{(1,k+1)}$
in terms of the rate of increase of $\tilde{m}_{n}^{(1,k+1)}$, using
also the fact that $\tilde{m}_{n}^{(1,k+1)}\lesssim r_{0}$ before
a trapped sphere is formed (see the relations (\ref{eq:UsefulBoundAlmostThere})
and (\ref{eq:ControlInMassration}) in the next section).

The bound (\ref{eq:BoundForMassIncrease}) is a technical version
of the bound (\ref{eq:InductiveEnIntro}), and its proof follows from
the ideas outlined in Section \ref{sub:Sketch-of-the-proof}. In obtaining
(\ref{eq:BoundForMassIncrease}), the lower bound of (\ref{eq:EnoughMassBehind})
is necessary.

Finally, the bound (\ref{eq:BoundForMaxBeamSeparation}) is a technical
version of the bound (\ref{eq:InductiveBoundRinIntro}) in Section
\ref{sub:Sketch-of-the-proof} and provides an estimate for the decrease
of the multiplicative factor in the right hand side of (\ref{eq:BoundForMassIncrease}).
In obtaining (\ref{eq:BoundForMaxBeamSeparation}) when $\frac{2m}{r}\simeq1$,
the fact that $\frac{2m}{r}$ is bounded away from $1$ everywhere
but on $\mathcal{R}_{\text{\textgreek{e}}n}^{(i,j)}$ is crucially
used (in particular, the bound (\ref{eq:BoundSecondBeamchanged})
is necessary for (\ref{eq:BoundForMaxBeamSeparation})).
\begin{rem*}
As is evident from the above discussion, most of the technical difficulties
in the proof of Proposition \ref{prop:TheMainBootstrapBeforeTrapping}
are associated to issues related with the near-trapped regime $\frac{2m}{r}\simeq1$.
In the case when, instead of the stronger bound (\ref{eq:TrappedSurfaceOccurs}),
one is merely interested in establishing the weaker instability estimate
(\ref{eq:WeakerinstabilitySatatement}), the proof of Proposition
\ref{prop:TheMainBootstrapBeforeTrapping} simplifies substantially:
In that case, it is not necessary to demand that the worst instability
scenario takes place in $\mathcal{R}_{\text{\textgreek{e}}n}^{(1,k+1)}$.
In particular, the bounds (\ref{eq:NotEnoughMassBehind}) and (\ref{eq:BoundSecondBeamchanged})
can be omitted from the proof. Moreover, the lower bound for $r_{0}$
in (\ref{eq:BoundMirror}) can be relaxed, and the exponentials in
the relations (\ref{eq:h_1_h_0_definition}) between $h_{0}(\text{\textgreek{e}}),h_{1}(\text{\textgreek{e}})$
can be replaced by polynomial functions.
\end{rem*}

\subsubsection{\label{sub:Proof-of-Proposition}Proof of Proposition \ref{prop:TheMainBootstrapBeforeTrapping}}

In this section, we will make use of the $O(\cdot)$ convention: For
any pair of functions $\mathcal{F},\mathcal{G}$ defined on the same
domain, with $\mathcal{G}\ge0$, the notation 
\[
\mathcal{F}=O(\mathcal{G})
\]
 will imply that 
\[
|\mathcal{F}|\le C\cdot\mathcal{G}
\]
 for some universal constant $C>0$ which is independent of all the
parameters in the statement of Theorem \ref{thm:TheTheorem}. We should
also remark that, throughout this proof, we will adopt the convention
on the indices stated under (\ref{eq:Domains_D}), i.\,e.: 

\begin{enumerate}

\item $U_{n}\big(v^{(-1)}\big)\doteq U_{n+1}\big(v^{(k)}\big)$. 

\item $V_{n}\big(v^{(k+l)}+c\big)\doteq V_{n-1}\big(v^{(l-1)}+c\big)$
for any integer $1\le l\le k$ and any $c\ge0$.

\end{enumerate}

In view of (\ref{eq:UpperUNonTrapping}), on $\mathcal{U}_{\text{\textgreek{e}}}^{+}$
we have 
\begin{equation}
\partial_{u}r<0<\partial_{v}r\label{eq:NonTrappingQualitativ}
\end{equation}
and 
\begin{equation}
\partial_{u}\tilde{m}\le0\le\partial_{v}\tilde{m}.\label{eq:NonTrappingMassSign}
\end{equation}

We will split the proof of Theorem \ref{thm:TheTheorem} into two
parts: In the first (and shortest) part, we will establish the bound
(\ref{eq:RoughBoundGeometry}) through a standard continuity argument.
The proof of (\ref{eq:RoughBoundGeometry}) will also yield (\ref{eq:BoundForRAwayInteractionProp})
and (\ref{eq:UpperBoundForAxisInteractionProp}). In the second (and
more extended) part, we will establish the bounds (\ref{eq:EnoughMassBehind})--(\ref{eq:BoundForMaxBeamSeparation})
by induction on $n$.

\subsubsection*{Part I: Proof of (\ref{eq:RoughBoundGeometry})--(\ref{eq:UpperBoundForAxisInteractionProp})}

Let $u_{*}>0$ be such, so that on 
\begin{equation}
\mathcal{U}_{\text{\textgreek{e}}}^{*}\doteq\mathcal{U}_{\text{\textgreek{e}}}^{+}\cap\{u<u_{*}\},
\end{equation}
 we can bound 
\begin{equation}
\Big|\log\Big(\frac{-\partial_{u}r}{1-\frac{1}{3}\Lambda r^{2}}\Big)\Big|+\Big|\log\Big(\frac{\partial_{v}r}{1-\frac{2m}{r}}\Big)\Big|\le2(h_{1}(\text{\textgreek{e}}))^{-4}\log\big((h_{3}(\text{\textgreek{e}}))^{-1}\big).\label{eq:RoughBoundGeometryBootstrap}
\end{equation}
 By showing that (\ref{eq:RoughBoundGeometry}) holds on $\mathcal{U}_{\text{\textgreek{e}}}^{*}$,
it will follow (by applying a standard continuity argument) that (\ref{eq:RoughBoundGeometry})
holds on the whole of $\mathcal{U}_{\text{\textgreek{e}}}^{+}$.

\paragraph*{\noindent Inductive formulas for $\partial_{u}r$ and $\partial_{v}r$
and proof of (\ref{eq:RoughBoundGeometry}).\emph{ }}

\noindent From equation (\ref{eq:EquationRForProof}), we can readily
derive the following renormalised equation: 
\begin{equation}
\partial_{v}\partial_{u}\Big\{\sqrt{-\frac{3}{\Lambda}}\tan^{-1}\Big(\sqrt{-\frac{\Lambda}{3}}r\Big)\Big\}=-2\frac{\tilde{m}}{r^{2}}\frac{(1-\Lambda r^{2})}{(1-\frac{1}{3}\Lambda r^{2})}\Big(\frac{-\partial_{u}r}{1-\frac{1}{3}\Lambda r^{2}}\Big)\Big(\frac{\partial_{v}r}{1-\frac{2m}{r}}\Big)\label{eq:RenormalisedREquation}
\end{equation}
Let $n\ge1$, $0\le i\le k$, $i\le j\le k+i$, $\bar{u}<u_{*}$ and
$v_{b}$ be such that 
\[
U_{n}\big(v^{(i)}+\frac{4}{\sqrt{-\Lambda}}h_{2}(\text{\textgreek{e}})\big)\le\bar{u}\le U_{n}\big(v^{(i-1)}\big)
\]
and 
\[
V_{n}\big(v^{(j)}\big)\le v_{b}\le V_{n}\big(v^{(j)}+\frac{4}{\sqrt{-\Lambda}}h_{2}(\text{\textgreek{e}})\big).
\]
 Integrating equation (\ref{eq:RenormalisedREquation}) for $u=\bar{u}$
from $v=V_{n}\big(v^{(j)}\big)$ up to $v=v_{b}$, using also the
fact that 
\[
\partial_{u}\Big\{\sqrt{-\frac{3}{\Lambda}}\tan^{-1}\Big(\sqrt{-\frac{\Lambda}{3}}r\Big)\Big\}=\frac{\partial_{u}r}{1-\frac{1}{3}\Lambda r^{2}},
\]
we obtain: 
\begin{align}
\frac{-\partial_{u}r}{1-\frac{1}{3}\Lambda r^{2}}\Bigg|_{(\bar{u},V_{n}(v^{(j)}))}=\frac{-\partial_{u}r}{1-\frac{1}{3}\Lambda r^{2}} & \Bigg|_{(\bar{u},v_{b})}+\label{eq:SimpleIntegrationInVOverBeam}\\
 & +O\Big(\sup_{\{u=\bar{u}\}\cap\{V_{n}(v^{(j)}+\frac{4}{\sqrt{-\Lambda}}h_{2}(\text{\textgreek{e}}))\le v\le V_{n}(v^{(j-1)})\}}\Big|\frac{\tilde{m}}{r^{2}}\Big(\frac{-\partial_{u}r}{1-\frac{1}{3}\Lambda r^{2}}\Big)\Big(\frac{\partial_{v}r}{1-\frac{2m}{r}}\Big)\Big|h_{2}(\text{\textgreek{e}})\Big).\nonumber 
\end{align}
Using the bootstrap bound (\ref{eq:RoughBoundGeometryBootstrap}),
combined with the trivial bounds 
\begin{equation}
\tilde{m}\le\tilde{m}|_{\mathcal{I}}=\frac{\text{\textgreek{e}}}{\sqrt{-\Lambda}}
\end{equation}
and 
\begin{equation}
r\ge r_{0}
\end{equation}
(following from (\ref{eq:DerivativeTildeUMass}), (\ref{eq:DerivativeTildeVMass})
and (\ref{eq:NonTrappingQualitativ})), as well as the relation (\ref{eq:h_2definition})
for $h_{2}(\text{\textgreek{e}})$, the relation (\ref{eq:SimpleIntegrationInVOverBeam})
yields: 
\begin{equation}
\frac{-\partial_{u}r}{1-\frac{1}{3}\Lambda r^{2}}\Bigg|_{(\bar{u},V_{n}(v^{(j)}))}=\frac{-\partial_{u}r}{1-\frac{1}{3}\Lambda r^{2}}\Bigg|_{(\bar{u},v_{b})}+O((h_{2}(\text{\textgreek{e}}))^{1/2}).\label{eq:D_uRdoesntchangemuchOverBeam}
\end{equation}
Similarly, integrating (\ref{eq:RenormalisedREquation}) for $v=\bar{v}$
from $u=U_{n}\big(v^{(i)}\big)$ up to $u=u_{b}$ for any $V_{n}\big(v^{(j)}+\frac{4}{\sqrt{-\Lambda}}h_{2}(\text{\textgreek{e}})\big)\le\bar{v}\le V_{n}\big(v^{(j-1)}\big)$
and any $U_{n}\big(v^{(i)}\big)\le u_{b}\le U_{n}\big(v^{(i)}+\frac{4}{\sqrt{-\Lambda}}h_{2}(\text{\textgreek{e}})\big)$
(assuming that $u_{b}<u_{*}$), we infer: 
\begin{equation}
\frac{\partial_{v}r}{1-\frac{1}{3}\Lambda r^{2}}\Bigg|_{(U_{n}(v^{(i)}),\bar{v})}=\frac{\partial_{v}r}{1-\frac{1}{3}\Lambda r^{2}}\Bigg|_{(u_{b},\bar{v})}+O((h_{2}(\text{\textgreek{e}}))^{1/2}).\label{eq:D_vRdoesntchangemuchOverBeam}
\end{equation}
By multiplying and dividing each factor with $1-\frac{2m}{r}=1-\frac{2\tilde{m}}{r}-\frac{1}{3}\Lambda r^{2}$,
the relations (\ref{eq:D_uRdoesntchangemuchOverBeam}) and (\ref{eq:D_vRdoesntchangemuchOverBeam})
are equivalent to 
\begin{align}
\frac{-\partial_{u}r}{1-\frac{2m}{r}}\Bigg|_{(\bar{u},V_{n}(v^{(j)}))}= & \frac{-\partial_{u}r}{1-\frac{2m}{r}}\Bigg|_{(\bar{u},v_{b})}\cdot\Bigg(\frac{1-\frac{2\tilde{m}}{r(1-\frac{1}{3}\Lambda r^{2})}\Big|_{(\bar{u},v_{b})}}{1-\frac{2\tilde{m}}{r(1-\frac{1}{3}\Lambda r^{2})}\Big|_{(\bar{u},V_{n}(v^{(j)}))}}+O((h_{2}(\text{\textgreek{e}}))^{1/2})\Bigg)\label{eq:KappaBarChangeOverBeam}
\end{align}
and 
\begin{align}
\frac{\partial_{v}r}{1-\frac{2m}{r}}\Bigg|_{(U_{n}(v^{(i)}),\bar{v})}= & \frac{\partial_{v}r}{1-\frac{2m}{r}}\Bigg|_{(u_{b},\bar{v})}\cdot\Bigg(\frac{1-\frac{2\tilde{m}}{r(1-\frac{1}{3}\Lambda r^{2})}\Big|_{(u_{b},\bar{v})}}{1-\frac{2\tilde{m}}{r(1-\frac{1}{3}\Lambda r^{2})}\Big|_{(U_{n}(v^{(i)}),\bar{v})}}+O((h_{2}(\text{\textgreek{e}}))^{1/2})\Bigg).\label{eq:KappaChangeOverBeam}
\end{align}

\begin{rem*}
In the vacuum case, where $\tilde{m}$ is contant, the factors in
the right hand side of (\ref{eq:KappaBarChangeOverBeam}) and (\ref{eq:KappaChangeOverBeam})
become identically $1$. In our case, however, where matter is present,
by relaxing our definition of $h_{2}$ and considering the limit $h_{2}\rightarrow0$
for fixed $\text{\textgreek{e}}$, the dominant terms in the factors
in the right hand side of (\ref{eq:KappaBarChangeOverBeam}) and (\ref{eq:KappaChangeOverBeam}),
i.\,e.~the first summands, do \emph{not} converge to $1$. This
is because, in this limit, while the function $r$ remains $C^{1}$,
the renormalised Hawking mass $\tilde{m}$ has a jump discontinuity
across the beam. 
\end{rem*}
Since $T_{uu}=T_{vv}=0$ on $\mathcal{R}_{\text{\textgreek{e}}n}^{(i,j)}$
for any $n\ge1$, any $0\le i\le k$ and any $i\le j\le k+i+1$, the
relations (\ref{eq:DerivativeInUDirectionKappa})--(\ref{eq:DerivativeInVDirectionKappaBar})
imply that: 
\begin{equation}
\partial_{v}\Big(\frac{-\partial_{u}r}{1-\frac{2m}{r}}\Big)\Bigg|_{\mathcal{R}_{\text{\textgreek{e}}n}^{(i,j)}}=\partial_{u}\Big(\frac{\partial_{v}r}{1-\frac{2m}{r}}\Big)\Bigg|_{\mathcal{R}_{\text{\textgreek{e}}n}^{(i,j)}}=0.
\end{equation}
 In particular, along lines of the form $\{u=\bar{u}\}$, the quantity
$\frac{-\partial_{u}r}{1-\frac{2m}{r}}$ remains constant on $\{u=\bar{u}\}\cap\mathcal{R}_{\text{\textgreek{e}}n}^{(i,j)}$
for any $\bar{u}<u_{*}$ such that $\{u=\bar{u}\}\cap\mathcal{R}_{\text{\textgreek{e}}n}^{(i,j)}$
is non-trivial. In view of (\ref{eq:KappaBarChangeOverBeam}), the
quantities $\frac{-\partial_{u}r}{1-\frac{2m}{r}}\Bigg|_{\{u=\bar{u}\}\cap\mathcal{R}_{\text{\textgreek{e}}n}^{(i,j+1)}}$
and $\frac{-\partial_{u}r}{1-\frac{2m}{r}}\Bigg|_{\{u=\bar{u}\}\cap\mathcal{R}_{\text{\textgreek{e}}n}^{(i,j)}}$
(for any $\bar{u}<u_{*}$ such that $\{u=\bar{u}\}\cap\mathcal{R}_{\text{\textgreek{e}}n}^{(i,j)}$
is non-trivial) are related by 
\begin{equation}
\frac{-\partial_{u}r}{1-\frac{2m}{r}}\Bigg|_{\{u=\bar{u}\}\cap\mathcal{R}_{\text{\textgreek{e}}n}^{(i,j)}}=\frac{-\partial_{u}r}{1-\frac{2m}{r}}\Bigg|_{\{u=\bar{u}\}\cap\mathcal{R}_{\text{\textgreek{e}}n}^{(i,j+1)}}\cdot\Bigg(\frac{1-\frac{2\tilde{m}}{r(1-\frac{1}{3}\Lambda r^{2})}\Big|_{(\bar{u},V_{n}(v^{(j)}))}}{1-\frac{2\tilde{m}}{r(1-\frac{1}{3}\Lambda r^{2})}\Big|_{(\bar{u},V_{n}(v^{(j)}+\frac{4}{\sqrt{-\Lambda}}h_{2}(\text{\textgreek{e}})))}}+O((h_{2}(\text{\textgreek{e}}))^{1/2})\Bigg).\label{eq:KappaBarChangeAdjacentDomains}
\end{equation}
Similarly, the quantity $\frac{\partial_{v}r}{1-\frac{2m}{r}}$ remains
constant along segments of the form $\{v=\bar{v}\}\cap\mathcal{R}_{\text{\textgreek{e}}n}^{(i,j)}$,
and $\frac{\partial_{v}r}{1-\frac{2m}{r}}\Bigg|_{\{v=\bar{v}\}\cap\mathcal{R}_{\text{\textgreek{e}}n}^{(i+1,j)}}$
and $\frac{\partial_{v}r}{1-\frac{2m}{r}}\Bigg|_{\{v=\bar{v}\}\cap\mathcal{R}_{\text{\textgreek{e}}n}^{(i,j)}}$
are related (in view of (\ref{eq:KappaChangeOverBeam})) by
\begin{align}
\frac{\partial_{v}r}{1-\frac{2m}{r}}\Bigg|_{\{v=\bar{v}\}\cap\mathcal{R}_{\text{\textgreek{e}}n}^{(i,j)}}= & \frac{\partial_{v}r}{1-\frac{2m}{r}}\Bigg|_{\{v=\bar{v}\}\cap\mathcal{R}_{\text{\textgreek{e}}n}^{(i+1,j)}}\cdot\Bigg(\frac{1-\frac{2\tilde{m}}{r(1-\frac{1}{3}\Lambda r^{2})}\Big|_{(U_{n}(v^{(i)}),\bar{v})}}{1-\frac{2\tilde{m}}{r(1-\frac{1}{3}\Lambda r^{2})}\Big|_{(U_{n}(v^{(i)}+\frac{4}{\sqrt{-\Lambda}}h_{2}(\text{\textgreek{e}})),\bar{v})}}+O((h_{2}(\text{\textgreek{e}}))^{1/2})\Bigg).\label{eq:KappaChangeAdjacentRegions}
\end{align}

We infer, therefore, that for any point $(\bar{u},\bar{v})\in\mathcal{R}_{\text{\textgreek{e}},n}^{(i,j)}$
for some $n\ge2$, $0\le i\le k$ and $i\le j\le k+i+1$ such that
$\bar{u}<u_{*}$, the following relations hold between $(\bar{u},\bar{v})$
and $(\bar{u}-v_{0},\bar{v}-v_{0})\in\mathcal{R}_{\text{\textgreek{e}},n-1}^{(i,j)}$:
\begin{align}
\frac{-\partial_{u}r}{1-\frac{2m}{r}}\Bigg|_{(\bar{u},\bar{v})}=\frac{-\partial_{u}r}{1-\frac{2m}{r}}\Bigg|_{(\bar{u}-v_{0},\bar{v}-v_{0})} & \times\prod_{\bar{j}=j}^{k+i}\Bigg(\frac{1-\frac{2\tilde{m}}{r(1-\frac{1}{3}\Lambda r^{2})}\Big|_{(\bar{u},V_{n}(v^{(\bar{j})}))}}{1-\frac{2\tilde{m}}{r(1-\frac{1}{3}\Lambda r^{2})}\Big|_{(\bar{u},V_{n}(v^{(\bar{j})}+\frac{4}{\sqrt{-\Lambda}}h_{2}(\text{\textgreek{e}})))}}+O((h_{2}(\text{\textgreek{e}}))^{1/2})\Bigg)\times\label{eq:TotalChangeKappaBarEachIteration}\\
 & \times\prod_{\bar{i}=i}^{k+i}\Bigg(\frac{1-\frac{2\tilde{m}}{r(1-\frac{1}{3}\Lambda r^{2})}\Big|_{(U_{n}(v^{(\bar{i})}),\bar{u})}}{1-\frac{2\tilde{m}}{r(1-\frac{1}{3}\Lambda r^{2})}\Big|_{(U_{n}(v^{(\bar{i})}+\frac{4}{\sqrt{-\Lambda}}h_{2}(\text{\textgreek{e}})),\bar{u})}}+O((h_{2}(\text{\textgreek{e}}))^{1/2})\Bigg)\times\nonumber \\
 & \times\prod_{\bar{j}=k+i+1}^{k+j}\Bigg(\frac{1-\frac{2\tilde{m}}{r(1-\frac{1}{3}\Lambda r^{2})}\Big|_{(\bar{u}-v_{0},V_{n}(v^{(\bar{j})}))}}{1-\frac{2\tilde{m}}{r(1-\frac{1}{3}\Lambda r^{2})}\Big|_{(\bar{u}-v_{0},V_{n}(v^{(\bar{j})}+\frac{4}{\sqrt{-\Lambda}}h_{2}(\text{\textgreek{e}})))}}+O((h_{2}(\text{\textgreek{e}}))^{1/2})\Bigg)\nonumber 
\end{align}
and 
\begin{align}
\frac{\partial_{v}r}{1-\frac{2m}{r}}\Bigg|_{(\bar{u},\bar{v})}=\frac{\partial_{v}r}{1-\frac{2m}{r}}\Bigg|_{(\bar{u}-v_{0},\bar{v}-v_{0})} & \times\prod_{\bar{i}=i}^{j-1}\Bigg(\frac{1-\frac{2\tilde{m}}{r(1-\frac{1}{3}\Lambda r^{2})}\Big|_{(U_{n}(v^{(\bar{i})}),\bar{v})}}{1-\frac{2\tilde{m}}{r(1-\frac{1}{3}\Lambda r^{2})}\Big|_{(U_{n}(v^{(\bar{i})}+\frac{4}{\sqrt{-\Lambda}}h_{2}(\text{\textgreek{e}})),\bar{v})}}+O((h_{2}(\text{\textgreek{e}}))^{1/2})\Bigg)\times\label{eq:TotalChangeKappaEachIteration}\\
 & \times\prod_{\bar{j}=j}^{k+j}\Bigg(\frac{1-\frac{2\tilde{m}}{r(1-\frac{1}{3}\Lambda r^{2})}\Big|_{(\bar{u}-v_{0},V_{n}(v^{(\bar{j})}))}}{1-\frac{2\tilde{m}}{r(1-\frac{1}{3}\Lambda r^{2})}\Big|_{(\bar{u}-v_{0},V_{n}(v^{(\bar{j})}+\frac{4}{\sqrt{-\Lambda}}h_{2}(\text{\textgreek{e}})))}}+O((h_{2}(\text{\textgreek{e}}))^{1/2})\Bigg)\times\nonumber \\
 & \times\prod_{\bar{i}=j}^{k+i}\Bigg(\frac{1-\frac{2\tilde{m}}{r(1-\frac{1}{3}\Lambda r^{2})}\Big|_{(U_{n}(v^{(\bar{i})}),\bar{v}-v_{0})}}{1-\frac{2\tilde{m}}{r(1-\frac{1}{3}\Lambda r^{2})}\Big|_{(U_{n}(v^{(\bar{i})}+\frac{4}{\sqrt{-\Lambda}}h_{2}(\text{\textgreek{e}})),\bar{v}-v_{0})}}+O((h_{2}(\text{\textgreek{e}}))^{1/2})\Bigg).\nonumber 
\end{align}

\begin{figure}[h] 
\centering 
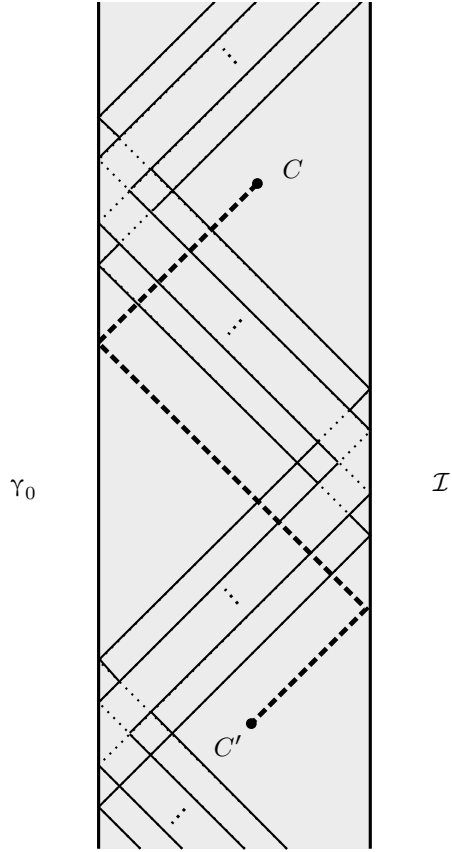 
\caption{In order to obtain the formula \eqref{eq:TotalChangeKappaBarEachIteration} relating $\frac{-\partial_{u}r}{1-\frac{2m}{r}}$ at the point $C=(\bar{u},\bar{v})$ with the same quantity at the point $C'=(\bar{u}-v_{0},\bar{v}-v_{0})$, we apply the relations \eqref{eq:KappaBarChangeAdjacentDomains} and \eqref{eq:KappaChangeAdjacentRegions} along the dashed path depicted above, using also the reflecting gauge condition on $\gamma_{0}$ and $\mathcal{I}$. }
\end{figure}

The relation (\ref{eq:TotalChangeKappaBarEachIteration}) is obtained
as follows (see also Figure 6.1): First, (\ref{eq:KappaBarChangeAdjacentDomains})
determines the evolution of $\frac{-\partial_{u}r}{1-\frac{2m}{r}}$
(according to (\ref{eq:KappaBarChangeAdjacentDomains})) along the
line $\{u=\bar{u}\}$ in the past direction, from $(\bar{u},\bar{v})$
up to $\text{\textgreek{g}}_{0}$. Then, using the boundary relation
\begin{equation}
\frac{-\partial_{u}r}{1-\frac{2m}{r}}\Big|_{\text{\textgreek{g}}_{0}}=\frac{\partial_{v}r}{1-\frac{2m}{r}}\Big|_{\text{\textgreek{g}}_{0}},
\end{equation}
one repeats the same procedure for $\frac{\partial_{v}r}{1-\frac{2m}{r}}$
along $\{v=\bar{u}\}$ from $\text{\textgreek{g}}_{0}$ up to $\mathcal{I}$.
Finally, using 
\begin{equation}
\frac{-\partial_{u}r}{1-\frac{2m}{r}}\Big|_{\mathcal{I}}=\frac{\partial_{v}r}{1-\frac{2m}{r}}\Big|_{\mathcal{I}},
\end{equation}
and following the evolution of $\frac{-\partial_{u}r}{1-\frac{2m}{r}}$
along $\{u=\bar{u}-v_{0}\}$ from $\mathcal{I}$ up to $(\bar{u}-v_{0},\bar{v}-v_{0})$,
one arrives at (\ref{eq:TotalChangeKappaBarEachIteration}). The relation
(\ref{eq:TotalChangeKappaEachIteration}) is similarly obtained by
following the same procedure along the lines $\{v=\bar{v}\}$ (up
to $\mathcal{I}$), $\{u=\bar{v}-v_{0}\}$ (from $\mathcal{I}$ up
to $\text{\textgreek{g}}_{0}$) and $\{v=\bar{v}-v_{0}\}$ (from $\mathcal{I}$
up to $(\bar{u}-v_{0},\bar{v}-v_{0})$).

In view of the bound 
\begin{equation}
1-\frac{2m}{r}\ge h_{3}(\text{\textgreek{e}})\label{eq:LowerBoundTrappingParameter}
\end{equation}
 on $\mathcal{U}_{\text{\textgreek{e}}}^{+}$ (see (\ref{eq:UpperUNonTrapping})),
we can estimate in the region $\{r\le\text{\textgreek{e}}^{1/2}(-\Lambda)^{1/2}\}\cap\mathcal{U}_{\text{\textgreek{e}}}^{+}$:
\begin{equation}
1-\frac{2\tilde{m}}{r(1-\frac{1}{3}\Lambda r^{2})}=\frac{1-\frac{2m}{r}}{1-\frac{1}{3}\Lambda r^{2}}\ge\frac{1}{2}h_{1}(\text{\textgreek{e}}).\label{eq:FirstTrivial}
\end{equation}
On the other hand, in the region $\{r\ge\text{\textgreek{e}}^{1/2}(-\Lambda)^{1/2}\}\cap\mathcal{U}_{\text{\textgreek{e}}}^{+}$,
using (\ref{eq:MassInfinity}) to bound $\tilde{m}$ we can trivially
estimate (in view also of (\ref{eq:h_1_h_0_definition})): 
\begin{equation}
1-\frac{2\tilde{m}}{r(1-\frac{1}{3}\Lambda r^{2})}\ge1-\frac{2\text{\textgreek{e}}}{\text{\textgreek{e}}^{\frac{1}{2}}}\ge h_{1}(\text{\textgreek{e}}).\label{eq:SecondTrivial}
\end{equation}
Combining (\ref{eq:FirstTrivial}) and (\ref{eq:SecondTrivial}),
using also the fact that $\tilde{m}\ge\tilde{m}|_{\text{\textgreek{g}}_{0}}=0$
on $\mathcal{U}_{\text{\textgreek{e}}}^{+}$, we can bound $1-\frac{2\tilde{m}}{r(1-\frac{1}{3}\Lambda r^{2})}$
from above and below everywhere on $\mathcal{U}_{\text{\textgreek{e}}}^{+}$
as: 
\begin{equation}
\frac{1}{2}h_{3}(\text{\textgreek{e}})\le1-\frac{2\tilde{m}}{r(1-\frac{1}{3}\Lambda r^{2})}\le1.\label{eq:TrivialFactor}
\end{equation}
Thus, by considering the logarithm of the relations (\ref{eq:TotalChangeKappaBarEachIteration})--(\ref{eq:TotalChangeKappaEachIteration})
and noting that the resulting right hand side contains $\sim k=\lceil1/h_{1}(\text{\textgreek{e}})\rceil$
summands, each controlled with the help of (\ref{eq:TrivialFactor}),
we readily obtain for any $n\ge2$, $0\le i\le k$ and $i\le j\le k+i+1$
and any point $(\bar{u},\bar{v})\in\mathcal{R}_{\text{\textgreek{e}}n}^{(i,j)}$
with $\bar{u}<u_{*}$: 
\begin{equation}
\Bigg|\log\Big(\frac{-\partial_{u}r}{1-\frac{2m}{r}}\Big)\Big|_{(\bar{u},\bar{v})}-\log\Big(\frac{-\partial_{u}r}{1-\frac{2m}{r}}\Big)\Big|_{(\bar{u}-v_{0},\bar{v}-v_{0})}\Bigg|\le\frac{C}{h_{1}(\text{\textgreek{e}})}\log\big((h_{3}(\text{\textgreek{e}}))^{-1}\big)\label{eq:RoughBoundForIteration}
\end{equation}
and 
\begin{equation}
\Bigg|\log\Big(\frac{\partial_{v}r}{1-\frac{2m}{r}}\Big)\Big|_{(\bar{u},\bar{v})}-\log\Big(\frac{\partial_{v}r}{1-\frac{2m}{r}}\Big)\Big|_{(\bar{u}-v_{0},\bar{v}-v_{0})}\Bigg|\le\frac{C}{h_{1}(\text{\textgreek{e}})}\log\big((h_{3}(\text{\textgreek{e}}))^{-1}\big).\label{eq:RoughBoundForIteration-1}
\end{equation}
In view of (\ref{eq:KappaBarChangeOverBeam})--(\ref{eq:KappaChangeOverBeam}),
the bounds (\ref{eq:RoughBoundForIteration}) and (\ref{eq:RoughBoundForIteration-1})
(stated in the case when $(\bar{u},\bar{v})$ belongs to a vacuum
region $\mathcal{R}_{\text{\textgreek{e}}n}^{(i,j)}$) also hold when
$(\bar{u},\bar{v})$ belongs to a beam, i.\,e.~when $U_{n}\big(v^{(i)}\big)\le\bar{u}\le U_{n}\big(v^{(i-1)}+\frac{4}{\sqrt{-\Lambda}}h_{2}(\text{\textgreek{e}})\big)$
or $V_{n}\big(v^{(j)}\big)\le\bar{v}\le V_{n}\big(v^{(i-1)}+\frac{4}{\sqrt{-\Lambda}}h_{2}(\text{\textgreek{e}})\big)$
for some $n\ge2$, $0\le i\le k$ and $i\le j\le k+i+1$. Therefore,
for any $n\ge2$, the bounds (\ref{eq:KappaBarChangeOverBeam})--(\ref{eq:KappaChangeOverBeam})
hold on the whole of 
\begin{equation}
\mathcal{U}_{\text{\textgreek{e}};n}^{*}\doteq\{U_{n}(v^{(k)})\le u\le U_{n+1}(v^{(k)})\}\cap\mathcal{U}_{\text{\textgreek{e}}}^{*}.
\end{equation}

From (\ref{eq:DefinitionUntrappedRegion}) and the definition (\ref{eq:defNf}),
it follows that 
\begin{equation}
n_{f}\le(h_{1}(\text{\textgreek{e}}))^{-2}.\label{eq:UpperBoundNf}
\end{equation}
Since $n\le n_{f}$ (because $\mathcal{U}_{\text{\textgreek{e}}}^{*}\subset\mathcal{U}_{\text{\textgreek{e}}}^{+}$),
by substituting $(\bar{u},\bar{v})\rightarrow(\bar{u}-v_{0},\bar{v}-v_{0})$
in (\ref{eq:RoughBoundForIteration})--(\ref{eq:RoughBoundForIteration-1})
$n-2$ times and using (\ref{eq:UpperBoundNf}), (\ref{eq:ConditionOnDvRInitiallyFamily}),
(\ref{eq:TrivialFactor}) as well as the Cauchy stability estimate
of Proposition \ref{prop:CauchyStabilityOfAdS} for the region $\{0\le u\le2v_{0}\}$,
we readily obtain 
\begin{equation}
\sup_{\mathcal{U}_{\text{\textgreek{e}}}^{*}}\Bigg\{\Big|\log\Big(\frac{-\partial_{u}r}{1-\frac{1}{3}\Lambda r^{2}}\Big)\Big|+\Big|\log\Big(\frac{\partial_{v}r}{1-\frac{2m}{r}}\Big)\Big|\Bigg\}\le\frac{C}{(h_{1}(\text{\textgreek{e}}))^{3}}\log\big((h_{3}(\text{\textgreek{e}}))^{-1}\big).\label{eq:ImprovedRoughBoundBootstrap}
\end{equation}
Thus, (\ref{eq:RoughBoundGeometry}) holds on $\mathcal{U}_{\text{\textgreek{e}}}^{*}$
in view of the relation (\ref{eq:h_1_h_0_definition}) for the parameter
$h_{1}(\text{\textgreek{e}})$ (provided $\text{\textgreek{e}}_{0}$
is small enough). Therefore (as explained in the beginning of the
proof), a standard continuity argument yields that (\ref{eq:RoughBoundGeometry})
actually holds on the whole of $\mathcal{U}_{\text{\textgreek{e}}}^{+}$.

\paragraph*{\noindent Proof of (\ref{eq:BoundForRAwayInteractionProp}) and (\ref{eq:UpperBoundForAxisInteractionProp}).\emph{ }}

\noindent For any $1\le n\le n_{f}$, we can bound in view of the
definition (\ref{eq:DefinitionV_j}) of $v^{(j)}$ and the bound (\ref{eq:RoughBoundGeometry}):
\begin{align}
\Bigg|\tan^{-1}\Big(\sqrt{-\frac{\Lambda}{3}}r\Big)\Big|_{\big(U_{n}(v^{(0)}+\frac{4}{\sqrt{-\Lambda}}h_{2}(\text{\textgreek{e}})),V_{n}(v^{(k)}+\frac{4}{\sqrt{-\Lambda}}h_{2}(\text{\textgreek{e}}))\big)}- & \tan^{-1}\Big(\sqrt{-\frac{\Lambda}{3}}r\Big)\Big|_{\mathcal{I}\cap\{v=V_{n}(v^{(k)}+\frac{4}{\sqrt{-\Lambda}}h_{2}(\text{\textgreek{e}}))\}}\Bigg|\label{eq:BoundInteractionRegionAway}\\
= & \sqrt{-\frac{\Lambda}{3}}\int_{U_{n}(v^{(k)}+\frac{4}{\sqrt{-\Lambda}}h_{2}(\text{\textgreek{e}}))}^{U_{n}(v^{(0)}+\frac{4}{\sqrt{-\Lambda}}h_{2}(\text{\textgreek{e}}))}\frac{-\partial_{u}r}{1-\frac{1}{3}\Lambda r^{2}}\Big|_{\big(u,V_{n}(v^{(k)}+\frac{4}{\sqrt{-\Lambda}}h_{2}(\text{\textgreek{e}}))\big)}\, du\nonumber \\
\le & \frac{C\sqrt{-\Lambda}}{(h_{1}(\text{\textgreek{e}}))^{4}}\log\big((h_{3}(\text{\textgreek{e}}))^{-1}\big)\big|v^{(k)}-v^{(0)}\big|\nonumber \\
\le & \frac{C\text{\textgreek{e}}}{(h_{1}(\text{\textgreek{e}}))^{6}}\log\big((h_{3}(\text{\textgreek{e}}))^{-1}\big)\nonumber 
\end{align}
and 
\begin{align}
\Bigg|\tan^{-1}\Big(\sqrt{-\frac{\Lambda}{3}}r\Big)\Big|_{\big(U_{n}(v^{(k)}+\frac{4}{\sqrt{-\Lambda}}h_{2}(\text{\textgreek{e}})),V_{n}(v^{(k+1)}+\frac{4}{\sqrt{-\Lambda}}h_{2}(\text{\textgreek{e}}))\big)}- & \tan^{-1}\Big(\sqrt{-\frac{\Lambda}{3}}r\Big)\Big|_{\text{\textgreek{g}}_{0}\cap\{v=V_{n}(v^{(k+1)}+\frac{4}{\sqrt{-\Lambda}}h_{2}(\text{\textgreek{e}}))\}}\Bigg|\label{eq:BoundInteractionRegionNear}\\
= & \sqrt{-\frac{\Lambda}{3}}\int_{U_{n}(v^{(k)}+\frac{4}{\sqrt{-\Lambda}}h_{2}(\text{\textgreek{e}}))}^{U_{n}(v^{(0)}+\frac{4}{\sqrt{-\Lambda}}h_{2}(\text{\textgreek{e}}))}\frac{-\partial_{u}r}{1-\frac{1}{3}\Lambda r^{2}}\Big|_{\big(u,V_{n}(v^{(k+1)}+\frac{4}{\sqrt{-\Lambda}}h_{2}(\text{\textgreek{e}}))\big)}\, du\nonumber \\
\le & \frac{C\sqrt{-\Lambda}}{(h_{1}(\text{\textgreek{e}}))^{4}}\log\big((h_{3}(\text{\textgreek{e}}))^{-1}\big)\big|v^{(k)}-v^{(0)}\big|\nonumber \\
\le & \frac{C\text{\textgreek{e}}}{(h_{1}(\text{\textgreek{e}}))^{6}}\log\big((h_{3}(\text{\textgreek{e}}))^{-1}\big).\nonumber 
\end{align}
From (\ref{eq:BoundInteractionRegionAway}) and (\ref{eq:BoundInteractionRegionNear})
we readily obtain (\ref{eq:BoundForRAwayInteractionProp}) and (\ref{eq:UpperBoundForAxisInteractionProp}),
respectively, in view of the relations (\ref{eq:h_1_h_0_definition})
and (\ref{eq:h_3definition}) for $h_{1},h_{3}$, respectively, and
the fact that $r|_{\text{\textgreek{g}}_{0}}=r_{0}$, $r|_{\mathcal{I}}=+\infty$.

\subsubsection*{Part II: Proof of (\ref{eq:EnoughMassBehind})--(\ref{eq:BoundForMaxBeamSeparation})}

We will now proceed to establish the bounds (\ref{eq:EnoughMassBehind})--(\ref{eq:BoundForMaxBeamSeparation}).
To this end, we will first derive some useful estimates for the differences
of the renormalied masses $\tilde{m}_{n}^{(i,j)}$ associated to the
vacuum regions around each interaction region $\mathcal{N}_{\text{\textgreek{e}}n}^{(i,j)}$.%
\footnote{A relation for the change the mass differences of two intersecting,
infinitely thin null dust beams was also obtained in \cite{PoissonIsrael1990}.%
}

\paragraph*{\noindent Relations for the change in the mass difference of the
beams. }

\noindent Let us introduce the notion of the mass difference for the
beams $\{U_{n}(v^{(i)})\le u\le U_{n}(v^{(i)}+\frac{4}{\sqrt{-\Lambda}}h_{2}(\text{\textgreek{e}}))\}$
and $\{V_{n}(v^{(j)})\le v\le V_{n}(v^{(j)}+\frac{4}{\sqrt{-\Lambda}}h_{2}(\text{\textgreek{e}}))\}$
around their interaction region $\mathcal{N}_{\text{\textgreek{e}}n}^{(i,j)}$:
For any $1\le n\le n_{f}$, $0\le i\le k$ and $i+1\le j\le k+i$,
we define the initial mass differences 
\begin{align}
(\mathfrak{D}_{-}\tilde{m})_{n}^{(i,j)} & \doteq\tilde{m}_{n}^{(i+1,j+1)}-\tilde{m}_{n}^{(i,j+1)}\label{eq:DefinitionIncomingEnergy}\\
(\overline{\mathfrak{D}}_{-}\tilde{m})_{n}^{(i,j)} & \doteq\tilde{m}_{n}^{(i+1,j)}-\tilde{m}_{n}^{(i+1,j+1)}\nonumber 
\end{align}
and the final mass differences 
\begin{align}
(\mathfrak{D}_{+}\tilde{m})_{n}^{(i,j)} & \doteq\tilde{m}_{n}^{(i+1,j)}-\tilde{m}_{n}^{(i,j)}\label{eq:DefinitionOutcomingEnergy}\\
(\overline{\mathfrak{D}}_{+}\tilde{m})_{n}^{(i,j)} & \doteq\tilde{m}_{n}^{(i,j)}-\tilde{m}_{n}^{(i,j+1)}.\nonumber 
\end{align}
Note that $(\mathfrak{D}_{-}\tilde{m})_{n}^{(i,j)}$ and $(\mathfrak{D}_{+}\tilde{m})_{n}^{(i,j)}$
are the mass differences around the outgoing beam $\{U_{n}(v^{(i)})\le u\le U_{n}(v^{(i)}+\frac{4}{\sqrt{-\Lambda}}h_{2}(\text{\textgreek{e}}))\}$
before and after crossing the region $\mathcal{N}_{\text{\textgreek{e}}n}^{(i,j)}$,
respectively, while $(\overline{\mathfrak{D}}_{-}\tilde{m})_{n}^{(i,j)}$
and $(\overline{\mathfrak{D}}_{+}\tilde{m})_{n}^{(i,j)}$ are the
mass differences around the ingoing beam $\{V_{n}(v^{(j)})\le v\le U_{n}(v^{(j)}+\frac{4}{\sqrt{-\Lambda}}h_{2}(\text{\textgreek{e}}))\}$
before and after crosssing the region $\mathcal{N}_{\text{\textgreek{e}}n}^{(i,j)}$.
Note the trivial identity 
\begin{equation}
(\mathfrak{D}_{-}\tilde{m})_{n}^{(i,j)}+(\overline{\mathfrak{D}}_{-}\tilde{m})_{n}^{(i,j)}=(\mathfrak{D}_{+}\tilde{m})_{n}^{(i,j)}+(\overline{\mathfrak{D}}_{+}\tilde{m})_{n}^{(i,j)}.\label{eq:ConservationOfMassDifference}
\end{equation}
\begin{figure}[h] 
\centering 
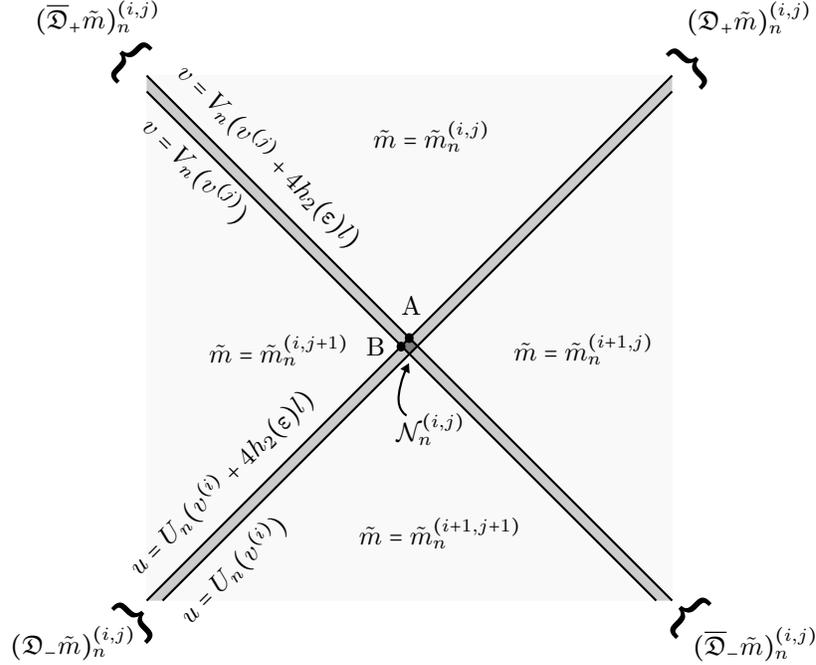 
\caption{Schematic depiction of two intersecting beams, with associated incoming and outcoming mass differences $(\mathfrak{D}_{-}\tilde{m})_{n}^{(i,j)}$, $(\overline{\mathfrak{D}}_{-}\tilde{m})_{n}^{(i,j)}$ and $(\mathfrak{D}_{+}\tilde{m})_{n}^{(i,j)}$, $(\overline{\mathfrak{D}}_{+}\tilde{m})_{n}^{(i,j)}$, respectively. The point $A$  satisfies $r(A)=r_{n}^{(i,j)}$, while the point B satisfies $r(B)=\bar{r}_{n}^{(i,j)}$. For simplicity, we have used the shorthand notation $l=(-\Lambda)^{-1/2}$.}
\end{figure}

We will establish the following bounds for any $1\le n\le n_{f}$,
$1\le i\le k$ and $i+1\le j\le k+i$: 
\begin{equation}
(\overline{\mathfrak{D}}_{+}\tilde{m})_{n}^{(i,j)}=(\overline{\mathfrak{D}}_{-}\tilde{m})_{n}^{(i,j)}\cdot\exp\Bigg(\frac{2}{\bar{r}_{n}^{(i,j)}}\frac{(\mathfrak{D}_{-}\tilde{m})_{n}^{(i,j)}}{1-\frac{2\tilde{m}_{n}^{(i+1,j)}}{\bar{r}_{n}^{(i,j)}}-\frac{1}{3}\Lambda(\bar{r}_{n}^{(i,j)})^{2}}\big(1-\mathfrak{Err}_{1,n}^{(i,j)}\big)\big(1-\mathfrak{Err}_{\backslash n}^{(i,j)}\big)\Bigg)\label{eq:MassDifferenceIncreaseInUDirection}
\end{equation}
and: 
\begin{equation}
(\mathfrak{D}_{+}\tilde{m})_{n}^{(i,j)}=(\mathfrak{D}_{-}\tilde{m})_{n}^{(i,j)}\cdot\exp\Bigg(-\frac{2}{\bar{r}_{n}^{(i,j)}}\frac{(\overline{\mathfrak{D}}_{+}\tilde{m})_{n}^{(i,j)}}{1-\frac{2\tilde{m}_{n}^{(i+1,j)}}{\bar{r}_{n}^{(i,j)}}-\frac{1}{3}\Lambda(\bar{r}_{n}^{(i,j)})^{2}}\big(1-\mathfrak{Err}_{1,n}^{(i,j)}\big)\big(1-\mathfrak{Err}_{/n}^{(i,j)}\big)\Bigg),\label{eq:MassDifferenceDecreaseInVDirection}
\end{equation}
where the terms $\mathfrak{Err}_{1,n}^{(i,j)}$ in (\ref{eq:MassDifferenceIncreaseInUDirection})
and (\ref{eq:MassDifferenceDecreaseInVDirection}) are allowed to
be different from each other, but they both satisfy the bound 
\begin{equation}
0\le\mathfrak{Err}_{1,n}^{(i,j)}\le1-\frac{\bar{r}_{n}^{(i,j)}-2\tilde{m}_{n}^{(i+1,j)}-\frac{1}{3}\Lambda(\bar{r}_{n}^{(i,j)})^{3}}{r\big|_{\big(U_{n}(v^{(i)}),V_{n}(v^{(j)}+\frac{4}{\sqrt{-\Lambda}}h_{2}(\text{\textgreek{e}}))\big)}-2\tilde{m}_{n}^{(i,j+1)}-\frac{1}{3}\Lambda r^{3}\big|_{\big(U_{n}(v^{(i)}),V_{n}(v^{(j)}+\frac{4}{\sqrt{-\Lambda}}h_{2}(\text{\textgreek{e}}))\big)}}\label{eq:BoundErrorTermForInteraction}
\end{equation}
and $\mathfrak{Err}_{\backslash n}^{(i,j)}$, $\mathfrak{Err}_{/n}^{(i,j)}$
satisfy the bounds 
\begin{equation}
0\le\mathfrak{Err}_{\backslash n}^{(i,j)}\le1-\frac{(\mathfrak{D}_{+}\tilde{m})_{n}^{(i,j)}}{(\mathfrak{D}_{-}\tilde{m})_{n}^{(i,j)}}\label{eq:BoundErrorTermIngoingInteraction}
\end{equation}
and 
\begin{equation}
0\le\mathfrak{Err}_{/n}^{(i,j)}\le1-\frac{(\overline{\mathfrak{D}}_{-}\tilde{m})_{n}^{(i,j)}}{(\overline{\mathfrak{D}}_{+}\tilde{m})_{n}^{(i,j)}}.\label{eq:BoundErrorTermOugoingInteraction}
\end{equation}
Moreover, the following estimate will be useful in the proof of (\ref{eq:BoundSecondBeamchanged}):
For any $1\le n\le n_{f}$, $1\le i\le k$ and $k+1\le j\le k+i$,
\begin{equation}
(\overline{\mathfrak{D}}_{+}\tilde{m})_{n}^{(i,j)}\ge(\overline{\mathfrak{D}}_{-}\tilde{m})_{n}^{(i,j)}\cdot\exp\Bigg(\frac{1}{5C_{0}}\frac{(\mathfrak{D}_{-}\tilde{m})_{n}^{(i,j)}}{\bar{r}_{n}^{(i,j)}}\Bigg).\label{eq:UsefulEstimate}
\end{equation}

\begin{rem*}
Notice that, as a consequence of (\ref{eq:MassDifferenceIncreaseInUDirection})
and (\ref{eq:MassDifferenceDecreaseInVDirection}), during the interaction
of the two beams at $\mathcal{N}_{\text{\textgreek{e}}n}^{(i,j)}$,
the mass difference $\overline{\mathfrak{D}}\tilde{m}$ of the ingoing
beam increases, while the mass difference $\mathfrak{D}\tilde{m}$
of the outgoing beam decrases.
\end{rem*}
\medskip{}

\noindent \emph{Proof of (\ref{eq:MassDifferenceIncreaseInUDirection})
and (\ref{eq:MassDifferenceDecreaseInVDirection}).} By differentiating
(\ref{eq:DerivativeTildeVMass}) in $u$ and using (\ref{eq:DerivativeInUDirectionKappa})
and (\ref{eq:ConservationT_vv}), we readily obtain the following
wave-type equation for $\tilde{m}$: 
\begin{equation}
\partial_{u}\partial_{v}\tilde{m}=-F(r,\tilde{m})\partial_{u}\tilde{m}\partial_{v}\tilde{m},\label{eq:WaveEquationMass}
\end{equation}
where 
\begin{equation}
F(r,\tilde{m})\doteq\frac{2}{r-2\tilde{m}-\frac{1}{3}\Lambda r^{3}}.\label{eq:OriginalFMassEquation}
\end{equation}
Note that, formally, equation (\ref{eq:WaveEquationMass}) can be
rewritten as 
\begin{equation}
\partial_{v}\log(-\partial_{u}\tilde{m})=-F(r,\tilde{m})\partial_{v}\tilde{m}\label{eq:OutgoingEquationMass}
\end{equation}
or 
\begin{equation}
\partial_{u}\log(\partial_{v}\tilde{m})=F(r,\tilde{m})(-\partial_{u}\tilde{m})\label{eq:IngoingEquationMass}
\end{equation}
(note, however, that $\log(-\partial_{u}\tilde{m})$, $\log(\partial_{v}\tilde{m})$
will not be well defined when $\partial_{u}\tilde{m}=0$ or $\partial_{v}\tilde{m}=0$).

For any $1\le n\le n_{f}$, $1\le i\le k$ and $i+1\le j\le k+i$,
integrating equation (\ref{eq:WaveEquationMass}) first in $u$, for
$U_{n}(v^{(i)}))\le u\le U_{n}(v^{(i)}+\frac{4}{\sqrt{-\Lambda}}h_{2}(\text{\textgreek{e}}))$,
and then in $v$, for $V_{n}(v^{(j)})\le v\le V_{n}(v^{(j)}+\frac{4}{\sqrt{-\Lambda}}h_{2}(\text{\textgreek{e}}))$,
we obtain: 
\begin{equation}
\tilde{m}_{n}^{(i,j)}-\tilde{m}_{n}^{(i,j+1)}=\int_{V_{n}(v^{(j)})}^{V_{n}(v^{(j)}+\frac{4}{\sqrt{-\Lambda}}h_{2}(\text{\textgreek{e}}))}\partial_{v}\tilde{m}|_{(U_{n}(v^{(i)})),v)}\cdot\exp\Big(2\int_{U_{n}(v^{(i)})}^{U_{n}(v^{(i)}+\frac{4}{\sqrt{-\Lambda}}h_{2}(\text{\textgreek{e}}))}\frac{-\partial_{u}\tilde{m}}{r-2m}\Big|_{(u,v)}\, du\Big)\, dv.\label{eq:BeforeIngoingDifference}
\end{equation}

\begin{rem*}
Note that, at the formal level, the derivation of (\ref{eq:BeforeIngoingDifference})
is easiest seen by integrating equation (\ref{eq:IngoingEquationMass})
first in $u$, then exponentiating, and then integrating in $v$.
This procedure can actually be done rigorously, since $\partial_{u}\tilde{m}<0<\partial_{v}\tilde{m}$
in the interior of $\mathcal{N}_{\text{\textgreek{e}}n}^{(i,j)}$,
in view of (\ref{eq:DerivativeTildeUMass}), (\ref{eq:DerivativeTildeVMass})
and (\ref{eq:TuuSupport})--(\ref{eq:TvvSupport}). 
\end{rem*}
In view of (\ref{eq:NonTrappingQualitativ})--(\ref{eq:NonTrappingMassSign}),
we can bound for any $U_{n}(v^{(i)})\le u\le U_{n}(v^{(i)}+\frac{4}{\sqrt{-\Lambda}}h_{2}(\text{\textgreek{e}}))$
and any $V_{n}(v^{(j)})\le v\le V_{n}(v^{(j)}+\frac{4}{\sqrt{-\Lambda}}h_{2}(\text{\textgreek{e}}))$:
\begin{align}
\bar{r}_{n}^{(i,j)}-2\tilde{m}_{n}^{(i+1,j)}-\frac{1}{3}\Lambda(\bar{r}_{n}^{(i,j)})^{3}\le\big(r- & 2m\big)\Big|_{(u,v)}\le\label{eq:TrivialBoundR}\\
 & \le r\big|_{\big(U_{n}(v^{(i)}),V_{n}(v^{(j)}+\frac{4}{\sqrt{-\Lambda}}h_{2}(\text{\textgreek{e}}))\big)}-2\tilde{m}_{n}^{(i,j+1)}-\frac{1}{3}\Lambda r^{3}\big|_{\big(U_{n}(v^{(i)}),V_{n}(v^{(j)}+\frac{4}{\sqrt{-\Lambda}}h_{2}(\text{\textgreek{e}}))\big)},\nonumber 
\end{align}
where
\[
\bar{r}_{n}^{(i,j)}\doteq r|_{(U_{n}(v^{(i)}+\frac{4}{\sqrt{-\Lambda}}h_{2}(\text{\textgreek{e}})),V_{n}(v^{(j)}))}.
\]
Therefore, using (\ref{eq:TrivialBoundR}) to estimate $\frac{1}{r-2m}$,
from (\ref{eq:BeforeIngoingDifference}) we readily infer that: 
\begin{equation}
\tilde{m}_{n}^{(i,j)}-\tilde{m}_{n}^{(i,j+1)}=\int_{V_{n}(v^{(j)})}^{V_{n}(v^{(j)}+\frac{4}{\sqrt{-\Lambda}}h_{2}(\text{\textgreek{e}}))}\partial_{v}\tilde{m}|_{(U_{n}(v^{(i)})),v)}\cdot\exp\Bigg(\frac{2}{\bar{r}_{n}^{(i,j)}}\frac{\tilde{m}|_{(U_{n}(v^{(i)}),v)}-\tilde{m}|_{(U_{n}(v^{(i)}+\frac{4}{\sqrt{-\Lambda}}h_{2}(\text{\textgreek{e}})),v)}}{1-\frac{2\tilde{m}_{n}^{(i+1,j)}}{\bar{r}_{n}^{(i,j)}}-\frac{1}{3}\Lambda(\bar{r}_{n}^{(i,j)})^{2}}\big(1-\mathfrak{Err}_{1,n}^{(i,j)}(v)\big)\Bigg)\, dv,\label{eq:bla}
\end{equation}
where, for any $V_{n}(v^{(j)})\le v\le V_{n}(v^{(j)}+\frac{4}{\sqrt{-\Lambda}}h_{2}(\text{\textgreek{e}}))$,
$\mathfrak{Err}_{1,n}^{(i,j)}(v)$ satisfies the bound (\ref{eq:BoundErrorTermForInteraction}). 

Equations (\ref{eq:DerivativeTildeUMass}), (\ref{eq:DerivativeInVDirectionKappaBar})
and (\ref{eq:EquationT_uu}) imply that, for any $V_{n}(v^{(j)})\le v\le V_{n}(v^{(j)}+\frac{4}{\sqrt{-\Lambda}}h_{2}(\text{\textgreek{e}}))$,
\begin{equation}
\partial_{v}\big(\tilde{m}|_{(U_{n}(v^{(i)}),v)}-\tilde{m}|_{(U_{n}(v^{(i)}+\frac{4}{\sqrt{-\Lambda}}h_{2}(\text{\textgreek{e}})),v)}\big)\le0
\end{equation}
and, therefore, for any $V_{n}(v^{(j)})\le v\le V_{n}(v^{(j)}+\frac{4}{\sqrt{-\Lambda}}h_{2}(\text{\textgreek{e}}))$:
\begin{equation}
(\mathfrak{D}_{+}\tilde{m})_{n}^{(i,j)}\le\tilde{m}|_{(U_{n}(v^{(i)}),v)}-\tilde{m}|_{(U_{n}(v^{(i)}+\frac{4}{\sqrt{-\Lambda}}h_{2}(\text{\textgreek{e}})),v)}\le(\mathfrak{D}_{-}\tilde{m})_{n}^{(i,j)}.\label{eq:MonotonicitymassDifference}
\end{equation}
The bound (\ref{eq:MonotonicitymassDifference}) implies that (\ref{eq:bla})
can be expressed as 
\begin{equation}
\tilde{m}_{n}^{(i,j)}-\tilde{m}_{n}^{(i,j+1)}=(\tilde{m}_{n}^{(i+1,j)}-\tilde{m}_{n}^{(i+1,j+1)})\cdot\exp\Bigg(\frac{2}{\bar{r}_{n}^{(i,j)}}\frac{(\mathfrak{D}_{-}\tilde{m})_{n}^{(i,j)}}{1-\frac{2\tilde{m}_{n}^{(i+1,j)}}{\bar{r}_{n}^{(i,j)}}-\frac{1}{3}\Lambda(\bar{r}_{n}^{(i,j)})^{2}}\big(1-\mathfrak{Err}_{1,n}^{(i,j)}\big)\big(1-\mathfrak{Err}_{\backslash n}^{(i,j)}\big)\Bigg)\label{eq:bla-1}
\end{equation}
where $\mathfrak{Err}_{1,n}^{(i,j)}$ satisfies the bound (\ref{eq:BoundErrorTermForInteraction})
and $\mathfrak{Err}_{\backslash n}^{(i,j)}$ satisfies the bound (\ref{eq:BoundErrorTermIngoingInteraction}).
In view of (\ref{eq:DefinitionIncomingEnergy}) and (\ref{eq:DefinitionOutcomingEnergy}),
(\ref{eq:bla-1}) is equivalent to (\ref{eq:MassDifferenceIncreaseInUDirection}). 

Similarly, integrating equation (\ref{eq:WaveEquationMass}) first
in $v$, for $V_{n}(v^{(j)})\le v\le V_{n}(v^{(j)}+\frac{4}{\sqrt{-\Lambda}}h_{2}(\text{\textgreek{e}}))$,
and then in $u$, for $U_{n}(v^{(i)}))\le u\le U_{n}(v^{(i)}+\frac{4}{\sqrt{-\Lambda}}h_{2}(\text{\textgreek{e}}))$
(see also (\ref{eq:OutgoingEquationMass})), we obtain (\ref{eq:MassDifferenceDecreaseInVDirection}). 

\medskip{}

\noindent \emph{Proof of (\ref{eq:UsefulEstimate}).} Recall $F$
defined by (\ref{eq:OriginalFMassEquation}) and let us define the
function $\bar{F}:\mathcal{D}_{\bar{F}}\rightarrow(0,+\infty)$, where
\begin{equation}
\mathcal{D}_{\bar{F}}=\big\{(x,y)\in\mathbb{R}^{2}\mbox{ }x>0\mbox{ and }x-y-\frac{2}{3}\Lambda x^{2}>0\big\},
\end{equation}
by the relation 
\begin{equation}
\bar{F}(x,y)\doteq\frac{2}{x-y-\frac{2}{3}\Lambda x^{2}}.\label{eq:ModifiedF}
\end{equation}
Note that, in view of (\ref{eq:RoughBoundGeometry}), (\ref{eq:UpperBoundForAxisInteractionProp}),
(\ref{eq:BoundMirror}) and (\ref{eq:h_2definition}), for any $\text{\textgreek{m}}\ge0$
for which 
\begin{equation}
\inf_{(u,v)\in\mathcal{N}_{\text{\textgreek{e}}n}^{(i,j)}}\big\{ r(u,v)-2\text{\textgreek{m}}-\frac{1}{3}\Lambda r^{2}(u,v)\big\}>h_{3}(\text{\textgreek{e}}),\label{eq:ConditionFWellDefined}
\end{equation}
we can readily bound: 
\begin{equation}
\max_{(u,v)\in\mathcal{N}_{\text{\textgreek{e}}n}^{(i,j)}}\bar{F}(r(u,v),\text{\textgreek{m}})<\min_{(u,v)\in\mathcal{N}_{\text{\textgreek{e}}n}^{(i,j)}}F(r(u,v),\text{\textgreek{m}})\label{eq:BoundForLemma}
\end{equation}
and 
\begin{equation}
\partial_{\text{\textgreek{m}}}\bar{F}(r(u,v),\text{\textgreek{m}}),\mbox{ }\partial_{\text{\textgreek{m}}}F(r(u,v),\text{\textgreek{m}})>0\label{eq:IncreasingForLemma}
\end{equation}
(note that $F(r|_{\mathcal{N}_{\text{\textgreek{e}}n}^{(i,j)}},\text{\textgreek{m}})$
and $\bar{F}(r|_{\mathcal{N}_{\text{\textgreek{e}}n}^{(i,j)}},\text{\textgreek{m}})$
are well-defined and positive under the condition (\ref{eq:ConditionFWellDefined})).

For any $1\le n\le n_{f}$, $1\le i\le k$ and $k+1\le j\le k+i$,
let us consider the following characteristic initial value problem
on $\mathcal{N}_{\text{\textgreek{e}}n}^{(i,j)}$:
\begin{equation}
\begin{cases}
\partial_{u}\partial_{v}\bar{m}=-\bar{F}(r,\bar{m})\partial_{u}\bar{m}\partial_{v}\bar{m} & \mbox{on }\mathcal{N}_{\text{\textgreek{e}}n}^{(i,j)},\\
\bar{m}=\tilde{m} & \mbox{on }[U_{n}(v^{(i)})),U_{n}(v^{(i)}+\frac{4}{\sqrt{-\Lambda}}h_{2}(\text{\textgreek{e}}))]\times\{V_{n}(v^{(j)}))\}\cup\\
 & \hphantom{\mbox{on }\cup}\cup\{U_{n}(v^{(i)}))\}\times[V_{n}(v^{(j)})),V_{n}(v^{(j)}+\frac{4}{\sqrt{-\Lambda}}h_{2}(\text{\textgreek{e}}))].
\end{cases}\label{eq:ModifiedInitialValueProblem}
\end{equation}
Note that $\tilde{m}$ satisfies the same characteristic initial value
problem with $F(r,\tilde{m})$ in place of $\bar{F}(r,\bar{m})$.
Notice also that, in view of (\ref{eq:DerivativeTildeUMass})--(\ref{eq:DerivativeTildeVMass})
and (\ref{eq:TuuSupport})--(\ref{eq:TvvSupport}), the initial data
for $\tilde{m}$ and $\bar{m}$ satisfy: 
\begin{equation}
\partial_{u}\tilde{m}<0\mbox{ on}\Big\{ U_{n}(v^{(i)})<u<U_{n}(v^{(i)}+\frac{4}{\sqrt{-\Lambda}}h_{2}(\text{\textgreek{e}}))\Big\}\label{eq:TuuSupport-1}
\end{equation}
and 
\begin{equation}
\partial_{v}\tilde{m}>0\mbox{ on}\Big\{ V_{n}(v^{(j)})<u<V_{n}(v^{(j)}+\frac{4}{\sqrt{-\Lambda}}h_{2}(\text{\textgreek{e}}))\Big\}.\label{eq:TvvSupport-1}
\end{equation}
Therefore, in view of (\ref{eq:BoundForLemma}), (\ref{eq:IncreasingForLemma})
and (\ref{eq:TuuSupport-1})--(\ref{eq:TvvSupport-1}), an application
of Lemma \ref{lem:HyperbolicMaximumPrinciple} (see Section \ref{sub:Auxiliary-lemmas})
with $\tilde{m},\bar{m}$ in place of $z_{2},z_{1}$, respectively,
yields the following a priori bounds for a solution $\bar{m}$ of
(\ref{eq:ModifiedInitialValueProblem}): 
\begin{equation}
\bar{m}\le\tilde{m}\mbox{ on }\mathcal{N}_{\text{\textgreek{e}}n}^{(i,j)}\label{eq:AprioriBoundMtilde'}
\end{equation}
and 
\begin{equation}
\partial_{u}\bar{m}<0<\partial_{v}\bar{m}\mbox{ in the interior of }\mathcal{N}_{\text{\textgreek{e}}n}^{(i,j)}.\label{eq:MonotonicityFromLemma}
\end{equation}
Notice that the a priori bound (\ref{eq:AprioriBoundMtilde'}) and
the initial data in (\ref{eq:ModifiedInitialValueProblem}) imply
that $\bar{m}\ge0$ and that (\ref{eq:ConditionFWellDefined}) holds
for $\text{\textgreek{m}}=\tilde{m}$ and $\text{\textgreek{m}}=\bar{m}$;
in particular, $\bar{F}(r,\bar{m})$ is well defined and positive
on $\mathcal{N}_{\text{\textgreek{e}}n}^{(i,j)}$. Thus, it readily
follows (using standard arguments) that (\ref{eq:ModifiedInitialValueProblem})
indeed has a unique smooth solution $\bar{m}$ satisfying (\ref{eq:AprioriBoundMtilde'}).

With $\bar{m}$ defined on $\mathcal{N}_{\text{\textgreek{e}}n}^{(i,j)}$
as above for any $1\le n\le n_{f}$, $1\le i\le k$ and $k+1\le j\le k+i$,
we will define the following modified versions of (\ref{eq:DefinitionIncomingEnergy})
and (\ref{eq:DefinitionOutcomingEnergy}):
\begin{align}
(\mathfrak{D}_{-}\bar{m})_{n}^{(i,j)} & \doteq\bar{m}|_{(U_{n}(v^{(i)}),V_{n}(v^{(j)}))}-\bar{m}|_{(U_{n}(v^{(i)}+\frac{4}{\sqrt{-\Lambda}}h_{2}(\text{\textgreek{e}})),V_{n}(v^{(j)}))}\label{eq:DefinitionIncomingEnergyModified}\\
(\overline{\mathfrak{D}}_{-}\bar{m})_{n}^{(i,j)} & \doteq\bar{m}|_{(U_{n}(v^{(i)}),V_{n}(v^{(j)}+\frac{4}{\sqrt{-\Lambda}}h_{2}(\text{\textgreek{e}})))}-\bar{m}|_{(U_{n}(v^{(i)}),V_{n}(v^{(j)}))}\nonumber 
\end{align}
and 
\begin{align}
(\mathfrak{D}_{+}\bar{m})_{n}^{(i,j)} & \doteq\bar{m}|_{(U_{n}(v^{(i)}),V_{n}(v^{(j)}+\frac{4}{\sqrt{-\Lambda}}h_{2}(\text{\textgreek{e}})))}-\bar{m}|_{(U_{n}(v^{(i)}+\frac{4}{\sqrt{-\Lambda}}h_{2}(\text{\textgreek{e}})),V_{n}(v^{(j)}+\frac{4}{\sqrt{-\Lambda}}h_{2}(\text{\textgreek{e}})))}\label{eq:DefinitionOutcomingEnergyModified}\\
(\overline{\mathfrak{D}}_{+}\bar{m})_{n}^{(i,j)} & \doteq\bar{m}|_{(U_{n}(v^{(i)}+\frac{4}{\sqrt{-\Lambda}}h_{2}(\text{\textgreek{e}})),V_{n}(v^{(j)}+\frac{4}{\sqrt{-\Lambda}}h_{2}(\text{\textgreek{e}})))}-\bar{m}|_{(U_{n}(v^{(i)}+\frac{4}{\sqrt{-\Lambda}}h_{2}(\text{\textgreek{e}})),V_{n}(v^{(j)}))}.\nonumber 
\end{align}
Note that, in view of the initial data for (\ref{eq:ModifiedInitialValueProblem}):
\begin{align}
(\mathfrak{D}_{-}\bar{m})_{n}^{(i,j)} & =(\mathfrak{D}_{-}\tilde{m})_{n}^{(i,j)}\label{eq:EqualityTwoEnergiesInitially}\\
(\overline{\mathfrak{D}}_{-}\bar{m})_{n}^{(i,j)} & =(\overline{\mathfrak{D}}_{-}\tilde{m})_{n}^{(i,j)},\nonumber 
\end{align}
while, in view of the bound (\ref{eq:AprioriBoundMtilde'}) (and the
initial data for (\ref{eq:ModifiedInitialValueProblem})):
\begin{align}
(\mathfrak{D}_{+}\bar{m})_{n}^{(i,j)} & \ge(\mathfrak{D}_{+}\tilde{m})_{n}^{(i,j)}\label{eq:InequalityOutcomingEnergies}\\
(\overline{\mathfrak{D}}_{+}\bar{m})_{n}^{(i,j)} & \le(\overline{\mathfrak{D}}_{+}\tilde{m})_{n}^{(i,j)}.\nonumber 
\end{align}

By repeating exactly the same steps that led to (\ref{eq:MassDifferenceIncreaseInUDirection})
and (\ref{eq:MassDifferenceDecreaseInVDirection}) but using (\ref{eq:ModifiedInitialValueProblem})
instead of (\ref{eq:WaveEquationMass}), we obtain for any $1\le n\le n_{f}$,
$1\le i\le k$ and $k+1\le j\le k+i$: 
\begin{equation}
(\overline{\mathfrak{D}}_{+}\bar{m})_{n}^{(i,j)}=(\overline{\mathfrak{D}}_{-}\bar{m})_{n}^{(i,j)}\cdot\exp\Bigg(\frac{2}{\bar{r}_{n}^{(i,j)}}\frac{(\mathfrak{D}_{-}\bar{m})_{n}^{(i,j)}}{1-\frac{\tilde{m}_{n}^{(i+1,j)}}{\bar{r}_{n}^{(i,j)}}-\frac{2}{3}\Lambda(\bar{r}_{n}^{(i,j)})^{2}}\big(1-\mathfrak{\overline{Err}}_{1,n}^{(i,j)}\big)\big(1-\mathfrak{\overline{Err}}_{\backslash n}^{(i,j)}\big)\Bigg)\label{eq:MassDifferenceIncreaseInUDirectionModified}
\end{equation}
and 
\begin{equation}
(\mathfrak{D}_{+}\bar{m})_{n}^{(i,j)}=(\mathfrak{D}_{-}\bar{m})_{n}^{(i,j)}\cdot\exp\Bigg(-\frac{2}{\bar{r}_{n}^{(i,j)}}\frac{(\overline{\mathfrak{D}}_{+}\bar{m})_{n}^{(i,j)}}{1-\frac{\tilde{m}_{n}^{(i+1,j)}}{\bar{r}_{n}^{(i,j)}}-\frac{2}{3}\Lambda(\bar{r}_{n}^{(i,j)})^{2}}\big(1-\mathfrak{\overline{Err}}_{1,n}^{(i,j)}\big)\big(1-\mathfrak{\overline{Err}}_{/n}^{(i,j)}\big)\Bigg),\label{eq:MassDifferenceDecreaseInVDirectionModified}
\end{equation}
where 
\begin{equation}
0\le\mathfrak{\overline{Err}}_{1,n}^{(i,j)}\le1-\frac{\bar{r}_{n}^{(i,j)}-\tilde{m}_{n}^{(i+1,j)}-\frac{2}{3}\Lambda(\bar{r}_{n}^{(i,j)})^{3}}{r\big|_{\big(U_{n}(v^{(i)}),V_{n}(v^{(j)}+\frac{4}{\sqrt{-\Lambda}}h_{2}(\text{\textgreek{e}}))\big)}-\tilde{m}_{n}^{(i,j+1)}-\frac{2}{3}\Lambda r^{3}\big|_{\big(U_{n}(v^{(i)}),V_{n}(v^{(j)}+\frac{4}{\sqrt{-\Lambda}}h_{2}(\text{\textgreek{e}}))\big)}},\label{eq:BoundErrorTermForInteraction-1}
\end{equation}
\begin{equation}
0\le\mathfrak{\overline{Err}}_{\backslash n}^{(i,j)}\le1-\frac{(\mathfrak{D}_{+}\bar{m})_{n}^{(i,j)}}{(\mathfrak{D}_{-}\bar{m})_{n}^{(i,j)}}\label{eq:BoundErrorTermIngoingInteraction-1-1}
\end{equation}
and 
\begin{equation}
0\le\mathfrak{\overline{Err}}_{/n}^{(i,j)}\le1-\frac{(\overline{\mathfrak{D}}_{-}\bar{m})_{n}^{(i,j)}}{(\overline{\mathfrak{D}}_{+}\bar{m})_{n}^{(i,j)}}\label{eq:BoundErrorTermOugoingInteraction-1}
\end{equation}
(and, as before, we allow the terms $\mathfrak{\overline{Err}}_{1,n}^{(i,j)}$
in (\ref{eq:MassDifferenceIncreaseInUDirectionModified}) and (\ref{eq:MassDifferenceDecreaseInVDirectionModified})
to be different). 

Because 
\begin{equation}
\tilde{m}_{n}^{(i+1,j)}\le\tilde{m}|_{\mathcal{I}}\le\frac{2}{3}r_{0}\le\frac{2}{3}\bar{r}_{n}^{(i,j)}\label{eq:UpperBoundMassFromR}
\end{equation}
(in view of (\ref{eq:MassInfinity}), (\ref{eq:BoundMirror}), (\ref{eq:NonTrappingQualitativ})
and (\ref{eq:NonTrappingMassSign})), from (\ref{eq:MassDifferenceDecreaseInVDirectionModified})
(using also (\ref{eq:UpperBoundForAxisInteractionProp}) and the fact
that $\mathfrak{\overline{Err}}{}_{1,n}^{(i,j)},\mathfrak{\overline{Err}}{}_{/n}^{(i,j)}\ge0$)
we can estimate for any $1\le n\le n_{f}$, $1\le i\le k$ and $k+1\le j\le k+i$:
\begin{align}
(\mathfrak{D}_{+}\bar{m})_{n}^{(i,j)} & =(\mathfrak{D}_{-}\bar{m})_{n}^{(i,j)}\cdot\exp\Bigg(-\frac{2}{\bar{r}_{n}^{(i,j)}}\frac{(\overline{\mathfrak{D}}_{+}\bar{m})_{n}^{(i,j)}}{1-\frac{\tilde{m}_{n}^{(i+1,j)}}{\bar{r}_{n}^{(i,j)}}-\frac{2}{3}\Lambda(\bar{r}_{n}^{(i,j)})^{2}}\big(1-\mathfrak{\overline{Err}}{}_{1,n}^{(i,j)}\big)\big(1-\mathfrak{\overline{Err}}{}_{/n}^{(i,j)}\big)\Bigg)\label{eq:OneMoreSillyEstimate}\\
 & \ge(\mathfrak{D}_{-}\bar{m})_{n}^{(i,j)}\cdot\exp\Bigg(-\frac{8(\overline{\mathfrak{D}}_{+}\bar{m})_{n}^{(i,j)}}{\bar{r}_{n}^{(i,j)}}\Bigg)\nonumber 
\end{align}
In view of the fact that 
\[
(\overline{\mathfrak{D}}_{+}\bar{m})_{n}^{(i,j)}\le(\overline{\mathfrak{D}}_{+}\tilde{m})_{n}^{(i,j)}=\tilde{m}_{n}^{(i,j)}-\tilde{m}_{n}^{(i,j+1)}\le\tilde{m}|_{\mathcal{I}}-0\le\frac{2}{3}r_{0}\le\frac{2}{3}\bar{r}_{n}^{(i,j)}
\]
(following from (\ref{eq:InequalityOutcomingEnergies})), (\ref{eq:OneMoreSillyEstimate})
yields 
\begin{equation}
(\mathfrak{D}_{+}\bar{m})_{n}^{(i,j)}\ge e^{-\frac{16}{3}}(\mathfrak{D}_{-}\bar{m})_{n}^{(i,j)}.\label{eq:DonAlready}
\end{equation}
In view of (\ref{eq:BoundErrorTermIngoingInteraction-1-1}), (\ref{eq:DonAlready})
implies that
\begin{equation}
1-\mathfrak{\overline{Err}}{}_{\backslash n}^{(i,j)}\ge\frac{1}{C_{0}}.\label{eq:BoundForErrorInteraction}
\end{equation}

Using (\ref{eq:BoundForErrorInteraction}) in (\ref{eq:MassDifferenceIncreaseInUDirectionModified}),
we obtain: 
\begin{equation}
(\overline{\mathfrak{D}}_{+}\bar{m})_{n}^{(i,j)}\ge(\overline{\mathfrak{D}}_{-}\bar{m})_{n}^{(i,j)}\cdot\exp\Bigg(\frac{2}{C_{0}\bar{r}_{n}^{(i,j)}}\frac{(\mathfrak{D}_{-}\bar{m})_{n}^{(i,j)}}{1-\frac{\tilde{m}_{n}^{(i+1,j)}}{\bar{r}_{n}^{(i,j)}}-\frac{2}{3}\Lambda(\bar{r}_{n}^{(i,j)})^{2}}\big(1-\mathfrak{\overline{Err}}_{1,n}^{(i,j)}\big)\Bigg).\label{eq:MassDifferenceIncreaseInUDirectionModified-1}
\end{equation}
In view of (\ref{eq:UpperBoundMassFromR}) and (\ref{eq:UpperBoundForAxisInteractionProp}),
we can also estimate 
\begin{equation}
\frac{1-\mathfrak{\overline{Err}}{}_{1,n}^{(i,j)}}{1-\frac{\tilde{m}_{n}^{(i+1,j)}}{\bar{r}_{n}^{(i,j)}}-\frac{2}{3}\Lambda(\bar{r}_{n}^{(i,j)})^{2}}\ge\frac{1}{10}
\end{equation}
and, thus, (\ref{eq:MassDifferenceIncreaseInUDirectionModified-1})
yields: 
\begin{equation}
(\overline{\mathfrak{D}}_{+}\bar{m})_{n}^{(i,j)}\ge(\overline{\mathfrak{D}}_{-}\bar{m})_{n}^{(i,j)}\cdot\exp\Bigg(\frac{1}{5C_{0}}\frac{(\mathfrak{D}_{-}\bar{m})_{n}^{(i,j)}}{\bar{r}_{n}^{(i,j)}}\Bigg).\label{eq:MassDifferenceIncreaseInUDirectionModified-1-1}
\end{equation}

From (\ref{eq:MassDifferenceIncreaseInUDirectionModified-1-1}) and
the relations (\ref{eq:EqualityTwoEnergiesInitially}) and (\ref{eq:InequalityOutcomingEnergies}),
we readily obtain (\ref{eq:UsefulEstimate}). \qed

\paragraph*{Proof of (\ref{eq:NotEnoughMassBehind}).}

For any $1\le n\le n_{f}$, from (\ref{eq:MassDifferenceIncreaseInUDirection})
we readily obtain that, for any $1\le i\le k$: 
\begin{equation}
(\overline{\mathfrak{D}}_{+}\tilde{m})_{n}^{(i,k+1)}\ge(\overline{\mathfrak{D}}_{-}\tilde{m})_{n}^{(i,k+1)}.\label{eq:Increasemassingoing}
\end{equation}
Applying (\ref{eq:Increasemassingoing}) successively for $i=1,2,\ldots k$,
using also the identity 
\begin{equation}
(\overline{\mathfrak{D}}_{-}\tilde{m})_{n}^{(i,j)}=(\overline{\mathfrak{D}}_{+}\tilde{m})_{n}^{(i+1,j)}\label{eq:EqualMassDifferenceWhennoIntraction}
\end{equation}
(which follows from the fact that $\tilde{m}$ is constant over each
$\mathcal{R}_{\text{\textgreek{e}}n}^{(i,j)}$), we thus infer that,
for any $1\le i\le k$: 
\begin{equation}
(\overline{\mathfrak{D}}_{+}\tilde{m})_{n}^{(i,k+1)}\ge(\overline{\mathfrak{D}}_{-}\tilde{m})_{n}^{(k,k+1)}=(\tilde{m}_{n-1}^{(0,0)}-\tilde{m}_{n-1}^{(0,1)}).\label{eq:TotalWeakBoundIngoing}
\end{equation}

Since 
\begin{equation}
\tilde{m}_{n-1}^{(0,0)}=\tilde{m}_{n-1}^{(1,1)}=\tilde{m}|_{\mathcal{I}}
\end{equation}
and 
\begin{equation}
(\mathfrak{D}_{+}\tilde{m})_{n-1}^{(0,1)}=\tilde{m}_{n-1}^{(1,1)}-\tilde{m}_{n-1}^{(0,1)}=(\overline{\mathfrak{D}}_{-}\tilde{m})_{n}^{(k,k+1)},\label{eq:EqualMassDifferenceSfterReflection}
\end{equation}
from (\ref{eq:TotalWeakBoundIngoing}) we infer that, for any $1\le i\le k$:
\begin{equation}
(\overline{\mathfrak{D}}_{+}\tilde{m})_{n}^{(i,k+1)}\ge(\mathfrak{D}_{+}\tilde{m})_{n-1}^{(0,1)}.\label{eq:WeakBoundIngoingFinalOutgoing}
\end{equation}

Similarly as for the derivation of (\ref{eq:Increasemassingoing}),
applying the relation (\ref{eq:MassDifferenceDecreaseInVDirection})
successively for $i=0$ and $j=1,2,\ldots k$ (with $n-1$ in place
of $n$), we infer: 
\begin{equation}
(\mathfrak{D}_{+}\tilde{m})_{n-1}^{(0,1)}=(\mathfrak{D}_{-}\tilde{m})_{n-1}^{(0,k)}\cdot\exp\Bigg(-\sum_{j=1}^{k}\frac{2}{\bar{r}_{n-1}^{(0,j)}}\frac{(\overline{\mathfrak{D}}_{-}\tilde{m})_{n-1}^{(0,j)}}{1-\frac{2\tilde{m}_{n-1}^{(0,j)}}{\bar{r}_{n-1}^{(0,j)}}-\frac{1}{3}\Lambda(\bar{r}_{n-1}^{(0,j)})^{2}}\big(1-\mathfrak{Err}_{1,n-1}^{(0,j)}\big)\big(1-\mathfrak{Err}_{/n-1}^{(0,j)}\big)\Bigg).\label{eq:MassDifferenceDecreaseInVDirection-1-1}
\end{equation}
In view of the bound (\ref{eq:MassInfinity}) for the total mass $\tilde{m}|_{\mathcal{I}}$,
the lower bound (\ref{eq:BoundForRAwayInteractionProp}) for $r_{n}^{(0,k)}$
and the fact that 
\[
\bar{r}_{n-1}^{(0,j)}\le r_{n-1}^{(0,k)}\le\bar{r}_{n-1}^{(0,j)}\big(1+(h_{2}(\text{\textgreek{e}}))^{1/2}\big)
\]
 for $1\le j\le k$ (following from (\ref{eq:RoughBoundGeometry}),
and (\ref{eq:h_2definition})), we can estimate 
\begin{equation}
\sum_{j=1}^{k}\frac{2}{\bar{r}_{n-1}^{(0,j)}}\frac{(\overline{\mathfrak{D}}_{-}\tilde{m})_{n-1}^{(0,j)}}{1-\frac{2\tilde{m}_{n-1}^{(0,j)}}{\bar{r}_{n-1}^{(0,j)}}-\frac{1}{3}\Lambda(\bar{r}_{n-1}^{(0,j)})^{2}}\big(1-\mathfrak{Err}_{1,n-1}^{(0,j)}\big)\big(1-\mathfrak{Err}_{/n-1}^{(0,j)}\big)\le\text{\textgreek{e}}^{\frac{3}{2}}.\label{eq:BoundForDecrease}
\end{equation}
Therefore, (\ref{eq:MassDifferenceDecreaseInVDirection-1-1}) yields:
\begin{equation}
\log\frac{(\mathfrak{D}_{+}\tilde{m})_{n-1}^{(0,1)}}{(\mathfrak{D}_{-}\tilde{m})_{n-1}^{(0,k)}}\ge-\text{\textgreek{e}}^{3/2}.\label{eq:TotalMassDecreaseTopInteraction-1}
\end{equation}
From (\ref{eq:WeakBoundIngoingFinalOutgoing}) and (\ref{eq:TotalMassDecreaseTopInteraction-1})
we thus infer that, for any $1\le i\le k$ and any $2\le n\le n_{f}$:
\begin{equation}
\log\frac{(\overline{\mathfrak{D}}_{+}\tilde{m})_{n}^{(i,k+1)}}{\tilde{m}_{n-1}^{(1,k+1)}}\ge-\text{\textgreek{e}}^{3/2}.\label{eq:LowerBoundGeneralInteraction}
\end{equation}
From (\ref{eq:LowerBoundGeneralInteraction}) for $i=1$ and the fact
that, for any $1\le n\le n_{f}$: 
\begin{equation}
(\overline{\mathfrak{D}}_{+}\tilde{m})_{n}^{(1,k+1)}=(\mathfrak{D}_{-}\tilde{m})_{n}^{(0,k)}=\tilde{m}_{n}^{(1,k+1)},
\end{equation}
we thus infer that, for all $2\le n\le n_{f}$: 
\begin{equation}
\log\frac{\tilde{m}_{n}^{(1,k+1)}}{\tilde{m}_{n-1}^{(1,k+1)}}\ge-\text{\textgreek{e}}^{3/2}.\label{eq:LowerboundQuotientMass}
\end{equation}
Applying (\ref{eq:LowerboundQuotientMass}) successively $n-1$ times,
we thus infer for any $2\le n\le n_{f}$: 
\begin{equation}
\log\frac{\tilde{m}_{n}^{(1,k+1)}}{\tilde{m}_{1}^{(1,k+1)}}\ge-\text{\textgreek{e}}^{3/2}(n-1).\label{eq:FinalLowerBoundTopMass}
\end{equation}

The bound (\ref{eq:BoundMirror}) for $r_{0}$ and the form (\ref{eq:TheIngoingVlasovInitially})
of the initial data imply that 
\begin{equation}
\frac{2(\tilde{m}_{/}(v^{(0)}+\frac{4}{\sqrt{-\Lambda}}h_{2}(\text{\textgreek{e}}))-\tilde{m}_{/}(v^{(0)}))}{r_{0}}\ge\frac{4}{C_{0}}h_{0}(\text{\textgreek{e}}).\label{eq:BoundForTheMirror-1}
\end{equation}
Therefore, from (\ref{eq:TotalWeakBoundIngoing}) and (\ref{eq:BoundForTheMirror-1})
we infer that, for all $1\le i\le1$: 
\begin{equation}
\frac{2(\overline{\mathfrak{D}}_{+}\tilde{m})_{1}^{(i,k+1)}}{r_{0}}\ge\frac{4}{C_{0}}h_{0}(\text{\textgreek{e}}).\label{eq:TrivialBoundMassFirst}
\end{equation}
From (\ref{eq:FinalLowerBoundTopMass}) and (\ref{eq:TrivialBoundMassFirst})
for $i=1$ (when $(\overline{\mathfrak{D}}_{+}\tilde{m})_{1}^{(1,k+1)}=\tilde{m}_{1}^{(1,k+1)}$),
using also the fact that $n_{f}\le(h_{1}(\text{\textgreek{e}}))^{-2}$,
we thus deduce that, for all $1\le n\le n_{f}$: 
\begin{equation}
\frac{2\tilde{m}_{n}^{(1,k+1)}}{r_{0}}\ge\frac{2}{C_{0}}h_{0}(\text{\textgreek{e}}).\label{eq:AllSoTrivialBounds}
\end{equation}
 The relations (\ref{eq:MassInfinity}) and (\ref{eq:AllSoTrivialBounds})
readily yield (\ref{eq:NotEnoughMassBehind}).

\paragraph*{Proof of (\ref{eq:EnoughMassBehind}).}

In view of the bound (\ref{eq:LowerBoundTrappingParameter}), we infer
that, for any $1\le n\le n_{f}$: 
\begin{equation}
1-\frac{2\tilde{m}_{n}^{(1,k+1)}}{r|_{\big(U_{n}(v^{(1)}),V_{n}(v^{(k+1)}+\frac{4}{\sqrt{-\Lambda}}h_{2}(\text{\textgreek{e}}))\big)}}-\frac{1}{3}\Lambda r^{2}|_{\big(U_{n}(v^{(1)}),V_{n}(v^{(k+1)}+\frac{4}{\sqrt{-\Lambda}}h_{2}(\text{\textgreek{e}}))\big)}\ge h_{3}(\text{\textgreek{e}}).\label{eq:FromTrivialBoundForTrapping}
\end{equation}
Using the bounds 
\begin{equation}
\frac{r|_{\big(U_{n}(v^{(1)}),V_{n}(v^{(k+1)}+\frac{4}{\sqrt{-\Lambda}}h_{2}(\text{\textgreek{e}}))\big)}}{r_{0}}\le1+(h_{2}(\text{\textgreek{e}}))^{1/2}
\end{equation}
(derived from (\ref{eq:RoughBoundGeometry}), (\ref{eq:UpperBoundForAxisInteractionProp})
and (\ref{eq:h_2definition})) and 
\begin{equation}
\frac{r_{0}}{\frac{2}{\sqrt{-\Lambda}}\text{\textgreek{e}}-\frac{1}{3}\Lambda r_{0}^{3}}<1-\frac{1}{2}\exp\big(-2(h_{0}(\text{\textgreek{e}}))^{-4}\big)\label{eq:UpperboundMirror}
\end{equation}
(from (\ref{eq:BoundMirror})), as well as the relation (\ref{eq:MassInfinity})
for $\tilde{m}|_{\mathcal{I}}$, we can readily derive from (\ref{eq:FromTrivialBoundForTrapping})
that: 
\begin{equation}
\frac{2(\tilde{m}|_{\mathcal{I}}-\tilde{m}_{n}^{(1,k+1)})}{r_{0}}\ge2\big(1+(h_{2}(\text{\textgreek{e}}))^{1/2}\big)^{-1}\exp\big(-2(h_{0}(\text{\textgreek{e}}))^{-4}\big)+\frac{1}{3}\Lambda r_{0}^{2}\big(1+(h_{2}(\text{\textgreek{e}}))^{1/2}\big).\label{eq:AlmostThereTrivial}
\end{equation}
The bound (\ref{eq:EnoughMassBehind}) follows readily from (\ref{eq:UpperboundMirror})
and (\ref{eq:AlmostThereTrivial}).

\paragraph*{\noindent Proof of (\ref{eq:BoundForMassIncrease}).\emph{ }}

\noindent For any $2\le n\le n_{f}$, applying the relation (\ref{eq:MassDifferenceIncreaseInUDirection})
successively for $j=k+1$ and $i=1,2,\ldots k$, using also the identity
(\ref{eq:EqualMassDifferenceWhennoIntraction}) and the trivial bound
\begin{equation}
(\mathfrak{D}_{-}\tilde{m})_{n}^{(i,j)}\big(1-\mathfrak{Err}_{\backslash n}^{(i,k+1)}\big)\ge(\mathfrak{D}_{+}\tilde{m})_{n}^{(i,k+1)}
\end{equation}
(following directly from (\ref{eq:BoundErrorTermIngoingInteraction})),
we obtain: 
\begin{align}
(\overline{\mathfrak{D}}_{+}\tilde{m})_{n}^{(1,k+1)} & =(\tilde{m}_{n-1}^{(0,0)}-\tilde{m}_{n-1}^{(0,1)})\cdot\exp\Bigg(\sum_{i=1}^{k}\frac{2}{\bar{r}_{n}^{(i,k+1)}}\frac{(\mathfrak{D}_{-}\tilde{m})_{n}^{(i,k+1)}}{1-\frac{2\tilde{m}_{n}^{(i+1,k+1)}}{\bar{r}_{n}^{(i,j)}}-\frac{1}{3}\Lambda(\bar{r}_{n}^{(i,k+1)})^{2}}\big(1-\mathfrak{Err}_{1,n}^{(i,k+1)}\big)\big(1-\mathfrak{Err}_{\backslash n}^{(i,k+1)}\big)\Bigg)\label{eq:MassDifferenceIncreaseInUDirection-1}\\
 & \ge(\tilde{m}_{n-1}^{(0,0)}-\tilde{m}_{n-1}^{(0,1)})\cdot\exp\Bigg(\sum_{i=1}^{k}\frac{2}{\bar{r}_{n}^{(i,k+1)}}\frac{(\mathfrak{D}_{+}\tilde{m})_{n}^{(i,k+1)}}{1-\frac{2\tilde{m}_{n}^{(i+1,k+1)}}{\bar{r}_{n}^{(i,j)}}-\frac{1}{3}\Lambda(\bar{r}_{n}^{(i,k+1)})^{2}}\big(1-\mathfrak{Err}_{1,n}^{(i,k+1)}\big)\Bigg).\nonumber 
\end{align}

\noindent In view of (\ref{eq:RoughBoundGeometry}), (\ref{eq:h_2definition}),
(\ref{eq:BoundErrorTermForInteraction}), (\ref{eq:UpperBoundForAxisInteractionProp})
and the fact that 
\[
r_{0}\le\bar{r}_{n}^{(i,k+1)}\le\bar{r}_{n}^{(k,k+1)},
\]
 for $1\le i\le k$ (following from (\ref{eq:NonTrappingQualitativ})),
we can bound for any $1\le i\le k$:
\begin{multline}
\frac{2}{\bar{r}_{n}^{(i,k+1)}}\frac{1}{1-\frac{2\tilde{m}_{n}^{(i+1,k+1)}}{\bar{r}_{n}^{(i,j)}}-\frac{1}{3}\Lambda(\bar{r}_{n}^{(i,k+1)})^{2}}\big(1-\mathfrak{Err}_{1,n}^{(i,k+1)}\big)\\
\ge2\min\Bigg\{\frac{1}{r\big|_{\big(U_{n}(v^{(i)}),V_{n}(v^{(j)}+\frac{4}{\sqrt{-\Lambda}}h_{2}(\text{\textgreek{e}}))\big)}-2\tilde{m}_{n}^{(i,j+1)}-\frac{1}{3}\Lambda r^{3}\big|_{\big(U_{n}(v^{(i)}),V_{n}(v^{(j)}+\frac{4}{\sqrt{-\Lambda}}h_{2}(\text{\textgreek{e}}))\big)}},\frac{1}{\bar{r}_{n}^{(i,j)}-2\tilde{m}_{n}^{(i+1,k+1)}-\frac{1}{3}\Lambda(\bar{r}_{n}^{(i,k+1)})^{3}}\Bigg\}\\
\ge\frac{2-O(\text{\textgreek{e}})}{r_{n}^{(k,k+1)}}.\label{eq:LowerBoundNearRInteraction}
\end{multline}
 Furthermore, 
\begin{equation}
\sum_{i=1}^{k}(\mathfrak{D}_{+}\tilde{m})_{n}^{(i,k+1)}=\sum_{i=1}^{k}(\tilde{m}_{n}^{(i+1,k+1)}-\tilde{m}_{n}^{(i,k+1)})=\tilde{m}|_{\mathcal{I}}-\tilde{m}_{n}^{(1,k+1)}.\label{eq:TotalMassDifferenceTopInteraction}
\end{equation}
Therefore, in view of (\ref{eq:EnoughMassBehind}), (\ref{eq:LowerBoundNearRInteraction}),
(\ref{eq:TotalMassDifferenceTopInteraction}) and the fact that 
\begin{equation}
(\overline{\mathfrak{D}}_{+}\tilde{m})_{n}^{(1,k+1)}=\tilde{m}_{n}^{(1,k+1)}-\tilde{m}_{n}^{(1,k+2)}=\tilde{m}_{n}^{(1,k+1)},
\end{equation}
the bound (\ref{eq:MassDifferenceIncreaseInUDirection-1}) yields
\begin{align}
\tilde{m}_{n}^{(1,k+1)} & \ge(\tilde{m}_{n-1}^{(0,0)}-\tilde{m}_{n-1}^{(0,1)})\cdot\exp\Big(\frac{2-O(\text{\textgreek{e}})}{r_{n}^{(k,k+1)}}\big(\tilde{m}|_{\mathcal{I}}-\tilde{m}_{n}^{(1,k+1)}\big)\Big)\label{eq:TotalMassIncreaseTopInteraction}\\
 & \ge(\tilde{m}_{n-1}^{(0,0)}-\tilde{m}_{n-1}^{(0,1)})\cdot\exp\Big(\frac{r_{0}}{2r_{n}^{(k,k+1)}}\exp\big(-2(h_{0}(\text{\textgreek{e}}))^{-4}\big)\Big).\nonumber 
\end{align}

Using the bound (\ref{eq:TotalMassDecreaseTopInteraction-1}) and
the fact that 
\[
(\mathfrak{D}_{+}\tilde{m})_{n-1}^{(0,1)}=\tilde{m}_{n-1}^{(1,1)}-\tilde{m}_{n-1}^{(0,1)}=\tilde{m}_{n-1}^{(0,0)}-\tilde{m}_{n-1}^{(0,1)}
\]
and 
\[
(\mathfrak{D}_{-}\tilde{m})_{n-1}^{(0,k)}=\tilde{m}_{n-1}^{(1,k+1)},
\]
we can estimate: 
\begin{equation}
(\tilde{m}_{n-1}^{(0,0)}-\tilde{m}_{n-1}^{(0,1)})\ge e^{-\text{\textgreek{e}}^{3/2}}\tilde{m}_{n-1}^{(1,k+1)}.\label{eq:TotalDecreaseOnceMore}
\end{equation}
From (\ref{eq:TotalMassIncreaseTopInteraction}) and (\ref{eq:TotalDecreaseOnceMore})
we thus obtain (in view also of (\ref{eq:UpperBoundForAxisInteractionProp})
and the properties (\ref{eq:h_1_h_0_definition}) of $h_{0}(\text{\textgreek{e}})$):
\begin{equation}
\tilde{m}_{n}^{(1,k+1)}\ge\tilde{m}_{n-1}^{(1,k+1)}\exp\Big(\frac{r_{0}}{4r_{n}^{(k,k+1)}}\exp\big(-2(h_{0}(\text{\textgreek{e}}))^{-4}\big)\Big).\label{eq:BoundToShowForMass}
\end{equation}
In particular, (\ref{eq:BoundForMassIncrease}) holds for all $2\le n\le n_{f}$.

\paragraph*{\noindent Proof of (\ref{eq:BoundSecondBeamchanged}).\emph{ }}

Combining (\ref{eq:LowerBoundGeneralInteraction}) and (\ref{eq:NotEnoughMassBehind})
(using also (\ref{eq:TrivialBoundMassFirst}) in the case $n=1$,
as well as (\ref{eq:MassInfinity}) and (\ref{eq:BoundMirror}) for
$\tilde{m}|_{\mathcal{I}}$, $r_{0}$), we can readily estimate for
any $1\le n\le n_{f}$ and any $1\le i\le k$: 
\begin{equation}
\frac{2(\overline{\mathfrak{D}}_{+}\tilde{m})_{n}^{(i,k+1)}}{r_{0}}\ge\frac{1}{C_{0}}h_{0}(\text{\textgreek{e}}).\label{eq:BoundForEnergyRemainingForTheOtherBemas}
\end{equation}
Similarly, in view of (\ref{eq:EqualMassDifferenceWhennoIntraction}),
(\ref{eq:EqualMassDifferenceSfterReflection}) and (\ref{eq:TotalMassDecreaseTopInteraction-1}),
we can bound for any $1\le n\le n_{f}$ and any $1\le i\le k$: 
\begin{equation}
\frac{2(\overline{\mathfrak{D}}_{-}\tilde{m})_{n}^{(i,k+1)}}{r_{0}}\ge\frac{1}{C_{0}}h_{0}(\text{\textgreek{e}}).\label{eq:BoundForInitialMassesInteraction}
\end{equation}
Using the relation 
\begin{equation}
(\overline{\mathfrak{D}}_{+}\tilde{m})_{n}^{(i,k+1)}=\tilde{m}_{n}^{(i,k+1)}-\tilde{m}_{n}^{(i,k+2)}
\end{equation}
and the trivial bounds 
\begin{equation}
\tilde{m}_{n}^{(i,k+1)}\le\tilde{m}|_{\mathcal{I}}
\end{equation}
and 
\begin{equation}
\max_{k+2\le j\le k+i+1}\tilde{m}_{n}^{(i,j)}=\tilde{m}_{n}^{(i,k+2)}
\end{equation}
(following from the monotonicity properties (\ref{eq:NonTrappingMassSign})
of $\tilde{m}$), from (\ref{eq:BoundForEnergyRemainingForTheOtherBemas})
we obtain for any $1\le n\le n_{f}$: 
\begin{equation}
\min_{\substack{1\le i\le k,\\
k+2\le j\le k+i+1
}
}\frac{2(\tilde{m}|_{\mathcal{I}}-\tilde{m}_{n}^{(i,j)})}{r_{0}}\ge\frac{1}{C_{0}}h_{0}(\text{\textgreek{e}}).\label{eq:ActualNotEnoughEnergy}
\end{equation}

In view of (\ref{eq:MassInfinity}) and (\ref{eq:BoundMirror}) and
the properties (\ref{eq:h_1_h_0_definition}) of $h_{0}(\text{\textgreek{e}})$,
from (\ref{eq:ActualNotEnoughEnergy}) we infer that, for any $1\le n\le n_{f}$:
\begin{equation}
\max_{\substack{1\le i\le k,\\
k+2\le j\le k+i+1
}
}\frac{2\tilde{m}_{n}^{(i,j)}}{r_{0}}\le1-\frac{1}{3C_{0}}h_{0}(\text{\textgreek{e}}).\label{eq:AwayFromTrappingBehindTheFirstBeam}
\end{equation}
In particular, for any $1\le n\le n_{f}$, (\ref{eq:AwayFromTrappingBehindTheFirstBeam})
implies that: 
\begin{equation}
\sup_{\{u\le v\le V_{n}(v^{(k+1)})\}\cap\{u\ge U_{n}(v^{(k)})\}}\frac{1}{1-\frac{2m}{r}}\le4C_{0}(h_{0}(\text{\textgreek{e}}))^{-1}.\label{eq:UpperBoundForTrappingBetween TwoBeams}
\end{equation}

The main estimate that will be used in the proof of (\ref{eq:BoundSecondBeamchanged})
is the following bound: For any $1\le n_{1}<n_{2}\le n_{f}$ and any
$V_{n_{2}}(v^{(k+2)}+\frac{4}{\sqrt{-\Lambda}}h_{2}(\text{\textgreek{e}}))\le v\le V_{n_{2}}(v^{(k+1)})$:

\begin{align}
\partial_{v}r\big|_{(U_{n_{2}}(v^{(1)}+\frac{4}{\sqrt{-\Lambda}}h_{2}(\text{\textgreek{e}})),v)}\ge & \partial_{v}r\big|_{(U_{n_{1}-1}(v^{(1)}+\frac{4}{\sqrt{-\Lambda}}h_{2}(\text{\textgreek{e}})),v-(n_{2}-n_{1}-1)v_{0})}\times\label{eq:UsefulBoundAlmostThere}\\
 & \hphantom{\partial_{v}r}\times\exp\Big(-C_{0}^{5}(h_{0}(\text{\textgreek{e}}))^{-3}\max_{n_{1}\le n\le n_{2}}\Big\{\frac{r_{n}^{(1,k+1)}}{r_{0}}\Big\}\log\Big(\frac{\tilde{m}{}_{n_{2}}^{(1,k+1)}}{\tilde{m}{}_{n_{1}}^{(1,k+1)}}\Big)-2\text{\textgreek{e}}^{1/2}\Big).\nonumber 
\end{align}

\medskip{}

\noindent \emph{Proof of (\ref{eq:UsefulBoundAlmostThere}).} For
any $1\le n\le n_{f}$, $1\le i\le k$ and $k+1\le j\le k+i$, integrating
(\ref{eq:DerivativeInUDirectionKappa}) from $u=U_{n}(v^{(i)})$ up
to $U_{n}(v^{(i)}+\frac{4}{\sqrt{-\Lambda}}h_{2}(\text{\textgreek{e}}))$
(and using (\ref{eq:DerivativeTildeUMass})), we infer that, for all
$V_{n}(v^{(j+1)}+\frac{4}{\sqrt{-\Lambda}}h_{2}(\text{\textgreek{e}}))\le\bar{v}\le V_{n}(v^{(j)})$:
\begin{align}
\log\Big(\frac{\partial_{v}r}{1-\frac{2m}{r}}\Big)\Big|_{(U_{n}(v^{(i)}),\bar{v})}-\log\Big(\frac{\partial_{v}r}{1-\frac{2m}{r}}\Big) & \Big|_{(U_{n}(v^{(i)}+\frac{4}{\sqrt{-\Lambda}}h_{2}(\text{\textgreek{e}})),\bar{v})}\label{eq:FormulaForComparisonMassDiffAndRChange-1}\\
= & 4\pi\int_{U_{n}(v^{(i)})}^{U_{n}(v^{(i)}+\frac{4}{\sqrt{-\Lambda}}h_{2}(\text{\textgreek{e}}))}\frac{rT_{uu}}{-\partial_{u}r}\Big|_{(u,\bar{v})}\, du\nonumber \\
= & 2\int_{U_{n}(v^{(i)})}^{U_{n}(v^{(i)}+\frac{4}{\sqrt{-\Lambda}}h_{2}(\text{\textgreek{e}}))}\frac{-\partial_{u}\tilde{m}}{r(1-\frac{2m}{r})}\Big|_{(u,\bar{v})}\, du.\nonumber 
\end{align}
In view of the monotonicity properties (\ref{eq:NonTrappingQualitativ}),
(\ref{eq:NonTrappingMassSign}) of $r$ and $\tilde{m}$, we can estimate:
\begin{align}
2\int_{U_{n}(v^{(i)})}^{U_{n}(v^{(i)}+\frac{4}{\sqrt{-\Lambda}}h_{2}(\text{\textgreek{e}}))}\frac{-\partial_{u}\tilde{m}}{r(1-\frac{2m}{r})}\Big|_{(u,\bar{v})}\, du & \le\sup_{U_{n}(v^{(i)})\le\bar{u}\le U_{n}(v^{(i)}+\frac{4}{\sqrt{-\Lambda}}h_{2}(\text{\textgreek{e}}))}\Big(\frac{2}{r(1-\frac{2m}{r})}\Big)\Big|_{(u,\bar{v})}\int_{U_{n}(v^{(i)})}^{U_{n}(v^{(i)}+\frac{4}{\sqrt{-\Lambda}}h_{2}(\text{\textgreek{e}}))}(-\partial_{u}\tilde{m})\Big|_{(u,\bar{v})}\, du\label{eq:OneMoreTrivialBound}\\
 & \le\frac{2}{r(U_{n}(v^{(i)}+\frac{4}{\sqrt{-\Lambda}}h_{2}(\text{\textgreek{e}})),\bar{v})}\sup_{U_{n}(v^{(i)})\le\bar{u}\le U_{n}(v^{(i)}+\frac{4}{\sqrt{-\Lambda}}h_{2}(\text{\textgreek{e}}))}\Bigg(\frac{1}{1-\frac{2m}{r}\big|_{(\bar{u},\bar{v})}}\Bigg)(\mathfrak{D}_{-}\tilde{m})_{n}^{(i,j)}.\nonumber 
\end{align}
From (\ref{eq:FormulaForComparisonMassDiffAndRChange-1}), (\ref{eq:OneMoreTrivialBound})
and the bound (\ref{eq:UpperBoundForTrappingBetween TwoBeams}) for
$1-\frac{2m}{r}$, we infer that
\begin{align}
\log\Big(\frac{\partial_{v}r}{1-\frac{2m}{r}}\Big)\Big|_{(U_{n}(v^{(i)}),\bar{v})}- & \log\Big(\frac{\partial_{v}r}{1-\frac{2m}{r}}\Big)\Big|_{(U_{n}(v^{(i)}+\frac{4}{\sqrt{-\Lambda}}h_{2}(\text{\textgreek{e}})),\bar{v})}\label{eq:FormulaForComparisonMassDiffAndRChange-1-1}\\
\le & 8C_{0}(h_{0}(\text{\textgreek{e}}))^{-1}\frac{\bar{r}_{n}^{(i,j)}}{r(U_{n}(v^{(i)}+\frac{4}{\sqrt{-\Lambda}}h_{2}(\text{\textgreek{e}})),\bar{v})}\cdot\frac{(\mathfrak{D}_{-}\tilde{m})_{n}^{(i,j)}}{\bar{r}_{n}^{(i,j)}}.\nonumber 
\end{align}
Notice that, in view of the bound (\ref{eq:UsefulEstimate}), we can
estimate: 
\begin{equation}
\frac{(\mathfrak{D}_{-}\tilde{m})_{n}^{(i,j)}}{\bar{r}_{n}^{(i,j)}}\le5C_{0}\log\Big(\frac{(\overline{\mathfrak{D}}_{+}\tilde{m})_{n}^{(i,j)}}{(\overline{\mathfrak{D}}_{-}\tilde{m})_{n}^{(i,j)}}\Big).\label{eq:ComeOn!!}
\end{equation}
Thus, from (\ref{eq:FormulaForComparisonMassDiffAndRChange-1-1})
and (\ref{eq:ComeOn!!}), we deduce that, for any $1\le n\le n_{f}$,
$1\le i\le k$ and $k+1\le j\le k+i$ and any $V_{n}(v^{(j+1)}+\frac{4}{\sqrt{-\Lambda}}h_{2}(\text{\textgreek{e}}))\le\bar{v}\le V_{n}(v^{(j)})$: 

\begin{align}
\log\Big(\frac{\partial_{v}r}{1-\frac{2m}{r}}\Big)\Big|_{(U_{n}(v^{(i)}),\bar{v})}- & \log\Big(\frac{\partial_{v}r}{1-\frac{2m}{r}}\Big)\Big|_{(U_{n}(v^{(i)}+\frac{4}{\sqrt{-\Lambda}}h_{2}(\text{\textgreek{e}})),\bar{v})}\le\label{eq:FormulaForComparisonMassDiffAndRChange}\\
\le & 40C_{0}^{2}(h_{0}(\text{\textgreek{e}}))^{-1}\frac{\bar{r}_{n}^{(i,j)}}{r(U_{n}(v^{(i)}+\frac{4}{\sqrt{-\Lambda}}h_{2}(\text{\textgreek{e}})),\bar{v})}\log\Big(\frac{(\overline{\mathfrak{D}}_{+}\tilde{m})_{n}^{(i,j)}}{(\overline{\mathfrak{D}}_{-}\tilde{m})_{n}^{(i,j)}}\Big).\nonumber 
\end{align}

Applying the relation (\ref{eq:FormulaForComparisonMassDiffAndRChange})
successively for $i=1,\ldots,k$ and $\bar{v}=v$, using also the
fact that $\partial_{u}\Big(\frac{\partial_{v}r}{1-\frac{2m}{r}}\Big)=0$
on each $\mathcal{R}_{\text{\textgreek{e}}n}^{(i,j)}$, we obtain:
\begin{align}
\frac{\partial_{v}r}{1-\frac{2m}{r}}\Big|_{(U_{n}(v^{(1)}+\frac{4}{\sqrt{-\Lambda}}h_{2}(\text{\textgreek{e}})),\bar{v})}\ge & \frac{\partial_{v}r}{1-\frac{2m}{r}}\Big|_{(U_{n-1}(v^{(0)}+\frac{4}{\sqrt{-\Lambda}}h_{2}(\text{\textgreek{e}})),\bar{v})}\times\label{eq:SecondRelationForRDifference}\\
 & \hphantom{\int_{B}^{B}}\times\exp\Bigg(-40C_{0}^{2}(h_{0}(\text{\textgreek{e}}))^{-1}\sum_{i=1}^{k}\frac{\bar{r}_{n}^{(i,k+1)}}{r(U_{n}(v^{(i)}+\frac{4}{\sqrt{-\Lambda}}h_{2}(\text{\textgreek{e}})),\bar{v})}\log\Big(\frac{(\overline{\mathfrak{D}}_{+}\tilde{m})_{n}^{(i,k+1)}}{(\overline{\mathfrak{D}}_{-}\tilde{m})_{n}^{(i,k+1)}}\Big)\Bigg).\nonumber 
\end{align}

For any $i=1,\ldots,k$, integrating (\ref{eq:DerivativeInUDirectionKappa})
in $u$ from $u=U_{n}(v^{(i)})$ up to $U_{n}(v^{(1)}+\frac{4}{\sqrt{-\Lambda}}h_{2}(\text{\textgreek{e}}))$
for $v=v_{*}\in[\bar{v},V_{n}(v^{(j)})]$ and using (\ref{eq:DerivativeTildeUMass})
and the fact that $\partial_{u}\tilde{m}=0$ on $\mathcal{R}_{\text{\textgreek{e}}n}^{(i,j)}$,
we infer: 
\begin{equation}
\frac{\partial_{v}r}{1-\frac{2m}{r}}\Big|_{(U_{n}(v^{(1)}+\frac{4}{\sqrt{-\Lambda}}h_{2}(\text{\textgreek{e}})),v_{*})}=\frac{\partial_{v}r}{1-\frac{2m}{r}}\Big|_{(U_{n}(v^{(i)}+\frac{4}{\sqrt{-\Lambda}}h_{2}(\text{\textgreek{e}})),v_{*})}\exp\Bigg(-2\int_{U_{n}(v^{(i)})}^{U_{n}(v^{(1)}+\frac{4}{\sqrt{-\Lambda}}h_{2}(\text{\textgreek{e}}))}\frac{-\partial_{u}\tilde{m}}{r-2\tilde{m}-\frac{1}{3}\Lambda r^{3}}\Big|_{(u,v_{*})}\, du\Bigg).\label{eq:BoundForDvRFirst}
\end{equation}
Using the fact that $r\ge r_{0}$, from (\ref{eq:BoundForDvRFirst})
we infer (in view of the monotonicity property (\ref{eq:NonTrappingMassSign})
for $\tilde{m}$) that 
\begin{align}
\frac{\partial_{v}r}{1-\frac{2m}{r}}\Big|_{(U_{n}(v^{(1)}+\frac{4}{\sqrt{-\Lambda}}h_{2}(\text{\textgreek{e}})),v_{*})}\ge & \frac{\partial_{v}r}{1-\frac{2m}{r}}\Big|_{(U_{n}(v^{(i)}+\frac{4}{\sqrt{-\Lambda}}h_{2}(\text{\textgreek{e}})),v_{*})}\exp\Bigg(-2\int_{U_{n}(v^{(i)})}^{U_{n}(v^{(1)}+\frac{4}{\sqrt{-\Lambda}}h_{2}(\text{\textgreek{e}}))}\frac{-\partial_{u}\tilde{m}|_{(u,v_{*})}}{r_{0}-2\tilde{m}|_{(u,v_{*})}-\frac{1}{3}\Lambda r_{0}^{3}}\, du\Bigg)\label{eq:BoundForDvRSecond}\\
 & =\frac{\partial_{v}r}{1-\frac{2m}{r}}\Big|_{(U_{n}(v^{(i)}+\frac{4}{\sqrt{-\Lambda}}h_{2}(\text{\textgreek{e}})),v_{*})}\exp\Bigg(-\int_{U_{n}(v^{(i)})}^{U_{n}(v^{(1)}+\frac{4}{\sqrt{-\Lambda}}h_{2}(\text{\textgreek{e}}))}\partial_{u}\Big(\log\big(1-\frac{2\tilde{m}|_{(u,v_{*})}}{r_{0}}-\frac{1}{3}\Lambda r_{0}^{2}\big)\Big)\, du\Bigg)\nonumber \\
 & =\frac{\partial_{v}r}{1-\frac{2m}{r}}\Big|_{(U_{n}(v^{(i)}+\frac{4}{\sqrt{-\Lambda}}h_{2}(\text{\textgreek{e}})),v_{*})}\cdot\Bigg(\frac{1-\frac{2\tilde{m}|_{(U_{n}(v^{(i)}),v_{*})}}{r_{0}}-\frac{1}{3}\Lambda r_{0}^{2}}{1-\frac{2\tilde{m}|_{(U_{n}(v^{(1)}+\frac{4}{\sqrt{-\Lambda}}h_{2}(\text{\textgreek{e}})),v_{*})}}{r_{0}}-\frac{1}{3}\Lambda r_{0}^{2}}\Bigg).\nonumber 
\end{align}
In view of the bounds (\ref{eq:UpperBoundForAxisInteractionProp})
and (\ref{eq:AwayFromTrappingBehindTheFirstBeam}), (\ref{eq:BoundForDvRSecond})
yields: 
\begin{equation}
\frac{\partial_{v}r}{1-\frac{2m}{r}}\Big|_{(U_{n}(v^{(1)}+\frac{4}{\sqrt{-\Lambda}}h_{2}(\text{\textgreek{e}})),v_{*})}\ge\frac{\partial_{v}r}{1-\frac{2m}{r}}\Big|_{(U_{n}(v^{(i)}+\frac{4}{\sqrt{-\Lambda}}h_{2}(\text{\textgreek{e}})),v_{*})}\frac{h_{0}(\text{\textgreek{e}})}{4C_{0}}.\label{eq:AtLast!!}
\end{equation}
Integrating (\ref{eq:AtLast!!}) in $v_{*}\in[\bar{v},V_{n}(v^{(j)})]$
and using (\ref{eq:UpperBoundForTrappingBetween TwoBeams}) (and (\ref{eq:UpperBoundForAxisInteractionProp}))
for the $\frac{1}{1-\frac{2m}{r}}$ factors, we thus obtain: 
\begin{equation}
\bar{r}_{n}^{(1,k+1)}-r_{0}\ge\Big(\bar{r}_{n}^{(i,k+1)}-r(U_{n}(v^{(i)}+\frac{4}{\sqrt{-\Lambda}}h_{2}(\text{\textgreek{e}})),\bar{v})\Big)\frac{(h_{0}(\text{\textgreek{e}}))^{2}}{16C_{0}^{2}}
\end{equation}
and, thus (in view of (\ref{eq:RoughBoundGeometry}), (\ref{eq:h_2definition})
and the fact that $r_{0}\le\min\{r_{n}^{(1,k+1)},r(U_{n}(v^{(i)}+\frac{4}{\sqrt{-\Lambda}}h_{2}(\text{\textgreek{e}})),\bar{v})\}$):
\begin{equation}
\frac{\bar{r}_{n}^{(i,k+1)}}{r(U_{n}(v^{(i)}+\frac{4}{\sqrt{-\Lambda}}h_{2}(\text{\textgreek{e}})),\bar{v})}\le16C_{0}^{2}(h_{0}(\text{\textgreek{e}}))^{-2}\frac{r_{n}^{(1,k+1)}}{r_{0}}.\label{eq:ForQuotientOfR}
\end{equation}

From (\ref{eq:ForQuotientOfR}) and (\ref{eq:EqualMassDifferenceWhennoIntraction}),
it follows that: 
\begin{align}
\sum_{i=1}^{k}\frac{\bar{r}_{n}^{(i,k+1)}}{r(U_{n}(v^{(i)}+\frac{4}{\sqrt{-\Lambda}}h_{2}(\text{\textgreek{e}})),v)} & \log\Big(\frac{(\overline{\mathfrak{D}}_{+}\tilde{m})_{n}^{(i,k+1)}}{(\overline{\mathfrak{D}}_{-}\tilde{m})_{n}^{(i,k+1)}}\Big)\label{eq:EstimateForTheExponentialInChange}\\
\le & 16C_{0}^{2}(h_{0}(\text{\textgreek{e}}))^{-2}\frac{r_{n}^{(1,k+1)}}{r_{0}}\sum_{i=1}^{k}\log\Big(\frac{(\overline{\mathfrak{D}}_{+}\tilde{m})_{n}^{(i,k+1)}}{(\overline{\mathfrak{D}}_{-}\tilde{m})_{n}^{(i,k+1)}}\Big)\nonumber \\
= & 16C_{0}^{2}(h_{0}(\text{\textgreek{e}}))^{-2}\frac{r_{n}^{(1,k+1)}}{r_{0}}\Big\{\sum_{i=1}^{k-1}\log\Big(\frac{(\overline{\mathfrak{D}}_{+}\tilde{m})_{n}^{(i,k+1)}}{(\overline{\mathfrak{D}}_{+}\tilde{m})_{n}^{(i+1,k+1)}}\Big)+\log\Big(\frac{(\overline{\mathfrak{D}}_{+}\tilde{m})_{n}^{(k,k+1)}}{\tilde{m}|_{\mathcal{I}}-\tilde{m}_{n-1}^{(0,1)}}\Big)\Big\}\nonumber \\
= & 16C_{0}^{2}(h_{0}(\text{\textgreek{e}}))^{-2}\frac{r_{n}^{(1,k+1)}}{r_{0}}\log\Big(\frac{(\overline{\mathfrak{D}}_{+}\tilde{m})_{n}^{(1,k+1)}}{\tilde{m}|_{\mathcal{I}}-\tilde{m}_{n-1}^{(0,1)}}\Big)\nonumber \\
= & 16C_{0}^{2}(h_{0}(\text{\textgreek{e}}))^{-2}\frac{r_{n}^{(1,k+1)}}{r_{0}}\log\Big(\frac{(\overline{\mathfrak{D}}_{+}\tilde{m})_{n}^{(1,k+1)}}{(\mathfrak{D}_{+}\tilde{m})_{n-1}^{(0,1)}}\Big)\nonumber \\
\le & 16C_{0}^{2}(h_{0}(\text{\textgreek{e}}))^{-2}\frac{r_{n}^{(1,k+1)}}{r_{0}}\Big\{\log\Big(\frac{(\overline{\mathfrak{D}}_{+}\tilde{m})_{n}^{(1,k+1)}}{(\mathfrak{D}_{-}\tilde{m})_{n-1}^{(0,k)}}\Big)+\text{\textgreek{e}}^{3/2}\Big\}\nonumber \\
= & 16C_{0}^{2}(h_{0}(\text{\textgreek{e}}))^{-2}\frac{r_{n}^{(1,k+1)}}{r_{0}}\Big\{\log\Big(\frac{\tilde{m}{}_{n}^{(1,k+1)}}{\tilde{m}{}_{n-1}^{(1,k+1)}}\Big)+\text{\textgreek{e}}^{3/2}\Big\}\nonumber 
\end{align}
(where the inequality at the sixth line of (\ref{eq:EstimateForTheExponentialInChange})
follows from (\ref{eq:TotalMassDecreaseTopInteraction-1})). Therefore,
(\ref{eq:SecondRelationForRDifference}) and (\ref{eq:EstimateForTheExponentialInChange})
yield for any $1\le n\le n_{f}$ and any $V_{n}(v^{(k+2)}+\frac{4}{\sqrt{-\Lambda}}h_{2}(\text{\textgreek{e}}))\le v\le V_{n}(v^{(k+1)})$:
\begin{equation}
\frac{\partial_{v}r}{1-\frac{2m}{r}}\Big|_{(U_{n}(v^{(1)}+\frac{4}{\sqrt{-\Lambda}}h_{2}(\text{\textgreek{e}})),v)}\ge\frac{\partial_{v}r}{1-\frac{2m}{r}}\Big|_{(U_{n-1}(v^{(0)}+\frac{4}{\sqrt{-\Lambda}}h_{2}(\text{\textgreek{e}})),v)}\exp\Big(-C_{0}^{5}(h_{0}(\text{\textgreek{e}}))^{-3}\frac{r_{n}^{(1,k+1)}}{r_{0}}\Big\{\log\Big(\frac{\tilde{m}{}_{n}^{(1,k+1)}}{\tilde{m}{}_{n-1}^{(1,k+1)}}\Big)+\text{\textgreek{e}}^{3/2}\Big\}\Big).\label{eq:LowerBoundDvRBetweenBeams}
\end{equation}

In view of (\ref{eq:DerivativeInVDirectionKappaBar}), we can bound
for any $U_{n-1}(v^{(1)}+\frac{4}{\sqrt{-\Lambda}}h_{2}(\text{\textgreek{e}}))\le u\le U_{n-1}(v^{(0)})$:
\begin{equation}
\frac{-\partial_{u}r}{1-\frac{2m}{r}}\Big|_{(u,V_{n-1}(v^{(1)}))}\ge\frac{-\partial_{u}r}{1-\frac{2m}{r}}\Big|_{(u,V_{n-1}(v^{(k+1)})}.\label{eq:LowerBoundDuRBetweenBeams}
\end{equation}
Hence, using the fact that 

\begin{itemize}

\item{ From (\ref{eq:ZeroMassNearAxis}), (\ref{eq:UpperBoundForAxisInteractionProp}):
\begin{equation}
\frac{\partial_{v}r}{1-\frac{2m}{r}}\Big|_{(U_{n}(v^{(1)}+\frac{4}{\sqrt{-\Lambda}}h_{2}(\text{\textgreek{e}})),v)}=\frac{\partial_{v}r}{1-\frac{1}{3}\Lambda r^{2}}\Big|_{(U_{n}(v^{(1)}+\frac{4}{\sqrt{-\Lambda}}h_{2}(\text{\textgreek{e}})),v)}=(1+O(\text{\textgreek{e}}))\partial_{v}r|_{(U_{n}(v^{(1)}+\frac{4}{\sqrt{-\Lambda}}h_{2}(\text{\textgreek{e}})),v)},
\end{equation}
}

\item{ From (\ref{eq:EquationRForProof}), (\ref{eq:MassInfinity}),
(\ref{eq:BoundForRAwayInteractionProp}) and (\ref{eq:RoughBoundGeometry})
\begin{equation}
\frac{\partial_{v}r}{1-\frac{2m}{r}}\Big|_{(U_{n-1}(v^{(0)}+\frac{4}{\sqrt{-\Lambda}}h_{2}(\text{\textgreek{e}})),v)}=\frac{\partial_{v}r}{1-\frac{2m}{r}}\Big|_{(U_{n-1}(v^{(0)}),v)}(1+O(\text{\textgreek{e}})),
\end{equation}
}

\item{ From (\ref{eq:DerivativeInVDirectionKappaBar}), (\ref{eq:DerivativeInUDirectionKappa})
and the gauge condition (\ref{eq:GaugeInfinityMaximal}):
\begin{equation}
\frac{\partial_{v}r}{1-\frac{2m}{r}}\Big|_{\mathcal{R}_{n-1}^{(1,1)}\cap\{v=\bar{v}\}}=\frac{-\partial_{u}r}{1-\frac{2m}{r}}\Big|_{\mathcal{R}_{n-1}^{(1,1)}\cap\{u=\bar{v}-v_{0}\}},
\end{equation}
}

\end{itemize}

\noindent the bounds (\ref{eq:LowerBoundDvRBetweenBeams}) and (\ref{eq:LowerBoundDuRBetweenBeams})
yield for any $V_{n}(v^{(k+2)}+\frac{4}{\sqrt{-\Lambda}}h_{2}(\text{\textgreek{e}}))\le v\le V_{n}(v^{(k+1)})$:
\begin{equation}
\partial_{v}r\big|_{(U_{n}(v^{(1)}+\frac{4}{\sqrt{-\Lambda}}h_{2}(\text{\textgreek{e}})),v)}\ge\frac{-\partial_{u}r}{1-\frac{2m}{r}}\Big|_{(v-v_{0},V_{n-1}(v^{(k+1)})}\exp\Big(-C_{0}^{5}(h_{0}(\text{\textgreek{e}}))^{-3}\frac{r_{n}^{(1,k+1)}}{r_{0}}\log\Big(\frac{\tilde{m}{}_{n}^{(1,k+1)}}{\tilde{m}{}_{n-1}^{(1,k+1)}}\Big)-\text{\textgreek{e}}^{1/2}\Big).\label{eq:AlmostUsefulBoundForBeamSeperation}
\end{equation}

In view of (\ref{eq:ZeroMassNearAxis}), (\ref{eq:UpperBoundForAxisInteractionProp})
and the fact that 
\begin{equation}
\frac{-\partial_{u}r}{1-\frac{2m}{r}}\Big|_{\mathcal{R}_{n-1}^{(1,k+2)}\cap\{u=\bar{u}\}}=\frac{\partial_{v}r}{1-\frac{2m}{r}}\Big|_{\mathcal{R}_{n-1}^{(1,k+2)}\cap\{v=\bar{u}\}}
\end{equation}
(following from the gauge condition (\ref{eq:GaugeMirrorMaximal})),
from (\ref{eq:AlmostUsefulBoundForBeamSeperation}) we infer for any
$V_{n}(v^{(k+2)}+\frac{4}{\sqrt{-\Lambda}}h_{2}(\text{\textgreek{e}}))\le v\le V_{n}(v^{(k+1)})$:
\begin{equation}
\partial_{v}r\big|_{(U_{n}(v^{(1)}+\frac{4}{\sqrt{-\Lambda}}h_{2}(\text{\textgreek{e}})),v)}\ge\partial_{v}r\big|_{(U_{n-1}(v^{(1)}+\frac{4}{\sqrt{-\Lambda}}h_{2}(\text{\textgreek{e}})),v-v_{0})}\exp\Big(-C_{0}^{5}(h_{0}(\text{\textgreek{e}}))^{-3}\frac{r_{n}^{(1,k+1)}}{r_{0}}\log\Big(\frac{\tilde{m}{}_{n}^{(1,k+1)}}{\tilde{m}{}_{n-1}^{(1,k+1)}}\Big)-2\text{\textgreek{e}}^{1/2}\Big).\label{eq:UsefulInductiveBoundForBeamSeperation}
\end{equation}
Iterating (\ref{eq:UsefulInductiveBoundForBeamSeperation}) for $n_{1}<n\le n_{2}$,
we thus obtain (\ref{eq:UsefulBoundAlmostThere}).

\medskip{}

Let $2\le n_{1}\le n_{f}$ be such so that 
\begin{equation}
r_{n_{1}}^{(1,k+1)}\le2r_{0}
\end{equation}
and 
\begin{equation}
r_{n_{1}-1}^{(1,k+1)}>2r_{0}.\label{eq:PreviousN1}
\end{equation}
Note that if no such $n_{1}$ exists, then (\ref{eq:BoundSecondBeamchanged})
is automatically true ((\ref{eq:PreviousN1}) holds for $r_{1}^{(1,k+1)}$
as a corollary of the Cauchy stability estimates of Proposition \ref{prop:CauchyStabilityOfAdS}
and the choice of the initial data). 

Let us also define 
\begin{equation}
n_{2}=\max\big\{ n_{1}\le n\le n_{f}:\, r_{l}^{(1,k+1)}\le2r_{0}\mbox{ for all }n_{1}\le l\le n\big\}.
\end{equation}
In order to establish (\ref{eq:BoundSecondBeamchanged}), it suffices
to establish that, for all $n_{1}\le n\le n_{2}$: 
\begin{equation}
\frac{r_{n}^{(1,k+1)}}{r_{0}}-1\ge\exp\Big(-C_{0}^{7}(h_{0}(\text{\textgreek{e}}))^{-3}\log\big((h_{0}(\text{\textgreek{e}}))^{-1}\big)\Big).\label{eq:BoundToShowSecondBeam}
\end{equation}

In view of the fact that 
\begin{equation}
\sup_{1\le l_{1}<l_{2}\le n_{*}}\frac{\tilde{m}_{l_{2}}^{(1,k+1)}}{\tilde{m}_{l_{1}}^{(1,k+1)}}\le C_{0}(h_{0}(\text{\textgreek{e}}))^{-1}\label{eq:ControlInMassration}
\end{equation}
(following from (\ref{eq:TheIngoingVlasovInitially}), (\ref{eq:UpperBoundMassFromR})
and the fact that the sequence $\tilde{m}_{n}^{(1,k+1)}$ is increasing
in $n$ as a consequence of (\ref{eq:BoundToShowForMass})), from
(\ref{eq:UsefulBoundAlmostThere}) we infer that, for any $n_{1}\le n\le n_{2}$
and any $V_{n_{2}}(v^{(k+2)}+\frac{4}{\sqrt{-\Lambda}}h_{2}(\text{\textgreek{e}}))\le v\le V_{n_{2}}(v^{(k+1)})$:
\begin{equation}
\partial_{v}r\big|_{(U_{n}(v^{(1)}+\frac{4}{\sqrt{-\Lambda}}h_{2}(\text{\textgreek{e}})),v)}\ge\partial_{v}r\big|_{(U_{n_{1}-1}(v^{(1)}+\frac{4}{\sqrt{-\Lambda}}h_{2}(\text{\textgreek{e}})),v-(n-n_{1})v_{0})}\exp\Big(-C_{0}^{6}(h_{0}(\text{\textgreek{e}}))^{-3}\log\big((h_{0}(\text{\textgreek{e}}))^{-1}\big)\Big).\label{eq:UsefulBoundAlmostThere-1}
\end{equation}
Thus, integrating (\ref{eq:UsefulBoundAlmostThere-1}) from $v=V_{n}(v^{(k+2)}+\frac{4}{\sqrt{-\Lambda}}h_{2}(\text{\textgreek{e}}))$
up to $v=V_{n}(v^{(k+1)})$ and using (\ref{eq:PreviousN1}), we immediately
infer (\ref{eq:BoundToShowSecondBeam}).

\paragraph*{\noindent Proof of (\ref{eq:BoundForMaxBeamSeparation}).\emph{ }}

\noindent In view of (\ref{eq:h_2definition}), (\ref{eq:KappaChangeAdjacentRegions}),
(\ref{eq:KappaBarChangeAdjacentDomains}), as well as the boundary
condition (\ref{eq:GaugeInfinityMaximal}) and the bounds (\ref{eq:BoundForRAwayInteractionProp})
and (\ref{eq:UpperBoundForAxisInteractionProp}), the following one-sided
bound holds for all $2\le n\le n_{f}$: 
\begin{align}
r_{n}^{(k,k+1)}-r_{0} & =\int_{V_{n}(v^{(2k)})}^{V_{n}(v^{(k+1)}+\frac{4}{\sqrt{-\Lambda}}h_{2}(\text{\textgreek{e}}))}\partial_{v}r\Big|_{(U_{n}(v^{(k)}),v)}\, dv\label{eq:IntegralForChangeOfTotalBeamWidth}\\
 & \le\int_{V_{n}(v^{(2k)})}^{V_{n}(v^{(k+1)}+\frac{4}{\sqrt{-\Lambda}}h_{2}(\text{\textgreek{e}}))}\frac{\partial_{v}r}{1-\frac{2m}{r}}\Big|_{(U_{n}(v^{(k)}),v)}\Big(1+O(\text{\textgreek{e}})\Big)\, dv\nonumber \\
 & =\int_{U_{n-1}(v^{(k)})}^{U_{n-1}(v^{(0)}+\frac{4}{\sqrt{-\Lambda}}h_{2}(\text{\textgreek{e}}))}\frac{-\partial_{u}r}{1-\frac{2m}{r}}\Big|_{(u,V_{n-1}(v^{(k+1)}+\frac{4}{\sqrt{-\Lambda}}h_{2}(\text{\textgreek{e}})))}\Big(1+O(\text{\textgreek{e}})\Big)\, du.\nonumber 
\end{align}

We can readily compute (using also (\ref{eq:h_2definition}), (\ref{eq:UpperBoundForAxisInteractionProp})
and (\ref{eq:ImprovedRoughBoundBootstrap})): 
\begin{align}
\int_{U_{n-1}(v^{(1)}+\frac{4}{\sqrt{-\Lambda}}h_{2}(\text{\textgreek{e}}))}^{U_{n-1}(v^{(0)})}\frac{-\partial_{u}r}{1-\frac{2m}{r}} & \Big|_{(u,V_{n-1}(v^{(k+1)}+\frac{4}{\sqrt{-\Lambda}}h_{2}(\text{\textgreek{e}})))}\, du\label{eq:TermThatWillGiveTheLogarithm}\\
= & \int_{U_{n-1}(v^{(1)}+\frac{4}{\sqrt{-\Lambda}}h_{2}(\text{\textgreek{e}}))}^{U_{n-1}(v^{(0)})}\Big(1-\frac{2\tilde{m}_{n-1}^{(1,k+1)}}{r|_{(u,V_{n-1}(v^{(k+1)}+\frac{4}{\sqrt{-\Lambda}}h_{2}(\text{\textgreek{e}})))}}-\frac{1}{3}\Lambda r^{2}|_{(u,V_{n-1}(v^{(k+1)}+\frac{4}{\sqrt{-\Lambda}}h_{2}(\text{\textgreek{e}})))}\Big)^{-1}\times\nonumber \\
 & \hphantom{\int_{U_{n-1}(v^{(1)}+\frac{4}{\sqrt{-\Lambda}}h_{2}(\text{\textgreek{e}}))}^{U_{n-1}(v^{(0)})}\Big(1-\frac{2\tilde{m}_{n-1}^{(1,k+1)}}{r|_{(u,V_{n-1}(v^{(k+1)}+\frac{4}{\sqrt{-\Lambda}}h_{2}(\text{\textgreek{e}})))}}-\frac{1}{3}\Lambda}\times(-\partial_{u}r)|_{(u,V_{n-1}(v^{(k+1)}+\frac{4}{\sqrt{-\Lambda}}h_{2}(\text{\textgreek{e}})))}\, du\nonumber \\
= & \int_{r_{0}+O((h_{2}(\text{\textgreek{e}}))^{1/2})}^{r_{n-1}^{(1,k+1)}}\Big(1-\frac{2\tilde{m}_{n-1}^{(1,k+1)}}{r}-\frac{1}{3}\Lambda r^{2}\Big)^{-1}\, dr\nonumber \\
\le & r_{n-1}^{(1,k+1)}-r_{0}+C_{0}\tilde{m}_{n-1}^{(1,k+1)}\Big|\log\big(1-\frac{2\tilde{m}_{n-1}^{(1,k+1)}}{r_{0}}\big)\Big|.\nonumber 
\end{align}
From (\ref{eq:BoundMirror}) and (\ref{eq:BoundSecondBeamchanged})
we can similarly estimate: 
\begin{align}
\int_{U_{n-1}(v^{(k)})}^{U_{n-1}(v^{(1)})}\frac{-\partial_{u}r}{1-\frac{2m}{r}} & \Big|_{(u,V_{n-1}(v^{(k+1)}+\frac{4}{\sqrt{-\Lambda}}h_{2}(\text{\textgreek{e}})))}\, du\le\label{eq:TermThatWillGiveTheLogarithm-1}\\
\le & \int_{r_{n-1}^{(1,k+1)}+O((h_{2}(\text{\textgreek{e}}))^{1/2})}^{r_{n-1}^{(k,k+1)}}\Big(1-\frac{2\tilde{m}|_{\mathcal{I}}}{r}+O(\text{\textgreek{e}})\Big)^{-1}\, dr\le\nonumber \\
\le & r_{n-1}^{(k,k+1)}+C_{0}\tilde{m}|_{\mathcal{I}}\Big|\log\big(\exp\big(h_{0}(\text{\textgreek{e}})\big)^{-4}\big)+1\Big|\le\nonumber \\
\le & r_{n-1}^{(k,k+1)}-r_{n-1}^{(1,k+1)}+C_{0}\tilde{m}|_{\mathcal{I}}\big(h_{0}(\text{\textgreek{e}})\big)^{-4}.\nonumber 
\end{align}
From (\ref{eq:h_2definition}), (\ref{eq:ImprovedRoughBoundBootstrap}),
(\ref{eq:IntegralForChangeOfTotalBeamWidth}), (\ref{eq:TermThatWillGiveTheLogarithm})
and (\ref{eq:TermThatWillGiveTheLogarithm-1}) one readily obtains
the bound (\ref{eq:BoundForMaxBeamSeparation}).\qed

\subsection{\label{sub:NearlyTrapped}Formation of a nearly-trapped sphere}

In this section, we will establish (\ref{eq:NearTrappingIsAchieved}),
using the bounds (\ref{eq:RoughBoundGeometry})--(\ref{eq:BoundForMaxBeamSeparation})
of Proposition \ref{prop:TheMainBootstrapBeforeTrapping}.

Let us set 
\begin{equation}
n_{max}\doteq\max\big\{ n_{*}\in\mathbb{N}:\,\mathcal{R}_{n}^{(1,k+1)}\subset\mathcal{U}_{\text{\textgreek{e}}}^{+}\mbox{ for all }n\le n_{*}\big\}.\label{eq:DefinitionNmax}
\end{equation}
Note that, in view of (\ref{eq:DefinitionUntrappedRegion}) and (\ref{eq:defNf}),
$n_{max}$ satisfies 
\begin{equation}
n_{f}\le n_{max}\le n_{f}+1.\label{eq:doubleBoundNmax}
\end{equation}
Thus, (\ref{eq:UpperUNonTrapping}) implies that 
\begin{equation}
n_{max}\le(h_{1}(\text{\textgreek{e}}))^{-2}.\label{eq:InitialUpperBoundForNmax}
\end{equation}

Notice that (\ref{eq:doubleBoundNmax}) and the definition (\ref{eq:DefinitionUntrappedRegion})
imply that, if 
\begin{equation}
n_{max}<\frac{1}{2}(h_{1}(\text{\textgreek{e}}))^{-2},\label{eq:NonTrivalBoundNmax}
\end{equation}
then, necessarily, (\ref{eq:NearTrappingIsAchieved}) holds. Thus,
in order to establish (\ref{eq:NearTrappingIsAchieved}), it suffices
to show (\ref{eq:NonTrivalBoundNmax}). 

We will show (\ref{eq:NonTrivalBoundNmax}) by applying Lemma \ref{lem:ForMassIncrease}
(see Section \ref{sub:Auxiliary-lemmas}). In particular, setting
for any $1\le n\le n_{max}+1$ 
\begin{align}
\text{\textgreek{m}}_{n} & \doteq\frac{2\tilde{m}_{n-1}^{(1,k+1)}}{r_{0}},\label{eq:mu_n}\\
\text{\textgreek{r}}_{n} & \doteq\frac{r_{n-1}^{(k,k+1)}}{r_{0}},\label{eq:rho_n}
\end{align}
the inductive bounds (\ref{eq:BoundForMassIncrease}) and (\ref{eq:BoundForMaxBeamSeparation})
imply that $\text{\textgreek{m}}_{n},\text{\textgreek{r}}_{n}$ satisfy
\begin{align}
\text{\textgreek{r}}_{n+1} & \le\text{\textgreek{r}}_{n}+C_{1}\log\big((1-\text{\textgreek{m}}_{n})^{-1}+1\big),\label{eq:InductiveRelationSequence-1}\\
\text{\textgreek{m}}_{n+1} & \ge\text{\textgreek{m}}_{n}\exp\big(\frac{c_{1}}{\text{\textgreek{r}}_{n+1}}\big),\nonumber 
\end{align}
 for any $1\le n\le n_{max}+1$, with 
\begin{equation}
C_{1}=(h_{0}(\text{\textgreek{e}}))^{-4}
\end{equation}
and 
\begin{equation}
c_{1}=\frac{1}{16}\exp\big(-2(h_{0}(\text{\textgreek{e}}))^{-4}\big).
\end{equation}
Furthermore, setting 
\begin{equation}
\text{\textgreek{d}}=h_{3}(\text{\textgreek{e}})
\end{equation}
(where $h_{3}(\cdot)$ is defined by (\ref{eq:h_3definition})), the
definition (\ref{eq:DefinitionNmax}) immediately implies that 
\begin{equation}
\max_{0\le n\le n_{max}}\text{\textgreek{m}}_{n}<1-\text{\textgreek{d}}.\label{eq:NerTrappedMu-1}
\end{equation}

Note that, in view of Definition \ref{def:ThefamilyOfInitialData}
and Proposition \ref{prop:CauchyStabilityOfAdS}, we have 
\begin{equation}
\text{\textgreek{m}}_{0}\sim h_{0}(\text{\textgreek{e}})
\end{equation}
and 
\begin{equation}
\text{\textgreek{r}}_{0}\sim(h_{1}(\text{\textgreek{e}}))^{-1}.
\end{equation}
Hence, as a consequence of (\ref{eq:h_1_h_0_definition}) and (\ref{eq:h_3definition}),
\begin{equation}
\Big(\frac{C_{1}}{c_{1}}\Big)^{8}\ll\frac{\text{\textgreek{r}}_{0}}{\text{\textgreek{m}}_{0}}\label{eq:LargenessRho0-1}
\end{equation}
and 
\begin{equation}
\text{\textgreek{d}}<\big(\frac{\text{\textgreek{m}}_{0}}{\text{\textgreek{r}}_{0}}\big)^{(C_{1}/c_{1})^{4}}.\label{eq:SmallnessDelta-1}
\end{equation}

The relations (\ref{eq:InductiveRelationSequence-1})--(\ref{eq:SmallnessDelta-1})
allow us to apply Lemma \ref{lem:ForMassIncrease} (see Section \ref{sub:Auxiliary-lemmas})
with $n_{*}=n_{max}+1$ for the sequence $\text{\textgreek{r}}_{n},\text{\textgreek{m}}_{n}$.
Thus, in view of Lemma \ref{lem:ForMassIncrease}, we obtain the following
upper bound for $n_{max}$: 
\begin{equation}
n_{max}+1\le\exp\Big(\exp\big(2(h_{0}(\text{\textgreek{e}}))^{-4}\big)\Big)(h_{1}(\text{\textgreek{e}}))^{-1}.\label{eq:UpperBoundNmax}
\end{equation}
In particular, (\ref{eq:NonTrivalBoundNmax}) (and, thus, (\ref{eq:NearTrappingIsAchieved}))
holds.

\subsection{\label{sub:FinalStep}The final step of the evolution}

\noindent In this section, we will complete the proof of Theorem \ref{thm:TheTheorem},
using the near-trapping bound (\ref{eq:NearTrappingIsAchieved}),
the bounds (\ref{eq:RoughBoundGeometry})--(\ref{eq:BoundForMaxBeamSeparation})
of Proposition \ref{prop:TheMainBootstrapBeforeTrapping}, as well
a backwards-in-time Cauchy stability estimate (see Lemma \ref{lem:PerturbationInitialData}
in Section \ref{sub:Cauchy-Stability-Backwards}).

The bound (\ref{eq:NearTrappingIsAchieved}), combined with the estimates
(\ref{eq:NotEnoughMassBehind}) and (\ref{eq:BoundSecondBeamchanged})
of Proposition \ref{prop:TheMainBootstrapBeforeTrapping}, imply that,
necessarily (in view also of (\ref{eq:h_1_h_0_definition}), (\ref{eq:h_2definition}),
(\ref{eq:RoughBoundGeometry}), (\ref{eq:BoundForMassIncrease}) and
(\ref{eq:DefinitionNmax})): 
\begin{equation}
\frac{2\tilde{m}_{n_{max}+1}^{(1,k+1)}}{r_{0}}\ge1-2h_{3}(\text{\textgreek{e}}).\label{eq:NearTrappedLastIteration}
\end{equation}
 Therefore, applying again Lemma \ref{lem:ForMassIncrease} for $\text{\textgreek{m}}_{n},\text{\textgreek{r}}_{n}$
(defined by (\ref{eq:mu_n}), (\ref{eq:rho_n})) and $n_{*}=n_{max}+1$
yields, in view of (\ref{eq:NearTrappedLastIteration}), that, either
\begin{equation}
\text{\textgreek{m}}_{n_{max}+1}>1+h_{3}(\text{\textgreek{e}}),\label{eq:TheCaseOfTrappedSurface}
\end{equation}
or 
\begin{equation}
1-2h_{3}(\text{\textgreek{e}})\le\text{\textgreek{m}}_{n_{max}+1}\le1+h_{3}(\text{\textgreek{e}})\label{eq:TheCaseOfNearTrappedSurface}
\end{equation}
and
\begin{align}
\text{\textgreek{m}}_{n_{max}} & \le1-\exp\Big(-\exp\big(2(h_{0}(\text{\textgreek{e}}))^{-4}\big)\Big)(h_{1}(\text{\textgreek{e}}))^{2}\label{eq:AlmostTrappedLastStep}\\
\max\{\text{\textgreek{r}}_{n_{max}+1},\text{\textgreek{r}}_{n_{max}}\} & \le\exp\Big(\exp\big(2(h_{0}(\text{\textgreek{e}}))^{-4}\big)\Big)(h_{1}(\text{\textgreek{e}}))^{-1}\log\big((h_{1}(\text{\textgreek{e}}))^{-1}\big).\label{eq:MaxSeperationLastStep}
\end{align}

Let us set 
\begin{equation}
\bar{v}_{*}\doteq V_{n_{max}+1}\big(v^{(k+1)}+\frac{4}{\sqrt{-\Lambda}}h_{2}(\text{\textgreek{e}})\big)\label{eq:vbar*}
\end{equation}
(recall that (\ref{eq:vbar*}) equals $V_{n_{max}}\big(v^{(0)}+\frac{4}{\sqrt{-\Lambda}}h_{2}(\text{\textgreek{e}})\big)$,
in view of our conventions on the indices). The proof of Theorem \ref{thm:TheTheorem}
will follow by showing that
\begin{itemize}
\item Either 
\begin{equation}
\inf_{\mathcal{U}_{\text{\textgreek{e}}}\cap\{v=\bar{v}_{*}\}}\big(1-\frac{2m}{r}\big)<0\label{eq:TrappedSurfaceFormed}
\end{equation}
(in which case $(r_{/}^{(\text{\textgreek{e}})},(\text{\textgreek{W}}_{/}^{(\text{\textgreek{e}})})^{2},\bar{f}_{in/}^{(\text{\textgreek{e}})},\bar{f}_{out/}^{(\text{\textgreek{e}})})=(r_{/\text{\textgreek{e}}},\text{\textgreek{W}}_{/\text{\textgreek{e}}}^{2},\bar{f}_{in/\text{\textgreek{e}}},\bar{f}_{out/\text{\textgreek{e}}})$
in the statement of Theorem \ref{thm:TheTheorem}).
\item Or 
\begin{equation}
\inf_{\mathcal{U}_{\text{\textgreek{e}}}\cap\{v=\bar{v}_{*}\}}\big(1-\frac{2m'}{r'}\big)<0,\label{eq:PerturbedTrappedSurfaceFormed}
\end{equation}
where $(r',(\text{\textgreek{W}}')^{2},\bar{f}_{in}^{\prime},\bar{f}_{out}^{\prime})$
is a (possibly different) smooth solution to the system (\ref{eq:RequationFinal})--(\ref{eq:OutgoingVlasovFinal})
arising as a future development of an asymptotically AdS boundary-characteristic
initial data set $(r_{/\text{\textgreek{e}}}^{\prime},(\text{\textgreek{W}}_{/\text{\textgreek{e}}}^{\prime})^{2},\bar{f}_{in/\text{\textgreek{e}}}^{\prime},\bar{f}_{out/\text{\textgreek{e}}}^{\prime})$
on $\{u=0\}\cap\{0\le v\le v_{0\text{\textgreek{e}}}\}$ (satisfying
the reflecting gauge condition at $r=r_{0},+\infty$) which is $(h_{1}(\text{\textgreek{e}}))^{2}$
close to $(r_{/\text{\textgreek{e}}},\text{\textgreek{W}}_{/\text{\textgreek{e}}}^{2},\bar{f}_{in/\text{\textgreek{e}}},\bar{f}_{out/\text{\textgreek{e}}})$
with respect to the norm (\ref{eq:GeometricNormForCauchyStability}),
i.\,e.~satisfies, in particular, (\ref{eq:GaugeDifferenceBoundCauchystability-1})
and (\ref{eq:DifferenceBoundCauchyStability-1}) (in which case $(r_{/}^{(\text{\textgreek{e}})},(\text{\textgreek{W}}_{/}^{(\text{\textgreek{e}})})^{2},\bar{f}_{in/}^{(\text{\textgreek{e}})},\bar{f}_{out/}^{(\text{\textgreek{e}})})=(r_{/\text{\textgreek{e}}}^{\prime},(\text{\textgreek{W}}_{/\text{\textgreek{e}}}^{\prime})^{2},\bar{f}_{in/\text{\textgreek{e}}}^{\prime},\bar{f}_{out/\text{\textgreek{e}}}^{\prime})$
in the statement of Theorem \ref{thm:TheTheorem}). 
\end{itemize}
Notice that, in both cases, (\ref{eq:DecayInInitialDataNorm}) follows
readily from (\ref{eq:CSnormFamily}) and (\ref{eq:h_1_h_0_definition}).
To this end, we will proceed to treat the cases (\ref{eq:TheCaseOfTrappedSurface})
and (\ref{eq:TheCaseOfNearTrappedSurface}) separately. 

\medskip{}

\noindent \emph{Case I.} Assume that (\ref{eq:TheCaseOfTrappedSurface})
holds. Then, we will show that (\ref{eq:TrappedSurfaceFormed}) also
holds. We will argue by contradiction, assuming that 
\begin{equation}
\inf_{\mathcal{U}_{\text{\textgreek{e}}}\cap\{v=\bar{v}_{*}\}}\big(1-\frac{2m}{r}\big)\ge0.\label{eq:UntrappedForContradiction}
\end{equation}

Let us set 
\begin{equation}
\mathcal{C}_{*}\doteq\{U_{n_{max}+1}(v^{(1)}+\frac{4}{\sqrt{-\Lambda}}h_{2}(\text{\textgreek{e}}))\le u<U_{n_{max}+1}(v^{(0)})\}\cap\{v=\bar{v}_{*}\}\cap\mathcal{U}_{\text{\textgreek{e}}}.
\end{equation}
The renormalised mass $\tilde{m}$ is constant on $\mathcal{C}_{*}$,
satisfying in particular 
\begin{equation}
\tilde{m}|_{\mathcal{C}_{*}}=\tilde{m}_{n_{max}+1}^{(1,k+1)}.\label{eq:MassOnIngoingFinalLine}
\end{equation}
Since $\partial_{u}r<0$ on $\mathcal{U}_{\text{\textgreek{e}}}$
(see (\ref{eq:NegativeDerivativeRMaximal})), from (\ref{eq:UntrappedForContradiction})
and the fact that $\mathcal{C}_{*}$ does not contain its future endpoint,
we infer the following stronger bound: 
\begin{equation}
1-\frac{2m}{r}\big|_{\mathcal{C}_{*}}>0.\label{eq:CompletelyUntrappedForContradiction}
\end{equation}
Thus, we also have 
\begin{equation}
\partial_{v}r|_{\mathcal{C}_{*}}>0.\label{eq:PositiveDvrForLater}
\end{equation}

We will now show that the future endpoint of $\mathcal{C}_{*}$ is
exactly $(U_{n_{max}+1}(v^{(0)}),\bar{v}_{*})$. If there existed
some $(u_{b},\bar{v}_{*})\in(\partial\mathcal{U}_{\text{\textgreek{e}}}\backslash\mathcal{I})$
such that $U_{n_{max}+1}(v^{(1)}+\frac{4}{\sqrt{-\Lambda}}h_{2}(\text{\textgreek{e}}))\le u_{b}<U_{n_{max}+1}(v^{(0)})$,
then Theorem \ref{thm:maximalExtension} on the strucure of the maximal
future development would imply that $r$ extends continuously on $(u_{b},\bar{v}_{*})$
with 
\begin{equation}
r(u_{b},\bar{v}_{*})=r_{0\text{\textgreek{e}}}.\label{eq:R0OnFutureBoundary}
\end{equation}
However, in that case, (\ref{eq:TheCaseOfTrappedSurface}), (\ref{eq:MassOnIngoingFinalLine})
and (\ref{eq:R0OnFutureBoundary}) would imply that, for some $u_{b*}$
close enough to $u_{b}$ 
\begin{equation}
1-\frac{2m}{r}\big|_{(u_{b*},\bar{v}_{*})}<0,
\end{equation}
which is a contradiction in view of (\ref{eq:UntrappedForContradiction}).
Therefore, 
\[
\{U_{n_{max}+1}(v^{(1)}+\frac{4}{\sqrt{-\Lambda}}h_{2}(\text{\textgreek{e}}))\le u<U_{n_{max}+1}(v^{(0)})\}\cap\{v=\bar{v}_{*}\}\cap(\partial\mathcal{U}_{\text{\textgreek{e}}}\backslash\mathcal{I})=\emptyset,
\]
and, thus 
\begin{equation}
\mathcal{C}_{*}=\{U_{n_{max}+1}(v^{(1)}+\frac{4}{\sqrt{-\Lambda}}h_{2}(\text{\textgreek{e}}))\le u<U_{n_{max}+1}(v^{(0)})\}\cap\{v=\bar{v}_{*}\}.\label{eq:TheWholeSegmentInTheDevelopment}
\end{equation}

In order to complete the proof in the case when (\ref{eq:TheCaseOfTrappedSurface})
holds, it suffices to establish that 
\begin{equation}
\limsup_{\bar{u}\rightarrow U_{n_{max}+1}(v^{(0)})}\frac{r|_{(\bar{u},\bar{v}_{*})}}{r_{0}}\le1+O\big((h_{2}(\text{\textgreek{e}}))^{1/2}\big).\label{eq:AlmostOnMirror}
\end{equation}
Assuming that (\ref{eq:AlmostOnMirror}) holds, from (\ref{eq:TheCaseOfTrappedSurface}),
(\ref{eq:MassOnIngoingFinalLine}) and (\ref{eq:AlmostOnMirror})
(in view also of (\ref{eq:h_2definition}), (\ref{eq:h_3definition}))
we readily obtain 
\begin{equation}
\liminf_{\bar{u}\rightarrow U_{n_{max}+1}(v^{(0)})}\Big(1-\frac{2m}{r}\Big)\Big|_{(\bar{u},\bar{v}_{*})}<-\frac{1}{2}h_{3}(\text{\textgreek{e}})<0,
\end{equation}
which is a contradiction in view of (\ref{eq:UntrappedForContradiction}).

Let us set 
\begin{align}
\mathcal{B}_{*}\doteq\{U_{n_{max}+1}(v^{(1)}+ & \frac{4}{\sqrt{-\Lambda}}h_{2}(\text{\textgreek{e}}))\le u<U_{n_{max}+1}(v^{(0)})\}\cap\{V_{n_{max}+1}(v^{(k)})\le v\le\bar{v}_{*}\}.
\end{align}
From (\ref{eq:TheWholeSegmentInTheDevelopment}) and the structure
of the maximal future development of general initial data sets for
(\ref{eq:RequationFinal})--(\ref{eq:OutgoingVlasovFinal}) (see Theorem
\ref{thm:maximalExtension}), we infer that 
\[
\mathcal{B}_{*}\subset\mathcal{U}_{\text{\textgreek{e}}}.
\]
Furthermore, in view of (\ref{eq:ConstrainVFinal}) and (\ref{eq:PositiveDvrForLater}),
we infer that 
\begin{equation}
\partial_{v}r\big|_{\mathcal{B}_{*}}>0
\end{equation}
 and, thus (in view of (\ref{eq:NegativeDerivativeRMaximal})): 
\begin{equation}
1-\frac{2m}{r}\big|_{\mathcal{B}_{*}}>0.\label{eq:UntrappedInRectangle}
\end{equation}

In view of (\ref{eq:DefinitionNmax}) and the bounds (\ref{eq:NotEnoughMassBehind})
and (\ref{eq:BoundSecondBeamchanged}), we have 
\begin{equation}
\big\{ u\le U_{n_{max}+1}(v^{(1)})\big\}\cap\mathcal{U}_{\text{\textgreek{e}}}\subset\mathcal{U}_{\text{\textgreek{e}}}^{+}.
\end{equation}
Therefore, as a consequence of (\ref{eq:RoughBoundGeometry}), we
can estimate 
\begin{equation}
\log\Big(\frac{\partial_{v}r}{1-\frac{2m}{r}}\Big)\Big|_{\big\{ u=U_{n_{max}+1}(v^{(1)})\big\}\cap\mathcal{U}_{\text{\textgreek{e}}}}\le\big(h_{1}(\text{\textgreek{e}})\big)^{-4}\log\big((h_{3}(\text{\textgreek{e}}))^{-1}\big).\label{eq:FromRoughBound}
\end{equation}
Since (\ref{eq:DerivativeInUDirectionKappa}) implies that 
\begin{equation}
\partial_{u}\log\Big(\frac{\partial_{v}r}{1-\frac{2m}{r}}\Big)\le0,
\end{equation}
from (\ref{eq:FromRoughBound}) (and (\ref{eq:UpperBoundForAxisInteractionProp})
) we infer the one sided bound: 
\begin{equation}
\partial_{v}r\big|_{\mathcal{B}_{*}}\le2\big(h_{1}(\text{\textgreek{e}})\big)^{-4}\log\big((h_{3}(\text{\textgreek{e}}))^{-1}\big).\label{eq:BoundForRseperationAlreadyTrapped}
\end{equation}

Integrating (\ref{eq:BoundForRseperationAlreadyTrapped}) from $v=V_{n_{max}+1}(v^{(k)})$
up to $V_{n_{max}+1}(v^{(k+1)}+\frac{4}{\sqrt{-\Lambda}}h_{2}(\text{\textgreek{e}}))$
using (\ref{eq:h_2definition}), we finally obtain (\ref{eq:AlmostOnMirror}).
Thus, the proof in the case when (\ref{eq:TheCaseOfTrappedSurface})
holds is complete.

\medskip{}

\noindent \emph{Case II.} Assume that (\ref{eq:TheCaseOfNearTrappedSurface})
holds. Then, (\ref{eq:AlmostTrappedLastStep}) and (\ref{eq:MaxSeperationLastStep})
also hold. 

As a consequence of (\ref{eq:NotEnoughMassBehind}), (\ref{eq:BoundSecondBeamchanged})
and (\ref{eq:BoundForMassIncrease}), the bound (\ref{eq:AlmostTrappedLastStep})
implies that 
\begin{equation}
\inf_{\{u\le U_{n_{max}}(v^{(0)}+\frac{4}{\sqrt{-\Lambda}}h_{2}(\text{\textgreek{e}}))\}\cap\mathcal{U}_{\text{\textgreek{e}}}}\Big(1-\frac{2\tilde{m}}{r}\Big)\ge\frac{1}{2}\exp\Big(-\exp\big(2(h_{0}(\text{\textgreek{e}}))^{-4}\big)\Big)(h_{1}(\text{\textgreek{e}}))^{2}.
\end{equation}
Therefore, using (\ref{eq:NotEnoughMassBehind}), (\ref{eq:BoundSecondBeamchanged})
and (\ref{eq:BoundForMassIncrease}) to estimate $(1-\frac{2m}{r})$
in the region 
\[
\{U_{n_{max}}(v^{(0)}+\frac{4}{\sqrt{-\Lambda}}h_{2}(\text{\textgreek{e}}))\le u\le U_{n_{max}+1}(v^{(0)})\}\backslash\mathcal{R}_{n_{max}+1}^{(1,k+1)},
\]
we infer that 
\begin{equation}
\inf_{\{u\le U_{n_{max}+1}(v^{(1)}+\frac{4}{\sqrt{-\Lambda}}h_{2}(\text{\textgreek{e}}))\}\cap\mathcal{U}_{\text{\textgreek{e}}}}\Big(1-\frac{2\tilde{m}}{r}\Big)\ge\frac{1}{2}\exp\Big(-\exp\big(2(h_{0}(\text{\textgreek{e}}))^{-4}\big)\Big)(h_{1}(\text{\textgreek{e}}))^{2}.\label{eq:AwayFromTrappedBeforePerurbation}
\end{equation}

\begin{rem*}
Notice that, while $1-\frac{2\tilde{m}}{r}$ becomes $\sim h_{3}(\text{\textgreek{e}})$
in $\{u\le U_{n_{max}+1}(v^{(0)})\}\cap\mathcal{U}_{\text{\textgreek{e}}}$
(in view of (\ref{eq:TheCaseOfNearTrappedSurface})), when restricting
to the subregion $\{u\le U_{n_{max}+1}(v^{(1)}+\frac{4}{\sqrt{-\Lambda}}h_{2}(\text{\textgreek{e}}))\}\cap\mathcal{U}_{\text{\textgreek{e}}}$,
the improved bound (\ref{eq:AwayFromTrappedBeforePerurbation}) holds.
\end{rem*}

Let us set 
\begin{equation}
u_{*}\doteq U_{n_{max}+1}(v^{(1)}+\frac{4}{\sqrt{-\Lambda}}h_{2}(\text{\textgreek{e}})),\label{eq:u_*}
\end{equation}
noticing that 
\begin{equation}
supp(r^{2}T_{vv})\cap\{u=u_{*}\}\subset\{r\le\text{\textgreek{e}}^{1/2}\}\label{eq:SupportNotAtInfinity}
\end{equation}
as a consequence of (\ref{eq:BoundForRAwayInteractionProp}). Let
us also fix a smooth cut-off function $\text{\textgreek{q}}_{\text{\textgreek{e}}}:[u_{*},u_{*}+v_{0\text{\textgreek{e}}})\rightarrow[0,1]$
such that 
\begin{equation}
\text{\textgreek{q}}_{\text{\textgreek{e}}}(v)=1\mbox{ for }v\in[V_{n_{max}+1}(v^{(k+1)}),V_{n_{max}+1}(v^{(k+1)}+\frac{4}{\sqrt{-\Lambda}}h_{2}(\text{\textgreek{e}}))]\label{eq:SupportOfCutOff}
\end{equation}
 and 
\[
\text{\textgreek{q}}_{\text{\textgreek{e}}}(v)=0\mbox{ for }v\in[u_{*},u_{*}+v_{0\text{\textgreek{e}}})\backslash[V_{n_{max}+1}(v^{(k+1)}-\frac{1}{\sqrt{-\Lambda}}h_{2}(\text{\textgreek{e}})),V_{n_{max}+1}(v^{(k+1)}+\frac{5}{\sqrt{-\Lambda}}h_{2}(\text{\textgreek{e}}))].
\]
 We will then define the function $\widetilde{T}_{vv}:[u_{*},u_{*}+v_{0\text{\textgreek{e}}})\rightarrow\mathbb{R}$
by the relation 
\begin{equation}
\widetilde{T}_{vv}(v)\doteq\exp\big(-2C_{\text{\textgreek{e}}}^{2}\frac{u_{*}}{v_{0}}\big)(h_{1}(\text{\textgreek{e}}))^{2}\text{\textgreek{q}}_{\text{\textgreek{e}}}(v)T_{vv}(u_{*},v),\label{eq:WidetildeT}
\end{equation}
where $C_{\text{\textgreek{e}}}$ is defined by (\ref{eq:Cepsilon}).
Notice that, since 
\begin{equation}
2\pi\int_{V_{n_{max}+1}(v^{(k+1)})}^{V_{n_{max}+1}(v^{(k+1)}+\frac{4}{\sqrt{-\Lambda}}h_{2}(\text{\textgreek{e}}))}\frac{(1-\frac{2m}{r})}{\partial_{v}r}r^{2}T_{vv}\Big|_{(u_{*},v)}\, dv=\tilde{m}_{n_{max}+1}^{(1,k+1)}-\tilde{m}_{n_{max}+1}^{(1,k+2)}=\tilde{m}_{n_{max}+1}^{(1,k+1)},\label{eq:MassDifferenceForPerturbation}
\end{equation}
we can readily bound in view of (\ref{eq:WidetildeT}), (\ref{eq:BoundSecondBeamchanged}),
(\ref{eq:TheCaseOfNearTrappedSurface}) and (\ref{eq:MassDifferenceForPerturbation}):
\begin{equation}
\sup_{u_{*}\le\bar{v}\le u_{*}+v_{0\text{\textgreek{e}}}}(-\Lambda)\int_{u_{*}}^{u_{*}+v_{0\text{\textgreek{e}}}}\frac{r^{2}(u_{*},v)\frac{|\tilde{T}_{vv}(v)|}{\partial_{v}\text{\textgreek{r}}(u_{*},v)}}{|\text{\textgreek{r}}(u_{*},v)-\text{\textgreek{r}}(u_{*},\bar{v})|+\text{\textgreek{r}}(u_{*},u_{*})}dv\le\exp\big(-C_{\text{\textgreek{e}}}^{2}\frac{u_{*}}{v_{0\text{\textgreek{e}}}}\big)(h_{1}(\text{\textgreek{e}}))^{2},\label{eq:SmallnessPerturbation-1}
\end{equation}
where $\text{\textgreek{r}}$ is defined in terms of $r$ by the relation
\begin{equation}
\text{\textgreek{r}}\doteq\tan^{-1}\Big(\sqrt{-\frac{\Lambda}{3}}r\Big).
\end{equation}

Applying the backwards-in-time Cauchy stability lemma \ref{lem:PerturbationInitialData}
(see Section \ref{sub:Auxiliary-lemmas}), for $u_{*}$ given by (\ref{eq:u_*})
and $\widetilde{T}_{vv}$ given by (\ref{eq:WidetildeT}) (in view
of (\ref{eq:AwayFromTrappedBeforePerurbation}), (\ref{eq:SupportNotAtInfinity})
and (\ref{eq:SmallnessPerturbation-1})), we infer that there exists
an smooth asymptotically AdS boundary-characteristic initial data
set $(r_{/\text{\textgreek{e}}}^{\prime},(\text{\textgreek{W}}_{/\text{\textgreek{e}}}^{\prime})^{2},\bar{f}_{in/\text{\textgreek{e}}}^{\prime},\bar{f}_{out/\text{\textgreek{e}}}^{\prime})$
on $\{u=0\}$ for the system (\ref{eq:RequationFinal})--(\ref{eq:OutgoingVlasovFinal})
satisfying the reflecting gauge condition at $r=r_{0\text{\textgreek{e}}},+\infty$
with the following properties:

\begin{enumerate}

\item The initial data sets $(r_{/\text{\textgreek{e}}},\text{\textgreek{W}}_{/\text{\textgreek{e}}}^{2},\bar{f}_{in/\text{\textgreek{e}}},\bar{f}_{out/\text{\textgreek{e}}})$
and $(r_{/\text{\textgreek{e}}}^{\prime},(\text{\textgreek{W}}_{/\text{\textgreek{e}}}^{\prime})^{2},\bar{f}_{in/\text{\textgreek{e}}}^{\prime},\bar{f}_{out/\text{\textgreek{e}}}^{\prime})$
satisfy (\ref{eq:GaugeDifferenceBoundCauchystability-1}) and (\ref{eq:DifferenceBoundCauchyStability-1}). 

\item The maximal development $(\mathcal{U}_{\text{\textgreek{e}}}^{\prime};r^{\prime},(\text{\textgreek{W}}^{\prime})^{2},\bar{f}_{in}^{\prime},\bar{f}_{out}^{\prime})$
of $(r_{/\text{\textgreek{e}}},\text{\textgreek{W}}_{/\text{\textgreek{e}}}^{2},\bar{f}_{in/\text{\textgreek{e}}},\bar{f}_{out/\text{\textgreek{e}}})$
satisfies (\ref{eq:ComparableDevelopments}), (\ref{eq:EqualROnTheSupport})
and (\ref{eq:NewIngoingEnergyMomentum}). 

\end{enumerate}

Using primes to denote quantities associated to $(r^{\prime},(\text{\textgreek{W}}^{\prime})^{2},\bar{f}_{in}^{\prime},\bar{f}_{out}^{\prime})$,
we can readily estimate in view of (\ref{eq:ComparableDevelopments}),
(\ref{eq:EqualROnTheSupport}), (\ref{eq:NewIngoingEnergyMomentum})
and (\ref{eq:WidetildeT}): 
\begin{align}
\tilde{m}^{\prime}\big|_{(u_{*},\bar{v}_{*})} & =\int_{V_{n_{max}+1}(v^{(k+1)})}^{V_{n_{max}+1}(v^{(k+1)}+\frac{4}{\sqrt{-\Lambda}}h_{2}(\text{\textgreek{e}}))}\frac{(1-\frac{2m^{\prime}}{r^{\prime}})}{\partial_{v}r^{\prime}}(r^{\prime})^{2}T_{vv}^{\prime}\Big|_{(u_{*},v)}\, dv\label{eq:NewMass}\\
 & =\int_{V_{n_{max}+1}(v^{(k+1)})}^{V_{n_{max}+1}(v^{(k+1)}+\frac{4}{\sqrt{-\Lambda}}h_{2}(\text{\textgreek{e}}))}\frac{(1-\frac{2m^{\prime}}{r})}{\partial_{v}r}r^{2}(T_{vv}+\widetilde{T}_{vv})\Big|_{(u_{*},v)}\, dv\nonumber \\
 & \ge(1+\exp\big(-2C_{\text{\textgreek{e}}}^{2}\frac{u_{*}}{v_{0}}\big)(h_{1}(\text{\textgreek{e}}))^{2})\tilde{m}\big|_{(u_{*},\bar{v}_{*})}.\nonumber 
\end{align}
Therefore, since $\tilde{m}\big|_{(u_{*},\bar{v}_{*})}=\tilde{m}_{n_{max}+1}^{(1,k+1)}$,
the bound (\ref{eq:TheCaseOfNearTrappedSurface}) (in view also of
(\ref{eq:h_3definition}) and (\ref{eq:Cepsilon})) implies that 
\begin{equation}
\frac{2\tilde{m}^{\prime}\big|_{(u_{*},\bar{v}_{*})}}{r_{0}}\ge(1+\exp\big(-2C_{\text{\textgreek{e}}}^{2}\frac{u_{*}}{v_{0}}\big)(h_{1}(\text{\textgreek{e}}))^{2})(1-2h_{3}(\text{\textgreek{e}}))\ge1+h_{3}(\text{\textgreek{e}}).\label{eq:PerturbationIsTrapped}
\end{equation}

Since $\tilde{m}\big|_{(u_{*},\bar{v}_{*})}$ is constant on 
\[
\mathcal{C}_{*}^{\prime}\doteq\{U_{n_{max}+1}(v^{(1)}+\frac{4}{\sqrt{-\Lambda}}h_{2}(\text{\textgreek{e}}))\le u<U_{n_{max}+1}(v^{(0)})\}\cap\{v=\bar{v}_{*}\}\cap\mathcal{U}_{\text{\textgreek{e}}}^{\prime}
\]
 and satisfies (\ref{eq:PerturbationIsTrapped}), we can now repeat
the same arguments as in Case I (i.\,e.~the case when (\ref{eq:TheCaseOfTrappedSurface})
holds) in order to infer that (\ref{eq:PerturbedTrappedSurfaceFormed})
holds. 

\medskip{}

Thus, the proof of Theorem \ref{thm:TheTheorem} is complete. \qed

\subsection{\label{sub:Auxiliary-lemmas}Some auxiliary lemmas}

In this section, we will prove some lemmas necessary for the proof
of Proposition \ref{prop:TheMainBootstrapBeforeTrapping} and Theorem
\ref{thm:TheTheorem}.

\subsubsection{\label{sub:MaximumPrinciple} A maximum principle for $1+1$ wave-type
equations}

The following lemma provides a comparison inequality for certain $1+1$
equations of wave type, and is used in the proof of Proposition \ref{prop:TheMainBootstrapBeforeTrapping}.
\begin{lem}
\label{lem:HyperbolicMaximumPrinciple}For any $u_{0}<u_{1}$, $v_{0}<v_{1}$
and $a\in\mathbb{R}$, let $F_{1},F_{2}:[u_{0},u_{1}]\times[v_{0},v_{1}]\times(-\infty,a]\rightarrow(0,+\infty)$
be smooth functions so that 
\begin{equation}
\max_{(u,v)\in[u_{0}u_{1}]\times[v_{0},v_{1}]}F_{1}(u,v,z)<\min_{(u,v)\in[u_{0}u_{1}]\times[v_{0},v_{1}]}F_{2}(u,v,z)\label{eq:RelationF_1F_2}
\end{equation}
for any $z\in(-\infty,a]$ and 
\begin{equation}
\partial_{z}F_{1}(u,v,z),\partial_{z}F_{2}(u,v,z)\ge0\label{eq:IncreasingInLastVariable}
\end{equation}
for any $(u,v,z)\in[u_{0},u_{1}]\times[v_{0},v_{1}]\times(-\infty,a]$.
Suppose also $z_{1},z_{2}:[u_{0},u_{1}]\times[v_{0},v_{1}]\rightarrow(-\infty,a]$
are smooth solutions to the equations 
\begin{equation}
\partial_{v}\partial_{u}z_{1}=-F_{1}(u,v,z_{1})\partial_{u}z_{1}\partial_{v}z_{1}\label{eq:Equation1}
\end{equation}
and 
\begin{equation}
\partial_{v}\partial_{u}z_{2}=-F_{2}(u,v,z_{2})\partial_{u}z_{2}\partial_{v}z_{2},\label{eq:Equation2}
\end{equation}
satisfying the same characteristic initial data 
\begin{equation}
z_{1}(u,v_{0})=z_{2}(u,v_{0})=z_{\backslash}(u),\label{eq:InData1}
\end{equation}
\begin{equation}
z_{1}(u_{0},v)=z_{2}(u_{0},v)=z_{/}(v).\label{eq:InData2}
\end{equation}
 where $z_{/}:[v_{0},v_{1}]\rightarrow(-\infty,a)$ and $z_{\backslash}:[u_{0},u_{1}]\rightarrow(-\infty,a)$
are smooth functions so that 
\begin{equation}
z_{/}(v_{0})=z_{\backslash}(v_{1}),
\end{equation}
\begin{equation}
\partial_{v}z_{/}|_{(v_{0},v_{1})}>0\label{eq:IncreasingRight}
\end{equation}
and 
\begin{equation}
\partial_{u}z_{\backslash}|_{(u_{0},u_{1})}<0.\label{eq:DecreasingLeft}
\end{equation}
Then, the functions $z_{1},z_{2}$ satisfy 
\begin{equation}
\partial_{u}z_{i}<0<\partial_{v}z_{i},\mbox{ }i=1,2\label{eq:Monotonicity}
\end{equation}
in $(u_{0},u_{1})\times(v_{0},v_{1})$ and 
\begin{equation}
z_{1}\le z_{2}\label{eq:InequalityZ1,2}
\end{equation}
everywhere on $[u_{0},u_{1}]\times[v_{0},v_{1}]$.\end{lem}
\begin{proof}
We will first establish (\ref{eq:Monotonicity}). By applying a standard
continuity argument, rewriting equation (\ref{eq:Equation1}) as 
\begin{equation}
\partial_{v}\log(-\partial_{u}z_{1})=-\partial_{v}z_{1}F_{1}(u,v,z_{1})\label{eq:DvDerivativeZ1}
\end{equation}
 and integrating in $v$, using also the property (\ref{eq:DecreasingLeft})
of the initial data, we obtain that 
\begin{equation}
\partial_{u}z_{1}<0
\end{equation}
everywhere on $(u_{0},u_{1})\times(v_{0},v_{1})$. Similarly, rewriting
(\ref{eq:Equation1}) as 
\begin{equation}
\partial_{u}\log(\partial_{v}z_{1})=-\partial_{u}z_{1}F_{1}(u,v,z_{1})
\end{equation}
and integrating in $u$, using (\ref{eq:DecreasingLeft}), and then
repeating the same procedure for $z_{2}$, we finally obtain (\ref{eq:Monotonicity}). 

In order to establish (\ref{eq:InequalityZ1,2}), we will argue by
continuity: Let $u_{*}\in[u_{0},u_{1})$ be such that (\ref{eq:InequalityZ1,2}),
\begin{equation}
\partial_{v}z_{1}\le\partial_{v}z_{2}\label{eq:ContinuityDv}
\end{equation}
and 
\begin{equation}
\partial_{u}z_{1}\le\partial_{u}z_{2}\label{eq:ContinuityDu}
\end{equation}
hold on $[u_{0},u_{*}]\times[v_{0},v_{1}]$. Note that $u_{*}=u_{0}$
satisfies this condition: In this case, (\ref{eq:InequalityZ1,2})
and (\ref{eq:ContinuityDv}) follow directly from (\ref{eq:InData2}),
while (\ref{eq:ContinuityDu}) follows by integrating (\ref{eq:DvDerivativeZ1})
(and its analogue for $z_{2}$) and using (\ref{eq:RelationF_1F_2}).
We will show that there exists a $\text{\textgreek{d}}>0$, such that
(\ref{eq:InequalityZ1,2}), (\ref{eq:ContinuityDv}) and (\ref{eq:ContinuityDu})
hold on $[u_{0},u_{*}+\text{\textgreek{d}})\times[v_{0},v_{1}]$.

For any $\bar{v}\in(v_{0},v_{1}]$, integrating (\ref{eq:Equation1})
and (\ref{eq:Equation2}) in $v$ along $\{u_{*}\}\times[v_{0},\bar{v}]$,
we obtain: 
\begin{equation}
\log(-\partial_{u}z_{1})(u_{*},\bar{v})=-\int_{v_{0}}^{\bar{v}}F_{1}(u_{*},v,z_{1})\partial_{v}z_{1}\, dv+\log(-\partial_{u}z_{\backslash})(u_{*})\label{eq:LogDuZ1}
\end{equation}
and 
\begin{equation}
\log(-\partial_{u}z_{2})(u_{*},\bar{v})=-\int_{v_{0}}^{\bar{v}}F_{2}(u_{*},v,z_{2})\partial_{v}z_{2}\, dv+\log(-\partial_{u}z_{\backslash})(u_{*}).\label{eq:LogDuZ2}
\end{equation}

Let us define the auxiliary functions $F_{1;u_{*}\bar{v}},F_{2;u_{*}\bar{v}}:(-\infty,a]\rightarrow(0,+\infty)$
by the relations 
\begin{equation}
F_{1;u_{*}\bar{v}}(z)=\max_{v\in[v_{0},\bar{v}]}F_{1}(u_{*},v,z)\label{eq:Aux1}
\end{equation}
and 
\begin{equation}
F_{2;u_{*}\bar{v}}(z)=\min_{v\in[v_{0},\bar{v}]}F_{2}(u_{*},v,z).\label{eq:Aux2}
\end{equation}

In view of (\ref{eq:RelationF_1F_2}), (\ref{eq:IncreasingInLastVariable})
and the fact that (\ref{eq:InequalityZ1,2}) holds on $\{u_{*}\}\times[v_{0},\bar{v}]$,%
\footnote{Note that we can immediately restrict from $[u_{0},u_{1}]\times[v_{0},v_{1}]$
to $\{u_{*}\}\times[v_{0},\bar{v}]$ in (\ref{eq:RelationF_1F_2}).%
} we can bound for any $v\in[v_{0},\bar{v}]$: 
\begin{equation}
F_{1;u_{*}\bar{v}}(z_{1}(u_{*},v))<F_{2;u_{*}\bar{v}}(z_{1}(u_{*},v))\le F_{2;u_{*}\bar{v}}(z_{2}(u_{*},v)).\label{eq:InequalityAuxiliary}
\end{equation}
 Thus, subtracting (\ref{eq:LogDuZ1}) and (\ref{eq:LogDuZ2}) and
using (\ref{eq:InequalityAuxiliary}) and (\ref{eq:ContinuityDv})
(and the fact that $\partial_{v}z_{2}>0$, $\bar{v}>v_{0}$), we readily
infer that 
\begin{align}
\log(-\partial_{u}z_{1}) & (u_{*},\bar{v})-\log(-\partial_{u}z_{2})(u_{*},\bar{v})\label{eq:ComparisonDu}\\
 & \ge\int_{v_{0}}^{\bar{v}}F_{2;u_{*}\bar{v}}(z_{2}(u_{*},v))\partial_{v}z_{2}(u_{*},v)\, dv-\int_{v_{0}}^{\bar{v}}F_{1;u_{*}\bar{v}}(z_{1}(u_{*},v))\partial_{v}z_{1}(u_{*},v)\, dv\nonumber \\
 & >0.\nonumber 
\end{align}
 From (\ref{eq:ComparisonDu}) we thus infer that, for any $v_{0}<\bar{v}\le v_{1}$:
\begin{equation}
\partial_{u}z_{1}(u_{*},\bar{v})<\partial_{u}z_{2}(u_{*},\bar{v}).\label{eq:AlmostThereComparison}
\end{equation}
Therefore, since $z_{1},z_{2}$ are smooth, there exists a continuous
function $\text{\textgreek{d}}_{u}:[v_{0},v_{1}]\rightarrow[0,1)$
with $\text{\textgreek{d}}_{u}|_{(v_{0},v_{1}]}>0$, such that 
\begin{equation}
\partial_{u}z_{1}(u,v)\le\partial_{u}z_{2}(u,v)\mbox{ for }\{v_{0}\le v\le v_{1}\}\cap\{u_{*}\le u\le u_{*}+\text{\textgreek{d}}_{u}(v))\}.\label{eq:AlmostThereComparison-2}
\end{equation}

Similarly, by integrating equations (\ref{eq:Equation1}) and (\ref{eq:Equation2})
in $u$ along $[u_{0},u_{1}]\times\{v_{0}\}$ and repeating a similar
procedure (using (\ref{eq:InData1})), we also obtain that there exists
a continuous function $\text{\textgreek{d}}_{v}:[u_{0},u_{1}]\rightarrow[0,1)$
with $\text{\textgreek{d}}_{v}|_{(u_{0},u_{1}]}>0$, such that 
\begin{equation}
\partial_{v}z_{1}(\bar{u},v_{0})\le\partial_{v}z_{2}(\bar{u},v_{0})\mbox{ for }\{u_{0}\le u\le u_{1}\}\cap\{v_{0}\le v\le v_{0}+\text{\textgreek{d}}_{v}(u))\}.\label{eq:AlmostThereComparison-1}
\end{equation}

From (\ref{eq:AlmostThereComparison}) and (\ref{eq:AlmostThereComparison-1}),
we infer that there exists some $\text{\textgreek{d}}>0$, such that
\begin{equation}
z_{1}\le z_{2}\mbox{ on }(u_{*},u_{*}+\text{\textgreek{d}})\times[v_{0},v_{1}].\label{eq:DoneFirstInequality}
\end{equation}
In particular, (\ref{eq:InequalityZ1,2}) holds on $[u_{0},u_{*}+\text{\textgreek{d}})\times[v_{0},v_{1}]$.
Furthermore, for any $\bar{u}\in(u_{*},u_{*}+\text{\textgreek{d}})$
and any $\bar{v}\in(v_{0},v_{0}+\text{\textgreek{d}}_{v}(\bar{u}))$,
repeating the procedure leading to (\ref{eq:ComparisonDu}) with $\bar{u}$
in place of $u_{*}$ and using (\ref{eq:AlmostThereComparison-1})
and (\ref{eq:DoneFirstInequality}) in place of (\ref{eq:ContinuityDv})
and (\ref{eq:InequalityZ1,2}), respectively, we infer that: 
\begin{equation}
\partial_{u}z_{1}(\bar{u},\bar{v})\le\partial_{u}z_{2}(\bar{u},\bar{v}).\label{eq:OneMoreComparison}
\end{equation}
Thus, combining (\ref{eq:AlmostThereComparison-2}) and (\ref{eq:OneMoreComparison}),
we infer that (\ref{eq:ContinuityDu}) holds on $[u_{0},u_{*}+\text{\textgreek{d}}')\times[v_{0},v_{1}]$,
for some $0<\text{\textgreek{d}}'\le\text{\textgreek{d}}$. Finally,
the bound (\ref{eq:ContinuityDv}) on $[u_{0},u_{*}+\text{\textgreek{d}}')\times[v_{0},v_{1}]$
follows in a similar way as the proof of (\ref{eq:ComparisonDu}),
by integrating equations (\ref{eq:Equation1}) and (\ref{eq:Equation2})
in $u\in[u_{0},u_{*}+\text{\textgreek{d}}')$ for any $\bar{v}\in(v_{0},v_{1})$
and using (\ref{eq:RelationF_1F_2}), (\ref{eq:InequalityZ1,2}) and
(\ref{eq:ContinuityDu}) (which we have shown that they hold on $[u_{0},u_{*}+\text{\textgreek{d}}')\times[v_{0},v_{1}]$).
We will omit the details.
\end{proof}

\subsubsection{\label{sub:TheInductionLemma}A lemma for a system of inductive inequalities}

The following lemma is used to show that the inductive bounds (\ref{eq:BoundForMassIncrease})
and (\ref{eq:BoundForMaxBeamSeparation}) for $\tilde{m}_{n}^{(1,k+1)}$
and $r_{n}^{(k,k+1)}$ indeed lead to the formation of an almost-trapped
surface.
\begin{lem}
\label{lem:ForMassIncrease}Let $0<c_{1}\ll1\ll C_{1}$, and $0<\text{\textgreek{m}}_{0}\ll1\ll\text{\textgreek{r}}_{0}$,
$0<\text{\textgreek{d}}\ll1$ be given variables, such that 
\begin{equation}
\Big(\frac{C_{1}}{c_{1}}\Big)^{8}\ll\frac{\text{\textgreek{r}}_{0}}{\text{\textgreek{m}}_{0}}\label{eq:LargenessRho0}
\end{equation}
and 
\begin{equation}
\text{\textgreek{d}}<\big(\frac{\text{\textgreek{m}}_{0}}{\text{\textgreek{r}}_{0}}\big)^{(C_{1}/c_{1})^{4}}.\label{eq:SmallnessDelta}
\end{equation}
Let also $\text{\textgreek{m}}_{n},\text{\textgreek{r}}_{n}>0$, be
sequences of positive numbers, with $\text{\textgreek{m}}_{n}$ increasing
in $n$, such that for $0\le n\le n_{*}$ they satisfy 
\begin{align}
\text{\textgreek{r}}_{n+1} & \le\text{\textgreek{r}}_{n}+C_{1}\log\big((1-\text{\textgreek{m}}_{n})^{-1}+1\big),\label{eq:InductiveRelationRho}\\
\text{\textgreek{m}}_{n+1} & \ge\text{\textgreek{m}}_{n}\exp\big(\frac{c_{1}}{\text{\textgreek{r}}_{n+1}}\big),\label{eq:InductiveRelationMu}
\end{align}
and 
\begin{equation}
\max_{0\le n\le n_{*}-1}\text{\textgreek{m}}_{n}<1-\text{\textgreek{d}}.\label{eq:NerTrappedMu}
\end{equation}
Then, 
\begin{equation}
n_{*}\le\big(\frac{C_{1}}{c_{1}}\big)^{3}\text{\textgreek{r}}_{0}\text{\textgreek{m}}_{0}^{-(C_{1}/c_{1})^{2}}.\label{eq:UpperBoundN0}
\end{equation}
Furthermore, if $1-\text{\textgreek{d}}\le\text{\textgreek{m}}_{n_{*}}\le1+\text{\textgreek{d}}$,
we can bound: 
\begin{equation}
\text{\textgreek{m}}_{n_{*}-1}\le1-\big(\frac{c_{1}}{C_{1}}\big)^{3}\text{\textgreek{r}}_{0}^{-2}\text{\textgreek{m}}_{0}^{2(C_{1}/c_{1})^{2}}\label{eq:UpperBoundMassN_*}
\end{equation}
and 
\begin{equation}
\max\{\text{\textgreek{r}}_{n_{*}},\text{\textgreek{r}}_{n_{*}-1}\}\le\big(\frac{C_{1}}{c_{1}}\big)^{4}\frac{\text{\textgreek{r}}_{0}}{\text{\textgreek{m}}_{0}^{(C_{1}/c_{1})^{2}}}\log\big(\frac{\text{\textgreek{r}}_{0}}{\text{\textgreek{m}}_{0}}\big).\label{eq:UpperBoundRhoN_*}
\end{equation}
\end{lem}
\begin{rem*}
Notice that the right hand side of (\ref{eq:UpperBoundN0}) is independent
of $\text{\textgreek{d}}$. \end{rem*}
\begin{proof}
Let us define for any integer $k\ge1$
\begin{equation}
n_{k}=\max\Big\{0\le n\le n_{*}:\mbox{ }\text{\textgreek{m}}_{l}\le1-\frac{1}{2^{k}}\mbox{ for all }0\le l\le n\Big\},\label{eq:DefinitionNk}
\end{equation}
using the convention 
\begin{equation}
n_{0}=0.
\end{equation}
Notice that, in view of the fact that the sequence $\text{\textgreek{m}}_{n}$
is increasing, for all $k\ge1$ and all $n_{k-1}<n\le n_{k}$ we can
estimate: 
\begin{equation}
1-\frac{1}{2^{k-1}}\le\text{\textgreek{m}}_{n}\le1-\frac{1}{2^{k}}\label{eq:TrivialBoundMuN}
\end{equation}
(note that, in the case $n_{k-1}=n_{k}$, there is no $n$ satisfying
$n_{k-1}<n\le n_{k}$ and (\ref{eq:TrivialBoundMuN})).

Using (\ref{eq:TrivialBoundMuN}), from (\ref{eq:InductiveRelationRho})
we can bound for any $k\ge1$ such that $n_{k-1}<n_{k}$ and any $n_{k-1}<n\le n_{k}$:
\begin{equation}
\text{\textgreek{r}}_{n}\le\text{\textgreek{r}}_{n_{k-1}}+2C_{1}(\log2)k(n-n_{k-1})\label{eq:InductiveRelationSeperation}
\end{equation}
and, therefore, for any $0\le n\le n_{k}$ we have: 
\begin{equation}
\text{\textgreek{r}}_{n}\le2C_{1}(\log2)\Big(\sum_{l=1}^{k-1}l(n_{l}-n_{l-1})+k(n-n_{k-1})\Big)+\text{\textgreek{r}}_{0}\label{eq:BestEstimateForRhoN}
\end{equation}
(note that $\eqref{eq:BestEstimateForRhoN}$ holds for all $0\le n\le n_{k}$,
while the bounds (\ref{eq:TrivialBoundMuN}) and (\ref{eq:InductiveRelationSeperation})
are non-trivial only for those values of $k$ for which $n_{k}>n_{k-1}$). 

Let us set 
\begin{equation}
k_{1}\doteq32\lceil\log\frac{C_{1}}{c_{_{1}}}\rceil.\label{eq:k_1}
\end{equation}
Then, (\ref{eq:InductiveRelationSeperation}) implies that, for all
$0\le n\le n_{k_{1}}$ 
\begin{equation}
\text{\textgreek{r}}_{n}\le\text{\textgreek{r}}_{0}+2C_{1}(\log2)k_{1}n\label{eq:BoundRhoBeforeN_k1}
\end{equation}
and, thus, (\ref{eq:InductiveRelationMu}) implies that 
\begin{equation}
\log\big(\frac{\text{\textgreek{m}}_{n_{k_{1}}}}{\text{\textgreek{m}}_{0}}\big)\ge c_{1}\sum_{n=1}^{n_{k_{1}}}\text{\textgreek{r}}_{n}^{-1}\ge c_{1}\sum_{n=1}^{n_{k_{1}}}\frac{1}{\text{\textgreek{r}}_{0}+2C_{1}(\log2)k_{1}n}\ge\frac{c_{1}\log\Big(\frac{\text{\textgreek{r}}_{0}+2C_{1}(\log2)k_{1}n_{k_{1}}}{\text{\textgreek{r}}_{0}+2C_{1}(\log2)k_{1}}\Big)}{4C_{1}(\log2)k_{1}}.\label{eq:AlmostBoundNk_1}
\end{equation}
From (\ref{eq:DefinitionNk}) and (\ref{eq:AlmostBoundNk_1}) we readily
infer that 
\begin{equation}
\frac{c_{1}\log\Big(\frac{\text{\textgreek{r}}_{0}+2C_{1}(\log2)k_{1}n_{k_{1}}}{\text{\textgreek{r}}_{0}+2C_{1}(\log2)k_{1}}\Big)}{4C_{1}(\log2)k_{1}}\le-\log(\text{\textgreek{m}}_{0})
\end{equation}
and, therefore (using also (\ref{eq:LargenessRho0})): 
\begin{equation}
n_{k_{1}}\le\frac{\text{\textgreek{r}}_{0}}{\text{\textgreek{m}}_{0}^{(C_{1}/c_{1})^{2}}}.\label{eq:FinalBoundNk_1}
\end{equation}

For any $k\ge2$ such that $n_{k}>n_{k-1}+1$, from (\ref{eq:InductiveRelationMu}),
(\ref{eq:DefinitionNk}), (\ref{eq:TrivialBoundMuN}) and (\ref{eq:BestEstimateForRhoN})
we readily infer: 
\begin{equation}
\frac{1}{2^{k-2}}\ge\log\frac{\text{\textgreek{m}}_{n_{k}}}{\text{\textgreek{m}}_{n_{k-1}+1}}\ge c_{1}\sum_{n=n_{k-1}+2}^{n_{k}}\text{\textgreek{r}}_{n}^{-1}\ge\frac{c_{1}}{4C_{1}(\log2)}\frac{1}{(k-1)+\sum_{l=2}^{k-1}(l-1)\frac{(n_{l}-n_{l-1})}{n_{k}-n_{k-1}-1}}
\end{equation}
and, hence: 
\begin{align}
n_{k}-n_{k-1}-1 & \le\frac{\sum_{l=2}^{k-1}(l-1)(n_{l}-n_{l-1})}{\frac{c_{1}}{4C_{1}(\log2)}2^{k-2}-(k-1)}\label{eq:InductiveRelationNk}\\
 & \le\frac{k(k-1)}{2}\frac{\max_{2\le l\le k-1}(n_{l}-n_{l-1})}{\frac{c_{1}}{4C_{1}(\log2)}2^{k-2}-(k-1)}\nonumber \\
 & \le\frac{k(k-1)}{2}\frac{1}{\frac{c_{1}}{4C_{1}(\log2)}2^{k-2}-(k-1)}n_{k-1}.\nonumber 
\end{align}
The relation (\ref{eq:InductiveRelationNk}) also holds (trivially)
when $n_{k}\le n_{k-1}+1$. Thus, for any $k\ge k_{1}$, the bound
(\ref{eq:InductiveRelationNk}) yields:
\begin{equation}
n_{k}\le\big(1+\frac{C_{1}}{c_{1}}2^{-\frac{k-2}{4}}\big)n_{k-1}+1
\end{equation}
and, therefore, for any $k\ge2$: 
\begin{equation}
n_{k}\le16\frac{C_{1}}{c_{1}}\big(n_{k_{1}}+\max\{k-k_{1},0\}\big).\label{eq:FirstTotalUpperBoundNk}
\end{equation}
In view of (\ref{eq:FinalBoundNk_1}), we thus obtain for any $k\ge2$:
\begin{equation}
n_{k}\le16\frac{C_{1}}{c_{1}}\big(\frac{\text{\textgreek{r}}_{0}}{\text{\textgreek{m}}_{0}^{(C_{1}/c_{1})^{2}}}+\max\{k-k_{1},0\}\big).\label{eq:RoughGrowthRateNk}
\end{equation}

Let us set 
\begin{equation}
k_{2}\doteq4k_{1}+2\frac{\log\big(\text{\textgreek{r}}_{0}/\text{\textgreek{m}}_{0}^{(C_{1}/c_{1})^{2}}\big)}{\log2}.\label{eq:DefinitionK2}
\end{equation}
Note that (\ref{eq:SmallnessDelta}), (\ref{eq:NerTrappedMu}) and
(\ref{eq:DefinitionNk}) implies that 
\begin{equation}
n_{k_{2}}\le n_{*}-1.
\end{equation}
In view of (\ref{eq:RoughGrowthRateNk}), we have 
\begin{equation}
n_{k_{2}}\le\big(\frac{C_{1}}{c_{1}}\big)^{3}\frac{\text{\textgreek{r}}_{0}}{\text{\textgreek{m}}_{0}^{(C_{1}/c_{1})^{2}}}\label{eq:UpperBoundNk_2}
\end{equation}
and, for all $k\ge k_{2}$ (in view of (\ref{eq:InductiveRelationNk})
and (\ref{eq:RoughGrowthRateNk})):
\begin{equation}
n_{k}-n_{k-1}\le1.
\end{equation}
Furthermore, (\ref{eq:BestEstimateForRhoN}) implies (in view of (\ref{eq:DefinitionK2})
and (\ref{eq:UpperBoundNk_2})) that 
\begin{equation}
\max\{\text{\textgreek{r}}_{n_{k_{2}}+1},\text{\textgreek{r}}_{n_{k_{2}}}\}\le\big(\frac{C_{1}}{c_{1}}\big)^{3}\frac{\text{\textgreek{r}}_{0}}{\text{\textgreek{m}}_{0}^{(C_{1}/c_{1})^{2}}}\log\big(\frac{\text{\textgreek{r}}_{0}}{\text{\textgreek{m}}_{0}}\big).\label{eq:UpperBoundRhok_2}
\end{equation}

In view of (\ref{eq:DefinitionNk}), we have 
\begin{equation}
\text{\textgreek{m}}_{n_{k_{2}}}\le1-2^{-k_{2}}<\text{\textgreek{m}}_{n_{k_{2}}+1}.\label{eq:SeperationMn_k_2}
\end{equation}
We will consider two cases, depending on whether $\text{\textgreek{m}}_{n_{k_{2}}+1}$
is larger than $1-\text{\textgreek{d}}$ or not.

\begin{enumerate}

\item In the case $\text{\textgreek{m}}_{n_{k_{2}}+1}\ge1-\text{\textgreek{d}}$,
(\ref{eq:NerTrappedMu}) implies that $n_{k_{2}}+1=n_{*}$. Thus,
(\ref{eq:UpperBoundN0}) follows from (\ref{eq:UpperBoundNk_2}).
Furthermore, (\ref{eq:UpperBoundRhoN_*}) follows from (\ref{eq:UpperBoundRhok_2}),
while (\ref{eq:UpperBoundMassN_*}) follows from (\ref{eq:SeperationMn_k_2}).

\item In the case $\text{\textgreek{m}}_{n_{k_{2}}+1}<1-\text{\textgreek{d}}$,
we can assume without loss of generality that $n_{k_{2}}\le n_{*}-2$
(otherwise, (\ref{eq:UpperBoundN0}) follows from (\ref{eq:UpperBoundNk_2})).
From (\ref{eq:InductiveRelationRho}), (\ref{eq:UpperBoundRhok_2})
and (\ref{eq:SeperationMn_k_2}), we thus infer that 
\begin{equation}
\text{\textgreek{r}}_{n_{k_{2}}+2}\le\big(\frac{C_{1}}{c_{1}}\big)^{3}\frac{\text{\textgreek{r}}_{0}}{\text{\textgreek{m}}_{0}^{(C_{1}/c_{1})^{2}}}\log\big(\frac{\text{\textgreek{r}}_{0}}{\text{\textgreek{m}}_{0}}\big)+C_{1}\log\big((1-\text{\textgreek{m}}_{n_{k_{2}}+1})^{-1}\big)\label{eq:UpperBoundNextRho}
\end{equation}
Hence, setting 
\begin{equation}
M\doteq\big(\frac{C_{1}}{c_{1}}\big)^{3}\frac{\text{\textgreek{r}}_{0}}{\text{\textgreek{m}}_{0}^{(C_{1}/c_{1})^{2}}}\log\big(\frac{\text{\textgreek{r}}_{0}}{\text{\textgreek{m}}_{0}}\big),\label{eq:DefinitionM}
\end{equation}
from (\ref{eq:InductiveRelationMu}) and (\ref{eq:SeperationMn_k_2})
we calculate: 
\begin{align}
\text{\textgreek{m}}_{n_{k_{2}}+2} & \ge\text{\textgreek{m}}_{n_{k_{2}}+1}\exp\Big(\frac{c_{1}}{\text{\textgreek{r}}_{n_{k_{2}}+2}}\Big)\label{eq:LowerBoundMnBiggerThan1}\\
 & \ge\begin{cases}
(1-2^{-k_{2}})e^{\frac{c_{1}}{2M}}, & \mbox{if }\log\big((1-\text{\textgreek{m}}_{n_{k_{2}}+1})^{-1}\big)\le\frac{M}{C_{1}}\\
\text{\textgreek{m}}_{n_{k_{2}}+1}\Big(1+\frac{c_{1}}{C_{1}\log\big((1-\text{\textgreek{m}}_{n_{k_{2}}+1})^{-1}+1\big)}\Big), & \mbox{if }\log\big((1-\text{\textgreek{m}}_{n_{k_{2}}+1})^{-1}\big)>\frac{M}{C_{1}}.
\end{cases}\nonumber 
\end{align}
If $\log\big((1-\text{\textgreek{m}}_{n_{k_{2}}+1})^{-1}\big)\le\frac{M}{C_{1}}$,
in view of (\ref{eq:DefinitionK2}) and (\ref{eq:DefinitionM}) we
can bound (using also (\ref{eq:LargenessRho0})) 
\begin{equation}
(1-2^{-k_{2}})e^{\frac{c_{1}}{2M}}\ge\Big(1-\frac{\text{\textgreek{m}}_{0}^{2(C_{1}/c_{1})^{2}}}{\text{\textgreek{r}}_{0}^{2}}\Big)\Big(1+\frac{c_{1}}{2}\big(\frac{c_{1}}{C_{1}}\big)^{3}\frac{\text{\textgreek{m}}_{0}^{(C_{1}/c_{1})^{2}}}{\text{\textgreek{r}}_{0}\log\big(\frac{\text{\textgreek{r}}_{0}}{\text{\textgreek{m}}_{0}}\big)}\Big)>1+\text{\textgreek{d}}.
\end{equation}
If $\log\big((1-\text{\textgreek{m}}_{n_{k_{2}}+1})^{-1}\big)>\frac{M}{C_{1}}\gg(C_{1}/c_{1})^{3}$,
we can also estimate: 
\begin{equation}
\text{\textgreek{m}}_{n_{k_{2}}+1}\Big(1+\frac{c_{1}}{C_{1}\log\big((1-\text{\textgreek{m}}_{n_{k_{2}}+1})^{-1}+1\big)}\Big)\ge\Big(1-e^{-\log\big((1-\text{\textgreek{m}}_{n_{k_{2}}+1})^{-1}\big)}\Big)\Big(1+\frac{c_{1}}{C_{1}\log\big((1-\text{\textgreek{m}}_{n_{k_{2}}+1})^{-1}+1\big)}\Big)>1+\text{\textgreek{d}}.
\end{equation}
Therefore, (\ref{eq:LowerBoundMnBiggerThan1}) implies that 
\begin{equation}
\text{\textgreek{m}}_{n_{k_{2}}+2}>1+\text{\textgreek{d}}\label{eq:MuBiggerThan1SecondCase}
\end{equation}
and, hence, $n_{k_{2}+2}=n_{*}$. Thus, (\ref{eq:UpperBoundN0}) follows
again from (\ref{eq:UpperBoundNk_2}).

\end{enumerate}
\end{proof}

\subsubsection{\label{sub:Cauchy-Stability-Backwards}Cauchy stability backwards
in time}

The following lemma, which is essentially a backwards-in-time Cauchy
stability estimate for late time perturbations of $(\mathcal{U}_{\text{\textgreek{e}}};r,\text{\textgreek{W}}^{2},\bar{f}_{in},\bar{f}_{out})$,
is an easy corollary of Theorem \ref{prop:CauchyStability}.
\begin{lem}
\label{lem:PerturbationInitialData} For any $0<\text{\textgreek{e}}<\text{\textgreek{e}}_{0}$
(provided $\text{\textgreek{e}}_{0}$ is sufficiently small), any
$r_{0}>0$ satisfying (\ref{eq:BoundMirror}), let $(\mathcal{U}_{\text{\textgreek{e}}};r,\text{\textgreek{W}}^{2},\bar{f}_{in},\bar{f}_{out})$
be the maximal future development of $(r_{/\text{\textgreek{e}}},\text{\textgreek{W}}_{/\text{\textgreek{e}}}^{2},\bar{f}_{in/\text{\textgreek{e}}},\bar{f}_{out/\text{\textgreek{e}}})$,
and let us set 
\begin{equation}
C_{\text{\textgreek{e}}}\doteq\exp\Big(\exp\big(-2(h_{0}(\text{\textgreek{e}}))^{-4}\big)(h_{1}(\text{\textgreek{e}}))^{-4}\Big).\label{eq:Cepsilon}
\end{equation}
Then, for any $0\le u_{*}\le(h_{1}(\text{\textgreek{e}}))^{-2}v_{0\text{\textgreek{e}}}$
such that 
\begin{equation}
\mathcal{W}_{u_{*}}\doteq\{0<u\le u_{*}\}\cap\{u<v<u+v_{0\text{\textgreek{e}}}\}\subset\mathcal{U}_{\text{\textgreek{e}}},
\end{equation}
\begin{equation}
\sup_{\mathcal{W}_{u_{*}}}\big(1-\frac{2\tilde{m}}{r}\big)^{-1}\le C_{\text{\textgreek{e}}}\label{eq:UpperBoundTrappingInTheRegionOfPerturbation}
\end{equation}
and 
\begin{equation}
u_{*}+v_{0\text{\textgreek{e}}}\notin supp\big(r^{2}T_{vv}(u_{*},\cdot)\big),\label{eq:EbdsNotInSupport}
\end{equation}
and for any $\tilde{T}_{vv}:(u_{*},u_{*}+v_{0\text{\textgreek{e}}})\rightarrow\mathbb{R}$
smooth and compactly supported satisfying $\tilde{T}_{vv}(\cdot)\ge-T_{vv}(u_{*},\cdot)$
and 
\begin{equation}
\sup_{u_{*}\le\bar{v}\le u_{*}+v_{0\text{\textgreek{e}}}}(-\Lambda)\int_{u_{*}}^{u_{*}+v_{0\text{\textgreek{e}}}}\frac{r^{2}(u_{*},v)\frac{|\tilde{T}_{vv}(v)|}{\partial_{v}\text{\textgreek{r}}(u_{*},v)}}{|\text{\textgreek{r}}(u_{*},v)-\text{\textgreek{r}}(u_{*},\bar{v})|+\text{\textgreek{r}}(u_{*},u_{*})}dv\le\exp\big(-C_{\text{\textgreek{e}}}^{2}\frac{u_{*}}{v_{0\text{\textgreek{e}}}}\big)(h_{1}(\text{\textgreek{e}}))^{2}\label{eq:SmallnessPerturbation}
\end{equation}
with 
\begin{equation}
\text{\textgreek{r}}(u,v)\doteq\tan^{-1}\Big(\sqrt{-\frac{\Lambda}{3}}r\Big),\label{eq:RhoFunctionAuxiliary}
\end{equation}
the following statement holds: There exists a smooth asymptotically
AdS boundary-characteristic initial data set $(r_{/\text{\textgreek{e}}}^{\prime},(\text{\textgreek{W}}_{/\text{\textgreek{e}}}^{\prime})^{2},\bar{f}_{in/\text{\textgreek{e}}}^{\prime},\bar{f}_{out/\text{\textgreek{e}}}^{\prime})$
on $\{u=0\}$ for the system (\ref{eq:RequationFinal})--(\ref{eq:OutgoingVlasovFinal})
satisfying the reflecting gauge condition at $r=r_{0\text{\textgreek{e}}},+\infty$
with the following properties:

\begin{enumerate}

\item The initial data sets $(r_{/\text{\textgreek{e}}},\text{\textgreek{W}}_{/\text{\textgreek{e}}}^{2},\bar{f}_{in/\text{\textgreek{e}}},\bar{f}_{out/\text{\textgreek{e}}})$
and $(r_{/\text{\textgreek{e}}}^{\prime},(\text{\textgreek{W}}_{/\text{\textgreek{e}}}^{\prime})^{2},\bar{f}_{in/\text{\textgreek{e}}}^{\prime},\bar{f}_{out/\text{\textgreek{e}}}^{\prime})$
are $(h_{1}(\text{\textgreek{e}}))^{2}$ close in the (\ref{eq:GeometricNormForCauchyStability})
norm, and in particular: 
\begin{align}
\sup_{v\in[0,v_{0\text{\textgreek{e}}})}\Bigg\{\Big|\log\big(\frac{\text{\textgreek{W}}_{/\text{\textgreek{e}}}^{2}}{1-\frac{1}{3}\Lambda r_{/\text{\textgreek{e}}}^{2}}\big)-\log\big(\frac{(\text{\textgreek{W}}_{/\text{\textgreek{e}}}^{\prime})^{2}}{1-\frac{1}{3}\Lambda(r_{/\text{\textgreek{e}}}^{\prime})^{2}}\big)\Big|+\Big|\log\Big(\frac{2\partial_{v}r_{/\text{\textgreek{e}}}}{1-\frac{2m_{/\text{\textgreek{e}}}}{r_{/\text{\textgreek{e}}}}}\Big)-\log\Big(\frac{2\partial_{v}r_{/\text{\textgreek{e}}}^{\prime}}{1-\frac{2m_{/\text{\textgreek{e}}}^{\prime}}{r_{/\text{\textgreek{e}}}^{\prime}}}\Big)\Big|+\label{eq:GaugeDifferenceBoundCauchystability-1}\\
+\Big|\log\Big(\frac{1-\frac{2m_{/_{\text{\textgreek{e}}}}}{r_{/\text{\textgreek{e}}}}}{1-\frac{1}{3}\Lambda r_{/\text{\textgreek{e}}}^{2}}\Big)-\log\Big(\frac{1-\frac{2m_{/\text{\textgreek{e}}}^{\prime}}{r_{/\text{\textgreek{e}}}^{\prime}}}{1-\frac{1}{3}\Lambda(r_{/\text{\textgreek{e}}}^{\prime})}\Big)\Big|+\sqrt{-\Lambda}|\tilde{m}_{/\text{\textgreek{e}}}-\tilde{m}_{/\text{\textgreek{e}}}^{\prime}|\Bigg\}(v) & \le(h_{1}(\text{\textgreek{e}}))^{2}\nonumber 
\end{align}
and 
\begin{equation}
\sup_{\bar{v}\in[0,v_{0\text{\textgreek{e}}}]}\int_{0}^{v_{0\text{\textgreek{e}}}}\frac{\big|r_{/\text{\textgreek{e}}}^{2}\frac{(T_{vv})_{/\text{\textgreek{e}}}}{\partial_{v}\text{\textgreek{r}}_{/\text{\textgreek{e}}}}(v)-(r_{/\text{\textgreek{e}}}^{\prime})^{2}\frac{(T_{vv})_{/\text{\textgreek{e}}}^{\prime}}{\partial_{v}\text{\textgreek{r}}_{/\text{\textgreek{e}}}^{\prime}}(v)\big|}{|\text{\textgreek{r}}_{/\text{\textgreek{e}}}(\bar{v})-\text{\textgreek{r}}_{/\text{\textgreek{e}}}(v)|+\text{\textgreek{r}}_{/\text{\textgreek{e}}}(0)}\, dv\le(h_{1}(\text{\textgreek{e}}))^{2}.\label{eq:DifferenceBoundCauchyStability-1}
\end{equation}

\item The maximal future development $(\mathcal{U}_{\text{\textgreek{e}}}^{\prime};r^{\prime},(\text{\textgreek{W}}^{\prime})^{2},\bar{f}_{in}^{\prime},\bar{f}_{out}^{\prime})$
of $(r_{/\text{\textgreek{e}}}^{\prime},(\text{\textgreek{W}}_{/\text{\textgreek{e}}}^{\prime})^{2},\bar{f}_{in/\text{\textgreek{e}}}^{\prime},\bar{f}_{out/\text{\textgreek{e}}}^{\prime})$
satisfies 
\begin{equation}
\mathcal{W}_{u_{*}}\subset\mathcal{U}_{\text{\textgreek{e}}}^{\prime},\label{eq:ComparableDevelopments}
\end{equation}
\begin{equation}
r^{\prime}|_{\{u=u_{*}\}\cap supp(T_{vv})}=r|_{\{u=u_{*}\}\cap supp(T_{vv})}\label{eq:EqualROnTheSupport}
\end{equation}
 and 
\begin{equation}
T_{vv}^{\prime}|_{\{u=u_{*}\}}=T_{vv}|_{\{u=u_{*}\}}+\tilde{T}_{vv}.\label{eq:NewIngoingEnergyMomentum}
\end{equation}

\end{enumerate}\end{lem}
\begin{proof}
In view of (\ref{eq:BoundMirror}), (\ref{eq:MassInfinity}), (\ref{eq:NotEnoughMassBehind})
and (\ref{eq:BoundSecondBeamchanged}), we can readily estimate 
\begin{equation}
\sup_{\mathcal{W}_{u_{*}}\backslash\cup_{n}\mathcal{R}_{n}^{(1,k+1)}}\Big(1-\frac{2\tilde{m}}{r}\Big)^{-1}\le2\exp\Big((h_{0}(\text{\textgreek{e}}))^{-4}\Big).\label{eq:TotalDistanceFromTrappingAway}
\end{equation}
Therefore, using (\ref{eq:UpperBoundTrappingInTheRegionOfPerturbation})
for $\cup_{n}\mathcal{R}_{n}^{(1,k+1)}$ and (\ref{eq:TotalDistanceFromTrappingAway})
for $\mathcal{W}_{u_{*}}\backslash\cup_{n}\mathcal{R}_{n}^{(1,k+1)}$,
the relations (\ref{eq:TotalChangeKappaBarEachIteration}) and (\ref{eq:TotalChangeKappaEachIteration})
imply (in view of (\ref{eq:h_2definition}), (\ref{eq:BoundForRAwayInteractionProp})
and the fact that $u_{*}\le(h_{1}(\text{\textgreek{e}}))^{-2}v_{0\text{\textgreek{e}}}$)
that 
\begin{equation}
\sup_{\mathcal{W}_{u_{*}}}\Bigg(\Big|\log\Big(\frac{-\partial_{u}r}{1-\frac{2m}{r}}\Big)\Big|+\Big|\log\Big(\frac{\partial_{v}r}{1-\frac{2m}{r}}\Big)\Big|\Bigg)\le(h_{1}(\text{\textgreek{e}}))^{-3}\exp\big((h_{0}(\text{\textgreek{e}}))^{-4}\big).\label{eq:BoundDuRDvRAwayFromSevereTrapping}
\end{equation}
Similarly, equations (\ref{eq:DerivativeInUDirectionKappa}) and (\ref{eq:DerivativeInVDirectionKappaBar}),
in view of the relations (\ref{eq:TotalChangeKappaBarEachIteration}),
(\ref{eq:TotalChangeKappaEachIteration}) (using again the bounds
(\ref{eq:UpperBoundTrappingInTheRegionOfPerturbation}) (\ref{eq:TotalDistanceFromTrappingAway}))
imply that 
\begin{equation}
\sup_{\bar{u}}\int_{\{u=\bar{u}\}\cap\mathcal{W}_{u_{*}}}rT_{vv}\, dv+\sup_{\bar{v}}\int_{\{v=\bar{v}\}\cap\mathcal{W}_{u_{*}}}rT_{uu}\, du\le(h_{1}(\text{\textgreek{e}}))^{-1}\exp\big((h_{0}(\text{\textgreek{e}}))^{-4}\big).\label{eq:BoundrTAwayFromSevereTrapping}
\end{equation}

Let us fix a set of smooth functions $r_{/}^{*},(\text{\textgreek{W}}_{/}^{*})^{2}:[u_{*},u_{*}+v_{0\text{\textgreek{e}}})\rightarrow(0,+\infty)$
and $\bar{f}_{in/}^{*},\bar{f}_{out/}^{*}:[u_{*},u_{*}+v_{0\text{\textgreek{e}}})\times(0,+\infty)\rightarrow[0,+\infty)$
satisfying the following requirements:

\begin{enumerate}

\item $(r_{/}^{*},(\text{\textgreek{W}}_{/}^{*})^{2},\bar{f}_{in/}^{*},\bar{f}_{out/}^{*})$
is a smooth asymptotically AdS boundary-characteristic initial data
set for for the system (\ref{eq:RequationFinal})--(\ref{eq:OutgoingVlasovFinal})
on $\{u_{*}\}\times[u_{*},u_{*}+v_{0\text{\textgreek{e}}})$ satisfying
the reflecting gauge condition at $r^{*}=r_{0\text{\textgreek{e}}},+\infty$.

\item The function $r_{/}^{*}$ satisfies for any $v$ such that
$(u_{*},v)\in supp(T_{vv})$ 
\begin{equation}
r_{/}^{*}(v)=r|_{supp(T_{vv})\cap\{u=u_{*}\}}.\label{eq:SameRLateTime}
\end{equation}

\item The function $\bar{f}_{in/}^{*}$ satisfies for all $v\in[u_{*},u_{*}+v_{0\text{\textgreek{e}}})$:
\begin{equation}
\int_{0}^{+\infty}\big((\text{\textgreek{W}}_{/}^{*})^{2}(v)p^{v}\big)^{2}\bar{f}_{in/}^{*}(v;p^{v})\,(r_{/}^{*})^{2}(v)\frac{dp^{v}}{p^{v}}=T_{vv}(u_{*},v)+\tilde{T}_{vv}(v).\label{eq:IngoingEnergyMomentumDeformed}
\end{equation}

\item The function $\bar{f}_{out/}^{*}$ satisfies for all $(v,p^{v})\in[u_{*},u_{*}+v_{0\text{\textgreek{e}}})\times(0,+\infty)$:
\begin{equation}
\bar{f}_{out/}^{*}(v;p^{v})=\bar{f}_{out}(u_{*},v;p^{v}).\label{eq:SameOutgoingF}
\end{equation}

\item The initial data sets $(r,\text{\textgreek{W}}^{2},\bar{f}_{in},\bar{f}_{out})|_{u=u_{*}}$
and $(r_{/}^{*},(\text{\textgreek{W}}_{/}^{*})^{2},\bar{f}_{in/}^{*},\bar{f}_{out/}^{*})$
satisfy 
\begin{align}
\sup_{v\in[u_{*},u_{*}+v_{0\text{\textgreek{e}}})}\Bigg\{\Big|\log\big(\frac{\text{\textgreek{W}}^{2}}{1-\frac{1}{3}\Lambda r^{2}}\big)\Big|_{u=u_{*}}-\log\big(\frac{(\text{\textgreek{W}}_{/}^{*})^{2}}{1-\frac{1}{3}\Lambda(r_{/}^{*})^{2}}\big)\Big|+\Big|\log\Big(\frac{2\partial_{v}r}{1-\frac{2m}{r}}\Big)\Big|_{u=u_{*}}-\log & \Big(\frac{2\partial_{v}r_{/}^{*}}{1-\frac{2m_{/}^{*}}{r_{/}^{*}}}\Big)\Big|+\label{eq:GaugeDifferenceBoundCauchystability-2}\\
+\Big|\log\Big(\frac{1-\frac{2m}{r}}{1-\frac{1}{3}\Lambda r^{2}}\Big)\Big|_{u=u_{*}}-\log\Big(\frac{1-\frac{2m_{/}^{*}}{r_{/}^{*}}}{1-\frac{1}{3}\Lambda(r_{/}^{*})}\Big)\Big|+\sqrt{-\Lambda}\big|\tilde{m}|_{u=u_{*}}-\tilde{m}_{/}^{*}\big|\Bigg\}(v) & \le\exp\big(-\frac{1}{2}C_{\text{\textgreek{e}}}^{2}\frac{u_{*}}{v_{0}}\big)(h_{1}(\text{\textgreek{e}}))^{2}\nonumber 
\end{align}
and 
\begin{equation}
\sup_{v\in[v_{1},v_{2}]}\int_{v_{1}}^{v_{2}}\frac{\big|r^{2}T_{vv}(u_{*},\bar{v})-(r_{/}^{*})^{2}(T_{vv})_{/}^{*}(\bar{v})\big|}{|\text{\textgreek{r}}_{/}(v)-\text{\textgreek{r}}_{/}(\bar{v})|+\text{\textgreek{r}}_{/}(v_{1})}\, d\bar{v}\le\exp\big(-\frac{1}{2}C_{\text{\textgreek{e}}}^{2}\frac{u_{*}}{v_{0}}\big)(h_{1}(\text{\textgreek{e}}))^{2}.\label{eq:DifferenceBoundCauchyStability-2}
\end{equation}

\end{enumerate}

\noindent \emph{Remark.} As a consequence of (\ref{eq:EbdsNotInSupport}),
by suitably deforming $r_{/}^{*}$ near $v=u_{*}+v_{0\text{\textgreek{e}}}$,
we can always arrange that (\ref{eq:GaugeInfinityInitialData}) and
(\ref{eq:SameRLateTime}) are satisfied simultaneously. Furthermore,
since $\tilde{T}_{vv}$ is compactly supported in $(u_{*},u_{*}+v_{0\text{\textgreek{e}}})$,
we can always choose $\bar{f}_{in/}^{*}=f_{in}|_{\{u=u_{*}\}}$ in
a neighborhood of $v=u_{*},u_{*}+v_{0\text{\textgreek{e}}}$, so that
(\ref{eq:LeftBoundaryConditionInitialData}) and (\ref{eq:RightBoundaryConditionInitialData})
are satisfied. Finally, $(r_{/}^{*},(\text{\textgreek{W}}_{/}^{*})^{2},\bar{f}_{in/}^{*},\bar{f}_{out/}^{*})$
can be chosen so that (\ref{eq:GaugeDifferenceBoundCauchystability-2})
and (\ref{eq:DifferenceBoundCauchyStability-2}) are satisfied because
of (\ref{eq:SmallnessPerturbation}) and the relations (\ref{eq:RelationHawkingMass}),
(\ref{eq:DerivativeInVDirectionKappaBar}) and (\ref{eq:DerivativeTildeVMass}).

\medskip{}

Let us now consider the two sets of initial data $(r,\text{\textgreek{W}}^{2},\bar{f}_{in},\bar{f}_{out})|_{u=u_{*}}$
and $(r_{/}^{*},(\text{\textgreek{W}}_{/}^{*})^{2},\bar{f}_{in/}^{*},\bar{f}_{out/}^{*})$
on $\{u=u_{*}\}\cap\{u_{*}\le v<u_{*}+v_{0\text{\textgreek{e}}}\}$.
The maximal \underline{past} development of $(r,\text{\textgreek{W}}^{2},\bar{f}_{in},\bar{f}_{out})|_{u=u_{*}}$
(see the remark below Theorem \ref{thm:maximalExtension}) coincides
with $(\mathcal{W}_{u_{*}};r,\text{\textgreek{W}}^{2},\bar{f}_{in},\bar{f}_{out})$
when restricted on $\{u\ge0\}$ and, in view of (\ref{eq:BoundDuRDvRAwayFromSevereTrapping})
and (\ref{eq:BoundrTAwayFromSevereTrapping}), satisfies 
\begin{align}
\sup_{\mathcal{W}_{u_{*}}}\Bigg\{\Big|\log\big(\frac{\text{\textgreek{W}}^{2}}{1-\frac{1}{3}\Lambda r^{2}}\big)\Big|+\Big|\log\Big(\frac{2\partial_{v}r}{1-\frac{2m}{r}}\Big)\Big|+\Big|\log\Big(\frac{1-\frac{2m}{r}}{1-\frac{1}{3}\Lambda r^{2}}\Big)\Big|+\sqrt{-\Lambda}|\tilde{m}|\Bigg\}+\label{eq:UpperBoundNonTrappingForCauchyStability-2}\\
+\sup_{\bar{u}}\int_{\{u=\bar{u}\}\cap\mathcal{W}_{u_{*}}}rT_{vv}\, dv+\sup_{\bar{v}}\int_{\{v=\bar{v}\}\cap\mathcal{W}_{u_{*}}}rT_{uu}\, du & \le4(h_{1}(\text{\textgreek{e}}))^{-3}\exp\big(-(h_{0}(\text{\textgreek{e}}))^{-4}\big).\nonumber 
\end{align}
Therefore, in view of (\ref{eq:UpperBoundNonTrappingForCauchyStability-2}),
(\ref{eq:GaugeDifferenceBoundCauchystability-2}) and (\ref{eq:DifferenceBoundCauchyStability-2}),
Theorem \ref{prop:CauchyStability} applied for the past development
of $(r,\text{\textgreek{W}}^{2},\bar{f}_{in},\bar{f}_{out})|_{u=u_{*}}$
on $\mathcal{W}_{u_{*}}$ (see the remark below Theorem \ref{prop:CauchyStability})
implies that the maximal past development $(\mathcal{U}^{*};r^{*},(\text{\textgreek{W}}^{*})^{2},\bar{f}_{in}^{*},\bar{f}_{out}^{*})$
of $(r_{/}^{*},(\text{\textgreek{W}}_{/}^{*})^{2},\bar{f}_{in/}^{*},\bar{f}_{out/}^{*})$
satsfies 
\begin{equation}
\mathcal{W}_{u_{*}}\subset\mathcal{U}^{*}
\end{equation}
and 
\begin{align}
\sup_{\mathcal{W}_{u_{*}}}\Bigg\{\Big|\log\big(\frac{\text{\textgreek{W}}^{2}}{1-\frac{1}{3}\Lambda r^{2}}\big)-\log\big(\frac{(\text{\textgreek{W}}^{*})^{2}}{1-\frac{1}{3}\Lambda(r^{*})^{2}}\big)\Big|+\Big|\log\Big(\frac{2\partial_{v}r}{1-\frac{2m}{r}}\Big)-\log\Big(\frac{2\partial_{v}r^{*}}{1-\frac{2m^{*}}{r^{*}}}\Big)\Big|+\label{eq:UpperBoundNonTrappingForCauchyStability-1-1}\\
+\Big|\log\Big(\frac{1-\frac{2m}{r}}{1-\frac{1}{3}\Lambda r^{2}}\Big)-\log\Big(\frac{1-\frac{2m^{*}}{r^{*}}}{1-\frac{1}{3}\Lambda(r^{*})^{2}}\Big)\Big|+\sqrt{-\Lambda}|\tilde{m}-\tilde{m}^{*}|\Bigg\}+\nonumber \\
+\sup_{\bar{u}}\int_{\{u=\bar{u}\}\cap\mathcal{W}_{u_{*}}}\big|rT_{vv}-r^{*}(T_{vv})^{*}\big|\, dv+\sup_{\bar{v}}\int_{\{v=\bar{v}\}\cap\mathcal{W}_{u_{*}}}\big|rT_{uu}-r^{*}(T_{uu})^{*}\big|\, du & \le(h_{1}(\text{\textgreek{e}}))^{3}.\nonumber 
\end{align}
Thus, the proof of the lemma concludes by setting 
\begin{equation}
(r_{/}^{\prime},(\text{\textgreek{W}}_{/}^{\prime})^{2},\bar{f}_{in/}^{\prime},\bar{f}_{out/}^{\prime})\doteq(r^{*},(\text{\textgreek{W}}^{*})^{2},\bar{f}_{in}^{*},\bar{f}_{out}^{*})|_{u=0}.
\end{equation}

\end{proof}
\bibliographystyle{plain}
\bibliography{DatabaseExample}

\end{document}